%% file: thesis.tex
\documentclass[12pt,a4paper,twoside,BCOR=1cm]{scrreprt}

\usepackage{amsmath, amsthm, amssymb}
\usepackage{mathtools}
\usepackage{grffile}
\usepackage[font=small,labelfont=bf]{caption}
\usepackage{subfig}
\usepackage{url}
\usepackage{graphicx}
\usepackage{float}
\usepackage{newapave}
\usepackage{xcolor}
\usepackage{hyperref} 
\usepackage{setspace}
\usepackage{units}
\usepackage[stable,perpage,symbol*,multiple]{footmisc}
\usepackage[tracking=true,kerning=true,spacing=true,babel]{microtype}
\usepackage[utf8]{inputenc}
\usepackage[final]{pdfpages}

\hypersetup{
 colorlinks=true,
 citecolor=teal,
 linkcolor=blue,
 urlcolor=cyan,
 hyperfigures=true}

\graphicspath{%
{figs/illustrations/}%
{figs/input/}%
{figs/results/}%
{figs/presentation/}%
}

\newtheorem{thm}{Theorem}[chapter]
\newtheorem{cor}[thm]{Corollary}
\newtheorem{lem}[thm]{Lemma}
\newtheorem{defn}{Definition}[chapter]

\DeclareMathOperator*{\vect}{\mathrm{vec}}
\renewcommand\Re{\operatorname{Re}}
\renewcommand\Im{\operatorname{Im}}

\begin{document}

% \maketitle
\include{front/cover}
\pagenumbering{roman}
\include{front/front}
\include{front/acknowledgments}

\tableofcontents
\listoffigures
%\listoftables

\begin{abstract}
  In this work we consider the problem of reconstruction of a signal
  from the magnitude of its Fourier transform, also known as phase
  retrieval. The problem arises in many areas of astronomy,
  crystallography, optics, and coherent diffraction imaging (CDI). Our
  main goal is to develop an efficient reconstruction method based on
  continuous optimization techniques. Unlike current reconstruction
  methods, which are based on alternating projections, our approach
  leads to a much faster and more robust method. However, all previous
  attempts to employ continuous optimization methods, such as
  Newton-type algorithms, to the phase retrieval problem failed. In this
  work we provide an explanation for this failure, and based on this
  explanation we devise a sufficient condition that allows development
  of new reconstruction methods---approximately known Fourier phase. We
  demonstrate that a rough (up to $\pi/2$ radians) Fourier phase
  estimate practically guarantees successful reconstruction by any
  reasonable method. We also present a new reconstruction method whose
  reconstruction time is orders of magnitude faster than that of the
  current method-of-choice in phase retrieval---Hybrid Input-Output
  (HIO). Moreover, our method is capable of successful reconstruction
  even in the situations where HIO is known to fail. We also extended
  our method to other applications: Fourier domain holography, and
  interferometry.

  Additionally we\footnote{The work on sub-wavelength CDI was done in
    collaboration with Prof. M.~Segev's group from the Technion
    Physics Department, Solid State Institute.} developed a new
  sparsity-based method for sub-wavelength CDI. Using this method we
  demonstrated experimental resolution exceeding several times the
  physical limit imposed by the diffraction light properties (so
  called diffraction limit).
\end{abstract}

\cleardoublepage{}
\pagenumbering{arabic}

%\begin{doublespacing}

\include{intro/intro}

%\include{physics/physics}

\include{mathematics/mathematics}

\include{current_methods/current_methods}

\include{optimization/optimization}

\include{phase-empirical/phase-empirical}

\include{phase-theory/phase-theory}

\include{phase-holography/phase-holography}

\include{boundary/boundary}

\include{sparsity/sparsity}

\include{afterword/afterword}

\bibliographystyle{newapave}
\bibliography{My_Library}

%\end{doublespacing}
%\includepdf[pages=-]{hebrew_reversed.pdf}
\end{document}

%% file: front/cover.tex
\thispagestyle{empty}

%----------------------Cover page-------------------------------%
\begin{center}
  \vspace*{0.25\textheight{}}
  \Huge {
    \textbf{Numerical methods for phase retrieval}
  }

  \vspace*{0.3\textheight{}}
  \LARGE{\textbf{Eliyahu Osherovich}}
\end{center}
\newpage
\thispagestyle{empty}
\mbox{}
%%% Local Variables: 
%%% mode: latex
%%% TeX-master: "../thesis"
%%% End: 

%% file: front/front.tex
\cleardoublepage{}

\thispagestyle{empty}
\begin{doublespace}
  \begin{center}
    \Huge
    \textbf {Numerical methods for phase retrieval}
  \end{center}
  \vfill{}
  \begin{center}
    \Large{Research Thesis}
  \end{center}
  \vfill{}
  \begin{center}
    \large{
      Submitted in Partial Fulfillment of the
      
      Requirement for the degree of
      
      Doctor of Philosophy} 
  \end{center}
  \vfill{}
  \begin{center}
    \LARGE\textbf{Eliyahu Osherovich}
  \end{center}
  \vfill{}
  \begin{center}  
    \mbox{\normalsize%
      Submitted to the Senate of the Technion --- %
      Israel Institute of Technology}
 
    Kislev 5772\hfill{}Haifa\hfill{}December 2011
  \end{center}
\end{doublespace}
\newpage
\thispagestyle{empty}
\mbox{}
%%% Local Variables: 
%%% mode: latex
%%% TeX-master: "../thesis"
%%% End: 

%% file: front/acknowledgments.tex
\thispagestyle{empty}
\begin{doublespace}
  \begin{center}
    The Research Thesis Was Done Under The Supervision of
    Prof. Irad Yavneh \\
    and Dr. Micheal Zibulevsky in the Department of Computer Science
  \end{center}
  % -------------- Publications list ---------------------------------%
  % \vspace{2cm}
  % \begin{singlespace}
  %   Parts of this thesis were published in
  %   \begin{enumerate}
  %   \item Osherovich, E., Zibulevsky, M., and Yavneh, I.
  %     Fast reconstruction method for diffraction imaging.
  %     In \textit{Advances in Visual Computing}, volume 5876 of
  %     \textit{Lecture Notes in Computer Science}
  %     pp. 1063--1072. Springer, 2009.
  %   \item Osherovich, E., Zibulevsky, M., and Yavneh, I. Approximate
  %     Fourier phase information in the phase retrieval problem: what it
  %     gives and how to use it. \textit{Journal of the Optical Society of
  %       America A}, 28(10), pp. 2124–2131, 2011.
  %   \end{enumerate}
  % \end{singlespace}
  % ------------- Acknowledgments ------------------------------------%
  \vspace{3cm}
  \begin{center}
    {\LARGE Acknowledgment}
  \end{center} 
  \vspace{2\baselineskip} I owe my deepest gratitude to my advisers:
  Prof. Irad Yavneh, and Dr. Michael Zibulevsky---without their help
  this work would not have been possible.

  I would like to thank Profs. Alfred Bruckstein, Avraham
  Sidi, and Marius Ungarish for our
  fruitful discussions and for their good company.

  I very much enjoyed our joint work with Prof. Mordechai Segev's
  group on sparsity-based subwavelength CDI. Special thanks to
  Profs. Yonina Eldar, and Mordechai Segev.

  Last, but not least, I am deeply indebted to my family and
  most of all to my wife Noa for their love and support.
\end{doublespace}
\vfill{}
The generous financial help of THE TECHNION, Ministry of Trade and
Industry MAGNET/IMG4 GRANT, and
ERC ADVANCED GRANT (via Prof. Mordechai Segev) is gratefully acknowledged.
\newpage\normalsize\thispagestyle{empty}

%%% Local Variables: 
%%% mode: latex
%%% TeX-master: "../thesis"
%%% End: 

%% file: intro/intro.tex
\chapter{Introduction}
\label{cha:introduction}

\section{Motivation}
\label{sec:motivation}
Recent development of nanotechnology has resulted in great interest in
imaging techniques suitable for visualization of nano-structures.  One
of the most promising techniques for such high resolution imaging is
Coherent Diffraction Imaging (CDI). In CDI, a highly coherent beam of
X-rays or electrons is incident on a specimen, generating a
diffraction pattern.  Under certain conditions the diffracted
wavefront is approximately equal (within a scale factor) to the
Fourier transform of the specimen.  After being recorded by a CCD
sensor, the diffraction pattern is used to reconstruct the
specimen~\shortcite{sayre52implications,miao99extending,quiney10coherent}.
Effectively, in CDI we replace the objective lens of a typical
microscope with a software algorithm. The advantage in using no lenses
is that the final image is aberration-free and the final resolution is
only diffraction and dose limited, that is, dependent only on the
wavelength, aperture size and exposure time. This process is
illustrated in Figure~\ref{fig:cdi-process}.
\begin{figure}[H]
  \centering
  \subfloat[]{
    \includegraphics[width=0.3\textwidth]{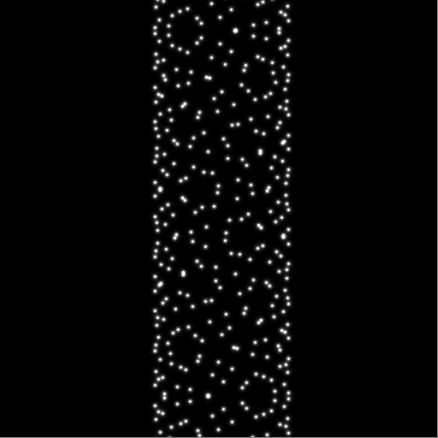}
  }
  \subfloat[]{
    \includegraphics[width=0.3\textwidth]{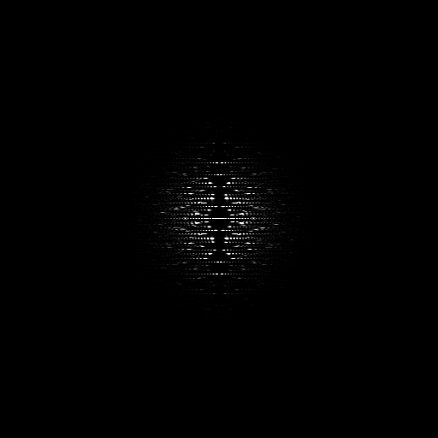}
  }
  \subfloat[]{
    \includegraphics[width=0.3\textwidth]{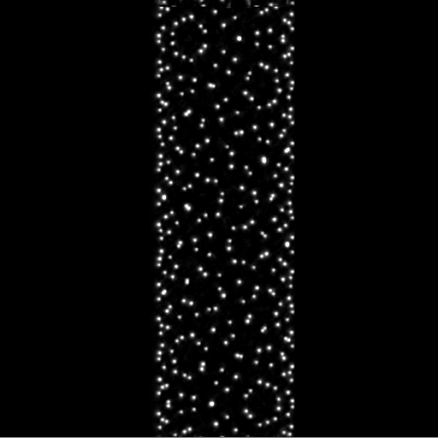}
  }
  \caption[CDI process]{CDI process: a specimen (a) scatters a wave of
    X-rays or electrons that otherwise diffracts freely; the
    diffraction pattern (b) is used for an algorithmic reconstruction
    (c) of the original specimen. Illustrations taken
    from~\shortcite{wikipedia11coherent}.}
  \label{fig:cdi-process}
\end{figure}
The method has been successfully applied to visualizing a variety of
nano-structures, such as carbon nano-tubes~\shortcite{zuo03atomic},
defects inside nano-crystals~\shortcite{pfeifer06three-dimensional},
proteins, and
more~\shortcite{neutze00potential,chapman06high-resolution,chapman07femtosecond}.
Furthermore, exactly the same problem---the reconstruction of a signal
from the magnitude of its Fourier transform---arises in may other
areas of science. Notable examples include astronomy, crystallography,
and speckle interferometry.

It is important to note that, due to the physical nature of the
sensor, we are limited to recording only the intensity (squared
amplitude) of the diffracted wave, hence its phase is lost. As will
be shown later, this loss of the phase leaves us with highly incomplete
data, which makes the problem of reconstruction hard.

\section{Data acquisition model}
\label{sec:data-acqu-model}
Of course, in the real world, the sought object $x(t)$ and its Fourier
transform (denoted by $\hat{x}(\omega)$) are both continuous functions
of $t$ and $w$, respectively, where $t$ and $w$ are multidimensional
coordinate vectors. Furthermore, the support of $x(t)$ is limited,
which means that $\hat{x}(\omega)$ is spread over all frequencies:
from $-\infty$ to $\infty$.  However, during the data acquisition
process we capture only a finite extent in the Fourier domain and all
further processing is done on digital computers. This naturally leads
to discrete approximations of $x(t)$ and $\hat{x}(\omega)$, that are
well justified in  view of the finite resolution that stems from
the measurements and, in general, from the fact that all optical
systems have resolution limits. Given that $x[n]$ (an adequate
sampling of $x(t)$) contains $N$ points (for simplicity we assume that
$x$ is one-dimensional---generalization for the multi-dimensional case
is straightforward) we assume that $x[n]$ vanishes outside the interval
$[0,N-1]$. Furthermore, if we assume, without loss of generality, that
the physical extent of $x$ is unity we immediately conclude that the
sampling rate in the Fourier domain must be $1/N$ to acquire the
measurements that are related to $x[n]$ via the Discrete Fourier
Transform (DFT). However, a more thorough examination of the problem
yields a higher sampling rate requirement. Recall that we record only
the intensity of the diffraction pattern. This intensity can be
represented as follows:
\begin{equation}
  \label{eq:intro-1}
  I(\omega) =
  |\hat{x}(\omega)|^{2}=\bar{\hat{x}}(\omega)\circ\hat{x}(\omega)\,, 
\end{equation}
where the overbar denotes the complex conjugate and $\circ$ denotes
element-wise multiplication. Hence, the inverse Fourier transform of
the measured intensity $I(\omega)$ results in the auto-correlation
function (denoted by $\star$) of the sought object,
\begin{equation}
  \label{eq:intro-2}
  \mathcal{F}^{-1}[I(\omega)] = x(t)\star x(t) \,. 
\end{equation}
Obviously, the autocorrelation $x(t)\star x(t)$ has support that is
twice as large as the support of $x(t)$ (in each dimension),
therefore, the diffraction pattern intensity must be sampled with the
rate two times higher than $1/N$ to capture all the information about
the auto-correlation function. To this end, we always assume that the
signal $x(t)$ (or $x[n]$) is ``padded'' with zeros so that its
size is doubled (in each dimension). This requirement is an
approximation to the physical constraint on $x(t)$ having finite
support. Without adding it into the reconstruction scheme, the problem
would be severely undetermined with multiple solutions that are
unrelated to the original signal $x$. To illustrate the last claim, one
can imagine the case where $x[n]$ is reconstructed from its Fourier
magnitude without additional constraints. Obviously, $\emph{any}$
choice of the Fourier phase will give rise to a valid solution which,
unfortunately,  has little to do  with the sought signal. 

\section{Reconstruction from incomplete Fourier data}
\label{sec:reconstr-from-incomp}
Before we proceed to the main subject of this work---the
reconstruction of a signal from the magnitude of its Fourier
transform---let us consider a number of toy problems, where we
evaluate the importance of different parts of the Fourier transform.

The Fourier transform is, in general, complex. There are two common
representations of a complex number: one is the sum of its real
($\Re$) and
imaginary ($\Im$) parts
\begin{equation}
  \label{eq:intro-3}
  z = \Re + j\Im \,, 
\end{equation}
and the other is the product of its magnitude ($r$) and the complex
exponent of its phase $e^{j\phi}$, 
\begin{equation}
  \label{eq:intro-4}
  z = re^{j\phi} \,. 
\end{equation}
From these formulas it is not clear whether one part of the
representation is more important
than the other. Below we demonstrate that the real and the
imaginary parts carry about equal amount of information and the loss
of one of them can often be recovered. However, the  phase carries
most of the information, and consequently, its loss is more difficult to
overcome. 

\subsection{Reconstruction from the real part}
\label{sec:reconstr-from-real}
Let us assume that $x[n]$ is a real one-dimensional signal\footnote{A
  generalization to a complex multidimensional  $x$ is
  straightforward.}. Furthermore, we assume that $x$ vanishes outside
the interval $[0,N-1]$, specifically, we assume that
$x[n]=0$ for $n = -1, -2, \ldots, -(N-1)$. Recall that any real signal
can be represented uniquely as a sum of two signals
\begin{equation}
  \label{eq:intro-5}
  x = x_{e}+x_{o} \,, 
\end{equation}
where $x_{e}$ is even and $x_{o}$ is odd\footnote{In the complex case,
  $x_{e}$ is Hermitian, and $x_{o}$ is anti-Hermitian.}
(see Figure~\ref{fig:even-odd}).  It can be easily
shown that
\begin{equation}
  \label{eq:intro-6}
  x_{e} = \frac{x[n] + x[-n]}{2}\,, \quad x_{o}=\frac{x[n] - x[-n]}{2}
  \,. 
\end{equation}
Recall also that the Fourier transform of an even signal is real and
that of an odd signal is purely imaginary. Hence, we conclude that the
real part of $\hat{x}$ is nothing but the Fourier transform of
$x_{e}$. Thus, we can obtain the even part $x_{e}$ by the inverse
Fourier transform of the real part of $\hat{x}$.  Furthermore,
reconstructing $x$ from $x_{e}$ is trivial---we should take the right-hand
side ($n\geq0$) of $x_{e}$ multiplied by two everywhere except the
origin (see Figure~\ref{fig:even-odd}).

\begin{figure}[H]
  \centering
  \subfloat[]{
    \includegraphics[width=0.8\textwidth{}]{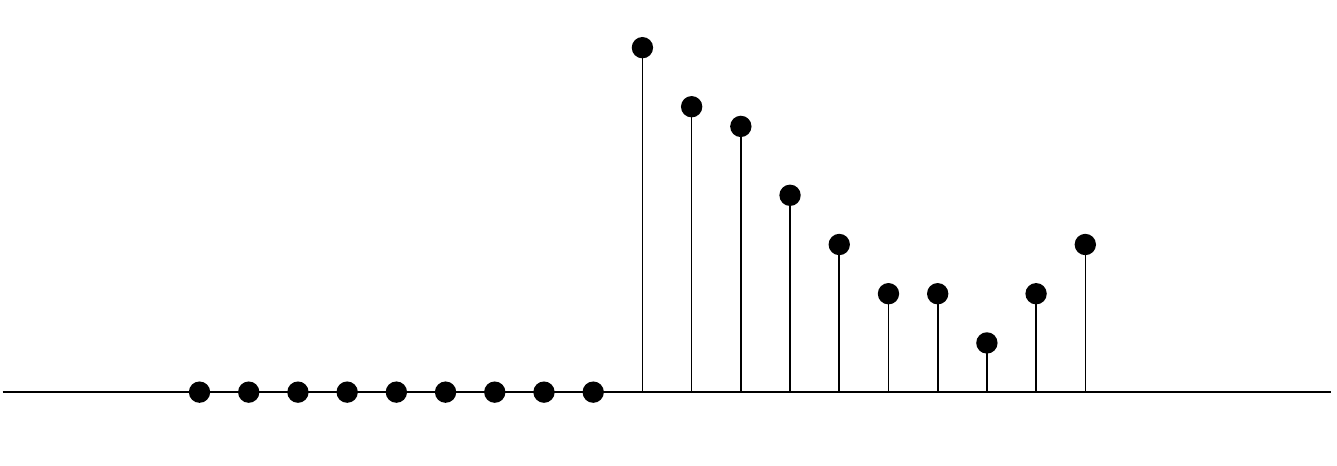}
  }\\
  \subfloat[]{
    \includegraphics[width=0.8\textwidth{}]{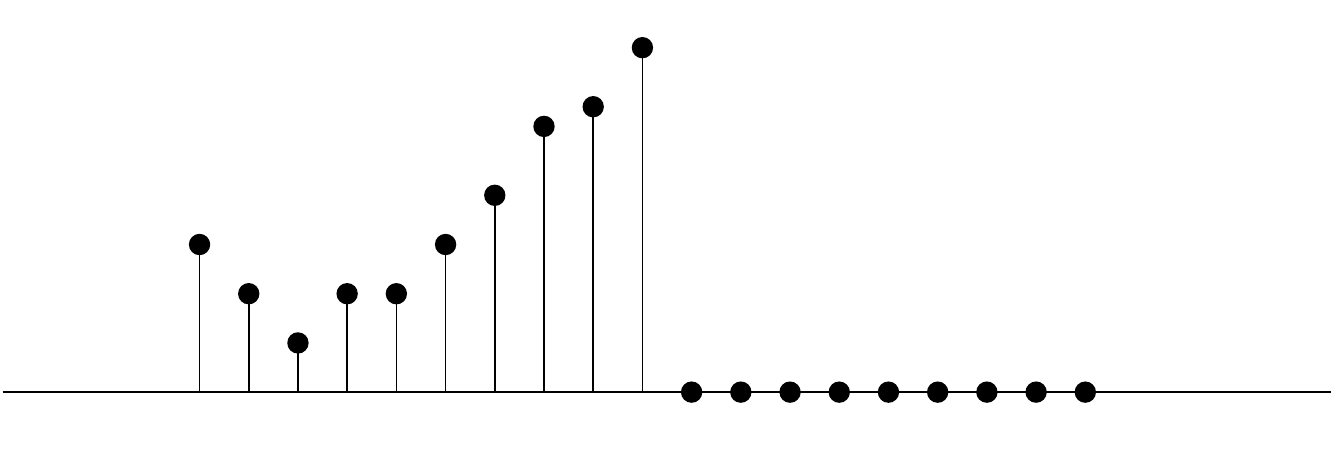}
  }\\
  \subfloat[]{
    \includegraphics[width=0.8\textwidth{}]{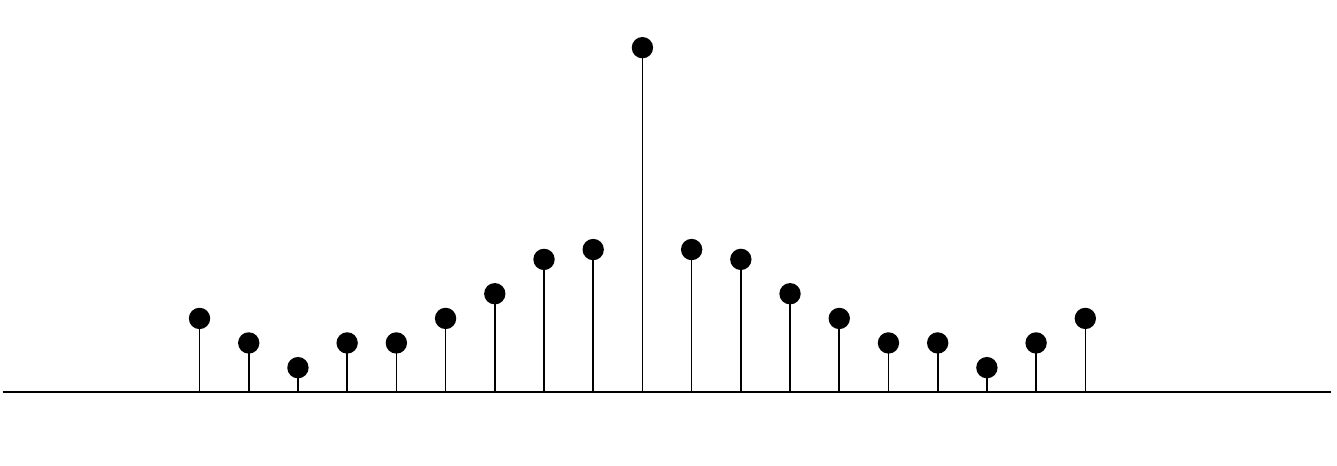}
  }\\
  \subfloat[]{
    \includegraphics[width=0.8\textwidth{}]{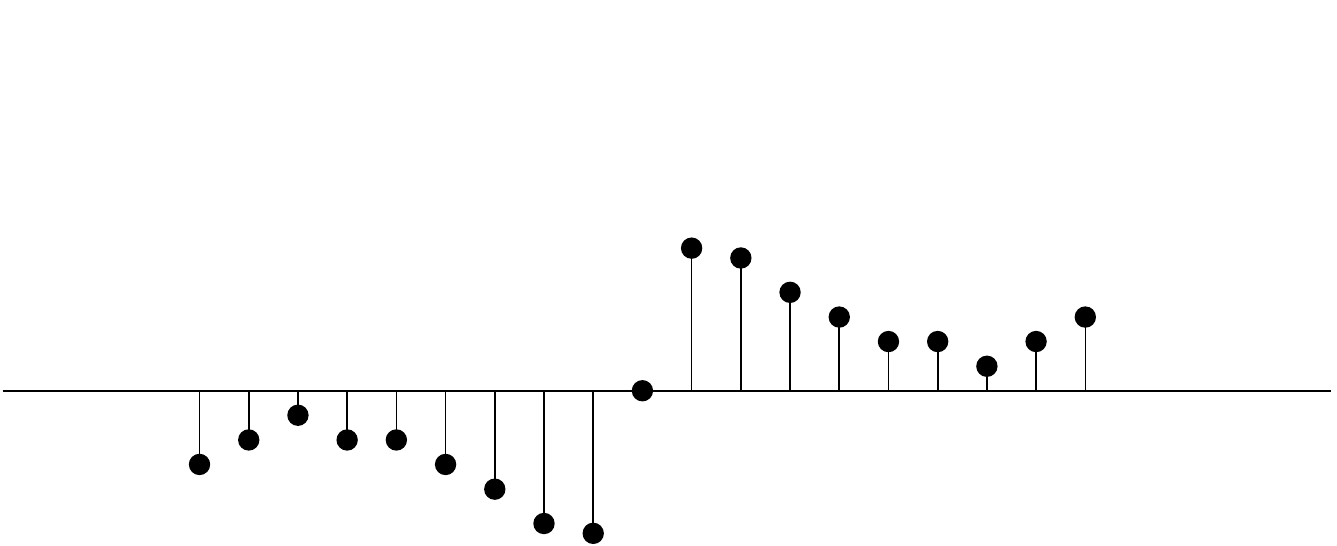}
  }
  \caption[Signal decomposition into even and odd parts]{Signal decomposition into even and odd parts: (a) original
    signal $x[n]$, (b) its reversed version $x[-n]$, (c) the even
    part, and (d) the odd part.}
  \label{fig:even-odd}
\end{figure}

\subsection{Reconstruction from the imaginary part}
\label{sec:reconstr-from-imag}
The reconstruction from the imaginary part is also easy. The method is
very similar to the reconstruction from the real part. The only
difference is that now we obtain the odd part of the sought
signal. But, again, taking the right-hand side of $x_{o}$ and multiplying
it by two we obtain the original signal $x$ everywhere besides the
origin.

For the most general case, where $x$ is complex, it is easy to show
that the missing imaginary part leads to almost perfect
reconstruction---only $\Im(x[0])$ is lost. Similarly, when
reconstructing from the imaginary part of $\hat{x}$---only $\Re(x[n])$
is lost. Therefore, we can conclude that the real and the imaginary
parts of the Fourier transform carry the same amount of information and
losing either one of them can be easily overcome if $x[n]$ is
sufficiently padded with zeros and vanishes for $n<0$.

\subsection{Reconstruction from the phase}
\label{sec:reconstr-from-phase}
Now we switch to the second (polar) representation of complex numbers and
consider first the situation where the magnitude is lost, namely,
reconstruction from the Fourier phase. Several researchers
(see~\shortcite{hayes82reconstruction,hayes80signal,millane96multidimensional,millane90phase,oppenheim81importance})
showed that sufficiently padded signals can be reconstructed from the
Fourier phase. However, the reconstruction is possible only up to
a scale factor, that is, one obtains  $\alpha x$ for some
real positive scalar $\alpha$. Furthermore, is was shown that this is
the only type of non-uniqueness possible in signal reconstruction from
the Fourier phase~\shortcite{hayes80signal}.

It seems that all previous works used a variation of the method
of alternating projections (see
Section~\ref{sec:gerchb-saxt-meth}). However, the problem can easily
be represented as linear programming, for which much more efficient
algorithms exist. Actually, we solve a variation of this problem, where
the Fourier phase is known to lie within  certain interval in
Chapters~\ref{cha:appr-four-phase-first} and
\ref{cha:appr-four-phase-explanation}, and our method is much faster
than the method of alternating projections. 

\subsection{Reconstruction from the magnitude}
\label{sec:reconstr-from-magn}
We finally arrive at the case that is the main subject of this
research---reconstruction from the Fourier magnitude alone. It turns
out that this case of incomplete Fourier information is the most
difficult among the four possibilities considered here. This
phenomenon can be related to the fact that the Fourier phase carries
the majority of the information in a signal. As an informal ``proof''
of the latter claim, look at Figure~\ref{fig:phase-importance}, where
we exchange the Fourier phase between two different images, while
keeping the Fourier magnitude intact. The results show, unmistakably,
an exchange of the images.

It turns out that this problem is very different in the
one-dimensional case and in the multi-dimensional case. The former
suffers from multiplicity of possible solutions (no matter how much
padding we use in $x$), while the latter is usually less prone to
multiple solutions (aside from trivial transformations: lateral
shifts, axis reversal, and constant phase factor), but finding a
solution is very difficult. An explanation for this phenomenon is
given in Chapter~\ref{cha:math-found}.

\begin{figure}[H]
  \centering
  \subfloat[]{
    \includegraphics[height=0.4\textheight{}]{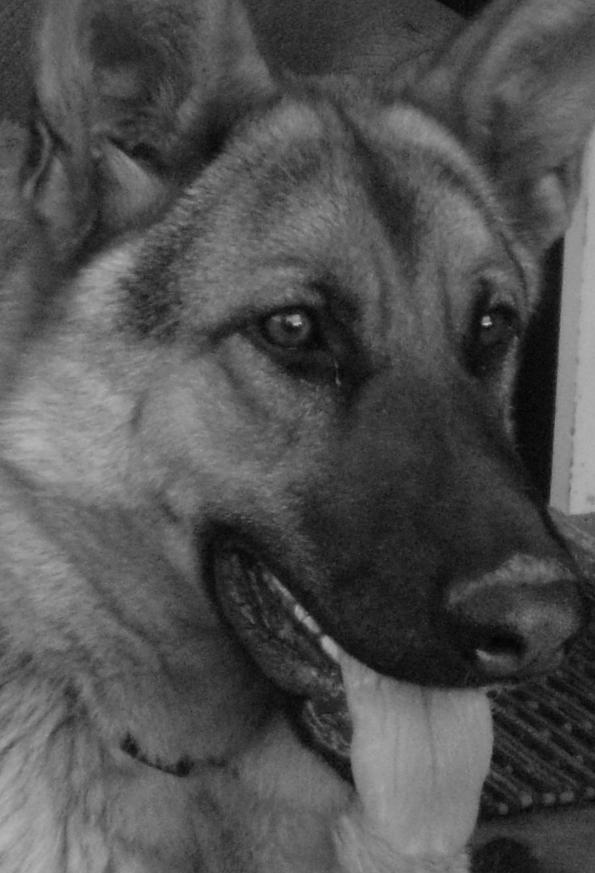}
  }\qquad{}
  \subfloat[]{
    \includegraphics[height=0.4\textheight{}]{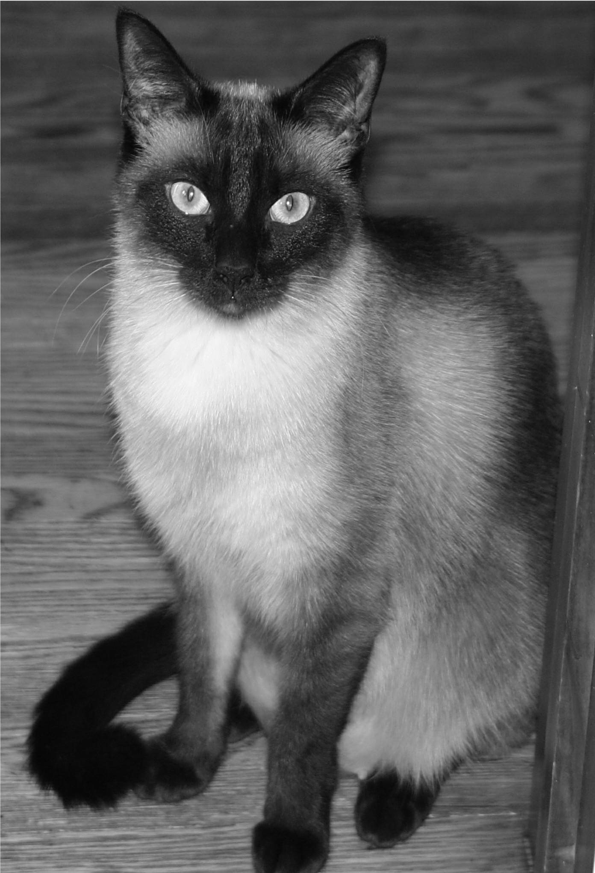}
  }\\
  \subfloat[]{
    \includegraphics[height=0.4\textheight{}]{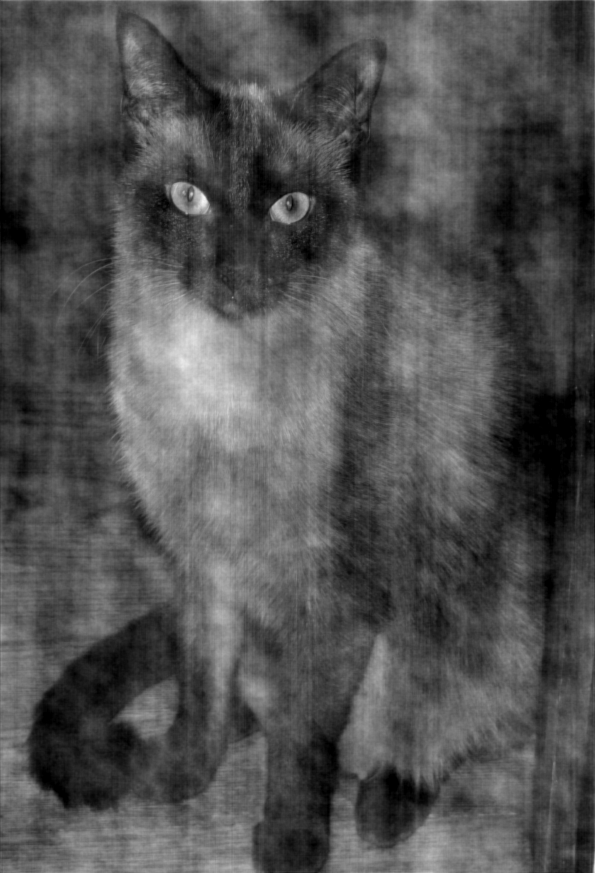}
  }\qquad{}
  \subfloat[]{
    \includegraphics[height=0.4\textheight{}]{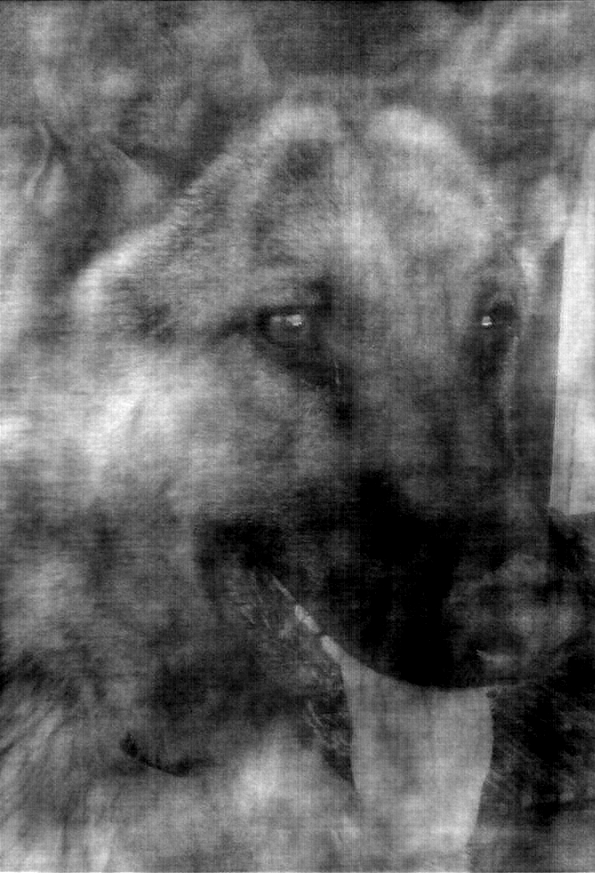}
  }
  \caption[The importnace of phase in signals]{The importnace of phase in signals---given two images (a) a
    dog, and (b) a cat, we exchange their Fourier phase while keeping
    the magnitude. The results, unmistakably, show the cat instead of
    the dog (c), and the dog instead of the cat (d).}
  \label{fig:phase-importance}
\end{figure}

\section{A short overview}
\label{sec:short-overview}
After this short introduction we are ready to proceed to the main
part of this work. In the next chapter we present the mathematical
foundations of the phase retrieval problem: our main concern is
uniqueness of the reconstruction, moreover, we demonstrate why
the one-dimensional and the multi-dimensional cases behave extremely
differently. In Chapter~\ref{cha:curr-reconstr-meth} we present the
methods that are used today for phase retrieval. These chapters are
based on a compilation of known facts and do not contain our original
research. Our development starts in
Chapter~\ref{cha:found-optim-meth}, where we introduce machinery for
efficient optimization of a real function of complex
argument. Moreover, we find and analyze the eigen-decomposition of the
Hessian of one of the most frequently used objective function in phase
retrieval. Chapters~\ref{cha:appr-four-phase-first} and
\ref{cha:appr-four-phase-explanation} are dedicated to a variant of
the phase retrieval problem where the Fourier phase is not lost
completely, but rather a rough estimate of it is available. We demonstrate
empirically and then prove mathematically that this scenario leads to
a new family of algorithms that are orders of magnitude faster than
the existing methods. The ideas from these chapters are taken into the
new field of the Fourier domain holography in
Chapters~\ref{cha:phase-retr-holography}, and
\ref{cha:design-bound-phase} where we develop a new method of
reconstruction and present guide rules for a reference beam design
that lead to fast and robust reconstruction. A new type of prior
knowledge---image sparsity---is exploited in
Chapter~\ref{cha:bandw-extr-using} where we present a CDI with
resolution several times higher than the maximal theoretical
limit. Finally, in Chapter~\ref{cha:afterword} we present concluding
remarks and provide reference to the works that was done in the course
of this Ph.D. research but have not been included in this thesis. For
example, we used Fienup's Hybrid Input-Output algorithm (see
Section~\ref{sec:fienup-algor-phase}) for creation of Grassmanian
matrices~\shortcite{osherovich09designing}---an interesting work
that was left outside because it is not related to the phase retrieval
problem.

%%% Local Variables: 
%%% mode: latex
%%% TeX-master: "../thesis"
%%% End: 

%% file: mathematics/mathematics.tex
\chapter{Mathematical foundations}
\label{cha:math-found}

In this chapter we review some mathematical properties of the phase
retrieval problem. Our main concern is the uniqueness properties of the solution,
which turns to be different in the one-dimensional and
multi-dimensional cases.

\section{One-dimensional signals}
\label{sec:one-dimensional-signals}

\subsection{Continuous-time case}
\label{sec:continuous-case}

What can be said about two signals $g_{1}(t)$ and $g_{2}(t)$ having
the same power spectrum? Obviously,
\begin{equation}
  \label{eq:28}
  |\hat{g}_{1}(\omega)|^{2} = |\hat{g}_{2}(\omega)|^{2}
\end{equation}
means that
\begin{equation}
  \label{eq:29}
  |\hat{g}_{1}(\omega)| = |\hat{g}_{2}(\omega)|\,.
\end{equation}
This, in turn, leads to the following relation
\begin{equation}
  \label{eq:30}
  \hat{g}_{2}(\omega) = \hat{g}_{1}(\omega)e^{j\phi(\omega)}\,.
\end{equation}
for some real-valued function $\phi(\omega)$. Hence, if $g_{2}(t)$ has
the same power spectrum as $g_{1}(t)$, then $g_{2}(t)$ must be a
result of passing $g_{1}(t)$ through an all-pass filter. Moreover, it
means that there are infinitely many (an uncountable set of) functions
having the same power spectrum, because any choice of $\phi(\omega)$ will
lead to some function
\begin{equation}
  \label{eq:31}
  g_{2}(t) =
  \frac{1}{2\pi}\int_{-\infty}^{\infty}\hat{g}_{1}(\omega)e^{j\phi(\omega)}
  e^{j\omega t}\, \mathrm{d} \omega \ .
\end{equation}
This observation, however, is not particularly helpful as we are not
interested in \textit{arbitrary} signals. Our interest is limited to
signals that are physically feasible. Such signals, for example,
must have finite support and finite energy. That is, we assume that
\begin{equation}
  \label{eq:7}
  g(t) = 0, \quad \text{for } |t|> \frac{T}{2}\,,
\end{equation}
and
\begin{equation}
  \label{eq:10}
  \int^{T/2}_{-T/2}|g(t)|^{2}\,\mathrm{d}(t) < \infty \,. 
\end{equation}
Under these restrictions it is not clear anymore that there still exist
infinitely many signals having the same power spectrum.\footnote{Here
  we do not count the trivial multiplication by a phase factor.} Nevertheless,
it was shown that even under these conditions the number of signals
having the same power spectrum can be infinite. A fairly complete
treatment of this problem is given
by~\shortciteA{hofstetter64construction} who considered the problem from a
different point of view: finding of all possible time-limited signals
having the same autocorrelation function as a given signal. However,
it is exactly the same problem due to the direct connection between
the autocorrelation and the power spectrum.  Although Hofstetter's
work was done for continuous-time (analogue) signals and not for
discrete signals that we encounter in computer algorithms, we believe it is
more instructive to start with this case and proceed then to the
discrete version of the problem. Hence, we present next the main
results and derivations found in~\shortcite{hofstetter64construction}.

Under the conditions~\eqref{eq:7} and~\eqref{eq:10}, $g(t)$ is
absolutely integrable
\begin{equation}
  \label{eq:11}
  \int^{T/2}_{-T/2}|g(t)|\,\mathrm{d}(t) < \infty \ .
\end{equation}
and its Laplace transform,
\begin{equation}
  \label{eq:12}
  G(s) \equiv \mathcal{L}[g(t)] =
  \int^{T/2}_{-T/2}g(t)e^{-st}\,\mathrm{d}(t)\,, 
\end{equation}
converges for all complex $s = \sigma + j\omega$. The function $G(s)$
is thus analytic in the entire $s$ plane and, because of
Equation~\eqref{eq:10}, square integrable. The signal $g(t)$ can
always be recovered from $G(s)$ by means of the inverse
transform% \footnote{Mathematically, Equation~\eqref{eq:13} is
  % interpreted in the sense of convergence almost everywhere. However,
  % for signals that represent physical quantities, the equality can be
  % assumed in a full and unconditional fashion.}
\begin{equation}
  \label{eq:13}
  g(t) = \frac{1}{2\pi}\int_{-\infty}^{\infty}G(j\omega)e^{j\omega
    t}\, \mathrm{d}\omega \,. 
\end{equation}
From the definition of the Laplace transform, it is obvious that it is
a generalization of the Fourier transform, and the relation between
the two is very simple:
\begin{equation}
  \label{eq:95}
  \hat{g}(\omega) = \mathcal{F}[g(t)] =
  \mathcal{L}[g(t)]\Big\vert_{s=j\omega}=G(j\omega)\,. 
\end{equation}

With these preliminaries in hand we proceed to the autocorrelation
function of $g(t)$ defined as
\begin{equation}
  \label{eq:14}
  r(\tau) \equiv g(t)\star g(t) =
  \int^{T/2}_{-T/2}\bar{g}(t)g(t+\tau)\,\mathrm{d}\tau\ .
\end{equation}
It can be readily shown that $r(\tau)$ also has finite support: it
vanishes outside the interval $[-T, T]$. Moreover, similarly to
$g(t)$, $r(\tau)$ is square and absolutely integrable. Therefore, its
Laplace transform,
\begin{equation}
  \label{eq:15}
  R(s) = \int^{T}_{-T}r(\tau) e^{-s\tau}\, \mathrm{d}\tau \,, 
\end{equation}
exists, and  is analytic in the entire $s$ plane. Of course, it and can be used to
recover $r(\tau)$ by the means of the inverse
transform
\begin{equation}
  \label{eq:16}
  r(\tau) =
  \frac{1}{2\pi}\int^{\infty}_{-\infty}R(j\omega)e^{j\omega\tau}\,
  \mathrm{d}\omega\ .
\end{equation}
Since our main goal is to design signals with the same autocorrelation
function, we need to find the relation between the Laplace transform of
a signal and the Laplace transform of its autocorrelation
function. The development is straightforward
\begin{align}
  \label{eq:19}
  R(s)
  & = \int^{\infty}_{-\infty} (g(t)\star g(t)) e^{-s\tau}\,\mathrm{d}\tau \\
  & = \int^{\infty}_{-\infty}
  \left[
    \int^{\infty}_{-\infty} \bar{g}(t) g(t+\tau)\, \mathrm{d}t
  \right]  e^{-s\tau} \,\mathrm{d}\tau\\
  & =  \int^{\infty}_{-\infty} \bar{g}(t) e^{st}\, \mathrm{d} t
  \int^{\infty}_{-\infty}g(t+\tau) e^{-s(t+\tau)}\, \mathrm{d} \tau\label{eq:21}\\
  & =  \overline{G(-\bar{s})}G(s)\label{eq:22} \ .
\end{align}
In the transition from \eqref{eq:21} to \eqref{eq:22} we used the fact
that
\begin{equation}
  \label{eq:23}
  \mathcal{L}[\bar{g}(t)] = \overline{G(\bar{s})}\ .
\end{equation}
Hence, if two signals $g_{1}(t)$, and $g_{2}(t)$ have the same
autocorrelation function they must obey the following equality for all $s$:
\begin{equation}
  \label{eq:18}
  \overline{G_{1}(-\bar{s})}G_{1}(s) = \overline{G_{2}(-\bar{s})}G_{2}(s) \ .
\end{equation}
This includes  $s = j\omega$, which gives the equality between their
power spectra
\begin{equation}
  \label{eq:24}
  \begin{split}
    \overline{G_{1}(-\overline{(j\omega)})}G_{1}(j\omega) & =
    \overline{G_{2}(-\overline{(j\omega)})}G_{2}(j\omega) \\
    \overline{G_{1}(j\omega)}G_{1}(j\omega) & =
    \overline{G_{2}(j\omega)}G_{2}(j\omega)\\
    |G_{1}(j\omega)|^{2} & =  |G_{2}(j\omega)|^{2} \ ,
  \end{split}
\end{equation}
as expected.

The necessity to introduce the Laplace transform will become clear
after we present a family of all-pass filters that preserve the
support bounds of $g(t)$.  It is shown
in~\shortcite{hofstetter64construction} that if $s_{0}$ is a zero of
$G_{1}(s)$, that is, if 
\begin{equation}
  \label{eq:17}
  G_{1}(s_{0}) = \int_{-\infty}^{\infty}g_{1}(t)e^{-s_{0}t}\,
  \mathrm{d}t = 0 \ ,
\end{equation}
then the all-pass filter $h(t)$ whose transfer function $H(s)$ is
given by
\begin{equation}
  \label{eq:20}
  H(s) = \frac{s+\bar{s}_{0}}{s - s_{0}}
\end{equation}
will not spread a given signal $g_{1}(t)$ outside the original
interval $[-T/2, T/2]$. It is easy to verify that
\begin{equation}
  \label{eq:32}
  |H(j\omega)|^{2} = H(j\omega)\overline{H(j\omega)} =
  \left(
    \frac{j\omega+\bar s_{0}}{j\omega - s_{0}}
  \right)
  \left(
    \frac{-j\omega + s_{0}}{-j\omega - \bar s_{0}}
  \right)  = 1 \ .
\end{equation}
Hence, $h(t)$ is indeed an all-pass filter because $|H(j\omega)| = 1$. To
demonstrate that $h(t)$ preserves the support, let us find its
explicit representation. To this end, it is more convenient to split
it into the following two filters
\begin{equation}
  \label{eq:33}
  H(s) = 1 + \frac{s_{0} + \bar s_{0}}{s - s_{0}}.
\end{equation}
In this form, it is easy to see that the impulse response $h(t)$ that
corresponds to $H(s)$ is given by
\begin{align}
  \label{eq:34}
  h(t)
  & = \frac{1}{2\pi}\int_{-\infty}^{\infty}
  \left(
    1 + \frac{s_{0} + \bar s_{0}}{j\omega - s_{0}}
  \right) e^{j\omega t}\, \mathrm{d} \omega\\
  & = \delta(t) + \frac{s_{0} + \bar s_{0}}{2\pi}
  e^{s_{0}t}\int_{-\infty}^{\infty}
  \frac{e^{j\omega t - s_{0}t}}{j\omega  - s_{0}}\, \mathrm{d}
  \omega\\
  & = \delta(t) + (s_{0} + \bar s_{0})e^{s_{0}t}U(t)\, ,
\end{align}
where $U(t)$ denotes the Heaviside step function. Hence, if a
time-limited input $g_{1}(t)$ were provided to the system described by
$H(s)$, the output $g_{2}(t)$ would be $g_{1}(t)$ plus the convolution
of $g_{1}(t)$ with the filter $h_{1}(t) = (s_{0} +
\bar s_{0})e^{s_{0}t}U(t)$. To show that the output $g_{2}(t)$ vanishes
outside the interval $[-T/2,T/2]$ we must show that the latter term
(convolution) vanishes outside this interval. The proof is
straightforward:
\begin{align}
  g_{21}(t)
  & = g_{1}(t)\otimes h_{1}(t)\\
  & = \int_{-T/2}^{T/2}g_{1}(\tau)h_{1}(t-\tau)\,\mathrm{d}\tau\\
  & =
  \begin{dcases}
    0 & \text{if $t < - \frac{T}{2}$}\,, \\
    (s_{0}+\bar s_{0}) \int_{-T/2}^{t}
    g_{1}(\tau)e^{s_{0}(t-\tau)}\,\mathrm{d}\tau & \text{if $-\frac{T}{2}\leq t
    \leq \frac{T}{2}$}\,,  \\
    (s_{0}+\bar s_{0})\int_{-T/2}^{T/2}
    g_{1}(\tau)e^{s_{0}(t-\tau)}\,\mathrm{d}\tau &  \text{if $t >
      \frac{T}{2}$}\,, 
  \end{dcases}
\end{align}
where $\otimes$ denotes convolution.
Note that $g_{21}(t)$ vanishes for $t > T/2$, because the last line in
the above equation reads
\begin{equation}
  \begin{split}
    (s_{0}+\bar s_{0})\int_{-T/2}^{T/2}
    g_{1}(\tau)e^{s_{0}(t-\tau)}\,\mathrm{d}\tau
    & =   (s_{0}+\bar s_{0}) e^{s_{0}t}\int_{-T/2}^{T/2}
    g_{1}(\tau)e^{s_{0}\tau}\,\mathrm{d}\tau \\
    & = (s_{0}+\bar s_{0}) e^{s_{0}t} G_{1}(s_{0}) \,, 
  \end{split}
\end{equation}
and $s_{0}$ is a root of $G_{1}(s)$. Hence, the support of $g_{2}$ is
not wider than the support of $g_{1}$,
that is, it vanishes outside the interval $[-T/2, T/2]$. It is important
to emphasize the way $G_{2}(s)$ (equivalently,  $g_{2}(t)$) is
obtained from $G_{1}(s)$ (equivalently,  $g_{1}(t)$):
\begin{equation}
  \label{eq:25}
  G_{2}(s) = G_{1}(s)\frac{s+\bar s_{0}}{s-s_{0}}\,,
\end{equation}
where $s_{0}$ is a root of $G_{1}(s)$. This formula gives an obvious
way to obtain a new time-limited signal $g_{2}(t)$ from a given
time-limited signal $g_{1}(t)$ such that the autocorrelation of the
two is equal. The algorithm is really simple, indeed:
\begin{enumerate}
\item Laplace transform $g_{1}(t)$ to obtain $G_{1}(s)$
\item Choose $k$ non-zero roots of $G_{1}(s)$:
  $\left\{s_{i}\right\}_{i=1}^{k}$ and replace them with their
  negative conjugates, to obtain $G_{2}(s)$
  \begin{equation}
    \label{eq:26}
     G_{2}(s)=G_{1}(s)\prod_{i=1}^{k}\frac{s+\bar s_{i}}{s-s_{i}}
  \end{equation}
\item Inverse Laplace transform $G_{2}(s)$ to obtain $g_{2}(s)$
\end{enumerate}
The fact that the new signal $g_{2}(t)$ vanishes outside the interval
$[-T/2, T/2]$ was proven above. The fact that the two signals have
the same autocorrelation function and, hence, the same power spectrum
is readily obtained by simple calculations:
\begin{equation}
  \label{eq:96}
  \begin{split}
    \mathcal{L}[g_{2}(t)\star g_{2}(t)]
    & = \overline{ G_{2}(-\bar s)}G_{2}(s)\\
    & =
    \overline{\left(
        G_{1}(-\bar s)\prod_{i=1}^{k}\frac{-\bar s+\bar s_{i}}{-\bar s-s_{i}}
      \right)} G_{1}(s)\prod_{i=1}^{k}\frac{s+\bar s_{i}}{s-s_{i}}\\
    & = \overline{G_{1}(-\bar s)}\prod_{i=1}^{k}\frac{-s+s_{i}}{-s-\bar s_{i}}
    G_{1}(s)\prod_{i=1}^{k}\frac{s+\bar s_{i}}{s-s_{i}}\\
    & = \overline{G_{1}(-\bar s)}G_{1}(s)
    \prod_{i=1}^{k}\frac{s-s_{i}}{s+\bar s_{i}}
    \prod_{i=1}^{k}\frac{s+\bar s_{i}}{s-s_{i}} \\
    & =  \overline{G_{1}(-\bar s)}G_{1}(s)\\
    & = \mathcal{L}[g_{1}(t)\star g_{1}(t)] \,. 
  \end{split}
\end{equation}
This, actually, leads us to an algorithm for generating new time-limited 
signals, all having the same autocorrelation function. The algorithm
simply replaces one or more non-zero roots of the Laplace transform of a
given signal with their negative conjugates. Moreover, it can be
proven that this approach can generate \emph{all possible} signals
with provided time support and autocorrelation function. A rigorous
proof of the last statement actually expands the current result to an
infinite set of roots. The proof is not difficult but we do not
present it here. The interested reader can find it in
\shortcite{hofstetter64construction}. For our discussion it is more
important to note that a set of $k$ non-zero roots of $G_{1}(s)$ gives
rise to $2^{k}$ new signals with the same time support and autocorrelation
function. Of course, depending on additional constraints, some of
these signals may not be feasible. For example, if our attention is
restricted to real-valued signals, the zeros of $G_{1}(s)$ must occur in
conjugate pairs. Hence, to generate a new real-valued signal we must
replace the corresponding pair and not a single root.

Below we present a number of examples of this technique (from
\shortcite{hofstetter64construction})
\begin{description}
\item[Example 1]
  Assume that $a < 0$ and let $g_{1}(t)$ be given by
  \begin{equation}
    \label{eq:27}
    g_{1}(t) =
    \begin{dcases}
      0 & \text{if $|t| >1$}\,,\\
      -e^{at} & \text{if $-1\leq t < 0$}\,,\\
      e^{at} & \text{if $0 < t \leq 1$}\,.
    \end{dcases}
  \end{equation}
  Calculating its Laplace transform gives us
  \begin{equation}
    \label{eq:113}
    \begin{split}
      G_{1}(s)
      & = \int_{-1}^{0}-e^{(a-s)t}\,\mathrm{d}t +
      \int_{0}^{1}e^{(a-s)t}\,\mathrm{d}t \\
      & = 2\frac{1-\cosh(s-a)}{s-a}\,.
    \end{split}
  \end{equation}
  Since $G_{1}(a)=0$, we can create a new signal $g_{2}(t)$ by passing
  $g_{1}(t)$ through the all-pass filter given by
  \begin{equation}
    \label{eq:36}
    H(s) = \frac{s+a}{s-a}\,.
  \end{equation}
  The impulse response of this filter is
  \begin{equation}
    \label{eq:37}
    h(t) = \delta(t) + 2aU(t)e^{at}\,.
  \end{equation}
  Hence, the output of the filter is
  \begin{equation}
    \label{eq:38}
    g_{2}(t) =
    \begin{dcases}
      0  & \text{if $|t| > 1$}\,,\\
      -e^{at}(2at+2a+1) & \text{if $-1\leq t < 0$}\,,\\
      e^{at}(2at-2a+1) & \text{if $ 0 < t \leq 1$}
    \end{dcases}
  \end{equation}
  Figure~\ref{fig:math-example1} below depicts the two signals
  $g_{1}(t)$, $g_{2}(t)$, and their common autocorrelation function.
  \begin{figure}[H]
    \centering
    \subfloat[]{
      \includegraphics[width=0.45\textwidth{}]{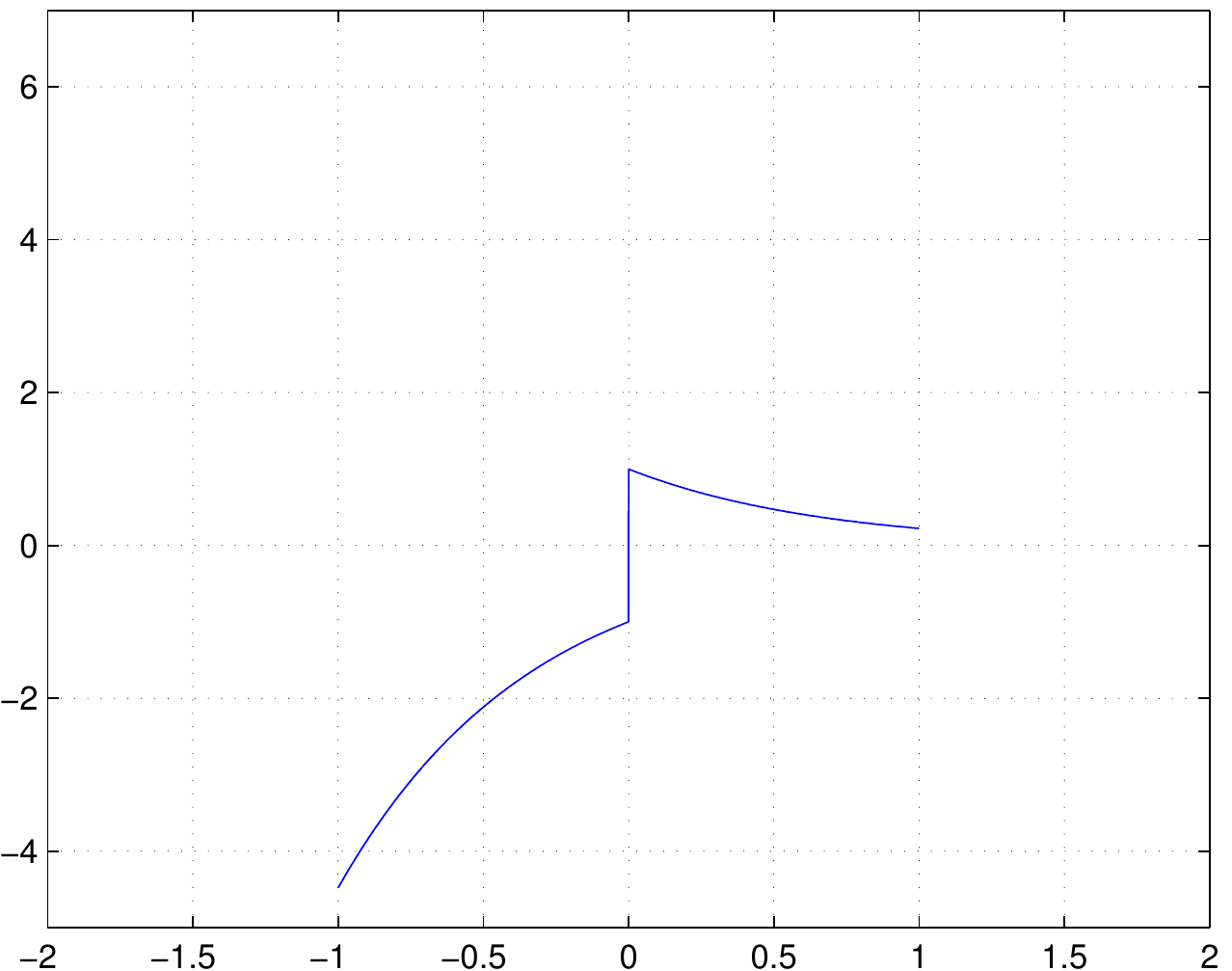}
    }\qquad{}
    \subfloat[]{
      \includegraphics[width=0.45\textwidth{}]{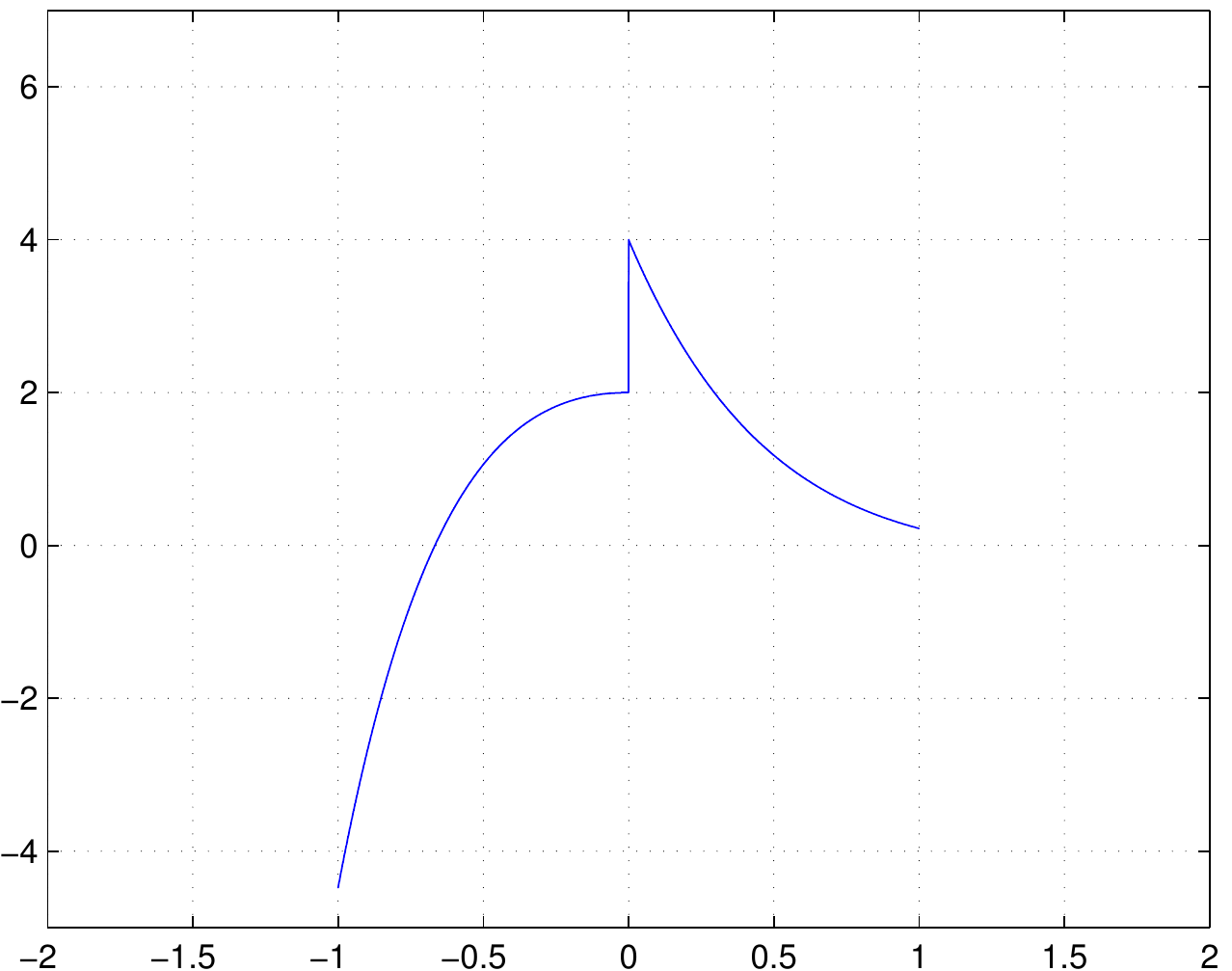}
    }\\
    \subfloat[]{
      \includegraphics[width=0.45\textwidth{}]{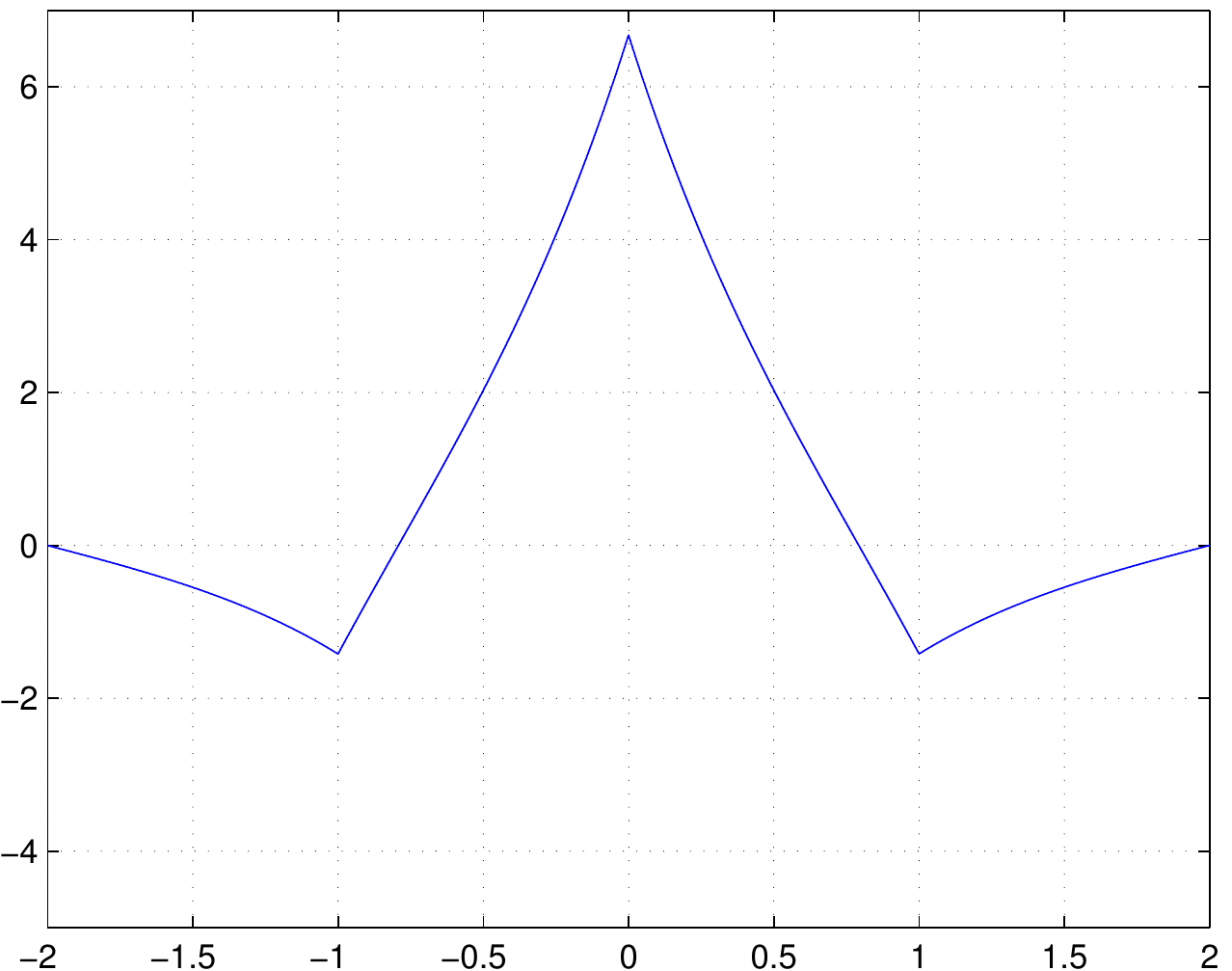}
    }
    \caption[Two pulses with the same autocorrelation]{Example 1: two pulses (a), and (b), and their common
      autocorrelation function (c).}
    \label{fig:math-example1}
  \end{figure}
\item[Example 2]
  Let $g_{1}(t)$ be given by
  \begin{equation}
    \label{eq:35}
    g_{1}(t) =
    \begin{dcases}
      0 & \text{if $|t| > 1$},\\
      e^{at} & \text{if $|t| \leq 1$}.
    \end{dcases}
  \end{equation}
  Then
  \begin{equation}
    \label{eq:39}
    G_{1}(t) = \int_{-1}^{1} e^{(a-s)t}\,\mathrm{d}t = 2\frac{\sinh(s-a)}{s-a}.
  \end{equation}
  The zeros of $G_{1}$ are located at the points $s_{k}$ given by
  \begin{equation}
    \label{eq:40}
    s_{k}=a + jk\pi, \qquad k = \pm 1, \pm 2.
  \end{equation}
  We can construct a real-valued signal having the same support and
  autocorrelation function as $g_{1}(t)$ by letting
  \begin{equation}
    \label{eq:41}
    G_{2}(s)=G_{1}(s)H(s),
  \end{equation}
  where
  \begin{equation}
    \label{eq:114}
    \begin{split}
      H(s)
      & = \frac{s + a - j\pi}{s-a-j\pi}\cdot\frac{s+a-j\pi}{s-a+j\pi}\\
      & = \frac{(s+a)^{2}+\pi^{2}}{(s-a)^{2}+\pi^{2}}.
    \end{split}
  \end{equation}
  The impulse response $h(t)$ of this all-pass filter is easily found
  to be
  \begin{equation}
    \label{eq:43}
    h(t) =  \delta(t) + 4aU(t)e^{at}
    \left(
      \cos(\pi t)+\frac{a}{\pi}\sin(\pi t) \,, 
    \right)
  \end{equation}
  and the convolution between $g_{1}(t)$ and $h(t)$ yields the result
  \begin{equation}
    \label{eq:44}
    g_{2}(t) =
    \begin{dcases}
      0 & \text{if $|t| > 1$}\,, \\
      \frac{4a}{\pi}e^{at}
      \left(
        \frac{a}{\pi}\cos(\pi t) - \sin(\pi t) + \frac{a}{\pi}
      \right) + e^{at} & \text{if $|t| \leq 1$} \,. 
    \end{dcases}
  \end{equation}
  Figure~\ref{fig:math-example2} depicts the signals $g_{1}(t)$,
  $g_{2}(t)$; and their common autocorrelation function.
   \begin{figure}[H]
     \centering
     \subfloat[]{
      \includegraphics[width=0.45\textwidth{}]{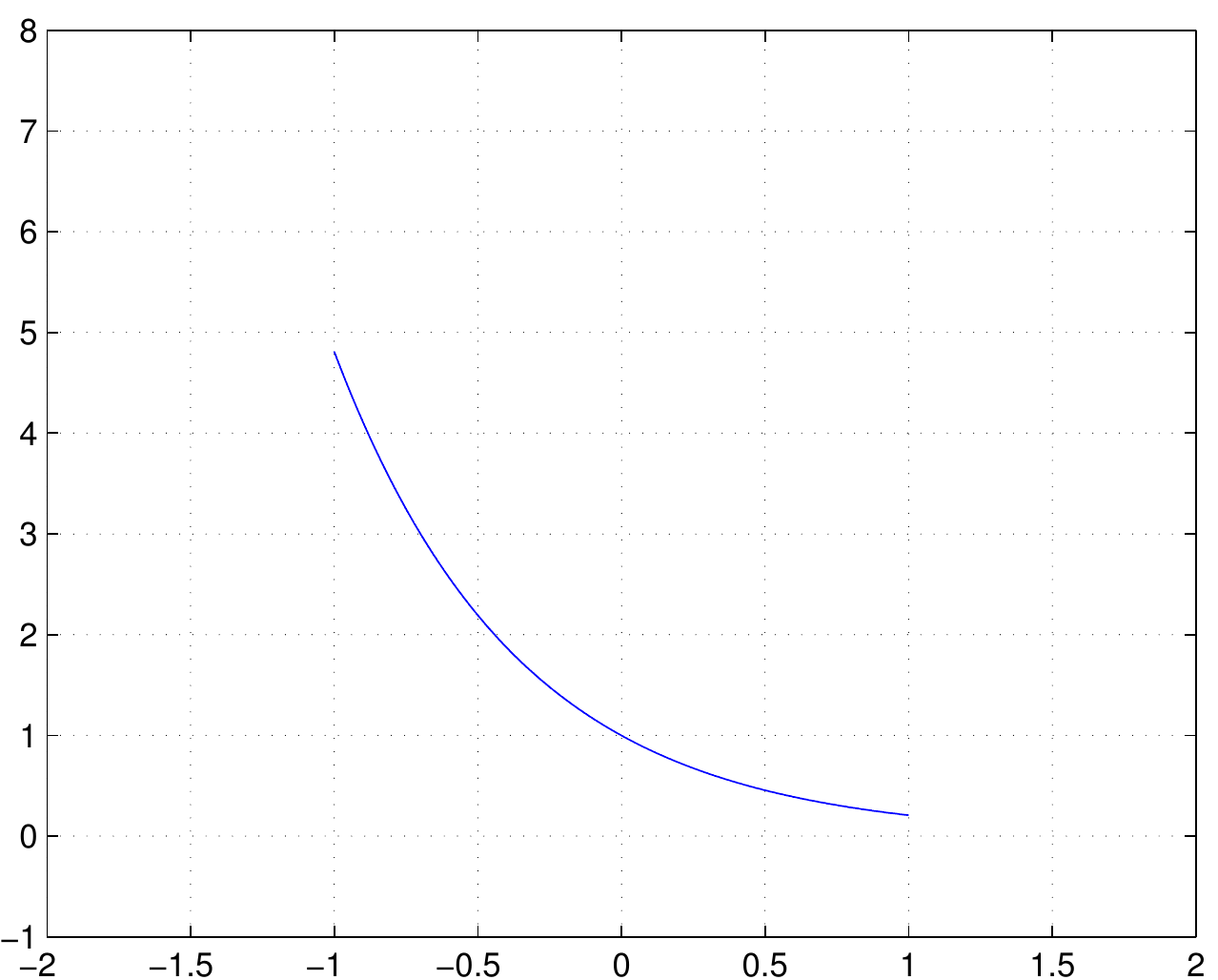}
    }\qquad{}
    \subfloat[]{
      \includegraphics[width=0.45\textwidth{}]{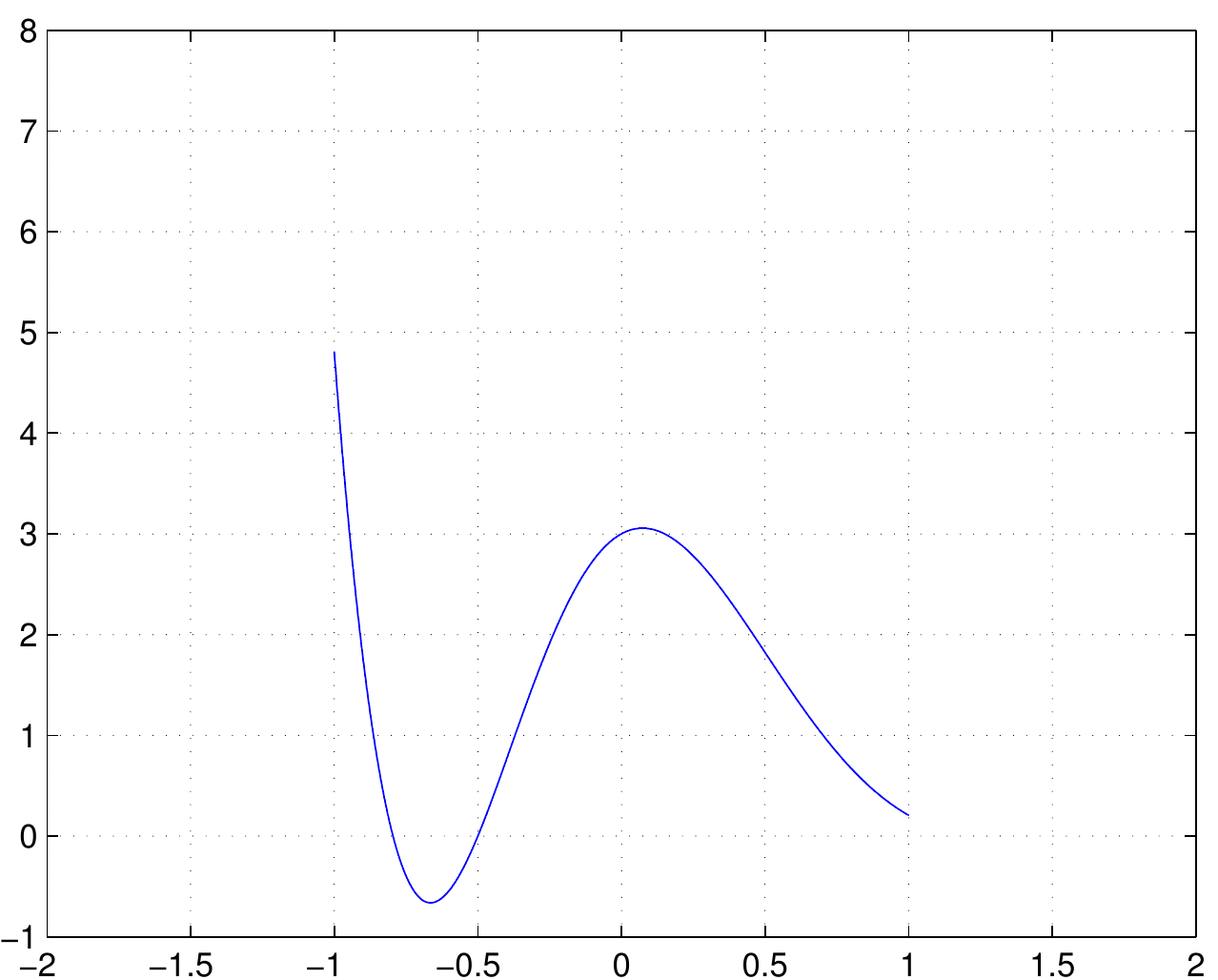}
    }\\
    \subfloat[]{
      \includegraphics[width=0.45\textwidth{}]{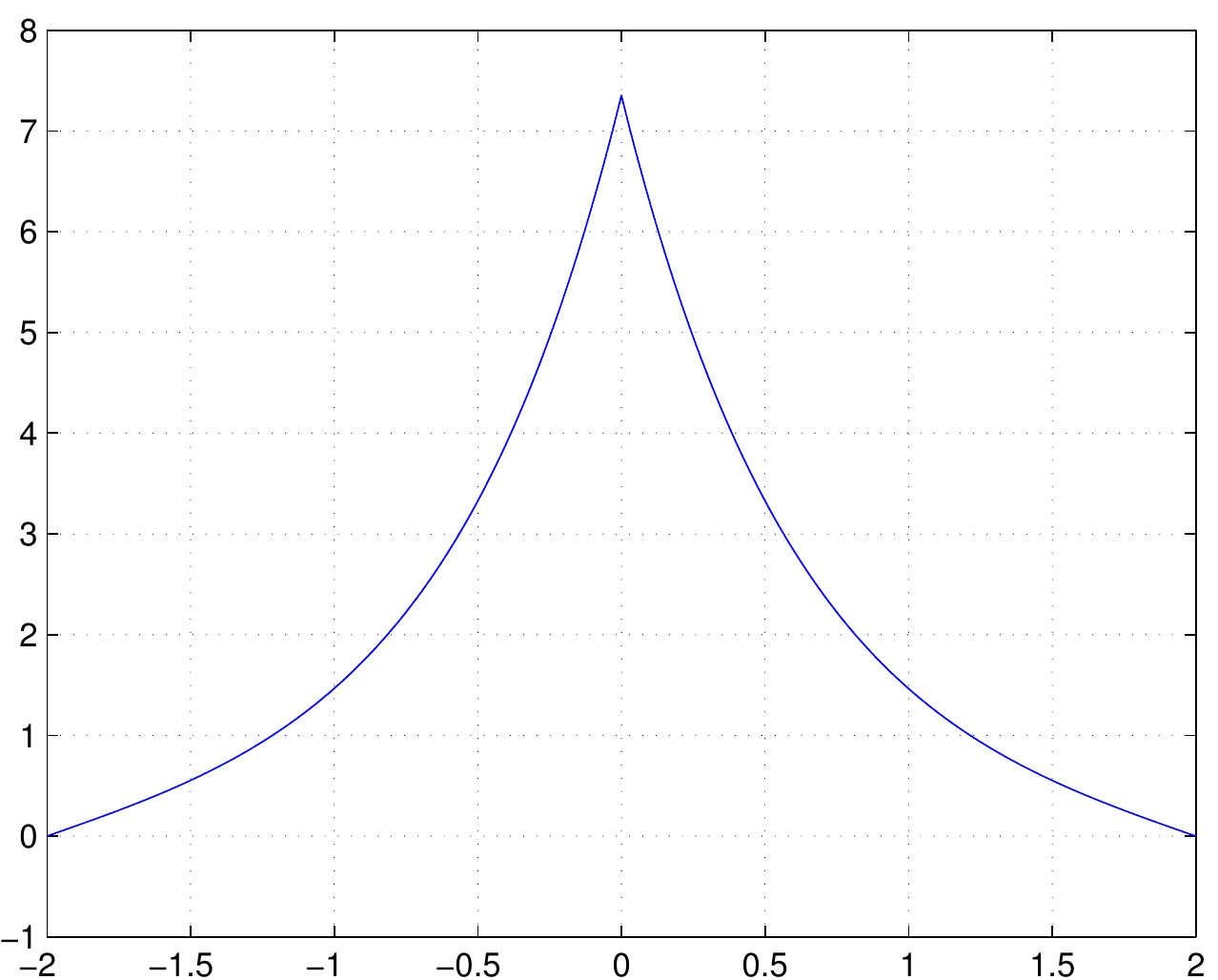}
    }
    \caption[Two pulses with the same autocorrelation]{Example 2:  two pulses (a), and (b), and their common
      autocorrelation function (c).}
    \label{fig:math-example2}
  \end{figure}
\end{description}

\subsection{Discrete-time case}
\label{sec:discrete-case}
The derivations in the previous section describe completely the theory
of continuous-time one-dimensional signals. However, for algorithms
that run on a digital computer, we need a similar theory for
discrete-time signals, that is, for $g(t)$ that is specified on a
finite set of points
\begin{equation}
  \label{eq:1}
  g(t) = \sum_{n}g_{n}\delta(t-t_{n})\,,
\end{equation}
where $t_{n} = n\Delta$. The corresponding Fourier
transform is, of course, the Discrete-Time Fourier Transform (DTFT)
\begin{equation}
  \label{eq:2}
  G(\omega) = \sum_{-\infty}^{\infty}g[n]e^{-j\omega n} \,. 
\end{equation}
By replacing $e^{-jw}$ with $z$ we obtain
\begin{equation}
  \label{eq:42}
  G(z) = \sum_{-\infty}^{\infty}g[n]z^{n}\,, 
\end{equation}
which is, of course, the usual $z$-transform of $g[n]$.\footnote{The notation
  where $G(z)$ is a polynomial in $z$ is common in
  geophysics. Electrical engineers, usually use $z^{-1}$ instead of
  $z$.} As before, $g[n]$ is assumed to have finite support. However,
this time it is more convenient to assume that the support is bounded
by $0$ and $N$, therefore, the summing limits in the
above formula should be replaced as follows
\begin{equation}
  \label{eq:45}
  G(z) = \sum_{n=0}^{N}g[n]z^{n}\,. 
\end{equation}
In this form the discrete $z$-transform looks almost exactly like the
continuous Laplace transform (see Equation~\eqref{eq:12}). Similarly
to Equation~\eqref{eq:18}, it is easy
to show that the $z$-transform of $g[n]$'s autocorrelation (denoted by $R(z)$) reads
\begin{equation}
  \label{eq:46}
  R(z) = \overline{G(1/\bar z)}G(z) =
  \bar G(z^{-1}) G(z)\,. 
\end{equation}
Note the different notation: $\overline{G(z)}$, and $\bar
G(z)$, the former means that $G(z)$ is computed and then conjugated,
the latter means that the coefficients of the polynomial $G(z)$ are
conjugated. That is
\begin{equation}
  \label{eq:118}
  \overline{G(z)} = \overline{
    \left(
      \sum_{n=0}^{N}g[n]z^{n}
    \right)} \,,
  \quad
  \bar G(z) =  \sum_{n=0}^{N}\bar g[n]z^{n} \,. 
\end{equation}
In other words, because
$G(z)$ is a polynomial in $z$, it can always be presented as a product
of simple factors
\begin{equation}
  \label{eq:116}
  G(z) = A\prod_{k=1}^{N}(z-z_{k}) \,, 
\end{equation}
where $A$ is a scalar, and $z_{k}$'s are the roots of
$G(z)$. Similarly, 
\begin{equation}
  \label{eq:119}
  \bar G(z^{-1}) = \bar A \prod_{k=1}^{N}
  \left(
    z^{-1}-\bar z_{k}
  \right) \,.
\end{equation}
This
yields
\begin{equation}
  \label{eq:117}
  R(z) = |A|^{2}\prod_{k=1}^{N}(z-z_{k})
  \left(
    z^{-1} - \bar z_{k}
  \right) \,. 
\end{equation}
Obviously, both $z_{k}$, and $1/\bar z_{k}$ are roots of $R(z)$.
Hence, if we consider a
new signal $g_{2}[n]$ whose $z$-transform is given by
\begin{equation}
  \label{eq:115}
  G_{2}(z)= G(z)\frac{z - 1/\bar {z}_{k}}{z - z_{k}} |z_{k}| \,, 
\end{equation}
where $z_{k}$ is a non-zero root of $G(z)$,
then the auto-correlations of $g_{2}[n]$ and $g[n]$ will be
equal, and so will be their power spectra. To show this, consider the
$z$-transform of $g_{2}[n]$'s autocorrelation (denoted by $R_{2}(z)$):
\begin{equation}
  \label{eq:120}
  \begin{split}
    R_{2}(z)
    & =  \overline{G_{2}(1/\bar z)}G_{2}(z)\\
    & =  \overline{
      \left(
        G(1/\bar z)\frac{1/\bar z - 1/\bar z_{k}}{1/\bar
          z - z_{k}}
      \right)|z_{k}|}
    G(z)\frac{z - 1/\bar {z}_{k}}{z - z_{k}}|z_{k}|\\
    & = R(z)
    \left(
      \frac{1/ z - 1/z_{k}} {1/z - \bar z_{k}}
    \right)
    \left(
      \frac{z - 1/\bar {z}_{k}}{z - z_{k}}
    \right)|z_{k}|^{2}\\
    & = R(z) \,. 
  \end{split}
\end{equation}
By applying the inverse $z$-transform we immediately conclude that
$g_{2}[n]$ and $g[n]$ have the same auto-correlation:
\begin{equation}
  \label{eq:121}
  r_{2}[n]\equiv g_{2}[n]\star g_{2}[n]= g[n]\star g[n]\equiv r[n]
  \,. 
\end{equation}

This analysis leads to a simple conclusion: each non-zero and
non-unitary ($|z_{k}|\not = 1$) root of $G(z)$ gives rise to two
possible solutions whose auto-correlation is the same. Hence, for a
general one-dimensional signal, whose support spreads over $N+1$
samples, there can be up to $2^{N}$ different signals\footnote{Here we
  do not count the trivial solutions that are obtained by multiplying
  by a constant phase factor.}  within the same support and same
auto-correlation. Of course, similarly, to the continuous-time case
that we considered in the previous section, in the case of real (and,
probably, non-negative) signal $g[n]$ the number of possible solutions
may be smaller because not all roots of $G(z)$ can be exchanged
freely---some will result in complex (or, probably, negative) signals.

Another important observation is that if a solution $x[n]$ has been bound,
then \emph{all} other solutions can be obtained from it by
systematically replacing the roots of $X[n]$, as described in
Equation~\eqref{eq:117}.\footnotemark[1]{}

Finally, recall that the sampling is done in the Fourier
domain. Hence, to capture $R(z)$ (the $z$-transform of $g[n]$'s
auto-correlation), we must sample it at $2N+1$ points. This requirement
directly follows from the fact that $R(z)$ is ``almost a polynomial'',
that is, 
\begin{equation}
  \label{eq:122}
  R(z) = \frac{P_{2N}(z)}{z^{N}} \,, 
\end{equation}
where $P_{2N}(z)$ is a polynomial of degree $2N$. From
Equation~\eqref{eq:117}, we have
\begin{equation}
  \label{eq:123}
  P_{2N} =  |A|^{2}\prod_{k=1}^{N}(z-z_{k})
  \left(
    1 - z\bar z_{k}
  \right) \,.
\end{equation}
Due to the uniqueness of the interpolation polynomial (see, for
example~\shortcite{dahlquist08numerical}), it is sufficient to sample
$P_{2N}$ at $2N+1$ points to fully determine its coefficients, and, thus,
to determine $R(z)$. For practical reasons the samples should be
performed at the points that correspond to the DFT frequencies. Hereon
we finish our treatment of the one-dimensional case and switch to
multi-dimensional signals.

\section{Multi-dimensional signals}
\label{sec:multi-dimens-case}
The analysis in the two- or higher-dimensional case
is very similar to what we have done in the one-dimensional case. The
main result is a straightforward generalization of
Equation~\eqref{eq:46}. That is, given a two-dimensional time-discrete
signal $g[n_{1}, n_{2}]$, whose support is given by
$[0,N_{1}]\times[0,N_{2}]$, the $z$-transform of its autocorrelation is
given by
\begin{equation}
  \label{eq:124}
  R(z_{1}, z_{2}) = \overline{G(1/\bar z_{1}, 1/\bar z_{2})}G(z_{1}, z_{2}) =
  \bar G(z_{1}^{-1},z_{2}^{-1} ) G(z_{1}, z_{2})\,,
\end{equation}
where $G(z_{1}, z_{2})$ is the two-dimensional $z$-transform of
$g[n_{1}, n_{2}]$. From this formula we can easily find a way to
generate a signal whose autocorrelation function is equal to that
of $g[n_{1}, n_{2}]$. Let us assume that $G(z_{1}, z_{2})$ can be
represented as a product of two polynomials of lower degree:
\begin{equation}
  \label{eq:125}
  G(z_{1}, z_{2}) = P(z_{1}, z_{2})Q(z_{1}, z_{2}) \,. 
\end{equation}
Assume further that the degree of $Q(z_{1}, z_{2})$ in $z_{1}$ and
$z_{2}$ is $d_{1}$ and $d_{2}$, respectively.
Now we have
\begin{equation}
  \label{eq:126}
  \begin{split}
    R(z_{1}, z_{2})
    & = \bar G(z_{1}^{-1},z_{2}^{-1} ) G(z_{1}, z_{2})\\
    & = \bar P(1/z_{1}, 1/z_{2}) \bar Q(1/z_{1}, 1/z_{2}) P(z_{1},
    z_{2})Q(z_{1}, z_{2})\\
    & = \bar P(1/z_{1}, 1/z_{2})Q(z_{1}, z_{2})
    P(z_{1},z_{2}) \bar Q(1/z_{1}, 1/z_{2})\\
    & = \bar G_{2}(1/z_{1}, 1/z_{2}) G_{2}(z_{1}, z_{2})\,, 
  \end{split}
\end{equation}
where
\begin{equation}
  \label{eq:127}
  G_{2}(z_{1}, z_{2}) = P(z_{1},z_{2}) \bar Q(1/z_{1}, 1/z_{2})z_{1}^{d_{1}}z_{2}^{d_{2}}.
\end{equation}
Now, by applying the inverse $z$-transform to  $G_{2}(z_{1},z_{2})$, we
obtain a new signal $g_{2}[n_{1}, n_{2}]$ whose autocorrelation is
equal to the autocorrelation of $g[n_{1}, n_{2}]$. Note that the
multiplicative factor $z_{1}^{d_{1}}z_{2}^{d_{2}}$ makes $G_{2}(z_{1},
z_{2})$ a proper polynomial in $z_{1}$ and $z_{2}$ whose degrees vary
from $0$ to $N_{1}$, and from $0$ to $N_{2}$, respectively. Hence, it ``shifts''
$g_{2}[n_{1}, n_{2}]$ to the same support region $[0,N_{1}]\times
[0,N_{2}]$ as occupied by $g_[n_{1}, n_{2}]$. Obviously
$g_{2}[n_{1}, n_{2}]\not=g_{1}[n_{1}, n_{2}]$ whenever
\begin{equation}
  \label{eq:128}
  \bar Q(1/z_{1}, 1/z_{2})z_{1}^{d_{1}}z_{2}^{d_{2}}
  \not = Q(z_{1}, z_{2})\,. 
\end{equation}
A similar result was obtained in~\shortcite{hayes82reconstruction},
however, the authors there considered only real-valued signals, and
their approach was slightly different from ours.

So far, the development is essentially the same as we have seen in the
one-dimensional case. The main result is also very similar: each
factor of $G(z_{1},z_{2})$ can give rise to two different solutions
(when inequality \eqref{eq:128} hold). The main difference however,
stems from the fact that multi-variate polynomials are, usually,
\emph{irreducible}, that is they cannot be factorized.
More specifically, the set of reducible
multi-variate polynomials is of measure zero. This fact was proved
in~\shortcite{hayes82reconstruction} for polynomials with real coefficients,
however, its generalization to polynomials with complex coefficient
is straightforward. In practical terms, this means that the
chances of getting a reducible two- or three-dimensional polynomial are
zero. This, in turn, means that the phase retrieval problem in the
multi-dimensional case has, usually, a unique solution (up to the
trivial transformations: lateral shifts, axis reversal, and constant
phase factor).

Despite this ``almost always'' guaranteed uniqueness one must apply a
critical judgment for every specific case, because physical signals may
not be considered ``purely random''. For example, the reducibility in
the $z$-space has a clear physical meaning: if
$G(z_{1},z_{2})=P(z_{1},z_{2})Q(z_{1},z_{2})$ then, $g[n_{1},n_{2}] =
p[n_{1},n_{2}]\otimes q[n_{1},n_{2}]$, where $\otimes$ denotes
convolution. Hence, if the sought signal is a result of a convolution of
some signal $p$ with a non-symmetric kernel $q$, the reconstruction
will not be unique (even without counting the trivial
transformations).

%%% Local Variables: 
%%% mode: latex
%%% TeX-master: "../thesis"
%%% End: 

%% file: current_methods/current_methods.tex
\chapter{Current reconstruction methods}
\label{cha:curr-reconstr-meth}
Current reconstruction methods date back to the pioneering work of
Gerchberg and Saxton (GS)~\shortcite{gerchberg72practical}. Their original
method was later improved significantly by Fienup~\shortcite{fienup82phase},
who introduced the Hybrid Input-Output (HIO) algorithm. The latter
algorithm, to the best of our knowledge, is the prevailing
numerical method today for phase retrieval. The HIO method will be
presented in Section~\ref{sec:fienup-algor-phase}. However, to better
understand it, we start with its progenitor---GS, which is a classical
example of optimization techniques known today as ``alternating
projections''.  Details of the methods are give in
Section~\ref{sec:gerchb-saxt-meth} below. Before we proceed further,
we need to define two basic terms:
\begin{defn}
  \label{def:current-1}
  The distance between a point $x$ and a closed set $\mathcal{S}$ is
  defined as
  \begin{equation}
    \label{eq:current-1}
    d(x,\mathcal{S}) = \min_{y\in\mathcal{S}}\|x - y\|\,. 
  \end{equation}
\end{defn}

\begin{defn}
  \label{def:current-2}
  For a given point $x$ and a closed set $\mathcal{S}$ we say
  $y\in{\mathcal{S}}$ is a projection of $x$ onto $\mathcal{S}$ if:
  \begin{equation}
    \label{eq:current-2}
    \|x - y\| = d(x, \mathcal{S})\,. 
  \end{equation}
  That is, $y$ is a solution of the following minimization problem:
  \begin{equation}
  \label{eq:current-3}
  \begin{split}
    \min_{y} &\quad \|x-y\|\,,  \\
    \mathrm{subject\ to} &\quad y\in\mathcal{S} \,. 
  \end{split}
\end{equation}
\end{defn}
It is important to note that a projection always exists, and
furthermore it is unique if $\mathcal{S}$ is convex. Otherwise, there
may be several solutions to Equation~\eqref{eq:current-3}. As we shall
see later, the constraints that appear in the phase retrieval problem
are not convex. Nevertheless, the projection is still well defined
(unique) in
all cases, except when the current estimate has zeros in the Fourier
domain.  The non-convexity of the constraints along with existence of
computationally cheap projections is the reason why current
reconstruction methods are based on projections and why continuous
optimization techniques, like gradient descent or Newton-type methods,
are not capable of successful phase retrieval. More details on that
will follow in Chapter~\ref{cha:appr-four-phase-explanation}.
Meanwhile we proceed to the current reconstruction methods.

\section{Gerchberg-Saxton method}
\label{sec:gerchb-saxt-meth}
Probably the first successful reconstruction method for phase
retrieval was suggested by Gerchberg and Saxton for a slightly
different problem---reconstruction of signals from two intensity
measurements \shortcite{gerchberg72practical}.  The authors considered a
situation that arises in electronic microscopes, where the intensity of
the sought signal\footnote{The signal is complex-valued, of course,
  otherwise the reconstruction would be trivial.}  can be measured
along with its diffraction pattern (Fourier domain intensity). For
this scenario, the authors suggested a reconstruction method that is
based on projections (hereinafter the method will be referred
to as GS method or GS algorithm). The algorithm is iterative, and each
iteration consists of the following four steps:
\begin{description}
\item[Step 1:] Fourier transform the current estimate of the signal.
\item[Step 2:] Replace the magnitude of the resulting computed Fourier
  transform with the measured Fourier magnitude to form a new estimate of
  the Fourier transform.
\item[Step 3:] Inverse Fourier transform the estimate of the Fourier transform.
\item[Step 4:] Replace the magnitude of the resulting computed signal with the
  measured signal modulus to form a new estimate of the signal.
\end{description}
It is a trivial exercise in basic calculus to show that the Steps
2 and 4 are, indeed, projections.

Of course, this algorithm can be (and, in fact, has been)
generalized to a large number of situations, where  the
constraints in both domains are such that lead to a well defined and,
preferably, computationally efficient projection. For example,
this happens in the phase retrieval problem, where the object domain
constraints are: limited support, that is, some parts of the signal
are known to be zero; and, often, non-negativity---the signal in the
support area is known to be real non-negative. A generalized version
of GS is depicted in Figure~\ref{fig:current-gs-method} below.

\begin{figure}[H]
  \centering
  \includegraphics{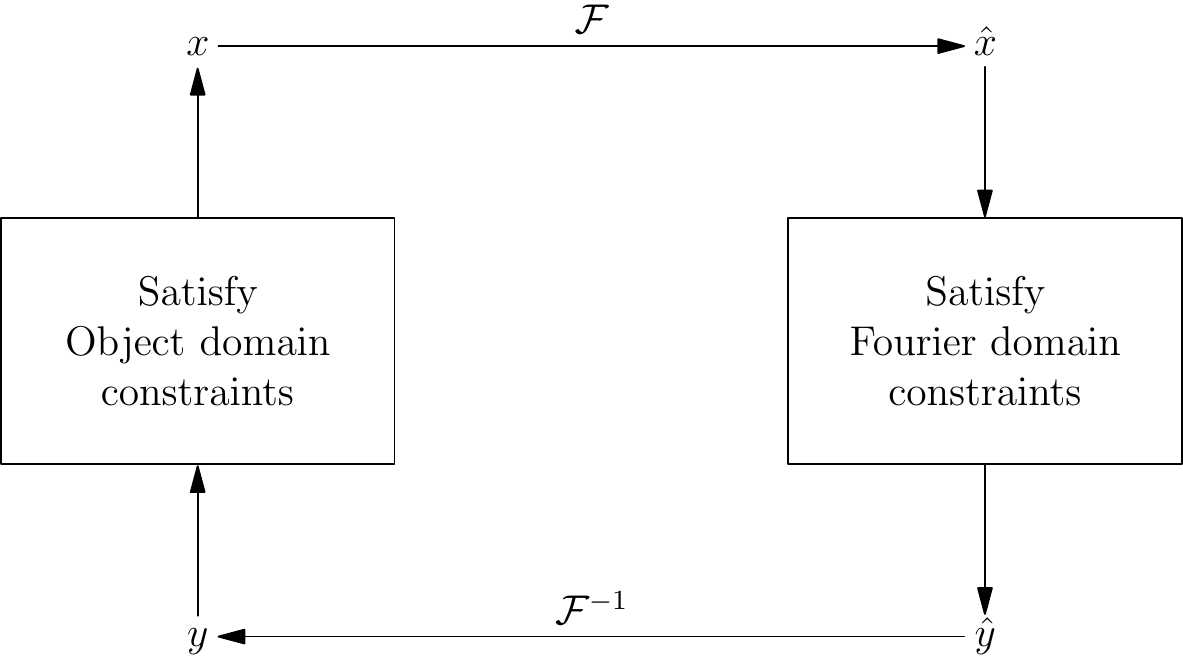}
  \caption{Schematic description of the Gerchberg-Saxton algorithm.}
\label{fig:current-gs-method}
\end{figure}
It is important to stress that the imposition of the constraints (both
in the Fourier and object domain) is performed via a
projection. Projections, unlike general transformations, guarantee
convergence of the algorithm as we prove below\footnote{A similar
  theorem was proved by Fienup for a specific set of constraints
  \shortcite{fienup82phase}. Our proof is much more general.}. Before
proceeding to the proof, note that convergence here means the lack of
progress of the algorithm.  It can happen in two different situations:
first, arriving at some stationary point (a solution is always a
stationary point, but not vice-versa); second, the algorithm can enter
into an endless loop, jumping from point to point without decreasing
the error.
\begin{thm}
  \label{thm:alternating-projections}
  Let $\mathcal{S}_{\mathcal{O}}$ and $\mathcal{S_{\mathcal{F}}}$ be
  the sets of feasible signals as defined by the constraints in the
  object and Fourier domains, respectively. Furthermore, assume
  that these sets are equipped with the corresponding  projection
  operators $P_{\mathcal{O}}$, and $P_{\mathcal{F}}$. Let
  $\left\{x^{k}\right\}_{k=0}^{\infty}$ and
  $\left\{y^{k}\right\}_{k=0}^{\infty}$ be the two sequences generated
  by the generalized GS method using the two projections:
  \begin{equation}
    \label{eq:current-4}
    y^{k} = P_{\mathcal{F}}[x^{k}], \quad
    x^{k+1}=P_{\mathcal{O}}[y^{k}]\,. 
  \end{equation}
  Then, the sequences
  $\left\{E_{\mathcal{F}}^{k}\right\}_{k=0}^{\infty}$, and
  $\left\{E_{\mathcal{O}}^{k}\right\}_{k=0}^{\infty}$, defined as
  \begin{equation}
    \label{eq:current-5}
    E_{\mathcal{F}}^{k} = d(x^{k}, S_\mathcal{F}),\quad
    E_{\mathcal{O}}^{k} = d(y^{k}, S_\mathcal{O})\,,
  \end{equation}
  are monotonically decreasing, that is
  \begin{equation}
    \label{eq:current-6}
    E_{\mathcal{F}}^{k+1} \leq E_{\mathcal{F}}^{k}
  \end{equation}
  \begin{equation}
    \label{eq:current-7}
    E_{\mathcal{O}}^{k+1} \leq E_{\mathcal{O}}^{k}
  \end{equation}
\end{thm}
\begin{proof}[Proof]
  Because $y^{k}$ is a projection of $x^{k}$ onto
  $\mathcal{S}_{\mathcal{F}}$ we have, by
  Definition~\ref{def:current-2},
  $d(x^{k}, \mathcal{S}_{\mathcal{F}})=\|x^{k}-y^{k}\|$. Furthermore,
  \begin{equation}
    \label{eq:current-8}
    d(x^{k}, \mathcal{S}_{\mathcal{F}}) = \|x^{k} - y^{k}\| \geq
    \|x^{k+1} - y^{k}||\,.  
  \end{equation}
  Note that the inequality in the above equation follows from the fact
  that both $x^{k}$, and $x^{k+1}$ belong to $\mathcal{S}_{\mathcal{O}}$
  and $x^{k+1}$ is a projection of $y^{k}$ onto
  $\mathcal{S}_\mathcal{O}$. Hence $\|y^{k}-x^{k+1}\| \leq
  \|y^{k}-x^{k}\|$. Similarly,
  \begin{equation}
    \label{eq:current-9}
    d(x^{k+1}, \mathcal{S}_{\mathcal{F}}) = \|x^{k+1} - y^{k+1}\| \leq
    \|x^{k+1} - y^{k}\|\,.  
  \end{equation}
  Again, the inequality follows from the fact that both $y^{k}$ and
  $y^{k+1}$ belong to $\mathcal{S}_{\mathcal{F}}$ and $y^{k+1}$ is a
  projection of $x^{k+1}$ onto $\mathcal{S}_{\mathcal{F}}$. By combining
  Equations~\eqref{eq:current-8} and \eqref{eq:current-9} we obtain
  \begin{equation}
    \label{eq:current-10}
    d(x^{k+1}, \mathcal{S}_{\mathcal{F}}) \leq  \|x^{k+1} - y^{k}\|
    \leq  d(x^{k}, \mathcal{S}_{\mathcal{F}})\,. 
  \end{equation}
  Hence
  \begin{equation}
    \label{eq:current-11}
     d(x^{k+1}, \mathcal{S}_{\mathcal{F}}) \leq  d(x^{k},
     \mathcal{S}_{\mathcal{F}})\,.
   \end{equation}
   The proof for the second claim follows immediately if we write down
   Equation~\eqref{eq:current-10} for the iterations $k$ and $k+1$:
   \begin{equation}
     \label{eq:current-12}
     d(x^{k+2}, \mathcal{S}_{\mathcal{F}})  \leq \|x^{k+2} - y^{k+1}\|
     \leq  d(x^{k+1}, \mathcal{S}_{\mathcal{F}}) \leq  \|x^{k+1} - y^{k}\|
     \leq  d(x^{k}, \mathcal{S}_{\mathcal{F}}) \,, 
   \end{equation}
   and note that
   \begin{equation}
     \label{eq:current-13}
     \|y^{k} - x^{k+1}\| = d(y^{k}, \mathcal{S}_{\mathcal{O}}),\quad
     \|y^{k+1} - x^{k+2}\| = d(y^{k+1}, \mathcal{S}_{\mathcal{O}})\,.
   \end{equation}
   Thus, we obtain
   \begin{equation}
     \label{eq:current-14}
     d(y^{k+1}, \mathcal{S}_{\mathcal{O}}) \leq d(y^{k},
     \mathcal{S}_{\mathcal{O}})\,.
   \end{equation}
\end{proof}
Note that all $x^{k}$ satisfy the object domain constraints and their
discrepancy with the Fourier domain constraints is ever
decreasing\footnote{Strictly speaking, it is a non-increasing
  sequence, however, in practice most algorithms terminate after the
  decrease in the current step is below some threshold. Hence, the
  sequence is strictly decreasing during the algorithm execution.}
with $k$. Similarly, all $y^{k}$ satisfy the Fourier domain
constraints and their discrepancy with the object domain constraints
is also ever decreasing. Hence,
Theorem~\ref{thm:alternating-projections} may suggest that the GS
method converges to a solution. This is true if the constraints are
convex. In our case, however, the Fourier domain constraints are
non-convex. Thus, the convergence to a solution is not guaranteed: the
decrease in the functions $d(x^{k}, \mathcal{S}_{\mathcal{F}})$ and
$d(y^{k}, \mathcal{S}_{\mathcal{O}})$ can be arbitrary small and even
zero if the algorithm gets stuck as some stationary point (usually, a
local minimum). Moreover, extensive experiments confirm that GS is not
suitable for the standard phase retrieval from a single intensity
measurement and support information (even for non-negative signals):
the algorithm typically stagnates at some point that is nowhere near
a solution.  In the next section we will review the Hybrid
Input-Output algorithm that was invented by Fienup to overcome the
stagnation problem of GS.

\section{Fienup's algorithms for phase retrieval }
\label{sec:fienup-algor-phase}
In 1982, Fineup suggested a family of iterative algorithms that are
based on a different interpretation of the GS method
\shortcite{fienup82phase}. These algorithms keep intact the right-hand side
(Fourier domain) of the diagram depicted in
Figure~\ref{fig:current-gs-method}, that is, the first three
operations of each iteration remains the same:
\begin{description}
\item[Step 1:] Fourier transforming
  $x^{k}\stackrel{\mathcal{F}}{\rightarrow}\hat{x}^{k}$.
\item[Step 2:] Satisfying the Fourier domain constraints
  $\hat{x}^{k} \rightarrow \hat{y}^{k}$.
\item[Step 3:] Inverse Fourier transforming the result $\hat{y}^{k} \stackrel{\mathcal{F}^{-1}}{\rightarrow} y^{k}$.
\end{description}
However, the further treatment (in the object domain) is
different. Fienup's insight was to group together the three steps
above into a non-linear system having an input $x$ and an output $y$
as depicted in Figure~\ref{fig:current-fienup-alg}. The useful
property of this system is that the output $y$ is always a signal
having a Fourier transform that satisfies the Fourier domain
constraints. Therefore, if the output also satisfies the object domain
constraints, it is a solution of the problem. Unlike in the GS
algorithm, the input $x$ need no longer be thought of as the
current estimate of the signal. Instead, it can be considered as a
driving function for the next output $y$. Hence, the input $x$ need
not necessarily satisfy the object domain constraints.
\begin{figure}[H]
  \centering
  \includegraphics{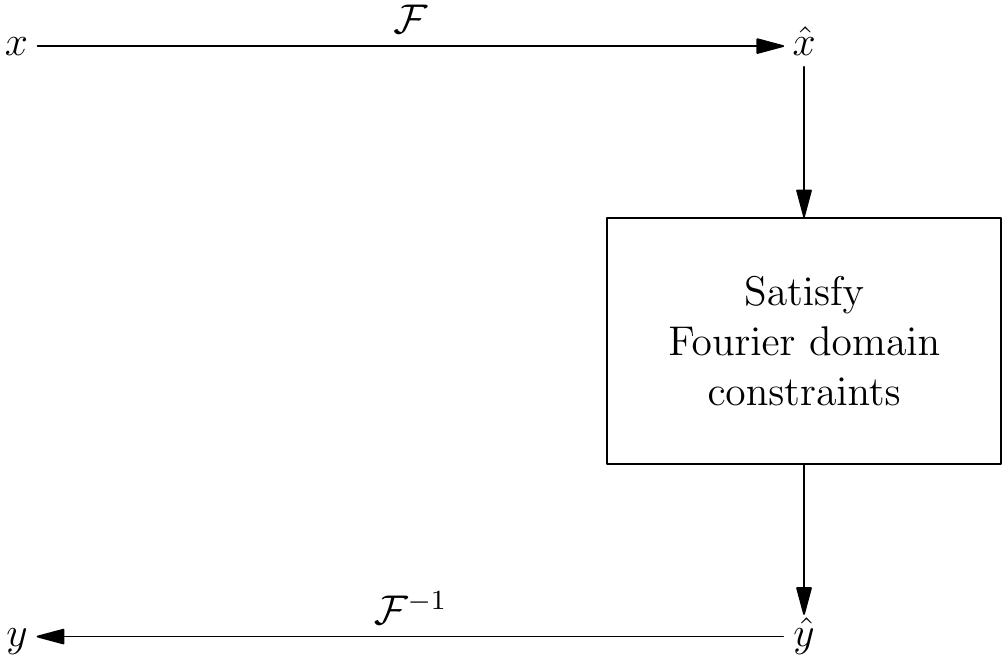}
  \caption[Projection in the Fourier domain]{Fienup's interpretation of the Fourier domain constraints
    imposing procedure as a non-linear system.}
  \label{fig:current-fienup-alg}
\end{figure}

Based on this novel interpretation, Fienup suggested three algorithms
for the phase retrieval problem from a single intensity and a priori
knowledge that the signal $x$ is non-negative
everywhere.\footnote{Originally, the algorithms were developed for
  real non-negative signals such as those that arise in
  astronomy. Later, the method was applied to complex-valued signals
  where it was discovered that precise support information is crucial
  for successful reconstruction \shortcite{fienup87reconstruction}.}

\begin{description}
\item [{input-output}] This algorithm is based on a claim (see
  \shortcite{fienup80iterative}) that a small change of the input results
  in a change of the output that is a constant $\alpha$ times the
  change in the input. Hence, if a change $\Delta x$ is desired in the
  output, a logical choice for the change of the input to achieve that
  change in the output would be $\beta\Delta x$, where $\beta$ is a
  constant, ideally equal to $\alpha^{-1}$. For the problem of phase
  retrieval from a single intensity measurement the desired change of
  the output is
  \begin{equation}
    \label{eq:57}
    \Delta x^{k}(t)=\begin{cases}
      0, & t\not\in\nu \,, \\
      -y^{k}(t), & t\in\nu \,, 
  \end{cases}
  \end{equation}
  where $\nu$ is the set of points at which $y(t)$ violates the object
  domain constraints. That is, at points where the constraints are
  satisfied, one does not require a change of the output. On the other
  hand; at points where the constraints are violated, the desired change of the
  output, required to satisfy the support and non-negativity
  constraints, is one that drives it towards the value of zero, and
  therefore, the desired change is the negated output at those
  points. Hence, the logical choice for the next input is
  \begin{align}
    x^{k+1}(t) & =x^{k}(t)+\beta\Delta x^{k}(t)\nonumber \\
    & =\begin{cases}
      x^{k}(t), & t\not\in\nu,\\
      x^{k}(t)-\beta y^{k}(t), & t\in\nu.
    \end{cases}\label{eq:ioalg}
  \end{align}

\item [{output-output}] This algorithm is based on the following
  observation with respect to the non-linear system depicted in
  Figure~\ref{fig:current-fienup-alg}.  If the output $y$ is used as an
  input, the resulting output will be $y$ itself, because it already
  satisfies the Fourier domain constraints.  Therefore, irrespective
  of what input actually resulted in the output $y$, the output $y$
  can be considered to have been resulted from itself as an input. From
  this point of view another logical choice for the next input is
  \begin{align}
    x^{k+1}(t) & =y^{k}(t)+\beta\Delta x^{k}(t)\nonumber \\
    & =\begin{cases}
      y^{k}(t), & t\not\in\nu,\\
      y^{k}(t)-\beta y^{k}(t), & t\in\nu.
    \end{cases}\label{eq:ooalg}
  \end{align}
  Note that if $\beta=1$ the output-output algorithm becomes the GS
  algorithm. Because best results are, in general, obtained with
  $\beta\not=1$ the GS algorithm can be viewed as a
  sub-optimal version of the input-output algorithm.
\item [{hybrid input-output}] Finally, we consider the third algorithm
  suggested by Fienup. This time the next input is formed by a
  combination of the upper line of Equation~\eqref{eq:ooalg} with the
  lower line of Equation~\eqref{eq:ioalg}:
  \begin{align}
    x^{k+1}(t) & =y^{k}(t)+\beta\Delta x^{k}(t)\nonumber \\
    & =\begin{cases}
      y^{k}(t), & t\not\in\nu,\\
      x^{k}(t)-\beta y^{k}(t), & t\in\nu.
    \end{cases}\label{eq:hioalg}
  \end{align}
\end{description}
The last algorithm, known as the Hybrid Input-Output (HIO) algorithm,
is currently the most widely used algorithm in industry due to
its simplicity and (usually) best convergence rate amongst the above
three algorithms.

In contrast to the GS method, there is no proof of convergence for the
HIO method. However, a large body of experiments indicates that the
algorithm is often successful in phase retrieval of real-valued
non-negative signals. Still, stagnation is possible in certain
situations
\shortcite{fienup86phase-retrieval,wackerman89avoiding}. Besides, it turns
out that phase retrieval of complex-valued objects whose support is
not known precisely poses a much more severe problem. In this case, HIO
is not capable of successful reconstruction
\shortcite{fienup87reconstruction}.

Several examples of reconstruction by HIO will be presented in the
following sections. We do not present results obtained by the GS
method because this method is not suitable for the phase retrieval
problem---it stagnates very fast and its results are not even close to
the sought signal.

%%% Local Variables: 
%%% mode: latex
%%% TeX-master: "../thesis"
%%% End: 

%% file: optimization/optimization.tex
\chapter{Fundamental developments in optimization
  methods\footnotemark{}}
\label{cha:found-optim-meth}

\footnotetext{The material presented in this section is currently in
  preparation for submission to a journal.}

Our research actually started with an idea to develop an efficient
phase retrieval method based on continuous optimization
technique. Unlike the current methods, described in the previous
chapter, the continuous optimization approach can potentially provide significantly
faster convergence rates and, probably even more important, can allow
easy introduction of additional knowledge/assumptions into the
computational scheme. The latter is especially difficult in the
projection-based framework that is used in the current methods.

However, continuous optimization techniques, such as gradient descent
or Newton-type methods, cannot be applied directly
because the objective function $f(z)$ is a real-valued function of
complex variables: $f:\mathbb{C}^{n} \mapsto \mathbb{R}$. Therefore, its
derivatives (of any order) with respect to $z$ are
not defined, as we show below. This can be circumvented by treating the
real and the imaginary parts of $z$ separately, that is, by looking at
the function $g:\mathbb R^{n}\times\mathbb R^{n} \mapsto \mathbb R$,
where $z = x+jy$, and $f(z) = f(z(x,y)) = g(x,y)$. This approach is
viable and widely applied, though it may be more convenient to work
with the original variable $z$ (and its complex conjugate $\bar{z}$)
rather than with its real and imaginary parts: $x$, and $y$. Moreover,
because most modern computer languages provide native support for
complex variables, this approach may be more efficient as well. Hence,
in the subsequent sections we shall develop an alternative definition
of the gradient and Hessian. Before that, let us demonstrate that
$f(z)$ is not differentiable with respect to $z$, except in the trivial
case where $f(z)$ is constant.
 
\begin{lem}
  \label{lem:nondiffer}
  Let $f(z)$ be a real function of complex argument $z$, then
  $f(z)$ cannot be holomorphic unless it is constant.
\end{lem}
\begin{proof}[Proof]
  Let us denote the complex argument $z=x+jy$ and $f(z) =
  u(z)+jv(z) = u(x,y) + jv(x,y)$, where $x$, $y$, $u$, and $v$ are
  real. If $f(z)$ is holomorphic it must satisfy the
  Cauchy-Riemann equations
  \begin{equation}
    \label{eq:math-1111}
    \frac{\partial u}{\partial x} = \frac{\partial v}{\partial y}\,,\qquad
    \frac{\partial u}{\partial y} = -\frac{\partial v}{\partial x}\,. 
  \end{equation}
  However, $f(z)$ is real, hence $v(x,y) = 0$ which, in
  turn means that
  \begin{equation}
    \label{eq:math-2222}
    \frac{\partial u}{\partial x} = \frac{\partial u}{\partial y} = 0
    \,. 
  \end{equation}
  Thus $u(x,y)$ is a constant and so is $f(z)$.
\end{proof}
Before proceeding to the next section it is pertinent to define the
following partial derivatives: $\partial f/\partial z$, and  $\partial
f/\partial \bar{z}$. Let $f(z) = h(z, \bar{z}) = u(x,y) + jv(x,y)$. If
$f(z)$ is real, that is, if $v(x,y) = 0$, then
\begin{equation}
  \label{eq:61}
  \begin{split}
    \frac{\partial f}{\partial z} & = \frac{1}{2}
    \left(
      \frac{\partial u}{\partial x} - j\frac{\partial u}{\partial y}
    \right)\,, \\
    \frac{\partial f}{\partial \bar z} & = \frac{1}{2}
    \left(
      \frac{\partial u}{\partial x} + j\frac{\partial u}{\partial y}
    \right) \,.  
  \end{split} 
\end{equation}
The above derivatives are, sometimes, called the Wirtinger derivatives
or Wirtinger operators~\shortcite{wirtinger27zur}.

\section{Complex gradient}
\label{sec:complex-gradient}
Let us consider the first differential of a differentiable function
$g: \mathbb R^{n} \mapsto \mathbb R$
\begin{equation}
  \label{eq:math-3333}
  \mathrm{d}g = \langle \nabla g, \mathrm{d}x \rangle \,,
\end{equation}
where $\langle \cdot,\cdot\rangle$ denotes the usual inner product.
This formula can be used for definition of a function's gradient (see,
for example, \shortcite{magnus99matrix}).  However, this approach is not
feasible in our case because the derivatives $\partial f/\partial
z_{i}$ are not defined, unless $f(z)$ is holomorphic. And this, of
course, is not possible except for some trivial cases where $f(z)$ is
constant, as was shown in Lemma~\ref{lem:nondiffer}. Therefore,
we suggest the following new definition for a real scalar function of
a complex vector $f:\mathbb C^{n} \mapsto \mathbb R$
\begin{equation}
  \label{eq:58}
  \mathrm{d}f = \Re \langle \nabla f, \mathrm{d}z \rangle \,,
\end{equation}
where $\Re$ denotes the real part of a complex number. This definition
preserves the most important properties of the gradient as we shall
see later.

Now, to obtain an expression for $\nabla f$, we use an
alternative form for the first differential using the partial derivatives
$\partial f/\partial z$, $\partial f/\partial \bar z$, as was done
in~\shortcite{brandwood83complex},
\begin{equation}
  \label{eq:63}
  \mathrm{d}f =
  \left(
    \nabla_{z}f
  \right) ^{T}\mathrm{d} z +
  \left(
    \nabla_{\bar z}f
  \right) ^{T}\mathrm{d} \bar z \,,
\end{equation}
where
\begin{equation}
  \label{eq:60}
  \nabla_{z}f =
  \begin{bmatrix}
    \frac{\partial f}{\partial z_{1}}\\
    \frac{\partial f}{\partial z_{2}}\\
    \vdots\\
    \frac{\partial f}{\partial z_{n}}\\
  \end{bmatrix} \,, 
  \qquad
  \nabla_{\bar{z}}f =
  \begin{bmatrix}
    \frac{\partial f}{\partial \bar{z}_{1}}\\
    \frac{\partial f}{\partial \bar{z}_{2}}\\
    \vdots\\
    \frac{\partial f}{\partial \bar{z}_{n}}\\
  \end{bmatrix} \,. 
\end{equation}
That is, the function $f(z)$ is assumed to be a function of two
\textit{independent} vectors $z$, $\bar z$. From
Equation~\eqref{eq:61}, it is obvious that
\begin{equation}
  \label{eq:64}
  \overline{\nabla_{z}f} =  \nabla_{\bar{z}}f \,. 
\end{equation}
Therefore, from Equation~\eqref{eq:63} we obtain
\begin{equation}
  \label{eq:65}
  \begin{split}
    \mathrm{d}f
    & =
    \left(
      \nabla_{z}f
    \right) ^{T}\mathrm{d} z +
    \left(
      \nabla_{\bar z}f
    \right) ^{T}\mathrm{d} \bar z\\
    & =  \left(
      \nabla_{z}f
    \right) ^{T}\mathrm{d} z +
    \overline{\left(
        \nabla_{z}f
      \right)} ^{T}\mathrm{d} \bar z\\
    & = \left(
      \nabla_{z}f
    \right) ^{T}\mathrm{d} z +
    \overline{\left(
        \nabla_{ z}f
      \right) ^{T}\mathrm{d}  z}\\
    & = 2 \Re \left( \left(
      \nabla_{z}f
    \right) ^{T}\mathrm{d} z \right) \\
    & = \Re \left\langle
    2\nabla_{\bar z}f, \mathrm{d} z \right\rangle \,.
  \end{split}
\end{equation}
Hence, according to our definition in Equation~\eqref{eq:58}, the
gradient of $f$ reads
\begin{equation}
  \label{eq:59}
  \nabla f(z) = 2 \nabla_{\bar{z}}f \,.
\end{equation}
Note that~\citeauthor{brandwood83complex} in his
paper~\citeyear{brandwood83complex} arrived at a 
slightly different definition: $\nabla f(z) = \nabla_{\bar{z}}f$,
which is incorrect. However, being different by only the factor
of two, it works in many situations because most algorithm use the
gradient direction only to perform a line search, while its length is used
exclusively as a termination criterion. 

The following two theorems from~\shortcite{brandwood83complex} prove
that  both definitions are consistent with the main gradient
properties used in optimization: (a) the gradient defines the
direction of maximal ascent, and (b) the gradient being zero is a
necessary and sufficient condition to determine a stationary point of
$f(z)$.

\begin{thm}
  \label{thm:optimizatoin-stationary}
  Let $f:\mathbb C^{n} \mapsto \mathbb R$ be a real-valued scalar
  function of a complex vector $z$. Let $f(z) = h(z, \bar{z})$, where
  $h: \mathbb C \times \mathbb C \mapsto \mathbb R$ is a real-valued
  scalar function of two complex vector variables and $h$ is analytic
  with respect to $z_{i}$ and $\bar{z}_{i}$. Then either of the
  conditions $\nabla_{z}h=0$ or $\nabla_{\bar{z}}h=0$ is necessary and
  sufficient to determine a stationary point of $f$.
\end{thm}
\begin{proof}[Proof]
  We can always express $f$ as a function of $2n$ real variables
  $x_{k}$, and $y_{k}$, by using $z_{k}=x_{k}+jy_{k}$:  $f(z) =
  u(x,y)$. Therefore, $u(x,y)$ (and hence $f(z)$) is stationary if,
  and only if, $\partial u/\partial x_{k} = \partial
  u/\partial y_{k} = 0$ for all $k$. From Equation~\eqref{eq:61} we immediately
  conclude that
  \begin{equation}
    \label{eq:62}
    \begin{split}
      \frac{\partial u}{\partial x_{k}} =  \frac{\partial u}{\partial y_{k}}
      = 0 
      & \Leftrightarrow \frac{\partial h}{\partial z_{i}}\\
      \frac{\partial u}{\partial x_{k}} =  \frac{\partial u}{\partial y_{k}}
      = 0 
      & \Leftrightarrow \frac{\partial h}{\partial \bar z_{k}}
    \end{split}
  \end{equation}
  Hence, $f(z)$ has a stationary point if and only if $\nabla_{z}(f) =
  0$. Similarly $\nabla_{\bar z}f=0$ is also necessary and sufficient
  to determine a stationary point of $f(z)$.
\end{proof}

\begin{thm}
  \label{thm:optimizatoin-steepestdir}
  Let $f(z)$ and $h(z, \bar z)$ be two functions as defined in
  Theorem~\ref{thm:optimizatoin-stationary}, then the gradient $\nabla
  f \equiv 2\nabla_{\bar z}f$ defines the direction of the maximal
  rate of change of $f$ with $z$.
\end{thm}
\begin{proof}[Proof]
  Consider Equations~\eqref{eq:58} and~\eqref{eq:65} that define the
  gradient $\nabla f$. Obviously
  \begin{equation}
    \label{eq:70}
    |\mathrm{d} f| =
    \left|
      \Re
      \left\langle
        2\nabla_{\bar z}f, \mathrm{d} z
      \right\rangle
    \right|
    \leq
    \left|
      \left\langle
        2\nabla_{\bar z}f, \mathrm{d} z
      \right\rangle
    \right| \,. 
  \end{equation}
\end{proof}
Furthermore, according to the Cauchy-Schwarz inequality
\begin{equation}
  \label{eq:71}
  \left|
      \left\langle
        2\nabla_{\bar z}f, \mathrm{d} z
      \right\rangle
    \right|
    \leq
    \| 2\nabla_{\bar z}f \|\, \| \mathrm{d} z \| \,.
\end{equation}
It is easy to verify that the equality in Equations~\eqref{eq:70}
and~\eqref{eq:71} holds if, and only if, $\nabla_{\bar z}f=\alpha
\mathrm{d} z$ for some real positive scalar $\alpha$.

Note that the result of Theorem~\ref{thm:optimizatoin-steepestdir}
follows from our \textit{definition} of the gradient via the first
differential: $\mathrm{d}f = \Re \langle \nabla f, \mathrm{d}z
\rangle$, and not from its particular formula.
% In fact, it can be used
% to prove Theorem~\ref{thm:optimizatoin-stationary} too.

\section{Complex Hessian}
\label{sec:complex-hessian}
The Hessian can also be obtained by treating the function $f:\mathbb
C^{n}\mapsto \mathbb R$ as a
function of $2n$ real variables that are the real and the imaginary part
of $f$'s complex argument $f(z) = u(x,y)$. Then, using
Equation~\eqref{eq:61}, one can express it as partial derivatives with
respect to $z$ and
$\bar z$. However, this time the order of variables is more
important. For example, in~\shortcite{van_den_bos94complex}, the author
defines the following two vectors $v\in\mathbb C^{2n}$, $w\in\mathbb R^{2n}$
\begin{equation}
  \label{eq:72}
  v = 
  \begin{bmatrix}
     z_{1}\\
     \bar{z}_{1}\\
      z_{2}\\
      \bar{z}_{2}\\
     \vdots\\
     z_{n}\\
     \bar{z}_{n}
   \end{bmatrix}, \qquad
   w = 
   \begin{bmatrix}
      x_{1}\\
      y_{1}\\
      x_{2}\\
      y_{2}\\
     \vdots\\
     x_{n}\\
     y_{n}
   \end{bmatrix} \,.
\end{equation}
Using this definition, and the fact that
\begin{equation}
  \label{eq:74}
  \begin{bmatrix}
    z_{n}\\
    \bar{z}_{n}
  \end{bmatrix}
  =
  \begin{pmatrix*}[r]
    1 & j \\
    1 & -j
  \end{pmatrix*}
  \begin{bmatrix}
    x_{k}\\
    y_{k}
  \end{bmatrix} \,, 
\end{equation}
we immediately obtain 
\begin{equation}
  \label{eq:75}
  v = Aw \,, 
\end{equation}
where $A$ is a block-diagonal matrix
\begin{equation}
  \label{eq:76}
  A = \mathrm{diag}
  \left(
    \begin{pmatrix*}[r]
      1 & j \\
      1 & -j
    \end{pmatrix*}
  \right) \,. 
\end{equation}
Hence~\citeauthor{van_den_bos94complex} easily concludes that
\begin{equation}
  \label{eq:73}
  \nabla^{2}_{w}f =
  A^{*}
  \left(
    \nabla^{2}_{v}f
  \right) A \,. 
\end{equation}
Furthermore, by using $A^{-1} = \frac{1}{2}A^{*}$, the relation can be
reversed
\begin{equation}
  \label{eq:77}
  \nabla^{2}_{v}f =
  \frac{1}{4}
  A
  \left(
    \nabla^{2}_{w}f
  \right) A^{*} \,. 
\end{equation}

However, we are not interested in the Hessian \textit{per se} because
we specifically aim for large-scale problems. Our
goal is to find an expression for the Hessian-vector product.  To this
end we consider the first differential of the gradient (again, by
treating $z$ and $\bar{z}$ as independent variables),
\begin{equation}
  \label{eq:78}
  \mathrm{d} (\nabla f) =  \left(\nabla^{2}f\right) \mathrm{d} z =
    \left(
      \nabla_{z}(\nabla f)
    \right) \mathrm{d} z +
    \left(
      \nabla_{\bar z}(\nabla f)
    \right) \mathrm{d} \bar z \,. 
\end{equation}
Hence, multiplying a vector $a$ with the Hessian $\nabla^{2}f$ reads
\begin{equation}
  \label{eq:79}
  \left(
    \nabla^{2}f
  \right) a =
  \left(
    \nabla_{z}(\nabla f)
  \right)  a +
  \left(
    \nabla_{\bar z}(\nabla f)
  \right)  \bar a \,. 
\end{equation}
In the next section we will see how to apply this formula to an
objective function associated with phase retrieval.
\section{Application to the  phase retrieval problem}
\label{sec:appl-phase-retr}

To exemplify our development with application to the phase retrieval
problem, let us use the following objective function
\begin{equation}
  \label{eq:math-1}
  E(s) = \frac{1}{2}\| |\hat{s}| - r\|^{2}\,,
\end{equation}
where $\hat{s}$ denotes the Fourier transform of a signal $s$, $r$
denotes the measured  magnitude of the Fourier transform, and $\|\cdot\|$
denotes the standard $l_{2}$ vector norm. Note, that $s$ and $r$ are not
necessarily one-dimensional vectors, hence, strictly speaking, the
$l_{2}$ norm is not properly defined in all cases. A proper notation would
be
\begin{equation}
  \label{eq:math-2}
  E(s) = \|\vect(|\hat{s}| - r)\|^{2}\,, 
\end{equation}
where the operator $\vect(\cdot)$ is a simple
rearrangement of a multidimensional argument into a column vector in
some predefined order. For example, let $s$ be a two-dimensional
$m\times n$ signal (matrix) with  $s_{i}$ being its $i$-th column.
Then, $\vect(s)$ is an $mn\times 1$ vector:
\begin{equation}
  \label{eq:math-3}
  \vect(s) =
  \begin{bmatrix}
    s_{1} \\
    s_{2} \\
    \vdots\\
    s_{n}
  \end{bmatrix} \,.
\end{equation}
Thus, in our convention the $\vect()$ operator transforms a matrix into a column
vector by stacking the matrix columns. Of course,
this operator is defined for signals of arbitrary (finite) dimensionality.
For the sake of brevity, hereinafter we shall use $s$ and $\vect(s)$
interchangeably and the appropriate form should be clear from the context.
Let us now review the objective function defined by
Equation~\eqref{eq:math-1}
\begin{equation}
  \label{eq:math-4}
  \begin{split}
    E(x)
    & = \frac{1}{2} \| |\hat{s}| - r \|^{2}\\
    & = \frac{1}{2} \| |\mathcal{F}[s]| - r \|^{2}\\
    & = \frac{1}{2} \| |Fs| - r \|^{2} \,.
  \end{split}
\end{equation}
Here, $\mathcal{F}[s]$ denotes the Discrete Fourier Transform (DFT)
operator applied to a (multidimensional) signal $s$, and $F$
represents the corresponding matrix, in the sense that
\begin{equation}
  \label{eq:math-5}
  \vect(\mathcal{F}[s]) = F\vect(s) \,. 
\end{equation}
We introduce the DFT matrix $F$ just for mathematical notation. In
practice, however, the DFT transform is performed by the Fast Fourier
Transform (FFT) algorithm that never creates this matrix. Note also
that $Fs$ means, actually, $F\vect(s)$, however, the shorter notation
is used, as we mentioned earlier. Consider now the final form of the
objective function we obtained in Equation~\eqref{eq:math-4}---it can
be viewed as a non-linear function of a complex argument $z\equiv
Fs$.%\footnote{We depart from our usual notation $Fs = \hat{s}$ to
  % stress that $z$ is complex-valued.}
\begin{equation}
  \label{eq:math-6}
  E(s) = \frac{1}{2} \| |Fs| - r\|^{2} = f(Fs) = f(z)\,.
\end{equation}
Hence,  $f:\mathbb C^{n} \mapsto \mathbb R$, where $n$ is the number
of elements in $Fs$ (equal to that in $s$, of course).  With the
theory developed in the previous section we can now find the gradient of
our objective function.
\begin{equation}
  \label{eq:math-10}
  \begin{split}
    \mathrm{d} (E(s))
    & = \mathrm{d} f(z)\\
    & = \Re \langle \nabla f, \mathrm{d} z \rangle\\
    & = \Re \langle \nabla f, \mathrm{d}(Fs) \rangle\\
    & = \Re \langle \nabla f, F \mathrm{d} s \rangle\\
    & = \Re \langle F^{*} \nabla f, \mathrm{d} s \rangle\\
  \end{split} \,. 
\end{equation}
Hence, using our definition we obtain
\begin{equation}
  \label{eq:math-11}
  \nabla E(s) = F^{*}\nabla f(z) \,,
\end{equation}
where $F^{*}$ denotes the Hermitian (conjugate) transpose of $F$. Now,
by using  Equation~\eqref{eq:math-4}, we have
\begin{equation}
  \label{eq:66}
  f(z) = \frac{1}{2} \| |z| - r\|^{2} \,.
\end{equation}
From which we obtain
\begin{equation}
  \label{eq:67}
  \begin{split}
    \nabla f
    & = 2 \nabla_{\bar z} f\\
    & = 2 \nabla_{\bar z} \left(\frac{1}{2} \| |z| - r\|^{2}\right)\\
    & = 2 (|z| - r) \circ \nabla_{\bar z} \left(|z|\right)\\
    & =  2 (|z| - r) \circ \nabla_{\bar z} \left((z\bar
      z)^{\frac{1}{2}}\right)\\
    & = 2 (|z| - r) \circ z^{\frac{1}{2}} \circ
    \frac{1}{2}\bar{z}^{\frac{1}{2}}\\
    & = \left(z - r\circ
      \frac{z^{\frac{1}{2}}}{\bar{z}^{\frac{1}{2}}}\right)\\
    & =
    \left(
      z - r\circ \frac{z}{|z|}
    \right) \,. 
  \end{split}
\end{equation}
Here $\circ$ denotes the element-wise (Hadamard) product. Moreover,
note that $r$, and $z$ are vectors, therefore quotients, and
exponents, like $z/|z|$, and $z^{\frac{1}{2}}$, are assumed to be
performed element-wise. This minor abuse of notation improves
readability, therefore we use it instead of introducing some special
notation. Substituting the above result into
Equation~\eqref{eq:math-11} we obtain (by using $z=Fs$)
\begin{equation}
  \label{eq:68}
  \begin{split}
    \nabla E(s)
    & = F^{*}\nabla f(z)\\
    & = F^{*} \left(
      z - r\circ \frac{z}{|z|}
    \right) \\
    & = F^{*} \left(
      Fs - r\circ \frac{Fs}{|Fs|}
    \right) \\
    & = s - F^{-1}
    \left(
      r\circ \frac{Fs}{|Fs|}
    \right)\,. 
  \end{split}
\end{equation}
In this derivation we used the fact that $F$ is unitary,
therefore $F^{-1} = F^{*}$. The expression for $\nabla E(s)$
is remarkable because it bears a clear physical meaning, which will be
discussed in the following chapters. Meanwhile we proceed
with developments required for our optimization approach.

We already have the gradient of our objective function. Hence, we can
deploy a variety of powerful optimization routines, such as
Quasi-Newton methods. However, our choice should be limited to those
that do not form a full approximation to the Hessian matrix because
typical signals may easily contain $10^{6}\text{--}10^{9}$ elements
which renders the problem of Hessian storage too costly for a
typical computer. In the following chapters we use the excellent
Quasi-Newtonian method called L-BFGS, which uses limited memory to
store an approximation to the Hessian matrix~\cite{liu89limited}.
However, to also allow for more powerful optimization methods we shall
consider the second derivatives of the objective function. Our main
goal is to devise the Hessian-vector product formula that is used in
many large scale optimization methods, for example, in the Conjugate
Gradients (CG) method~\cite{hestenes52methods}, and in the Sequential
Subspace Optimization (SESOP) method~\shortcite{narkiss05sequential}.
To this end we consider the first differential of the gradient $\nabla
E(s)$
\begin{equation}
  \label{eq:math-12}
  \begin{split}
    \mathrm{d}(\nabla E)
    & = \mathrm{d} \left( F^{*}\nabla f \right)\\
    & = F^{*} \mathrm{d} \left( \nabla f \right)\\
    & = F^{*} \left( \nabla^{2}f \right) \mathrm{d} z \\
    & = F^{*} \left( \nabla^{2}f \right) \mathrm{d} (Fs) \\
    & = F^{*} \left( \nabla^{2}f\right) F \mathrm{d} s \,. 
  \end{split}
\end{equation}
Hence, according to the definition of the Hessian we get
\begin{equation}
  \label{eq:math-13}
  \nabla^{2}E(s) = F^{*}\left(\nabla^{2}f(z)\right)F \,.
\end{equation}
Recall that the Hessian $\nabla^{2}f(z)$ has not been defined,
instead we focus on the Hessian-vector product. Based on our
development we can compute $(\nabla^{2}E(s))a$ for any vector $a$
\begin{equation}
  \label{eq:80}
  (\nabla^{2}E(s))a = F^{-1}\left(\nabla^{2}f(z)\right)Fa =
  F^{*}\left[(\nabla^{2}f(z)) (Fa)\right]\,. 
\end{equation}
The brackets in the last expression are added to emphasize the order of
efficient computation: first, the Fourier transform $\hat{a}=Fa$ is computed;
second, the Hessian-vector product $(\nabla^{2}f(z))\hat{a}$ is
computed (described below); finally, the result undergoes an inverse
Fourier transform. It is important to note that the first and the
third steps in the above calculation are independent of the objective
function and only the second step has this dependence. Let us now
devise the formula for the Hessian-vector product
$(\nabla^{2}f(z))a$. Using Equation~\eqref{eq:79} we have
\begin{equation}
  \label{eq:81}
  \begin{split}
   \left(
    \nabla^{2}f
  \right) a
  & =
  \left(
    \nabla_{z}(\nabla f)
  \right)  a +
  \left(
    \nabla_{\bar z}(\nabla f)
  \right)  \bar a \\
  & = \left(\nabla_{z} \left(
      z -  r\circ z^\frac{1}{2} \circ \bar{z}^{-\frac{1}{2}}
    \right)\right) a
  +  \left(\nabla_{\bar z} \left(
      z -  r\circ z^\frac{1}{2} \circ \bar{z}^{-\frac{1}{2}}
    \right)\right) \bar a \\
  & = \mathrm{diag}
  \left(
    1 - \frac{r}{2|z|}
  \right) a
  +
  \mathrm{diag}
  \left(
    \frac{r\circ z^{2}}{2|z|^{3}} 
  \right) \bar a \\
  & = \left(
    1 - \frac{r}{2|z|}
  \right)\circ a
  +
  \left(
    \frac{r\circ z^{2}}{2|z|^{3}} 
  \right) \circ \bar a\,.
  \end{split}
\end{equation}
Note that we again use the quotient and exponent, like $r/z$ and
$z^{2}$ in the element-wise manner.

\subsection{Special properties}
\label{sec:special-properties}

Let us consider some mathematical properties of the gradient and the
Hessian of our objective function. First, let us look at the equation
that defines the Newton direction $d$
\begin{equation}
  \label{eq:69}
  \left(
    \nabla^{2}E
  \right) d
  = -\nabla E \,. 
\end{equation}
Even if we assume that the Hessian $\nabla^{2}E$ is
invertible, finding $d$ is not straightforward as we do not form
$\nabla^{2}E$ explicitly. Fortunately, the Hessian-vector product
routine is sufficient. For example, we can use the CG method to find
$d$. However, this will require a fair amount of iterations. To find a
better (faster) way, let us consider the product $(\nabla^{2}E) \nabla
E$
\begin{equation}
  \label{eq:82}
  \begin{split}
    (\nabla^{2}E)\nabla E
    & = F^{*}\left(\nabla^{2}f(z)\right)F\nabla E\\
    & = F^{*}\left(\nabla^{2}f(z)\right)FF^{*}\nabla f\\
    & = F^{*}\left(\nabla^{2}f(z)\right)\nabla f\\
    & = F^{*}\left(\nabla^{2}f(z)\right) \left( z - r\circ
      \frac{z}{|z|}\right)\\
    & = F^{*}
    \left(
      \left(
        1 - \frac{r}{2|z|}
      \right)
      \circ
      \left(
          z - r\circ
          \frac{z}{|z|}
        \right)
      + 
      \frac{r\circ z^{2}}{2|z|^{3}} 
      \circ
      \overline{
        \left(
          z - r\circ
          \frac{z}{|z|}
        \right)
      }
    \right) \\
    & = F^{*}
    \left(
      z - \frac{r\circ z}{2|z|} - \frac{r\circ z}{|z|} +
      \frac{r^{2}\circ z}{2|z|^{2}} + \frac{r\circ z}{2|z|} -
      \frac{r^{2}\circ z}{2|z|^{2}}
    \right) \\
    & = F^{*}
    \left(
      z - \frac{r\circ z}{|z|}
    \right) \\
    & = \nabla E \,. 
  \end{split}
\end{equation}
% Hence, we got
% \begin{equation}
%   \label{eq:83}
%   (\nabla^{2}E)\nabla E = \nabla E \,. 
% \end{equation}
Namely, the gradient $\nabla E$ is an eigenvector of the Hessian
$\nabla^{2}E$ with the corresponding eigenvalue equal to one. This
means that $-\nabla E$ is the Newton step. That is, the
gradient descent method is equivalent to the Newton method in this
case.  Let us consider a single gradient descent (Newton) step with
unit step-length
\begin{equation}
  \label{eq:84}
  s - \nabla E(s) = F^{-1}
  \left(
    r\circ \frac{Fs}{|Fs|}
  \right) \,. 
\end{equation}
Consider the above result from a physical point of view: the current
signal estimate $s$ undergoes the Fourier transform $Fs$, then the
(generally incorrect) magnitude $|Fs|$ is replaced with the correct one $r$,
and the resulting signal is inverse transformed by $F^{-1}$. This is
exactly the projection step that we saw in
Chapter~\ref{cha:curr-reconstr-meth}. Is then  a single gradient descent
step enough to solve the phase retrieval problem? The answer is yes,
though the result is usually meaningless because it does not satisfy
additional constraints that are usually imposed on the sought signal,
for example, support information. The relation between the projection
and the gradient descent has long been known
(see~\cite{fienup82phase}), however, the relation to the Newton method
is new, to the best of our knowledge.

We have found one eigenvalue (1) and eigenvector ($\nabla E$) of the
Hessian.  We may get even deeper insight into the problem if we
look at the eigendecomposition of the Hessian. To this end we need the
full Hessian matrix. It can be obtained, using our Hessian-vector
product,  for the real-valued
case, that is $s\in\mathbb R^{n}$. Consider the Hessian-vector product for some real vector $t$
\begin{equation}
  \label{eq:86}
  \begin{split}
    (\nabla^{2}E)t
    & = F^{*}(\nabla^{2}f)Ft\\
    & = F^{*}(\nabla^{2}f)(Ft)\\
    & = F^{*}
    \left(
      \mathrm{diag}
      \left(
        1 - \frac{r}{2|\hat{s}|}
      \right) (Ft)
      +
      \mathrm{diag}
      \left(
        \frac{r\circ \hat{s}^{2}}{2|\hat{s}|^{3}} 
      \right)
      \left(
        \overline{Ft}
      \right)
    \right)\\
    & =  F^{*}
    \left(
      \mathrm{diag}
      \left(
        1 - \frac{r}{2|\hat{s}|}
      \right) F
      +
      \mathrm{diag}
      \left(
        \frac{r\circ \hat{s}^{2}}{2|\hat{s}|^{3}} 
      \right)
      \bar F
    \right) t \\
    & =  F^{*}
    \left(
      \mathrm{diag}
      \left(
        1 - \frac{r}{2|\hat{s}|}
      \right) F
      +
      \mathrm{diag}
      \left(
        \frac{r\circ \hat{s}^{2}}{2|\hat{s}|^{3}} 
      \right)
      F^{*}
    \right) t \,.
  \end{split}
\end{equation}
Hence, we obtain
\begin{equation}
  \label{eq:87}
  \begin{split}
    \nabla^{2}E
    & =
    F^{*}\mathrm{diag}\left(1 - \frac{r}{2|\hat{s}|}\right) F
    + 
    F^{*}\mathrm{diag}\left(\frac{r\circ \hat{s}^{2}}{2|\hat{s}|^{3}}\right)F^{*}
    \\
    & = I -  F^{*}\mathrm{diag}\left(\frac{r}{2|\hat{s}|}\right) F
    +
    F^{*}\mathrm{diag}\left(\frac{r\circ \hat{s}^{2}}{2|\hat{s}|^{3}}\right)F^{*}
\end{split}
\end{equation}
In this form, it is easy to perform an eigenanalysis of the Hessian. The
main results are proven in the following two theorems.
\begin{thm}
  \label{thm:hessian-eigenvalues}
  Let $\nabla^{2} E$ be as defined in Equation~\eqref{eq:87}, where
  $\hat{s}$ represents the Fourier transform of a real signal $s$, and $r$ denotes
  the absolute value of the Fourier transform of a real signal. Then, the
  eigenvalues of the Hessian are given by
  \begin{equation}
    \label{eq:88}
    \lambda(\nabla^{2}E) = 1 - \frac{r}{2|\hat{s}|} \pm \frac{r}{2|\hat{s}|} \,.
  \end{equation}
\end{thm}
\begin{proof}[Proof]
  To prove the claim we must recall certain properties of the DFT
  matrix $F$.
  \begin{enumerate}
  \item $F$ is unitary: $F^{*}F = FF^{*}=I$.
  \item $F$ is symmetric: $F^{T}=F$.
  \item if $t$ is real, then $Ft$ is Hermitian, that is, conjugate symmetric.
  \item $F^{2}$ is a permutation matrix that ``reverses'' its
    argument. As a consequence, if $t$ is real, then
    $F^{2}(Ft)=\overline{Ft}$.
  \item Applying the Fourier transform four times results in the
    original signal, namely $F^{4}=I$. 
  \end{enumerate}
  Now, let us consider a real signal $t$, and the following matrix $A$
  obtained from it
  \begin{equation}
    \label{eq:89}
    \begin{split}
      A
      & = F^{*} \mathrm{diag}(Ft) F^{*}\\
      & = F^{*} F^4 \mathrm{diag}(Ft) F^{4} F^{*}\\
      & = F^{*} F^{2} F^{2} \mathrm{diag}(Ft)F^{2} F^{2} F^{*} \\
      & = F \mathrm{diag}(\overline{Ft})F\\
      & = \overline{\bar{F} \mathrm{diag}(Ft) \bar{F}}\\
      & = \overline{F^{*} \mathrm{diag}(Ft) F^{*}}\\
      & = \bar{A} \,. 
    \end{split}
  \end{equation}
  We find that  $A$ is symmetric. Therefore,  $A^{2}$ can be written
  as follows:
  \begin{equation}
    \label{eq:90}
    \begin{split}
      A^{2}
      & = AA\\
      & = A \bar A \\
      & = F^{*} \mathrm{diag}(Ft) F^{*}  F
      \mathrm{diag}(\overline{Ft})F\\
      & = F^{*} \mathrm{diag} (|Ft|^{2}) F \,. 
    \end{split}
  \end{equation}
  Thus, the eigenvalues of $A^{2}$ are $|Ft|^{2}$, and the
  eigenvalues of $A$ are $\pm|Ft|$. Moreover, the matrix can be written
  as
  \begin{equation}
    \label{eq:91}
    A = \sqrt{A^{2}} = F^{*} \mathrm{diag} (\pm|Ft|) F \,.  
  \end{equation}
  In this form, it is obvious that our proof is complete once we note
  that $(r\circ z^{2})/(2|z|^{3})=Ft$ for some real vector $t$. This is
  quite obvious, because $(r\circ z^{2})/(2|z|^{3})$ is Hermitian, hence
  the inverse Fourier transform will result in a real vector.
\end{proof}

Note that this proof does not tell us whether '$+$' or '$-$' should be
taken in Equation~\eqref{eq:88}, and it says nothing about the
eigenvectors of the Hessian. These questions will be addressed in
Theorem~\ref{thm:hessian-eigenanalysis}. However, this result can
already provide some interesting insights into the problem. First,
about half (see below) of the Hessian eigenvalues are unity (these
correspond to the choice of '$+$' in Equation~\eqref{eq:88}); the rest
are equal $1-r/|\hat{s}|$.  Furthermore, if an exact solution is
found, that is, if $|\hat{s}| = r$, then the eigenvalues of the
Hessian become $1$ and $0$, with multiplicities equal, respectively,
to the number of pluses and minuses in Equation~\eqref{eq:88}.  Hence,
about half of the eigenvalues will be zero at a solution, which might
make the problem quite difficult because the Hessian is highly singular
at a solution and ill-conditioned in its neighborhood. Another
observation shows that if $|\hat{s}| > r$ (the relation is taken
element-wise), then the Hessian is positive definite, which is
beneficial in optimization problem.

We next  prove a theorem that is much stronger than
Theorem~\ref{thm:hessian-eigenvalues}. This time we devise an
unambiguous expression for the Hessian eigenvalues and also 
find its eigenvectors. The complete derivation is split into the
following two theorems.
\begin{thm}
  \label{thm:hessian-eigenanalysis-1}
  Let $\hat{t}\in\mathbb C^{n}$ be the Fourier transform of a real-valued signal
  $t$ (either one- or multi-dimensional): $\hat{t}= Ft$. Let us also
  denote by $\mathcal{K} = \{k_{1}, k_{2},\ldots, k_{n_{1}}\}$ and
  $\mathcal L = \{l_{1}, l_{2}, \ldots, l_{n_{1}}\}$ the two sets of
  indices that are exchanged upon an application of the permutation
  matrix $F^{2}$, while $\mathcal{M}=\{m_{1}, m_{2},\ldots,
  m_{n_{2}}\}$, where $n_{2} = n - 2n_{1}$, 
  is the set of indices that are invariant under the
  permutation $F^{2}$. That is, if $e_{i}$ is the $i$-th column
  of the identity matrix $I_{n}$, then
  \begin{equation}
    \label{eq:92}
    \begin{split}
      F^{2}e_{k_{i}}  & =
      e_{l_{i}}, \quad \forall i\in 1,2, \ldots,  n_{1} \,. \\
      F^{2}e_{l_{i}}  & =
      e_{k_{i}}, \quad \forall i\in 1,2, \ldots,  n_{1} \,. \\
      F^{2}e_{m_{i}}  & =
      e_{m_{i}}, \quad \forall i\in 1,2, \ldots n_{2} \,. 
    \end{split}
  \end{equation}
  Then, the eigenvalues of the matrix $C =
  \mathrm{diag}(\hat{t})F^{2}$ are as follows
  \begin{equation}
    \label{eq:94}
    \begin{cases}
      \lambda_{i} = \hat{t}_{i} & \text{if } i\in \mathcal{M} \,, \\
      \lambda_{i} = |\hat{t}_{i}| & \text{if } i\in \mathcal{K} \,, \\
      \lambda_{i} = -|\hat{t}_{i}| & \text{if } i\in \mathcal{L} \,. \\
    \end{cases}
  \end{equation}
  The corresponding eigenvectors of $C$ are given by
  \begin{equation}
    \label{eq:97}
    \begin{cases}
      v_{m_{j}} = e_{m_{j}} & \text{for } j=1,2,\ldots, n_{2} \,, \\
      v_{k_{j}} = e_{k_{j}} + a_{k_{j}}e_{l_{j}} &\text{for } j = 1,2,\ldots,
      n_1 \,, \\
      v_{l_{j}} = e_{k_{j}} - a_{l_{j}}e_{l_{j}} &\text{for } j = 1,2,\ldots,
       n_1 \,,
    \end{cases}
  \end{equation}
  where
  \begin{equation}
    \label{eq:105}
    a_{k_{j}} = a_{l_{j}}
    = \frac{|\hat{t}_{k_{j}}|}{\hat{t}_{k_{j}}}
    = \frac{|\hat{t}_{l_{j}}|}{\bar{\hat{t}}_{l_{j}}} \,. 
  \end{equation}
\end{thm}
Before we proceed to the proof, note that the set $\mathcal{M}$ is
defined uniquely---it includes the zero-frequency and the
half-Nyquist
frequencies. The latter present if, and only if, the number
of samples, along some dimension, is even. The two other sets:
$\mathcal{K}$, and $\mathcal{L}$ are not unique---one can exchange
$k_{j}$ and $l_{j}$. This non-uniqueness, however, does not have any
special effect. The theorem, in fact, claims that the conjugate
symmetric signal $\hat{t}$ defines uniquely the set of eigenvalues
and eigenvectors of $C$ in the following manner. If $\hat{t}_{i}$
has no conjugate counterpart (zero frequency, or half-Nyquist
frequency) then it contributes the eigenvalue $\hat{t}_{i}$ and the
corresponding eigenvector $e_{i}$. If, on the other hand,
$\hat{t}_{i}, \quad i = k_{j}$ has a conjugate counterpart, then
the pair  $\hat{t}_{i}=\hat{t}_{k_{j}}$, and $\bar{\hat{t}}_{i} =
\hat{t}_{l_{j}}$ contributes two eigenvalues:
$|\hat{t}_{i}|$, and $-|\hat{t}_{i}|$ along with the corresponding
eigenvectors

\begin{proof}[Proof]
  Let us start with $i\in\mathcal{M}$:
  \begin{equation}
    \label{eq:98}
    \begin{split}
      Ce_{m_{j}}
      & = \mathrm{diag}(\hat{t})F^{2}e_{m_{j}}\\
      & = \mathrm{diag}(\hat{t})e_{m_{j}}\\
      & = \hat{t}_{m_{j}}e_{m_{j}} \,,
    \end{split}
  \end{equation}
  which completes the proof for this case. Let us now consider
  $i\in\mathcal{K}$, that is, $i=k_{j}$:
  \begin{equation}
    \label{eq:99}
    \begin{split}
      \lambda_{i}v_{i}
      & = \lambda_{k_{j}}v_{k_{j}} \\
      & = Cv_{k_{j}}\\
      & = \mathrm{diag}(\hat{t})F^{2}(e_{k_{j}} + a_{k_{j}} e_{l_{j}})\\
      & = \mathrm{diag}(\hat{t})(F^{2}e_{k_{j}} +
      a_{k_{j}}F^{2}e_{l_j})\\
      & =  \mathrm{diag}(\hat{t})(e_{l_{j}} + a_{k_{j}}e_{k_{j}})\\
      & = \hat{t}_{l_{j}} e_{l_{j}} +
      a_{k_{j}}\hat{t}_{k_{j}}e_{k_{j}} \\
      & = \bar{\hat{t}}_{k_{j}} e_{l_{j}} +
      a_{k_{j}}\hat{t}_{k_{j}}e_{k_{j}} \,. 
    \end{split}
  \end{equation}
  Hence, we have
  \begin{equation}
    \label{eq:100}
    \begin{split}
      \lambda_{k_{j}} & = a_{k_{j}}\hat{t}_{k_{j}} \,, \\
      \lambda_{k_{j}} a_{k_{j}} & = \bar{\hat{t}}_{k_{j}} \,,
    \end{split}
  \end{equation}
  which leads immediately to
  \begin{equation}
    \label{eq:101}
    a^{2}_{k_{j}}\hat{t}_{k_{j}} =  \bar{\hat{t}}_{k_{j}}
    \Rightarrow
    a^{2}_{k_{j}} = \frac{\bar{\hat{t}}_{k_{j}}}{\hat{t}_{k_{j}}}
    \Rightarrow
    a_{k_{j}} =  \pm \frac{|\hat{t}_{k_{j}}|}{\hat{t}_{k_{j}}}
    \,. 
  \end{equation}
  To decide upon the sign of $a_{k_{j}}$ in this equation, we use the
  fact that  $\lambda_{k_{j}}$ is positive, according to the
  definition in Equation~\eqref{eq:94}.  Now, if we look at
  Equation~\eqref{eq:100}, we immediately obtain $\lambda_{k_{j}} =
  a_{k_{j}}\hat{t}_{k_{j}} = \pm|\hat{t}_{k_{j}}|$, hence, we must choose '$+$'. The proof
  for the last case, $i\in\mathcal{L}$, is very similar.
  \begin{equation}
    \label{eq:102}
    \begin{split}
      \lambda_{i}v_{i}
      & = \lambda_{l_{j}}v_{l_{j}} \\
      & = Cv_{l_{j}}\\
      & = \mathrm{diag}(\hat{t})F^{2}(e_{k_{j}} - a_{l_{j}} e_{l_{j}})\\
      & = \mathrm{diag}(\hat{t})(F^{2}e_{k_{j}} - a_{l_{j}}F^{2}e_{l_j})\\
      & =  \mathrm{diag}(\hat{t})(e_{l_{j}} - a_{l_{j}}e_{k_{j}})\\
      & = \hat{t}_{l_{j}} e_{l_{j}} - a_{l_{j}}\hat{t}_{k_{j}}e_{k_{j}} \\
      & = \hat{t}_{l_{j}} e_{l_{j}} - a_{l_{j}}\bar{\hat{t}}_{l_{j}}e_{k_{j}} \,.  
    \end{split}
  \end{equation}
  Hence, we have
  \begin{equation}
    \label{eq:103}
    \begin{split}
      \lambda_{l_{j}} & = -a_{l_{j}}\bar{\hat{t}}_{l_{j}} \,, \\
      -\lambda_{l_{j}} a_{l_{j}} & = \hat{t}_{l_{j}} \,,
    \end{split}
  \end{equation}
  which gives us
  \begin{equation}
    \label{eq:104}
    a_{l_{j}} = \pm\frac{|\hat{t}_{l_{j}}|}{\bar{\hat{t}}_{l_{j}}}.
  \end{equation}
  Again, we must choose '$+$', because $\lambda_{l_{j}} =
  -a_{l_{j}}\bar{\hat{t}}_{l_{j}}=-|\hat t_{l_{j}}|$ must be
  negative. It is worthwhile to note that the eigenvectors $v_{i}$ are
  mutually orthogonal.
\end{proof}

Now we can prove the main result.

\begin{thm}
  \label{thm:hessian-eigenanalysis}
  Let  $\hat{s}, \hat{t}\in\mathbb C^{n}$ be Fourier
  transforms of real-valued signals $s$, and $t$, respectively.
  Let the sets of indices $\mathcal K$, $\mathcal L$, and $\mathcal
  M$; and vectors $\{v_{i}\}_{i=1}^{n}$
  be defined as in Theorem~\ref{thm:hessian-eigenanalysis-1}. Let
  $(\lambda_{i}, u_{i})$ be an eigenpair of the following matrix
  \begin{equation}
    \label{eq:93}
    A = F^{*}\mathrm{diag}(|\hat{s}|)F + F^{*}
    \mathrm{diag}(\hat{t})F^{*} \,.
  \end{equation}
  Then, the eigenvalues are given by
  \begin{equation}
    \label{eq:106}
    \begin{cases}
      \lambda_{i} = |\hat{s}|_{i} + \hat{t}_{i} \,,  & \text{if } i\in \mathcal{M} \,, \\
      \lambda_{i} = |\hat{s}|_{i} -|\hat{t}_{i}| \,, & \text{if } i\in \mathcal{K} \,, \\
      \lambda_{i} = |\hat{s}|_{i} + |\hat{t}_{i}| \,,  & \text{if } i\in \mathcal{L} \,.
    \end{cases}
  \end{equation}
  And the eigenvectors are given by
  \begin{equation}
    \label{eq:107}
    U = [u_{1}, u_{2},\ldots, u_{n}] = VF \,, 
  \end{equation}
  where
  \begin{equation}
    \label{eq:108}
    V = [v_{1}, v_{2}, \ldots, v_{n}]
  \end{equation}
\end{thm}

\begin{proof}[Proof]
  The proof is trivial once we note that
  \begin{equation}
    \label{eq:109}
    \begin{split}
      A
      & = F^{*}\mathrm{diag}(|\hat{s}|)F + F^{*}
      \mathrm{diag}(\hat{t})F^{*} \\
      & = F^{*}\mathrm{diag}(|\hat{s}|)F + F^{*}
      \mathrm{diag}(\hat{t})F^{4}F^{*}\\
      & = F^{*}\mathrm{diag}(|\hat{s}|)F + F^{*}
      \mathrm{diag}(\hat{t})F^{4}F^{*} \\
      & = F^{*}
      \left(
        \mathrm{diag}(|\hat{s}|) +
        \mathrm{diag}(\hat{t})F^{2}
      \right) F \\
      & = F^{*} (B + C) F \,, 
    \end{split}
  \end{equation}
  where
  \begin{equation}
    \label{eq:110}
    B = \mathrm{diag}(|\hat{s}|), \quad C =
    \mathrm{diag}(\hat{t})F^{2} \,. 
  \end{equation}
  The eigendecomposition of $C$ was found in
  Theorem~\ref{thm:hessian-eigenanalysis-1}. Thus, if we show that
  $v_{i}$ is also an eigenvector of $B$, with the corresponding
  eigenvalue of $|\hat{s}|_{i}$ the theorem will be proved. Indeed, it
  is easy to verify that
  \begin{equation}
    \label{eq:111}
    \mathrm{diag}(|\hat{s}|) = |\hat{s}|_{i} v_{i}\,,
  \end{equation}
  which completes the proof.
\end{proof}

\begin{cor}
  The eigenvalues of the Hessian $\nabla^{2}E$ are given by
  \begin{equation}
    \label{eq:112}
    \begin{cases}
      \lambda_{i} = 1 \,, & \text{if } i\in\mathcal{M} \,,\\
      \lambda_{i} = 1 - \frac{r_{i}}{|\hat{x}|_{i}} \,, & \text{if }
      i\in\mathcal K \,, \\
      \lambda_{i} = 1 \,, & \text{if }
      i\in\mathcal L \,, \\
    \end{cases}
  \end{equation}
  where $\mathcal M$, $\mathcal K$, and $\mathcal L$ are set of
  indices, as defined in
  Theorem~\ref{thm:hessian-eigenanalysis-1}. The corresponding  eigenvectors 
  are given by $\{u_{i}\}_{i=1}^{n}$, as defined in
  Theorem~\ref{thm:hessian-eigenanalysis}.
\end{cor}

This analysis can potentially lead to development of new methods. For
example, because the eigenvalues of the Hessian is known, we can add a
regularization term with a weight that guarantees that the problem is
convex all the time. Although we have not developed such a method in
the course of this work, it definitely appears on our ``todo'' list.

\section{Concluding remarks (disappointing)}
\label{sec:concl-remarks-disapp}
After developing the machinery for efficient optimization methods, we
tested this apparatus on the classical phase retrieval problem for a
real non-negative two-dimensional signal with known off-support
locations ($\mathcal O$)
\begin{equation}
  \label{eq:85}
  \begin{split}
    \min_{x} & \quad \||Fx|-r\|^{2} \,, \\
    \mathrm{subject\ to} &\quad x \geq 0 \,, \\
    &\quad x_{o\in\mathcal O} = 0 \,. 
  \end{split}
\end{equation}
Unfortunately, the classical Newton-type methods failed so solve this
problem. This failure was not surprising because the phase retrieval
problem is known to be ``tough'' for continuous optimization
techniques. Let us cite a concluding excerpt from a
study~\cite{nieto-vesperinas86study}:

``\textit{The efficiency of an important class of Newton methods
(the Levenberg-Marquardt algorithm) for solving
overdetermined sets of nonlinear equations is tested in
finding the solution to the two-dimensional phase problem.
It is seen that the nonlinearity and number of local
minima of the cost function increases dramatically with
the size of the object array, making these methods of little
practical use for sizes greater than $6\times 6$\ldots}''

Obviously, straightforward application of existing methods will not
work. We hope that the eigenanalysis we performed in this chapter will
eventually lead to a development of new efficient methods. Unfortunately, we
did not do that in the framework of this research. 

\section{Concluding remarks (encouraging)}
\label{sec:concl-remarks-enco}
The failure of continuous optimization techniques when applied to the
multi-dimensional phase retrieval problem is well known and, so to
say, ``widely accepted'' amongst researchers. However, the exact
reason for this is not fully understood. In
Chapter~\ref{cha:appr-four-phase-explanation} we provide an
explanation for this failure. Moreover, in the following chapters we
demonstrate that additional information can change things
dramatically. For example, if a rough Fourier phase estimate is
available, we demonstrate (first, empirically in
Chapter~\ref{cha:appr-four-phase-first}, then theoretically in
Chapter~\ref{cha:appr-four-phase-explanation}) that continuous
optimization techniques succeed very well. Furthermore, in
Chapter~\ref{cha:appr-four-phase-explanation} we provide a rigorous
mathematical reasoning that explains why \textit{any} reasonable
method is expected to succeed in the case where the Fourier phase
uncertainty is below $\pi/2$ radians. In
Chapter~\ref{cha:bandw-extr-using} we demonstrate that sparsity prior
lets us solve successfully an even more difficult
problem---simultaneous phase retrieval and bandwidth extrapolation.

%%% Local Variables: 
%%% TeX-master: "../thesis"
%%% End: 

%% file: phase-empirical/phase-empirical.tex
\chapter{Approximate Fourier phase knowledge for non-negative
  signals---first success\footnotemark}
\label{cha:appr-four-phase-first}

\footnotetext{The material presented in this chapter was published in
  \shortcite{osherovich09fast}.}

As was mentioned in Chapter~\ref{cha:found-optim-meth}, classical
continuous optimization
methods are known to fail miserably when applied to the phase
retrieval problem. In Section~\ref{cha:appr-four-phase-explanation} we
shall give an explanation for this failure for a wide class of
methods---monotone line-search optimization algorithms. In this
chapter, however, we develop a new reconstruction method based on a
Quasi-Newton optimization algorithm for a variation of the classical
phase retrieval problem where very little additional information about the
Fourier phase is available. In many situations, this information is
readily available or can be obtained by an appropriate experimental
arrangement.

A more extended discussion on some possible ways to obtain a rough
phase estimate is delayed until
Section~\ref{sec:approx-phase1-conclusions}. Here we demonstrate that
a rough (up to $\pi$ radians) phase estimate allows us to develop a
new method whose convergence rate is several orders of magnitude
faster than that of the current reconstruction techniques. Unlike
current methods, which are based on alternating projections, our
approach is based on continuous optimization. Therefore, besides fast
convergence, our method allows a great deal of flexibility in choosing
appropriate objective functions as well as introducing additional
information or prior assumptions about the sought signal like, for
example, smoothness. The speed of
convergence is important in many applications. For example, in microscopy a
real-time algorithm would have a clear advantage. The ability to
incorporate additional information may have an even greater effect: starting
from a vast improvement in the reconstruction speed and going all the
way to the very chances of successful reconstruction.

\section{Developing an efficient optimization method}
\label{sec:approx-phase1-devel-cont-optim}

Let us start by formulating the optimization problem for real
non-negative signals. The very common formulation is as follows
\begin{equation}
  \label{eq:approx-phase1-1}
  \begin{split}
    \min_{x} &\quad \||FPx| - r\|^{2} \,,   \\
    \mathrm{subject\ to} &\quad x\geq 0 \,, 
  \end{split}
\end{equation}
where $F$ denotes the Fourier transform operator (a matrix in the discrete case),
$P$ represents zero-padding (note that any support information can be
represented as zero-padding), and $r$ denotes the measured Fourier
magnitude. 

Of
course, there is an endless number of ways to choose the objective
function. The particular choice  may affect the convergence speed
and numerical stability. However, in our view, it is more important to
choose the objective function that properly reflects the underlying
physical phenomena. For example, the choice of
Equation~\eqref{eq:approx-phase1-1} is especially suitable when the
measured quantity is $r$ and the noise in the measurements has
a (close to) Gaussian distribution with zero mean.

As we have already seen, applying a Newton-type method to the problem
in Equation~\eqref{eq:approx-phase1-1} fails. However, we have not
introduced yet the additional information available in the
setup we consider here---the approximate Fourier phase. Let us consider one
pixel in the Fourier domain.  If the phase is known to lie within
a certain interval $[\alpha,\beta]$, the correct complex number must
belong to the arc $\hat{A}\hat{B}$ defined by $\alpha$ and $\beta$ as
depicted in Figure~\ref{fig:phase-interval}. Even with this additional
information, the problem still remains non-convex and cannot directly be solved
efficiently. However, if we perform a \emph{convex relaxation}. That
is, if we relax our requirements on the Fourier modulus and let the
complex number lie in the convex region $\mathcal{C}$ defined by
$\alpha$ and $\beta$ as shown in Figure~\ref{fig:convex-area}, the
problem now becomes convex.
\begin{figure}[H]
  \centering
  \subfloat[]{\label{fig:phase-interval}
    \includegraphics[width=0.4\textwidth]{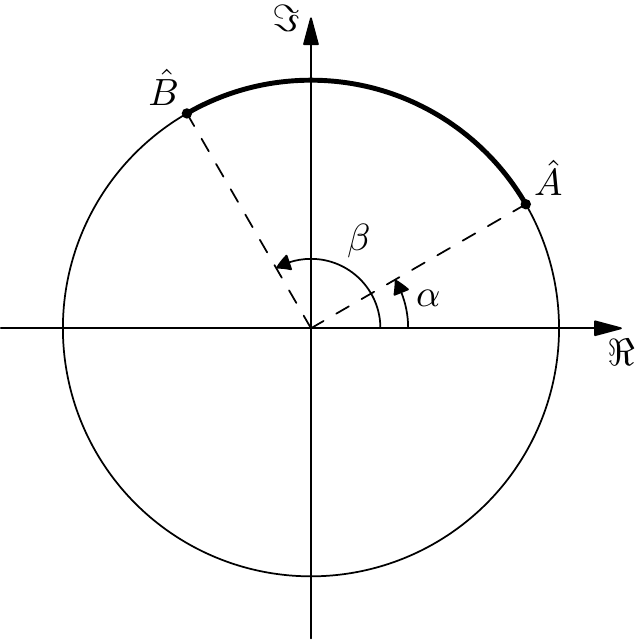}}%
  \qquad%
  \subfloat[]{\label{fig:convex-area}
    \includegraphics[width=0.4\textwidth]{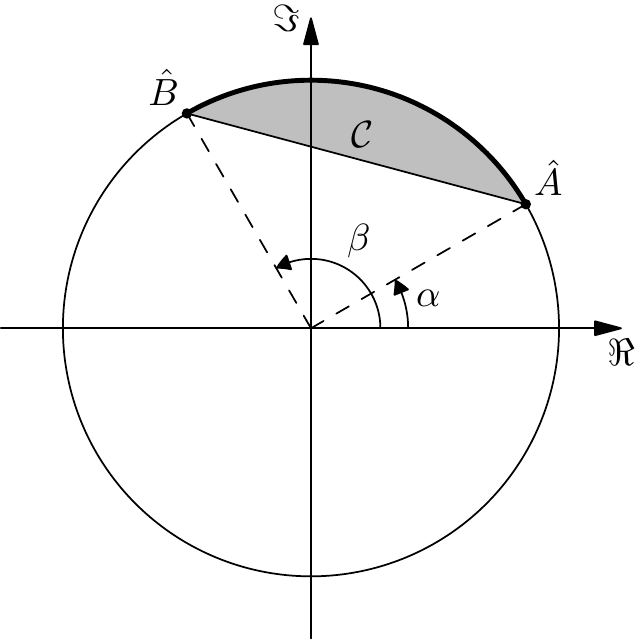}}
  \caption[Convex relaxation of phase bounds]{Convex relaxation: (a)
    the original phase uncertainty interval results in an arc of a
    circle of known radius; (b) after relaxation, the allowed region
    is convex.}
  \label{fig:convex}
\end{figure}

Note that $\mathcal{C}$ is the smallest convex region that contains
the original constraint (the arc $\hat{A}\hat{B}$). The formal
definition of the relaxed problem is as follows:
\begin{equation}
  \label{eq:approx-phase1-convex-func}
  \begin{split}
    \min &\quad d^{2}(FPx, \mathcal{C})\\
    \mathrm{subject\ to} &\quad x\geq 0
  \end{split}
\end{equation}
where $d(a,\mathcal{C})$ denotes the Euclidean distance from point $a$
to the convex set $\mathcal{C}$ (see
Definition~\ref{def:current-1}). From our experience,  a few
dozen iterations are sufficient to solve this convex problem (see
Figure~\ref{fig:stage1}). Of course, the solution does not usually match the
original image because both the phase and the magnitude may vary
significantly. However, we suggest the following method for the
solutions of the original problem.
\begin{description}
\item[Stage 1:] Starting with a random\footnote{In fact, $x^{0}$ can
    be chosen in any reasonable way, as our method is insensitive to
    the choice of the starting point.} $x^{0}$,
solve the problem defined by Equation~(\ref{eq:approx-phase1-convex-func}).
\item[Stage 2:] Use the solution obtained in Stage 1 (denoted $x^{1}$)
  as the starting point for the minimization problem that
  combines both the convex and non-convex parts, as defined below
  \begin{equation}
    \label{eq:approx-phase1-min-func}
    \begin{split}
      \min & \quad \||FPx|-r\|^{2} + d^{2}(FPx, \mathcal{C})\\
      \mathrm{subject\ to} &\quad x \geq 0
    \end{split}
  \end{equation}
\end{description}

More precisely, in our implementation we use the unconstrained
minimization formulation, that is , instead of
Equations~\eqref{eq:approx-phase1-convex-func} and
\eqref{eq:approx-phase1-min-func} we minimize the following
convex, and non-convex functionals, respectively.
\begin{align}
  \label{eq:approx-phase1-4}
  E_c(x) & = d^{2}(FPx, \mathcal{C}) + \|[x]_{-}\|^{2}\,, \\
  E(x) & = \||FPx|-r\|^{2} + \mu_{1}d^{2}(FPx, \mathcal{C}) + \mu_{2}\|[x]_{-}\|^{2}
  \,,  \label{eq:approx-phase1-5}
\end{align}
where $[x]_{-}$ is defined as follows
\begin{equation}
  \label{eq:approx-phase1-6}
  [x]_{-}=
  \begin{cases}
    0,& x\geq 0 \,. \\
    x,& x < 0 \,. 
  \end{cases}
\end{equation}
The weights $\mu_{1}$, and $\mu_{2}$ are usually set to unity.
Results of our simulations are presented in the next section.

\section{Simulations and Results}
\label{sec:approx-phase1-results}
Due to the high dimensionality of the problem (especially in the 3D
case) we  limit our choice to methods that do not require the
Hessian matrix or its approximation.  Hence, in our implementation we
use a modified version of the SESOP
algorithm~\shortcite{narkiss05sequential} and the L-BFGS
method~\shortcite{liu89limited}. Both algorithms demonstrate very similar
results. The main difference is that SESOP guarantees that
there are two Fourier transforms per iteration just like in the GS and
HIO methods. The L-BFGS method, on the other hand, cannot guarantee
that. However, in practice the average number of the Fourier
transforms per iteration is
very close to that of SESOP and HIO.

The method was tested across a variety of data. In this section we
present some of these examples. The first example is a molecule of
caffeine whose 3D model along with a 2D projection of its electron density, and
the corresponding Fourier magnitudes are shown in
Figure~\ref{fig:caffeine}. This information was obtained from a
PDB\footnote{See http://www.pdb.org for more
  information.} (protein database) file. In addition, we use a
``natural'' image which represents a class of images with rich texture
and tight support. Moreover, it may be easier to estimate the visual
reconstruction quality of such images. This image and its Fourier
modulus are given in Figure~\ref{fig:lena}.
\begin{figure}[H]
  \centering
  \subfloat[]{\label{fig:caffeine-pdb}%
    \includegraphics[height=.35\linewidth]{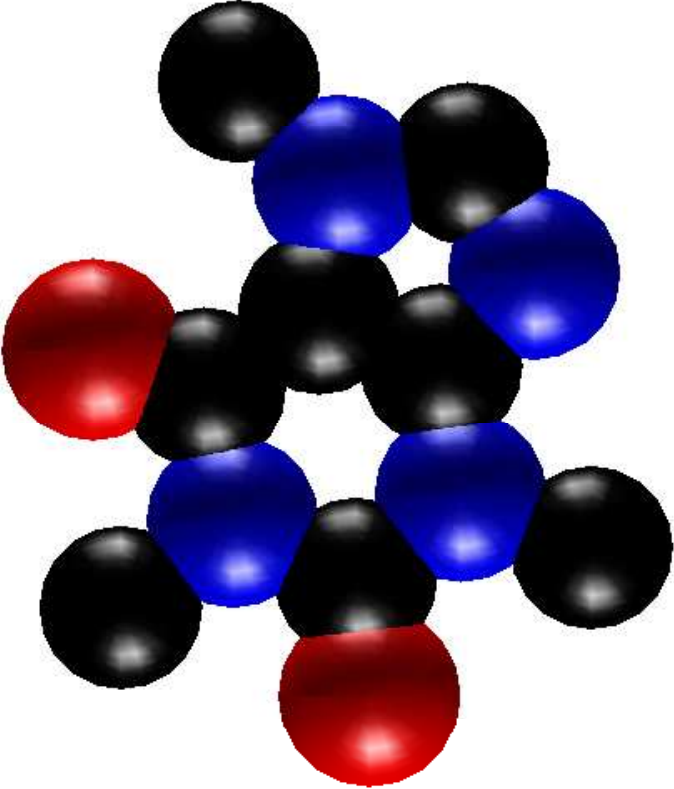}}%
  \hspace{.18\linewidth}
  \subfloat[]{\label{fig:caffeine-2drho}%
    \includegraphics[width=.35\linewidth]{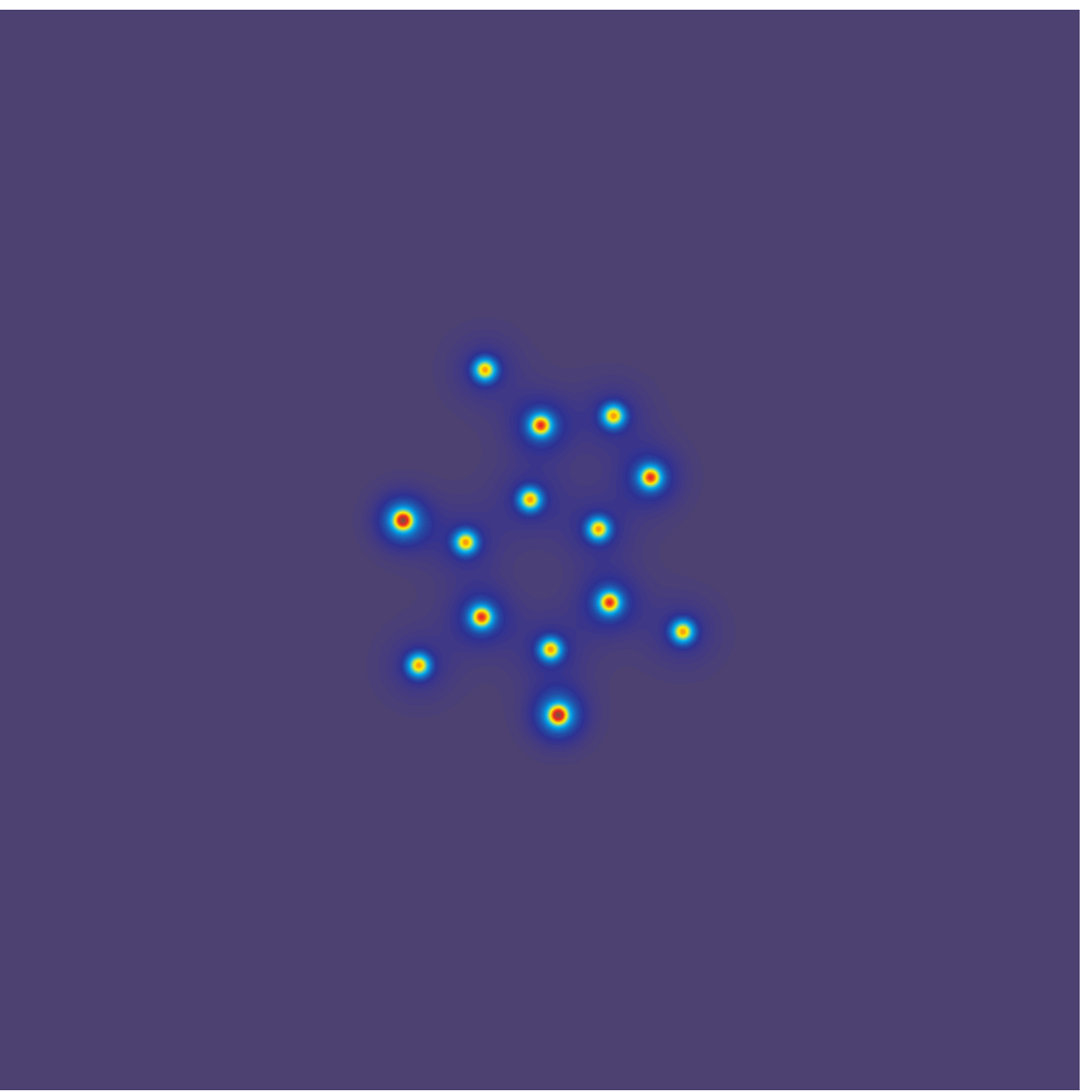}}\\
  \subfloat[]{\label{fig:caffeine-3dfourier}%
    \includegraphics[width=.35\linewidth]{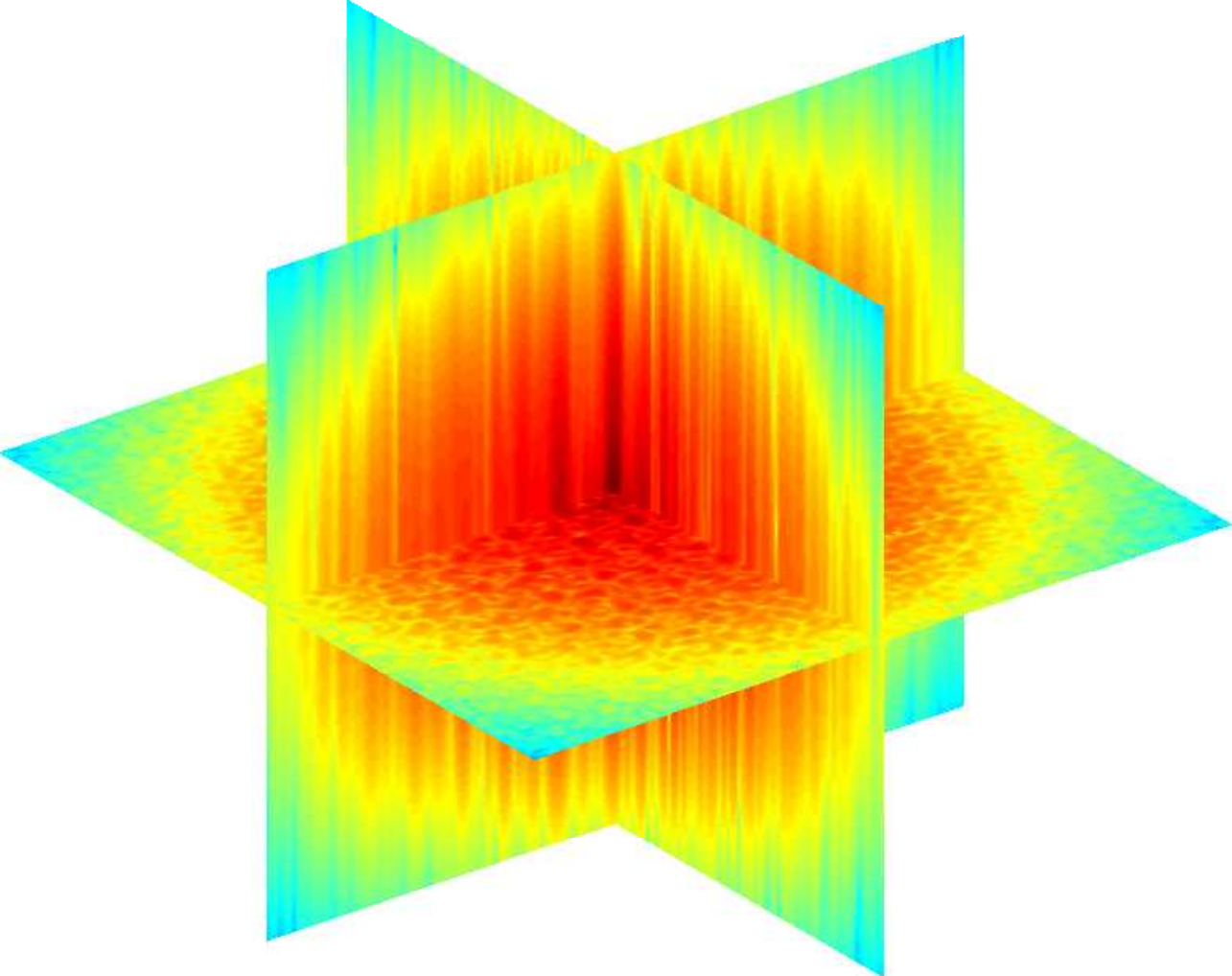}}%
  \hspace{.15\linewidth}
  \subfloat[]{\label{fig:caffeine-2dfourier}%
  \includegraphics[width=.35\linewidth]{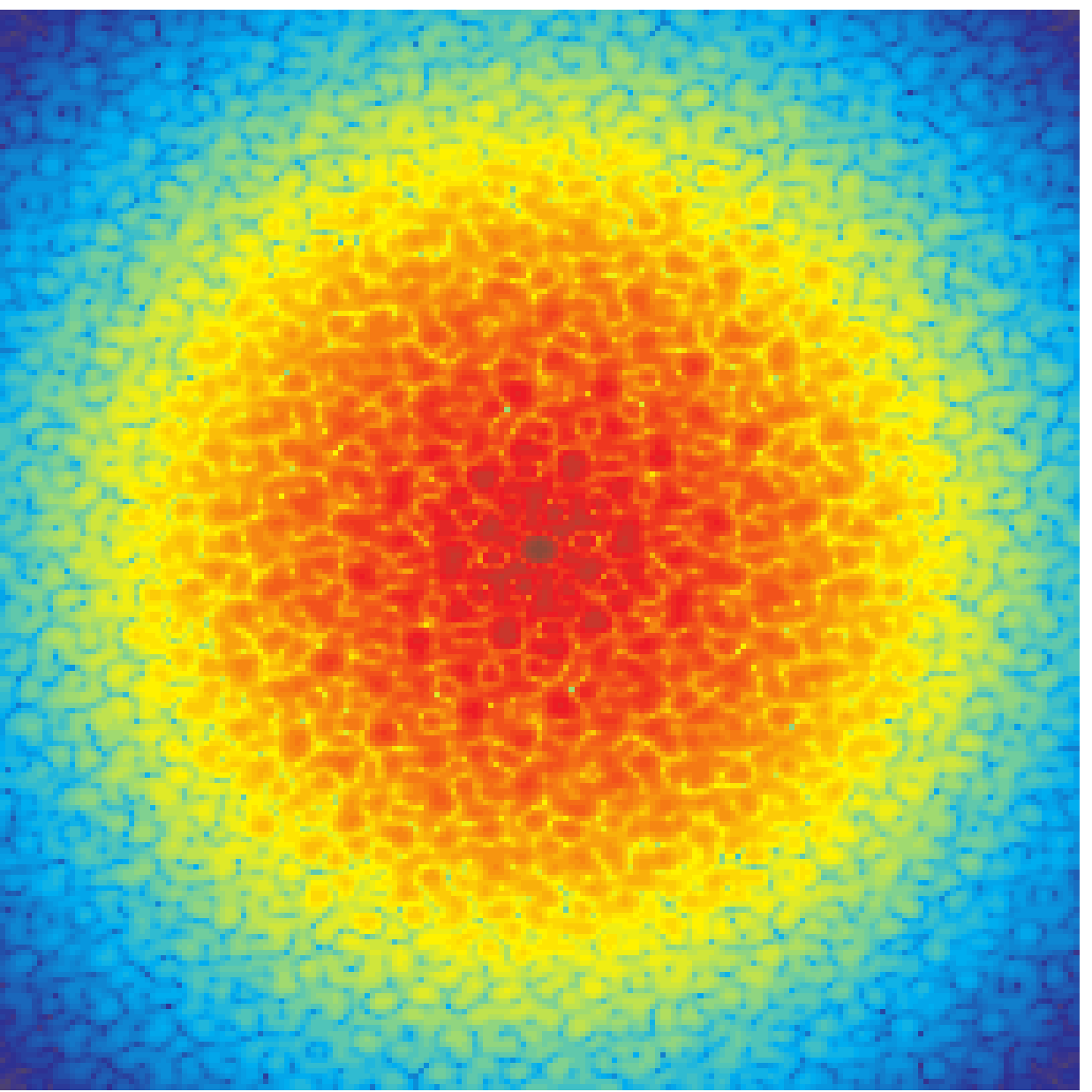}}
\caption[Caffeine molecule]{Caffeine molecule: (a) 3D model (PDB), (b) 2D projection of
  its electron density; and their corresponding Fourier magnitudes:
  (c) 3D, and (d) 2D.  }
  \label{fig:caffeine}
\end{figure}
\begin{figure}[H]
  \centering
  \subfloat[]{\label{fig:lena-image}%
    \includegraphics[width=.35\linewidth]{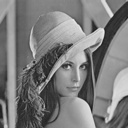}}
  \hspace{.15\linewidth}
  \subfloat[]{\label{fig:lena-fourier}%
    \includegraphics[width=.35\linewidth]{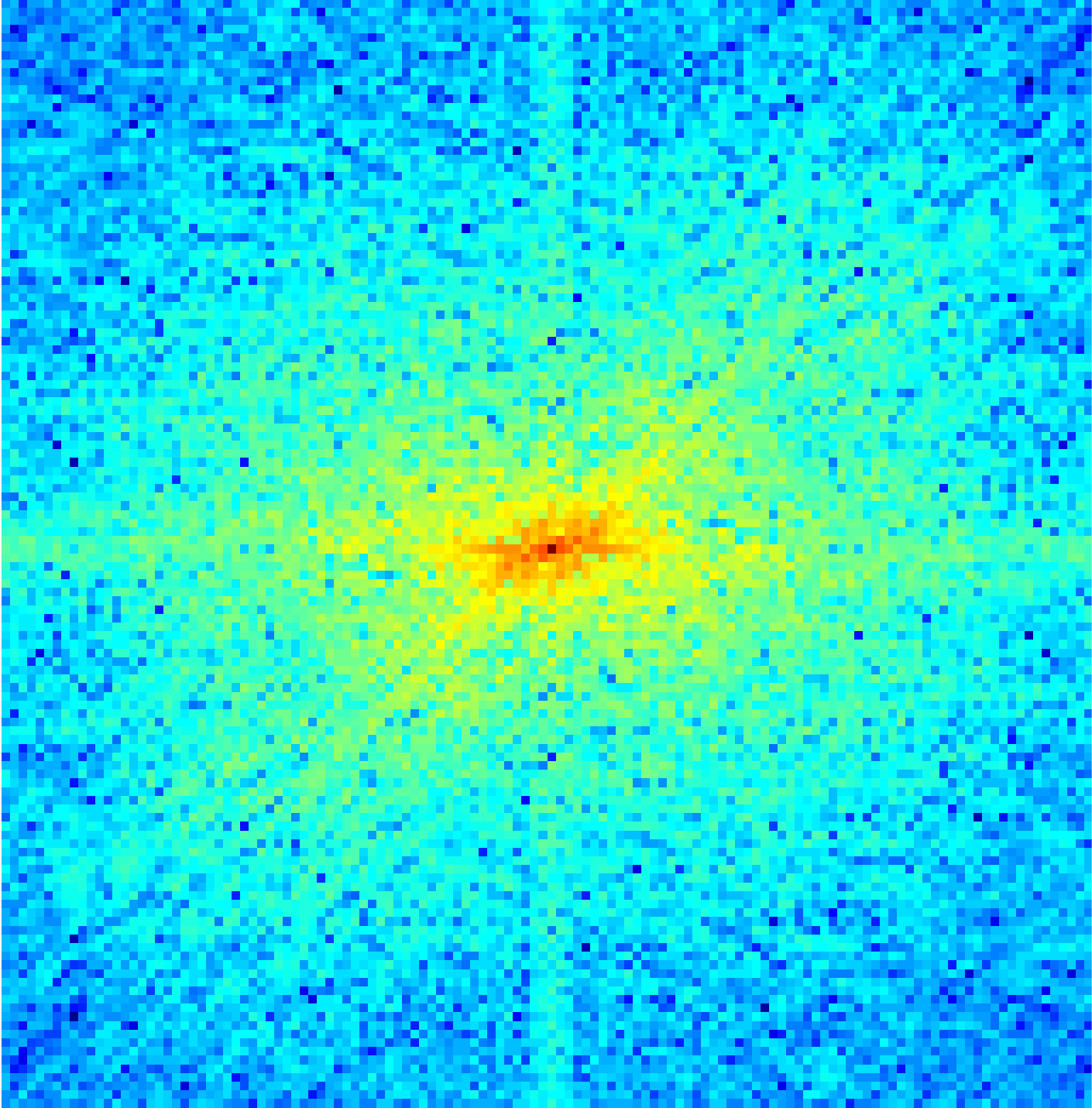}}
  \caption[A natural image (Lena)]{A natural image (Lena): (a) original image, and (b) its
    Fourier magnitude.}
\label{fig:lena}
\end{figure}
Note that we assume that a rectilinear sampling is available in the 3D
case. In practice, however, the sensors measure a two-dimensional
slice of the 3D volume. Provided that a sufficient number of such
slices were measured, an interpolation can be used to form a
rectilinear array of measurements~\shortcite{miao01approach}. However, the
slices can be incorporated directly into our minimization scheme. This
will be addressed in future work.

In our experiments we tested a phase uncertainty of up to 3 radians. The
bounds were chosen at random at every measured pixel (voxel) such that
the true phase had a uniform distribution inside the interval. The
starting point ($x^{0}$) was also chosen randomly. Of course, there is
an obvious way to make a more educated guess: by choosing the middle
of the uncertainty interval, however, this choice will generally violate the
object domain constraints. Fortunately, our experiments indicate that
the starting point has little influence on the reconstruction. In all cases the
reconstructed images obtained with our method were visually
indistinguishable from the original. Therefore, we only present the values
of $E_c(x)$ and $E(x)$ as defined in
Equations~\eqref{eq:approx-phase1-4} and~\eqref{eq:approx-phase1-5}
to visualize the progress of the first and the second stages,
respectively. The second stage is compared with the HIO algorithm for
which the error term is $E(x)$ without the phase bounds constraint,
that is,
\begin{equation}
  \label{eq:1}
  E_{\mathrm{HIO}}=\||FPx|-r\|^2 + \|[x]_{-}\|^{2} \ .
\end{equation}
The first experiment is as follows. First, we run 60 iterations of
Stage 1, that is, the convex problem defined
by~\eqref{eq:approx-phase1-convex-func}. The
progress of different images is shown in Figure~\ref{fig:stage1}. In
the second stage we run 200 iterations of our algorithm (SESOP)
starting at the solution obtained in the previous stage ($x^{1}$). To
compare the convergence rate with current methods, we ran twice
the HIO algorithm: once, starting at $x^{0}$, whereby the algorithm is
unaware of the additional phase information. Another run was started
at $x^{1}$, hence, the phase information was made (indirectly) available
to the algorithm. The results for 2D and 3D reconstruction of the
caffeine molecule are shown in Figures.~\ref{fig:stage2_caff2d}
and~\ref{fig:stage2_caff3d}, respectively. The results of the natural
image are shown in~\ref{fig:stage2_lena}.

It is evident from these results that our method significantly
outperforms the HIO algorithm is all experiments. Moreover, its
superiority for the ``Lena'' image is tremendous. 
\begin{figure}[H]
  \centering
  \subfloat[]{\label{fig:stage1}%
    \includegraphics[width=0.4\linewidth]{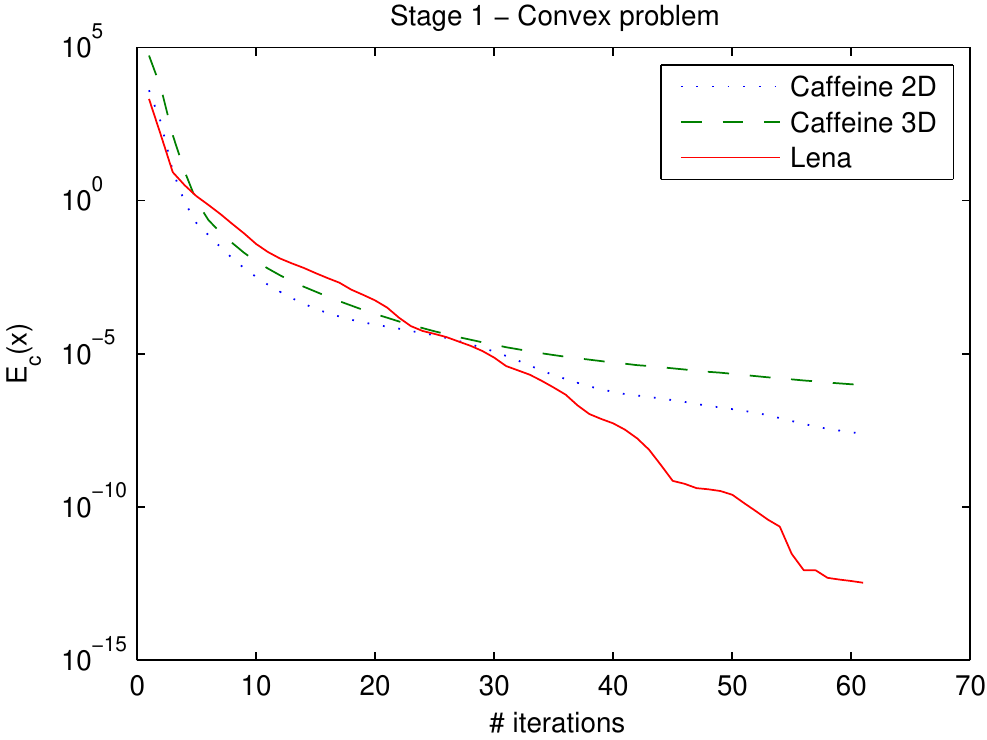}}%
  \hspace{.1\linewidth}
  \subfloat[]{\label{fig:stage2_caff2d}%
    \includegraphics[width=0.4\linewidth]{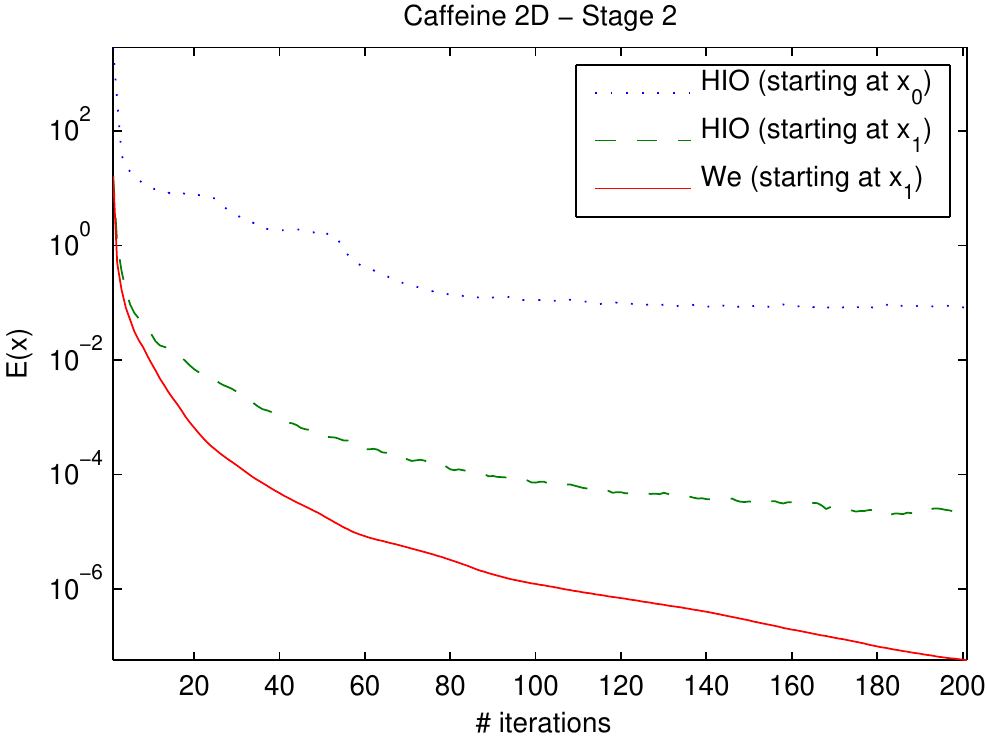}}\\
  \subfloat[]{\label{fig:stage2_caff3d}%
    \includegraphics[width=0.4\linewidth]{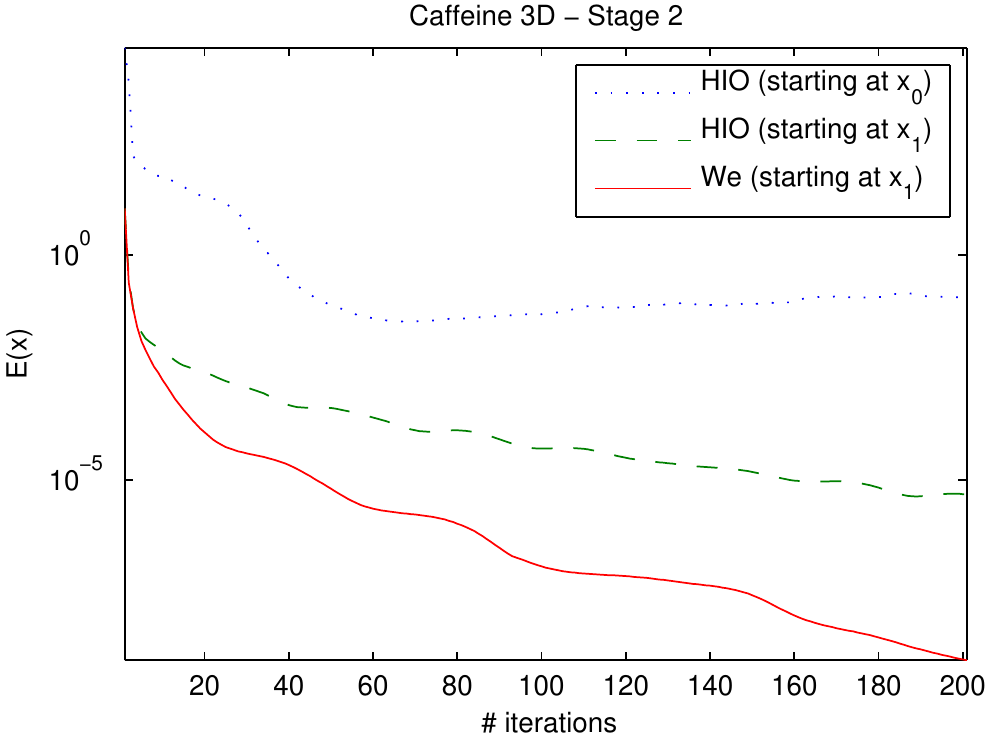}}
  \hspace{.1\linewidth}
  \subfloat[]{\label{fig:stage2_lena}%
    \includegraphics[width=0.4\linewidth]{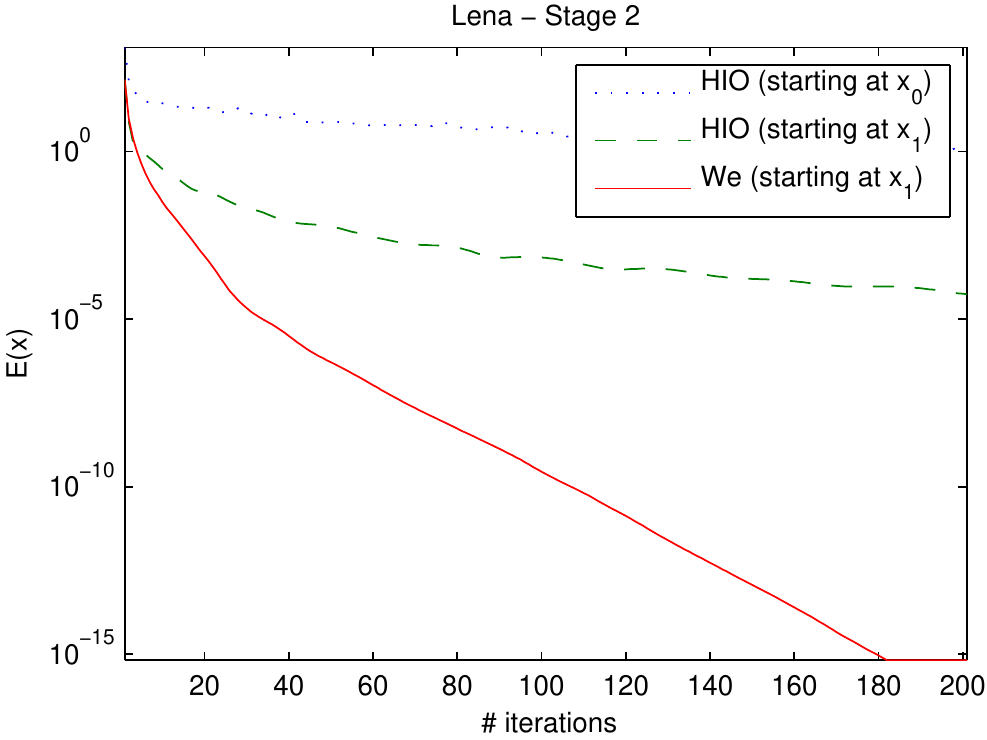}}
  \caption[Reconstruction with approximately known
  phase]{Reconstruction results: (a) stage 1 convergence rate, (b)
    stage 2 convergence rate of the 2D caffeine model, stage 2 of the
    3D caffeine model, and (d) stage 2 of Lena.}
  \label{fig:stage2}
\end{figure}

In addition to the examples shown in this chapter we have studied
a number of other examples. Based on our observations we conclude that
our algorithm demonstrates a significantly better convergence rate so
long as
the interval of phase uncertainty is not too close to $\pi$
radians.%  Otherwise, the algorithm may stagnate at local minima. At
% this moment is seems that images with tight support give better gain
% in the convergence rate.

Besides the fast convergence rate, our method allows us to incorporate
additional information about the image or the noise distribution in
the measurements. For example, in practice  we measure
$r^2$ and not $r$, and the noise distribution is Poissonian rather than
Gaussian. In this case the
maximum-likelihood criterion implies the functional for the error
measure in the Fourier domain to be as follows:
\begin{equation}
  \label{eq:approx-phase1-poisson-func}
  E_P(x) = \mathbf{1}^T \left(
    |FPx|^2-r^2\ln\left(|FPx|^2\right)
  \right).
\end{equation}
To demonstrate the performance of our method we contaminated the
measurements ($r^2$) of the ``Lena'' image with Poissonian noise such that
the signal to noise ratio (SNR) was 53.6 dB. The phase uncertainty was
3 radians as before. First, we started by solving the convex problem,
as defined by Equation~(\ref{eq:approx-phase1-convex-func}). The solution obtained
was then used as the starting point for the second stage of our
method using  the non-convex functional defined in
Equation~(\ref{eq:approx-phase1-poisson-func}). The HIO algorithm also started at
this solution. In addition to using the objective function that fits
the noise distribution we also included a regularization term in the
object space. In this example, we used the total variation
functional~\shortcite{rudin92nonlinear}
\begin{equation}
  \label{eq:total-variation}
 \mathrm{TV}(x) = \int |\nabla x|\,, 
\end{equation}
with a small weight. Total variation is a good prior for a broad range
of images, especially for images that are approximately piece-wise
constant. In our case, introduction of this regularization added about
3 dB to the reconstruction SNR. The reconstruction results are shown in
Figure~\ref{fig:reconstruction-noise}. Our method achieved the SNR of
30dB, while the HIO algorithm produced a significantly inferior
result. Its SNR was only 16.7dB.
\begin{figure}[t]
  \centering
  \subfloat[]{\label{fig:lena-reconstruction-us}%
    \includegraphics[width=0.35\linewidth]%
    {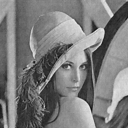}}
  \hspace{.1\linewidth}
  \subfloat[]{\label{fig:lena-reconstruction-hio}%
    \includegraphics[width=0.35\linewidth]%
    {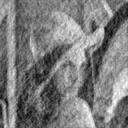}}
\caption[Reconstruction from noisy data]{Reconstruction from noisy
  data: (a) our method (30dB), and (b) HIO (16.7dB).}
\label{fig:reconstruction-noise}
\end{figure}
Note that the SNR values given above were obtained using different
measures. The  SNR in the measurements was
obtained with respect to the measured intensity, that is,
\emph{squared} Fourier magnitude, while the SNR values reported for the
reconstruction results was measured with respect to the image
magnitude. More correct would be to measure it with respect to the
image intensity---in this case the reported SNR should be multiplied by
approximately 2. Hence, our reconstruction provides SNR that is better
than the SNR of the measurement. This improvement was achieved with
the aid of the regularization term (TV)---something that is impossible
to incorporate into methods based on alternating projections. 

It is also worthwhile to note that
Poissonian noise of small intensity can be well approximated by
Gaussian noise. However, if one uses the objective function implied by
the Gaussian noise in $r^2$, that is,
\begin{equation}
  \label{eq:3}
  \||FPx|^2 - r^2\|^2 \,, 
\end{equation}
the reconstruction results are a few dB's worse than those we got with
the proper choice of the objective function. 

\section{Concluding remarks}
\label{sec:approx-phase1-conclusions}
In this chapter we presented the first successful method based on
continuous optimization for the phase retrieval problem for
non-negative objects whose phase is known to lie within a certain
interval.  It is important to note, however, that straightforward
incorporation of this information does not lead automatically to a
successful method of reconstruction. Therefore, we designed a
two-stage algorithm.  At the first stage we perform convex relaxation
and solve the resulting convex problem. At the second stage the
original objective function is re-introduced into the scheme and the
reconstruction continues from the solution of the first stage.  The
algorithm demonstrates a significantly better convergence rate compared
to current reconstruction methods. Moreover, in contrast to these
methods, our technique is flexible enough to allow incorporation of
additional information. Practical examples of such information include
measurement noise distribution and  knowledge that the sought image
is piece-wise smooth.

It is worthwhile to discuss possible sources of the approximate
Fourier phase information. Probably, most obvious way to obtain it is
to introduce into the scene an object whose Fourier transform is
known. In this case the recorded data is the modulus of a sum of two
complex numbers, one of which is known, and we are actually required
to perform some sort of holographic/interferometric
reconstruction. However, unlike in the classical
holograph/interferometry, the known object must not be known
precisely. This is a very important property, because calibrating and
maintaining an interferometer so as to keep the reference beam with
high precision during a series of experiments is a time-consuming and
expensive procedure. We shall return to the holographic setup in
Chapters~\ref{cha:phase-retr-holography}
and~\ref{cha:design-bound-phase}.
But there are other sources
of the phase information that do not require a physical modification
of the experiments. For example, as explained
in~\shortcite[Section~5.2]{goodman05introduction}, an ideal lens can perform the
Fourier transform. This property can be used to test lenses by
illuminating them with a known beam and measuring the resulting
intensity in the Fourier plane. By reconstructing the illuminating
beam from this intensity and comparing it with the actual beam, one
can estimate the quality of the lens. Note that the imperfections in
the lens production affect directly the Fourier phase by diverting it
from the known expected values. Another possible way to obtain a phase
estimate is to use current reconstruction methods such as HIO in 
situations where they are known to be able to reconstruct the sought
signal. If a successful reconstruction is guaranteed, at some point
the phase must be close enough to the true value, and if we manage to
identify such a moment, then we can switch to our algorithm and get a faster
reconstruction. As a last resort, the whole interval of $2\pi$ radians
can be split into several, small enough, intervals and each one can be
tested separately. In this case, once a correct partitioning is found
our method will recover the image. This approach is closely related to
combinatorial optimization, as we have to choose a correct combination
of the phase intervals. Thus it is very important that an assessment
of the current hypothesis (a particular choice of the phase intervals)
can be performed very fast.

%%% Local Variables: 
%%% mode: latex
%%% TeX-master: "../thesis"
%%% End: 

%% file: phase-theory/phase-theory.tex
\chapter{Approximate Fourier phase data for arbitrary
  signals---why does it help?\footnotemark{}}
\label{cha:appr-four-phase-explanation}

\footnotetext{The material presented in this chapter was published in \shortcite{osherovich11approximate}.}

In the previous chapter we presented an algorithm for fast phase
retrieval in situations where a rough Fourier phase estimate is
available. Our experiments demonstrated excellent convergence rates for
phase uncertainty intervals of up to $\pi$ radians. This algorithm is
potentially important, as it is the first successful application of
continuous optimization technique to the phase retrieval problem.
However, these  results were purely empirical---no explanation
has been provided so far about why and when this method is expected to
solve the problem . Moreover, all these results were obtained for
real-valued non-negative signals. When we tried to apply the method to
complex-valued signals the results were less encouraging. The reason
for this gap lies in our experimental setup: the phase uncertainty intervals
were generated independently for each pixel in the Fourier
space, ignoring the conjugate symmetry exhibited by real-valued
signals in the Fourier domain. Hence, the \textit{effective} phase
uncertainty 
in our previous experiments was approximately $\pi/2$ radians and not
$\pi$. After we refined the allowed phase uncertainty interval, the
results in both the real-valued and complex-valued cases returned to be
excellent. This chapter provides an explanation for this phenomenon.

Thus, in this section we continue to evaluate the importance of
approximate phase information in the phase retrieval problem. However,
now we address more theoretical questions, the main one being
when and why a rough phase estimate (up to $\pi/2$ radians) can lead
to a guaranteed reconstruction by the algorithm presented in the
previous section. Our main discovery is that the above phase
uncertainty limit practically guarantees a successful reconstruction
by \textit{any} reasonable algorithm for the reason described
below. This is 
an  important property as it allows development of very
efficient algorithms whose reconstruction time is orders of magnitude
faster than that of the current method of choice---the Hybrid
Input-Output (HIO) algorithm. We have already presented one such
algorithm in the previous section, however, we believe that its
results can be further improved by using more sophisticated interior
point methods. Using the new algorithms we were able to reconstruct
signals that cannot be successfully reconstructed by HIO, namely,
complex-valued signals without tight support information.

Additionally, we provide a heuristic explanation of why continuous
optimization methods such as gradient descent or Newton-type algorithms
fail when applied to the phase retrieval problem and how the
approximate phase information can remedy this situation. We actually
show the reason for the failure for a very large family of optimization
methods---monotone line-search optimization algorithms.
Notwithstanding this failure in the general case, we argue that a
rough phase estimate leads to an important property: local minima of a
functional associated with the phase retrieval problem are likely to
be global minima. This is the reason for our previous claim: chances
are that \textit{any} algorithm capable of finding a local minimum
will successfully reconstruct the image.

Additional numerical simulations are provided in
Section~\ref{sec:approx-phase2-results} to demonstrate the validity of our analysis
and success of our reconstruction method.

\section{Optimization methods: the problem and the remedy}
\label{sec:optim-meth-probl}
Let  $z$ represent the true (unknown) signal, with
$\hat{z}$ being its Fourier transform, in accordance with our
notational conventions. The current estimate (obtained after $k$
iterations) is denoted by $x^{k}$. With this notation, the
classical phase retrieval problem reads: find $x$ such that
$|Fx|=|\hat{z}|$ subject to certain constraints in the object
domain. When the Fourier phase is known to lie within a given interval
$[\alpha,\beta]$ ($\alpha$, and $\beta$ are vectors of appropriate
size) the problem to be solved is: find $x$ such that $|Fx| =
|\hat{z}|$ and $\alpha \leq \angle(Fx) \leq \beta$ subject to certain
constraints in the object domain. Note that such $x$ is not
necessarily equal to $z$; even when the Fourier domain data is
sufficiently ``oversampled'', the reconstruction in the classical
phase retrieval problem is not
unique~\shortcite{hayes82reconstruction}, and it may not be unique
even when the Fourier phase estimates are available. 

Let us now try to understand why the classical Newton-type and
gradient descent methods fail for the phase retrieval
problem. Actually, we can address a very wide family of optimization
methods that includes these two methods---monotone line-search
algorithms. Each iteration of these algorithms has a common 3-step
template:
\begin{description}
\item[Step 1:] Find a descent direction $p^{k}$.
\item[Step 2:] Along that direction find a step-length $\Delta^{k}$
  that sufficiently decreases the objective function value.
\item[Step 3:] Move to the new location: $x^{k+1} = x^{k}+\Delta^{k}p^{k}$.
\end{description}
Descent direction is defined as a vector whose
inner product with the gradient is negative (for obvious reasons). This
guarantees that there always exists a step along that direction that decreases the
objective function value.

To be specific we choose the most popular objective
function for the discrepancy minimization in the Fourier domain:
\begin{equation}
  \label{eq:approx-phase2-1}
  f(x) = \frac{1}{2}
  \left\|
    |Fx| - |\hat{z}|
  \right\|^{2}\ .
\end{equation}
The gradient of $f(x)$ is given by (see
Chapter~\ref{cha:found-optim-meth}):
\begin{equation}
  \label{eq:approx-phase2-2}
  \nabla f(x) = 
  x - F^{-1}
  \left(
    |\hat{z}|\circ\frac{Fx}{|Fx|}
  \right)\,.
\end{equation}
As before, $a\circ b$ and $\frac{a}{b}$ denote the Hadamard (element-wise)
product and quotient, respectively. It is interesting to note that the
signal $F^{-1}\left(|\hat{z}|\circ\frac{Fx}{|Fx|}\right)$ has a clear
physical meaning: it is obtained from $x$ by substituting the (wrong)
Fourier magnitude $|Fx|$ with the correct one $|\hat{z}|$. Thus, it is
nothing but the signal denoted by $y$ in
Figure~\ref{fig:current-fienup-alg}. However, the main observation
about the gradient is that $\nabla f(x)=0$ if and only if
$|Fx|=|\hat{z}|$, that is, if and only if $x$ is a solution.  Of course,
this is true only if there are no additional constraints. In practice,
however, the optimization is done while respecting certain restrictions on
$x$. The constraints are often implemented as penalty functions that
augment the original functional, and the augmented gradient may
vanish when $|Fx|\not=|\hat{z}|$. Also, we may find ourselves in a
situation where no feasible descent direction exists, if the
constraints are kept as ideal barriers. Such situations
usually 
indicate a local minimum and make further progress by such standard
optimization methods impossible.  In this discussion we are being
deliberately vague about the exact nature of the object domain
constraints and enforcement thereof in the optimization process. We
only assume that imposing these constraints on an image estimate will
take that estimate to be closer\footnote{``Closer'' here refers to the
Euclidean distance, however it can be readily generalized to other
metrics.} to the true signal---a natural
assumption for monotone optimization.

Let us now consider a single element in the Fourier
domain. Using our notation, the true value is $\hat{z}_{i}$, whose
magnitude $|\hat{z}_{i}|$ is known and whose phase $\theta_{i}$ is
unknown. We distinguish three possible scenarios where $|\hat{x}_{i}|
\not= \hat{z}_{i}$. First, the
Fourier magnitude of the current estimate is smaller than
$|\hat{z}_{i}|$ and the phase
error $\alpha_{i}$ is greater than $\pi/2$ radians. Second, the current
Fourier magnitude is greater than $|\hat{z}_{i}|$ (phase error is
unimportant in this case). Finally, the third possibility: the
current estimated Fourier magnitude is less than $|\hat{z}_{i}|$ and
the phase 
error $\alpha_{i}$ is less than $\pi/2$ radians. These scenarios are
illustrated in Figures~\ref{fig:wrong-direction},
\ref{fig:large-magnitude}, and~\ref{fig:small-magnitude}, respectively.
\begin{figure}[H]
  \centering
  \subfloat[]{
    \label{fig:wrong-direction}
    \includegraphics[width=0.45\textwidth{}]{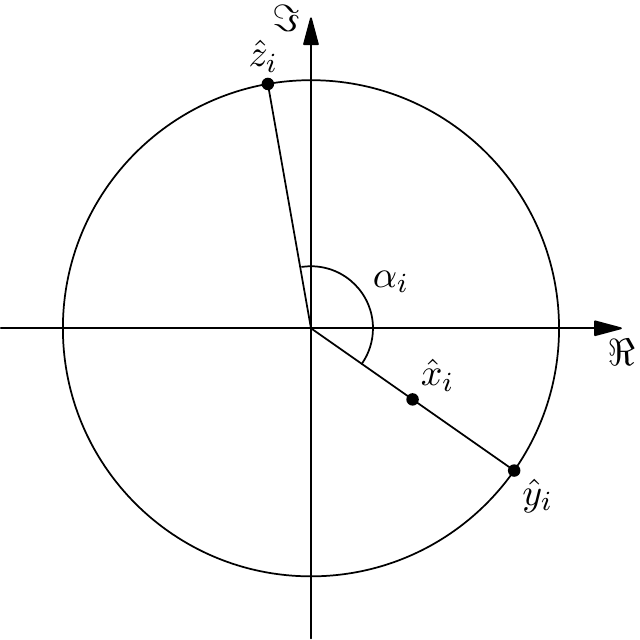}
  }
  \quad{}
  \subfloat[]{%
    \label{fig:large-magnitude}
    \includegraphics[width=0.45\textwidth]{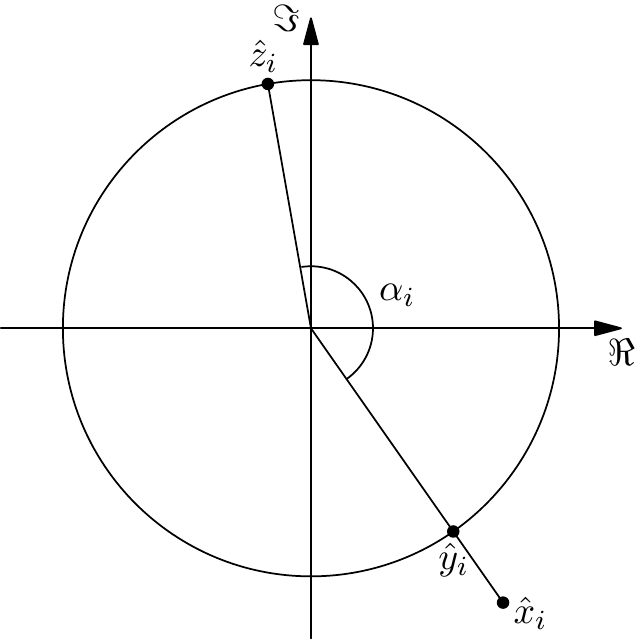}}
  \quad{}
  \subfloat[Small magnitude: good direction.]{%
    \label{fig:small-magnitude}
    \includegraphics[width=0.45\textwidth]{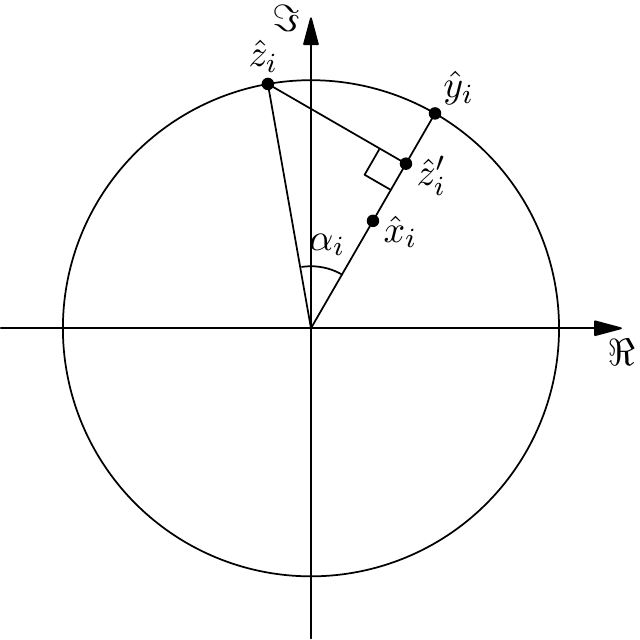}}
  \caption[Three possible scenarios in the Fourier domain]{Three
    possible scenarios in the Fourier domain: (a) small magnitude =
    bad direction, (b) large magnitude = good direction, and (c) small
    magnitude = good direction.}
  \label{fig:possible-scenarios}
\end{figure}
Recall also that, by Equation~\eqref{eq:approx-phase2-2}, we know that the gradient
descent direction is always taking us from $\hat{x}$ towards $\hat{y}$
(see Equation~\eqref{eq:approx-phase2-2} and Figure~\ref{fig:possible-scenarios}).
Let us consider the first case. In the Fourier domain, when we move
from $\hat{x}_{i}$ towards $\hat{y}_{i}$ we are actually moving away
from the correct value $\hat{z}_{i}$ (see
Figure~\ref{fig:wrong-direction}).  Therefore, due to the unitarity of
$F$, $x$ is pulled farther away from $z$. On the other hand, object
domain constraints shall pull us towards the correct value, as assumed
before. Hence, the two forces may cancel each other, resulting thereby
in the stagnation of the algorithm. In numerical tests this stagnation
is observed in all but very tiny problems. Worse still, it
happens at very early stages, long before the reconstruction algorithm
gets close to the correct value. Hence, the results are usually
worthless.  In the second case, where the magnitude of the current
estimate is greater than the correct value $\hat{z}_{i}$, moving from
$\hat{x}_{i}$ towards $\hat{y}_{i}$ will necessarily bring us closer
to the correct value $\hat{z}_{i}$ (see
Figure~\ref{fig:large-magnitude}). The last case is more
interesting. If the estimated Fourier magnitude is sufficiently
smaller than the correct one, and the phase error $\alpha_{i}$ is less
than $\pi/2$ radians, then moving along the gradient descent
direction, that is, from $\hat{x}_{i}$ towards $\hat{y}_{i}$, will get us
closer to $\hat{z}_{i}$, so long as we do not pass the point
$\hat{z}_{i}'$ which is the projection of $\hat{z}_{i}$ onto the line
segment $[0, \hat{y}_{i}]$ as shown in
Figure~\ref{fig:small-magnitude}.  Note that the fact that moving from
$\hat{z}'_{i}$ towards $\hat{y}_{i}$ takes us farther away from the
correct value $\hat{z}_{i}$ keeps us from claiming that any
optimization algorithm will necessarily converge to a
solution. However, in the discussion that follows we prove that the
situation where an optimization algorithm gets stuck at some $x$ such
that all $\hat{x}_{i}$ lie between $\hat{z}'_{i}$ and $\hat{y}_{i}$ is
impossible. In fact, we argue below that in the situation where the Fourier
phase errors are limited by $\pi/2$ radians any local minimum is
likely to be a global one. That is, any algorithm capable of finding a
(constrained) local minimum can be expected to solve the phase
retrieval problem in this case. To give a reasoning behind this claim
we must be more specific about the constraints in the object
domain. As will be evident from the argument below, our assumptions
are very general and fit all commonly encountered cases.

Let us
consider the most frequent object domain constraint: limited support
information.
\begin{equation}
  \label{eq:phase-approx2-3}
  x_{o} = 0, \quad \forall o\in\mathcal{O}\ ,
\end{equation}
where $\mathcal{O}$ denotes the set of off-support locations, where
$z$ is known to be zero. Note that zero-padding is a special case of
such support information. In certain situations the sought signal can
be assumed to be real non-negative. In these situations the above constraint
can be extended by the non-negativity requirement:
\begin{equation}
  \label{eq:phase-approx2-4}
  x_{s} \geq 0, \quad \forall s\not\in\mathcal{O}\ .
\end{equation}
What is important for our discussion is that in both cases the set of
all feasible signals ($\mathcal{Z}$) is a convex set that contains a
proper cone $\mathcal{K}$. That is, if $x$ and $y$ are feasible, then
$\lambda x + \gamma y$ is also feasible for any non-negative scalars
$\lambda, \gamma$. $\mathcal{Z}$ contains 
$\mathcal{K}$ but not necessarily equals it, but this is
unimportant in the following discussion. Let us now consider the
minimization of the objective function from Equation~\eqref{eq:approx-phase2-1}
subject to the following two conditions: first, the set of object
domain constraints is convex and contains the proper cone
$\mathcal{K}$; second, the phase error (of all elements) in the
Fourier domain is bounded by $\pi/2$, that is, $|\angle(\hat{x}) -
\angle(\hat{z})| \leq \pi/2-\epsilon$ for some small positive
$\epsilon$.  Obviously, the Fourier domain constraints define a convex
set which is a proper cone too (a Cartesian product of proper
cones). Hence, the optimization in this case is done over a convex
set. Assume now that the algorithm converges at some local minimum
$x$ which is not a solution (global minimum). Following a basic
theorem from constrained-optimization 
theory we conclude that the following inequality must hold for any
feasible point $w$ (see, for example,~\shortcite{bertsekas99nonlinear}):
\begin{equation}
  \label{eq:approx-phase2-5}
  \langle \nabla f(x), (w-x)\rangle \geq 0\ .
\end{equation}
In words, that means that no feasible descent direction exists at the
point $x$. Let us consider a (small) subset of all feasible points: $w
= \lambda z + \gamma x$, where $\lambda,\gamma\geq 0$ (note that $x$
and $z$ are feasible by definition, hence, $x,z\in\mathcal{K}$ which
implies that $w\in \mathcal {K}$). The choice of this subset of
feasible directions is stipulated by the fact that the phase error of
$w$ in all frequencies is less than or equal to the phase error of
$x$ (it is strictly smaller when $\lambda\not=0$). With $w$ as above,
Equation~\eqref{eq:approx-phase2-5} becomes
\begin{equation}
  \label{eq:approx-phase2-6}
  \langle \nabla f(x), (w-x) \rangle =
  \lambda\langle\nabla f(x), z\rangle +
  (\gamma - 1) \langle\nabla f(x),x\rangle\ .
\end{equation}
If $\langle\nabla f(x),x \rangle < 0$, we can set $\lambda=0$,
$\gamma=2$ and get a feasible descent direction (this is, actually,
scaling up of $x$). If  $\langle\nabla f(x),x \rangle > 0$ we can set
$\lambda =0$, $\gamma=0.5$ and, again, get a feasible descent
direction (scaling down of $x$). Hence, any local minimum must
obey $\langle\nabla f(x),x \rangle = 0$ (we shall call such $x$
\textit{optimally scaled}). In this situation, the sign of the
left-hand size of 
Equation~\eqref{eq:approx-phase2-5} depends solely on the sign of the inner product
$\langle\nabla f(x), z\rangle$. It is more convenient to consider the
above inner product in the Fourier domain. Due to unitarity of $F$ we have:
\begin{align}
  \label{eq:approx-phase2-7}
  \langle\nabla f(x), z\rangle
  & = \langle F\nabla f(x), Fz\rangle \\
  & = \sum_{i} (|\hat{x}_{i}|
  -|\hat{z}_{i}|)\,|\hat{z}_{i}|\cos\alpha_{i}\label{eq:approx-phase2-8}
  \ .
\end{align}
Recall that the above formula is considered in
the context of an optimally scaled $x$, that is:
\begin{align}
  \label{eq:approx-phase2-9}
  0 = \langle   \nabla f(x), x\rangle
  & = \langle F \nabla f(x), Fx\rangle \\
  & = \sum_{i} (|\hat{x}_{i}|
  -|\hat{z}_{i}|)\,|\hat{x}_{i}|\label{eq:approx-phase2-10}
  \ .
\end{align}
For our following discussion, it is convenient to consider
Equations~\eqref{eq:approx-phase2-8} and \eqref{eq:approx-phase2-10} as \textit{weighted sums} of
$|\hat{x_{i}}|-|\hat{z}_{i}|$. Thus, for example, it becomes obvious
that an $x$ that belongs to the convex region $\mathcal{C}$ (see
Figure~\ref{fig:convex-area}), as was required in the 
algorithm in Chapter~\ref{cha:appr-four-phase-first}, cannot be a
local minimum unless $|\hat{x}| = |\hat{z}|$ (which makes it a global
one), because this is the only way to get zero by summing non-positive
numbers associated with strictly positive weights. This explains
the success of the algorithm. However, even without the
restrictions on $|\hat{x}|$ used in the original algorithm, we can
expect the sum in Equation~\eqref{eq:approx-phase2-8} to be negative.  To
understand why let 
us split it into three disjoint sets of indices,
\begin{multline}
  \label{eq:approx-phase2-11}
  \sum_{i} (|\hat{x}_{i}|
  -|\hat{z}_{i}|)\,|\hat{z}_{i}|\cos\alpha_{i} =
  \sum_{i\in\mathcal{B}}(|\hat{x}_{i}|
  -|\hat{z}_{i}|)\,|\hat{z}_{i}|\cos\alpha_{i} +\\
  \sum_{i\in\mathcal{S}}(|\hat{x}_{i}|
  -|\hat{z}_{i}|)\,|\hat{z}_{i}|\cos\alpha_{i} +
  \sum_{i\in\mathcal{A}}(|\hat{x}_{i}|
  -|\hat{z}_{i}|)\,|\hat{z}_{i}|\cos\alpha_{i}
  \ ,
\end{multline}
where $\mathcal{B} = \{i \ |\ |\hat{x}_{i}| > |\hat{z}_{i}|\}$,
$\mathcal{S} =\{i \ |\ |\hat{x}_{i}| \leq |\hat{z}_{i}|\cos\alpha_{i}\}$,
and $\mathcal{A} =\{i \ |\ |\hat{z}_{i}|\cos\alpha_{i} < |\hat{x}_{i}|
<|\hat{z}_{i}|\}$. With this subdivision it is easy to compare the sum
in Equation~\eqref{eq:approx-phase2-8} (for which a negative sign indicates
the presence of a
feasible descent direction) and the sum in Equation~\eqref{eq:approx-phase2-10}
(which is zero). The weight in these weighted sums 
changes from 
$|\hat{x}_{i}|$ in~\eqref{eq:approx-phase2-10} to $|\hat{z}_{i}|\cos\alpha_{i}$ in
\eqref{eq:approx-phase2-8}. Hence, for 
$i\in\mathcal B$, $|\hat{x}_{i}|-|\hat{z}_{i}|$ is positive and its
weight has decreased, thus, pulling the total sum towards a negative
value. For $i\in\mathcal S$, $|\hat{x}_{i}|-|\hat{z}_{i}|$ is
negative, and its weight has increased, thus again contributing to
the negativeness of the result. The only subset of indices that
increases the sum is $i\in\mathcal A$. From our experience it is very
unlikely to encounter a situation where the contribution of
$i\in\mathcal A$ outweighs the joint contribution of $i\in\mathcal B$
and $i\in\mathcal S$. Hence it is very unlikely to get stuck in a local
minimum with no descent direction. Note also that if the phase error
of all frequencies of $x$ is strictly less than $\pi/2$ radians then the
sum in Equation~\eqref{eq:approx-phase2-8} must be zero for $x$ to be a local
minimum,  because if the direction towards $w$ is an ascent direction,
one can simply reverse it to get a descent direction.

This discussion provides a heuristic explanation why a carefully designed
optimization method can be expected to converge to a  global
minimum, that is, to solve the phase retrieval problem when the Fourier
phase is known up to $\pi/2$ radians and the object domain constraints
are given in the form of (possibly loose) support information.

\section{Experimental results}
\label{sec:approx-phase2-results}
The method was tested on many images with consistent results. Here,
for demonstration purposes, we chose a natural image with a lot of
features so as to allow easy perception of the reconstruction quality
by the naked eye. We demonstrate our results on two different cases: one
with loose support information, that is, part of the image is zero and
we do not know that a priori; another with tight support.  Both
images are complex-valued and their original size is $128\times128$
pixels, padded with zeros to the size of $256\times256$
pixels.  The intensity (squared magnitude) of the images (without the
zero-padding) is shown in
Figure~\ref{fig:lena-images}.

\begin{figure}[H]
  \centering
  \subfloat[]{
    \label{fig:lena-loosesupport}
    \includegraphics{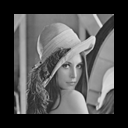}
  }
  \qquad{}
  \subfloat[]{
    \label{fig:lena-tightsupport}
    \includegraphics{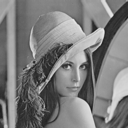}
  }
  \caption[Test images]{Test images (intensity): (a) loose support,
    (b) tight support.}
  \label{fig:lena-images}
\end{figure}
All our experiments show that the phase distribution in the object
domain does not affect the reconstruction. Therefore, the actual phase
distribution was chosen to be random to avoid any possible assumptions
of smoothness.  We compared three
reconstruction methods.  First, a slight modification of the
quasi-Newton method from Chapter~\ref{cha:appr-four-phase-first}.
Second, we created the following
phase-aware modification of the HIO algorithm (PA-HIO). When the
Fourier phase is known to lie in the interval $[\alpha,\beta]$,
PA-HIO's correction in the Fourier domain forces the current estimate
$\hat{x}$ to lie on the arc $\hat{A}\hat{B}$ (see
Figure~\ref{fig:convexrelaxation}), hence $\hat{y}$ (see
Figure~\ref{fig:projection-fourier}) is the point closest to $\hat{x}$
that lies on the arc $\hat{A}\hat{B}$. The third algorithm is the
classical HIO method without any alterations. In fact, we tested also
a phase-aware modification of the GS algorithm, however, its results were
consistently worse than those of PA-HIO, so we omit them.
The current modification of our method uses only one
stage. That is, we abandoned the first stage that was  used
in the original method from Chapter~\ref{cha:appr-four-phase-first} to
find a 
feasible $x$ by solving the convex problem in the original
algorithm. This stage is no longer necessarily because the phase
bounds in the Fourier domain are the main reasons for success. The
additional restrictions on $|\hat{x}|$ used in the original method can
be viewed as heuristic constraint added (with smaller weight) for the
reasons given in the discussion that follows
Equation~\eqref{eq:approx-phase2-10}. This also made the choice of the initial $x$
straightforward: $x^{0}=0$.

We first demonstrate how the phase uncertainty interval affects our
ability to reconstruct the images. A set of 51 uncertainty intervals
was chosen in the range from zero to 2.5 radians. For each interval of
uncertainty its endpoints were chosen at random so that the true phase
was uniformly distributed inside it. We ran our quasi-Newton
optimization algorithm one hundred times (each time generating new
phase bounds), checking at each run whether it was successful
or not.  A run was considered successful if the error in the Fourier
domain (as defined by Equation~\eqref{eq:approx-phase2-1}) was below $0.5\times
10^{-6}$ after 1000 iterations. As is evident from
Figure~\ref{fig:rec-success-rate} the reconstruction always succeeded as
long as the phase uncertainty was below $\pi/2$, in perfect agreement with
our analysis. It is also evident that, for images with tight support
information, successful reconstruction can be expected even for
significantly rougher phase estimation. Moreover, the algorithm
converges very fast and the 
above threshold is usually reached after 250-300 iterations for the
loose-support image and only 80 iterations for the image with tight
support as is apparent from Figure~\ref{fig:rec-speed-losesup} and
\ref{fig:rec-speed-tightsup}.

\begin{figure}[H]
  \centering
  \includegraphics[width=\textwidth{}]{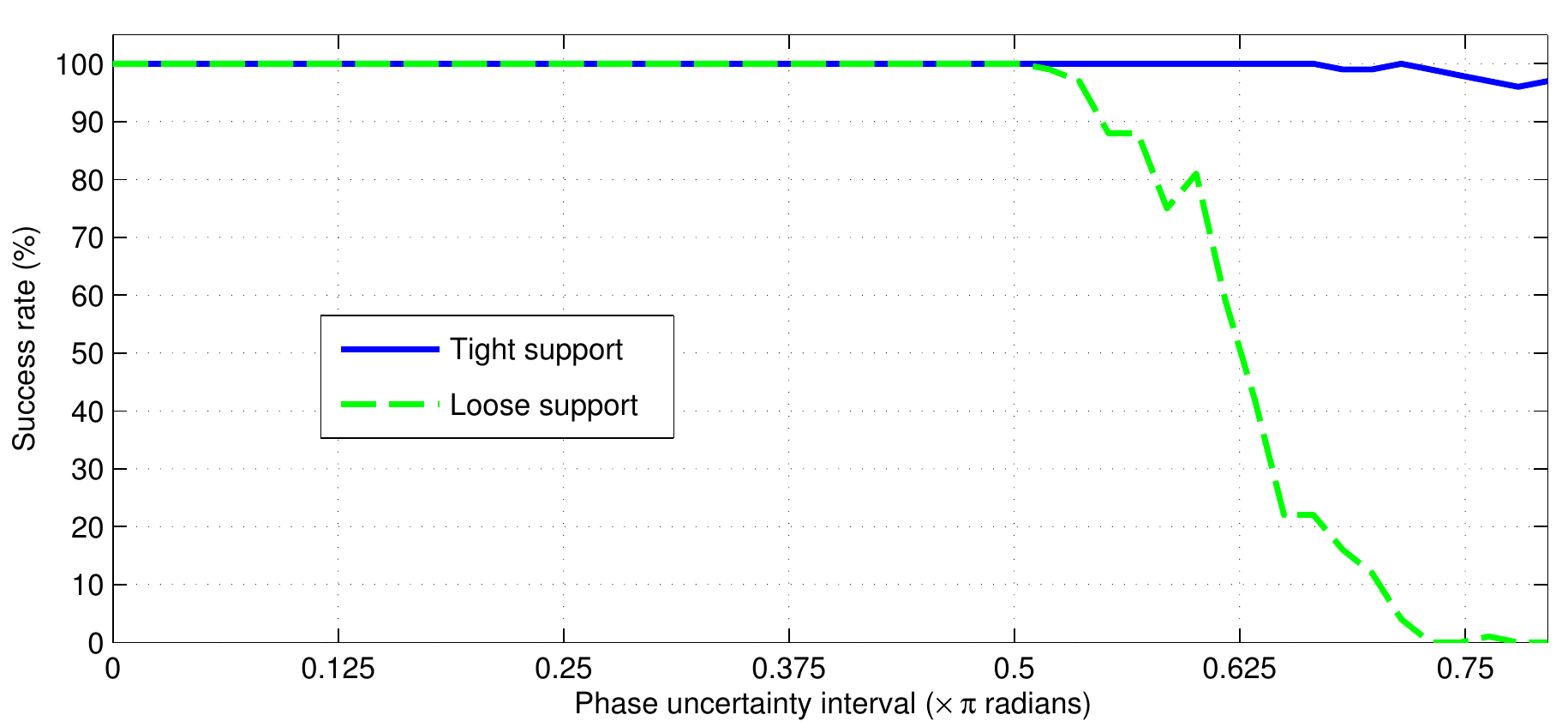}
  \caption[Success rate of our method]{Reconstruction success rate of our method as a function of
    phase uncertainty interval.}
  \label{fig:rec-success-rate}
\end{figure}

Next, we demonstrate in detail the reconstruction results in the case
where the Fourier phase uncertainty interval is 1.57 radians. Note
that without the phase bounds the HIO method cannot reconstruct images
with loose support. Images with tight support information are usually
reconstructed successfully, though they may undergo some trivial
transformation, for example, axis reversal. As is evident from
Figure~\ref{fig:rec-results} our methods (quasi-Newton and PA-HIO)
produce very good visual quality. HIO, on the other hand has problems
with the loose-support image (as expected) but the second case seems
to yield acceptable quality. However, visual assessment does not
provide full insight into the quality of reconstruction and tells
nothing about its speed. Quantitative results are given in
Figure~\ref{fig:rec-speed-losesup}
and \ref{fig:rec-speed-tightsup}, from which
it is evident that our quasi-Newton method significantly outperforms
HIO and PA-HIO
in terms of speed.
\begin{figure}[H]
  \centering
  \subfloat[]{%
    \label{fig:padded-we-rec}
    \includegraphics{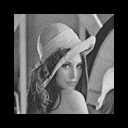}}%
  \qquad{}%
  \subfloat[]{%
    \label{fig:padded-pahio-rec}
    \includegraphics{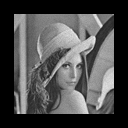}}%
  \qquad{}%
  \subfloat[]{%
    \label{fig:padded-hio-rec}
    \includegraphics{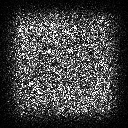}}\\
  \subfloat[]{%
    \label{fig:tight-we-rec}
    \includegraphics{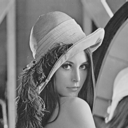}}%
  \qquad{}%
  \subfloat[]{%
    \label{fig:tight-pahio-rec}
    \includegraphics{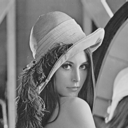}}%
  \qquad{}%
  \subfloat[]{%
    \label{fig:tight-hio-rec}
    \includegraphics{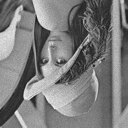}}\\
  \caption[Reconstruction results for phase uncertainty of 1.57
  radians]{Reconstruction results with phase uncertainty interval of
    1.57 radians.

    Upper row---loose support: (a) our method (after 250
    iterations), (b) PA-HIO (after 1000 iterations), and (c) HIO
    (after 1000 iterations).

    Lower row---tight support: (d) our method
    (after 80 iterations), (e) PA-HIO (after 1000 iterations), and (f)
    HIO (after 1000 iterations).}
  \label{fig:rec-results}
\end{figure}
It is important to point out two things before evaluating the
quantitative results presented in Figures~\ref{fig:rec-speed-losesup},
and \ref{fig:rec-speed-tightsup}. First, all presented methods are
iterative by nature and every one of them uses two Fourier transforms
per iteration. Hence, comparing the number of iterations is well
justified and gives a good estimation of the reconstruction speed
because the Fourier transforms constitute the major computational
burden. Second, it is obvious that images with loose support lead to
non-unique solutions. Hence, a small value of the objective function
does not necessarily mean small error in the object domain. This
explains the results in Figure~\ref{fig:rec-speed-losesup}. Another
important observation is that the HIO and PA-HIO methods do not
enforce the off-support areas (padding) to be zero. Hence, these
methods may give a large error in the Fourier domain, while the error
in the object domain (after we discard the off-support parts) may be
small. This phenomena is also evident in
Figure~\ref{fig:rec-speed-losesup}, and \ref{fig:rec-speed-tightsup}.

\begin{figure}[H]
  \centering
  \subfloat[]{
    \includegraphics[width=0.5\textwidth]{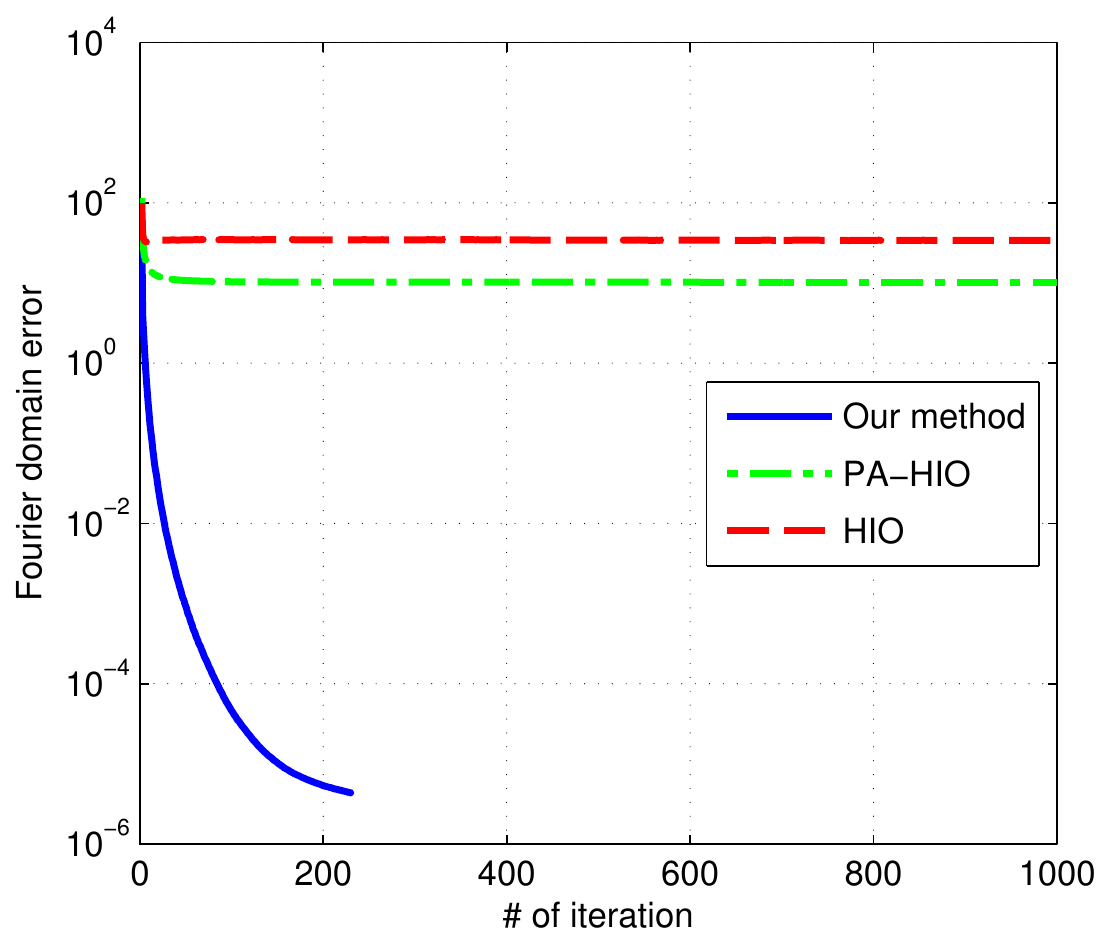}}
  \subfloat[]{
    \includegraphics[width=0.5\textwidth]{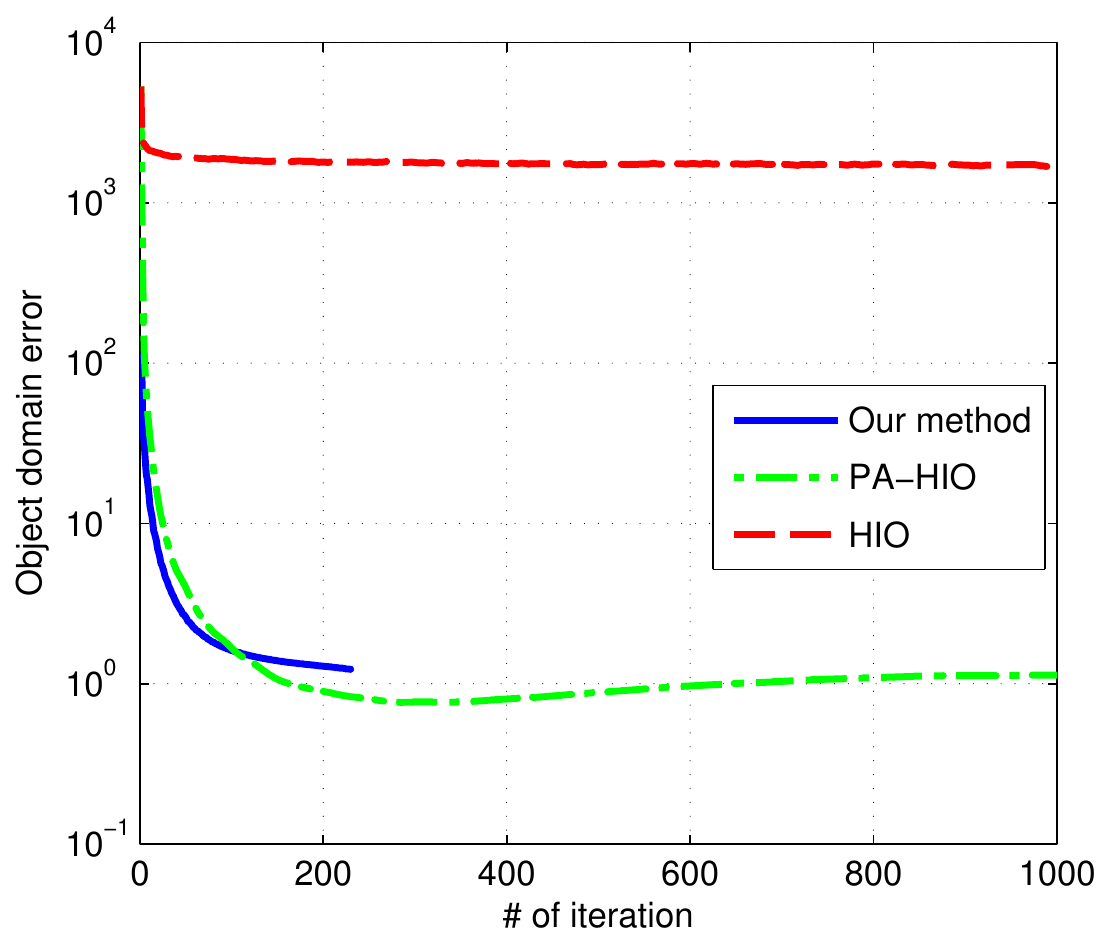}}
  \caption[Error behavior in the case of loose support]{Error behavior in the case of loose support:
    (a) Fourier domain error $\||Fx| - |\hat{z}|\|^{2}$,
    (b) object domain error $\|x| - |z|\|^{2}$.}
  \label{fig:rec-speed-losesup}
\end{figure}

\begin{figure}[H]
  \centering{}
  \subfloat[]{
    \includegraphics[width=0.5\textwidth]{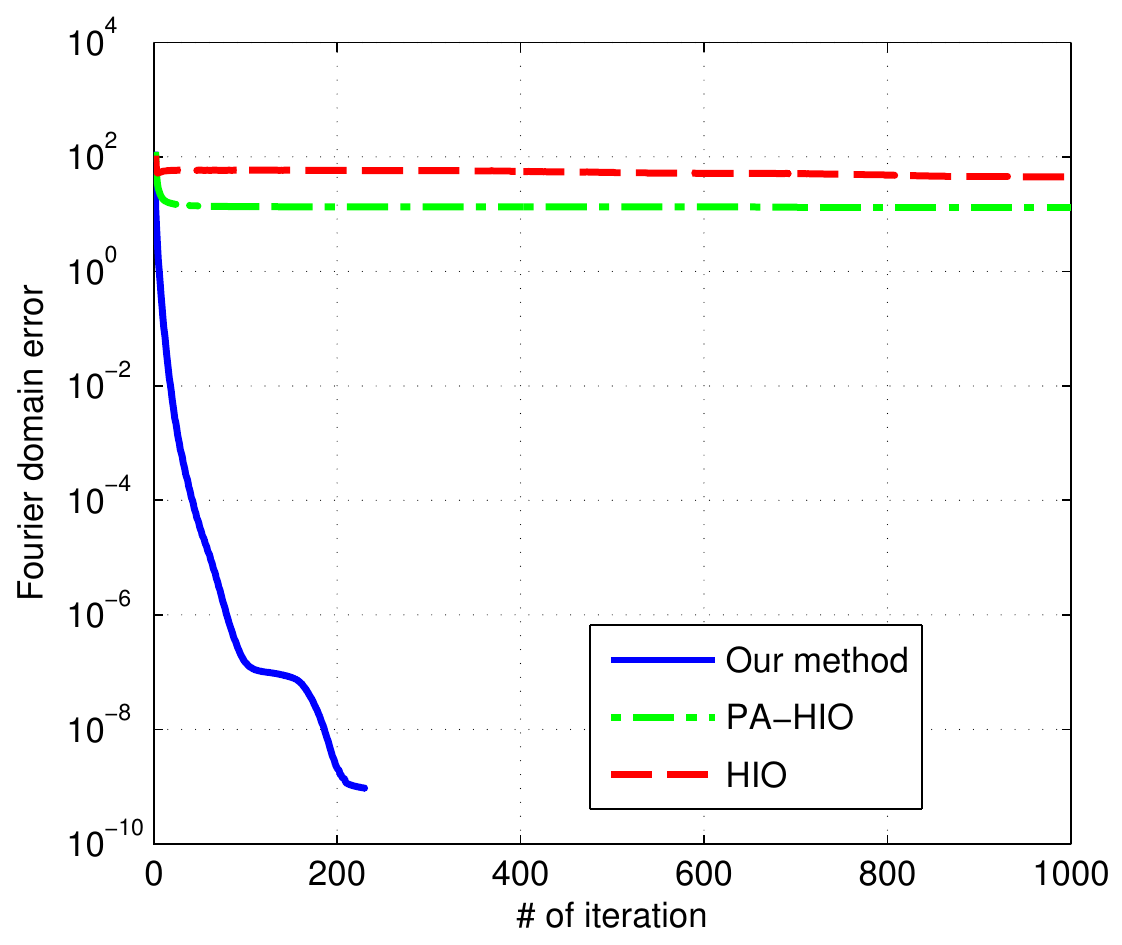}}
   \subfloat[]{
    \includegraphics[width=0.5\textwidth]{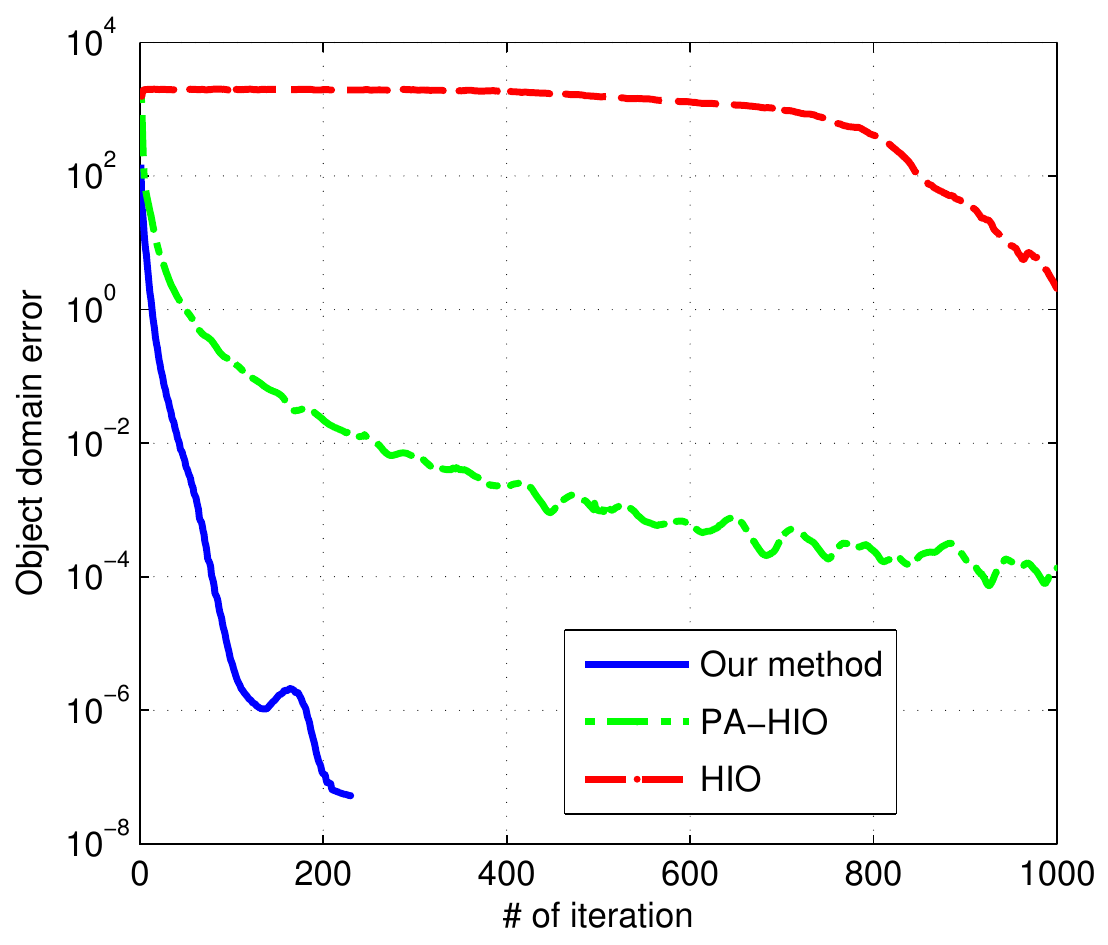}}
  \caption[Error behavior in the case of tight support]{Error behavior in the case of tight support:
    (a) Fourier domain error: $\||Fx| - |\hat{z}|\|^{2}$,
    (b) object domain error: $\|x| - |z|\|^{2}$.}
  \label{fig:rec-speed-tightsup}
\end{figure}

\section{Concluding remarks}
\label{sec:approx-phase2-conclusions}
In this chapter we presented a new analysis explaining
why continuous optimization methods fail when applied to the phase
retrieval problem. On the basis of this observation we gave a
heuristic explanation why local minima of a functional associated with
the phase retrieval problem can be expected to be global ones in the
situation where the Fourier phase error does not surpass $\pi/2$
radians. This, in turn, opens the door for continuous optimization
methods whose rate of convergence and ability to incorporate
additional information in the computational scheme significantly
exceeds those of HIO. We also present such an algorithm and demonstrate
that its reconstruction speed is significantly faster than that of
HIO, even when the phase constraints are also employed in the latter.

The analysis and the methods developed in this chapter are used in the
next chapter to perform phase retrieval in situations where the
Fourier phase uncertainty is greater than $\pi/2$ radians by using
a type of bootstrapping approach.

%%% Local Variables: 
%%% mode: latex
%%% TeX-master: "../thesis"
%%% End: 

%% file: phase-holography/phase-holography.tex
\chapter{Phase retrieval combined with digital
  holography\footnotemark{}}
\label{cha:phase-retr-holography}

\footnotetext{The material presented in this section is currently in
  preparation for submission to a journal.}

In this chapter we take our algorithm developed in the previous
chapters into a new niche: signal reconstruction from two intensity
measurements made in the Fourier plane. One is the Fourier magnitude
of the sought image, as in classical phase retrieval, and the second
is the intensity pattern resulting from the interference of the
original signal with a known reference beam, as in the Fourier domain
holography. Although either one of these measurements may, in theory,
be sufficient for successful reconstruction of the unknown image, our
method provides significant advantages over such reconstructions. For
example, comparing with reconstruction from the Fourier magnitude
alone by HIO, our method gives a much faster speed and better quality
in case of noisy measurements as we showed earlier. Furthermore,
unlike classical holography methods, our algorithm does not require
any special design of the reference beam. Finally and most
importantly, very good reconstruction quality is obtained even when
the reference beam contains severe errors.

\section{Basic reconstruction algorithm}
\label{sec:basic-reconstr-algor} Let us start with the notation used
throughout the chapter. The unknown two-dimensional signal that we wish
to reconstruct is represented by the complex-valued function
$z(\xi,\eta)=|z(\xi,\eta)|\exp[j\varphi(\xi,\eta)]$. To address the
phase of a complex-valued number we use the angle notation, as before:
$\angle(z(\xi,\eta)) \equiv \varphi(\xi,\eta)$. Our measurements are
done in the Fourier plane $(\xi',\eta')$, hence the transformation
that $z$ undergoes when transforming from the $(\xi,\eta)$ plane to
the $(\xi',\eta')$ plane is simply the unitary Fourier transform
\begin{equation}
  \label{eq:phase-holo-1}
  \hat{z}(\xi',\eta') = \mathcal{F}\{z(\xi,\eta)\}\ .
\end{equation}
% Strictly speaking, the actual transformation is a bit more
% complicated as it includes some constant multipliers and scale
% factors \shortcite{goodman05introduction}.  However, these are
% immaterial for our discussion. 
Hereinafter, we use  the
usual convention that a pair of symbols like $x$ and $\hat{x}$ denotes a
signal in the $(\xi,\eta)$ plane (also known as the object domain) and
its counterpart in the $(\xi',\eta')$ plane (also referred to as the
Fourier domain), respectively. For the sake of brevity, we may
omit the location designator $(\xi,\eta)$  or $(\xi',\eta')$ and use $x$
or $\hat{x}$ when the entire signal is considered.

The main purpose of our work is to develop a robust reconstruction
method that can tolerate severe errors in the reference beam. To
this end we use the reference beam only for \textit{estimating}
the Fourier phase of the sought image. Once a rough phase estimate
is available we can use the method of phase retrieval with
approximately known Fourier phase that was developed in
the previous chapters.

The two measurements available at our disposal are used as follows.
$I_{1}$ provides the Fourier magnitude of the sought image via the
simple relationship between the two:
\begin{equation}
  \label{eq:phase-holo-2}
  I_{1}(\xi',\eta') = |\hat{z}(\xi',\eta')|^{2}\ .
\end{equation}
The second measurement reads
\begin{equation}
  \label{eq:phase-holo-3}
  I_{2}(\xi',\eta') =|\hat{z}(\xi',\eta') + \hat{R}(\xi',\eta')|^{2}\ ,
\end{equation}
where $\hat{R}(\xi, \eta)$ denotes a known reference beam that is
used to obtain the Fourier phase estimate as described below. One
possible schematic setup that provides these measurements is shown
in the next figure.
\begin{figure}[H]
  \centering
  \includegraphics[width=\textwidth{}]{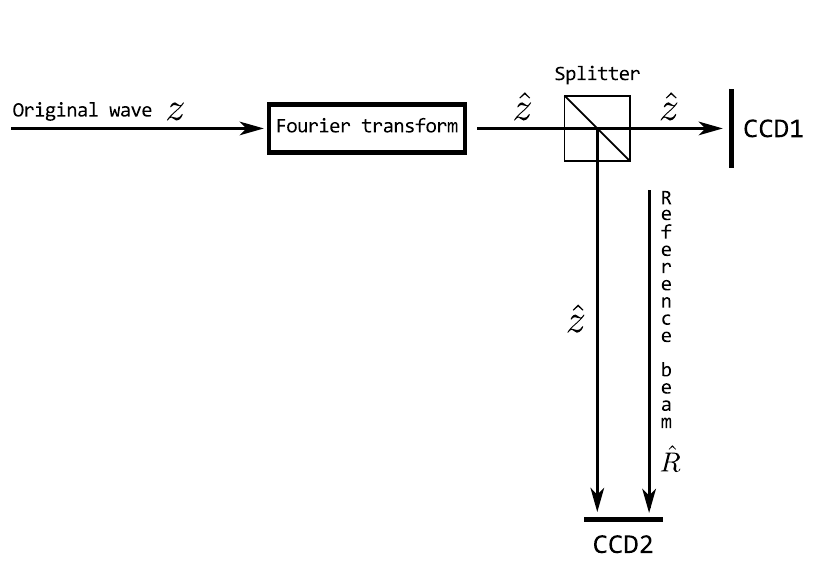}
  \caption{Schematic representation of the experiment.}
  \label{fig:experiment-schematic}
\end{figure}
Note that $\hat{R}$ is not
necessarily a Fourier transform of some physical signal $R$ in the
object domain. This means that $\hat{R}$ can be formed directly in
the Fourier plane without forming first $R$ and applying then an
optical Fourier transform to obtain $\hat{R}$.  Nevertheless,
there exists the mathematical inverse
\begin{equation}
  \label{eq:phase-holo-4}
  R(\xi,\eta) = \mathcal{F}^{-1}\{\hat{R}(\xi',\eta')\}\ ,
\end{equation}
whose properties, such as extent, magnitude, etc. can be
considered. The only requirement of $\hat{R}$ is that it must not
vanish in the region of our measurements. This is an important
point that provides an advantage for our method over the classical
holography techniques. We shall elaborate more on this in
Section~\ref{sec:relation-holography}.

Let us now describe  how $\hat{z}$'s phase information is extracted
from $I_{2}$, and, more importantly, how it is used in our
reconstruction method. Consider the two signals:
\begin{equation}
  \label{eq:phase-holo-5}
  \hat{z}(\xi',\eta') = |\hat{z}(\xi',\eta')|\exp[j\phi(\xi',\eta')]\ ,
\end{equation}
and
\begin{equation}
  \label{eq:phase-holo-6}
  \hat{R}(\xi',\eta') = |\hat{R}(\xi',\eta')|\exp[j\psi(\xi',\eta')]\ .
\end{equation}
The intensity pattern of their interference can be written as:
\begin{equation}
  \label{eq:phase-holo-7}
\begin{split}
  I_{2}(\xi',\eta')
  = &|\hat{z}(\xi',\eta') + \hat{R}(\xi',\eta')|^{2}\\
  = &|\hat{z}(\xi',\eta')|^{2} + |\hat{R}(\xi',\eta')|^{2}
  + \hat{z}^{*}(\xi',\eta')\hat{R}(\xi',\eta')
  + \hat{z}(\xi',\eta')\hat{R}^{*}(\xi',\eta')\\
  = & |\hat{z}(\xi',\eta')|^{2} + |\hat{R}(\xi',\eta')|^{2}
  +
  2|\hat{z}(\xi',\eta')|\,|\hat{R}(\xi',\eta')|\cos[\phi(\xi',\eta')-\psi(\xi',\eta')] 
  \ .
\end{split}
\end{equation}
From this formula we can easily extract the difference between the
unknown phase $\phi(\xi',\eta')$ and the known phase
$\psi(\xi',\eta')$:
\begin{equation}
  \label{eq:phase-holo-8}
  \cos[\phi(\xi',\eta')-\psi(\xi',\eta')] =
  \frac
  {I_{2}(\xi',\eta') -|\hat{z}(\xi',\eta')|^{2} - |\hat{R}(\xi',\eta')|^{2}}
  {2|\hat{z}(\xi',\eta')|\,|\hat{R}(\xi',\eta')|} \ .
\end{equation}
This gives us:
\begin{equation}
  \label{eq:phase-holo-9}
  \phi(\xi',\eta') = \psi(\xi',\eta') \pm \alpha(\xi',\eta')\ ,
\end{equation}
where
\begin{equation}
  \label{eq:phase-holo-10}
  \alpha(\xi',\eta') = \arccos
  \left[
    \frac
    {I_{2}(\xi',\eta') -|\hat{z}(\xi',\eta')|^{2} - |\hat{R}(\xi',\eta')|^{2}}
    {2|\hat{z}(\xi',\eta')|\,|\hat{R}(\xi',\eta')|}
  \right]\ .
\end{equation}
This expression is well defined, as $|\hat{R}(\xi',\eta')|$ is
assumed to be non-zero everywhere in the region of interest, and
the places where $|\hat{z}(\xi',\eta')|=0$ can be simply excluded
from our consideration as there is nothing to be recovered because
their phase has no influence. We assume that $\pm\alpha$, that is,
the difference between the phases $\phi$ and $\psi$, lies within
the interval $[-\pi,\pi]$, hence, no phase unwrapping is
necessary. The phase $\phi(\xi',\eta')$ can assume either one
(rarely) or two possible values at every location. The two
possible situations are shown in
Figure~\ref{fig:holography-possible-situations}.
\begin{figure}[H]
  \centering
  \subfloat[]{
    \label{fig:projection-fourier}
    \includegraphics[width=0.45\textwidth]{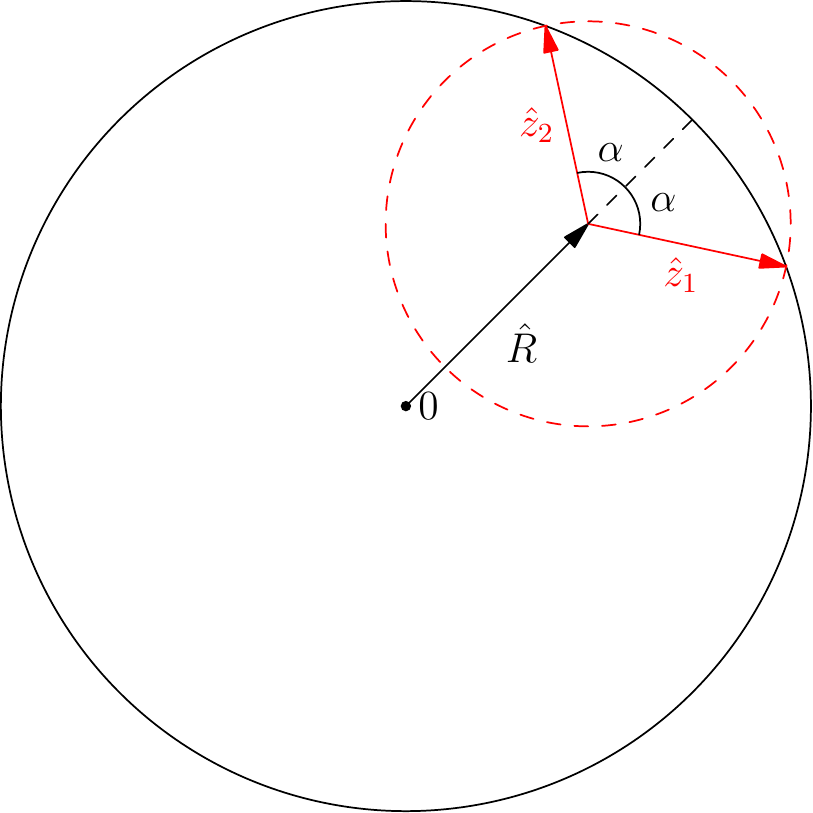}
  }
  \qquad{}
  \subfloat[]{
    \label{fig:convexrelaxation}
    \includegraphics[width=0.45\textwidth]{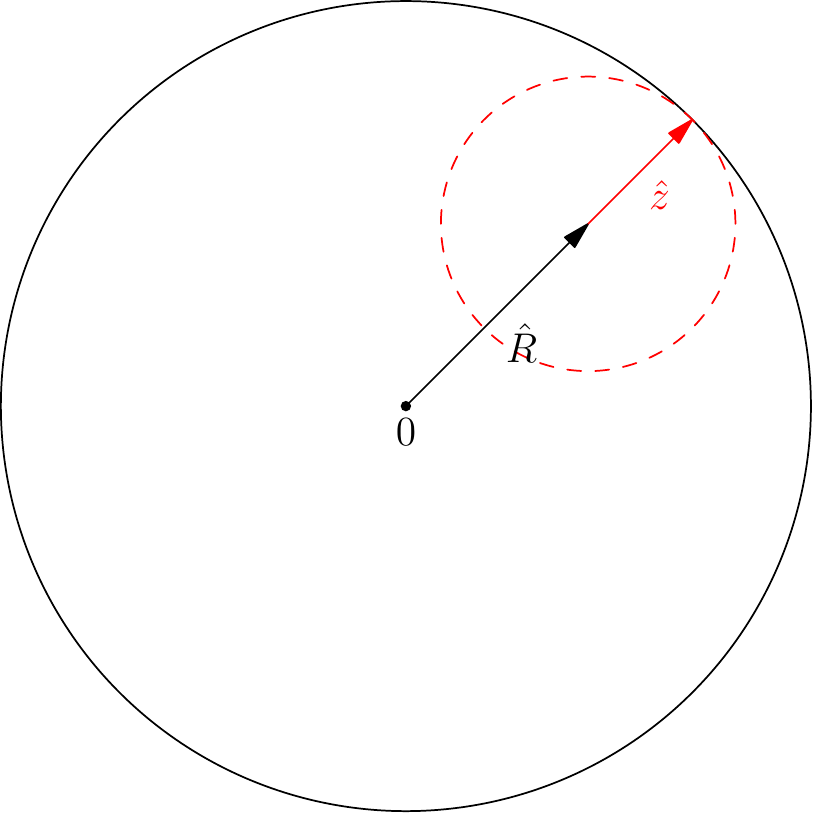}
  }
  \caption[Possible scenarios in Fourier domain holography]{Given a reference beam (black) whose magnitude and phase
    are known, and an unknown signal of known magnitude (dotted
    circle radius), one can try to find the phase of the unknown
    signal by measuring the magnitude of the sum (solid circle
    radius). Evidently, in most cases there are two possible
    solutions (a). However, in certain cases, the is only one
    solution (b).}
  \label{fig:holography-possible-situations}
\end{figure}

Hence, if the
intensity $I_{2}(\xi',\eta')$ is sampled at $N$ points there are
generally $2^{N}$ possible solutions $\hat{z}(\xi',\eta')$ and,
consequently, the same number of possible reconstructions
$z(\xi,\eta)$. (Here we consider the worst case scenario where all
sampled values give rise to two solutions.) To guarantee a unique
and meaningful reconstruction we must use additional information
about the sought signal $z(\xi,\eta)$. In the phase retrieval problem,
as well as in holography, 
it is usually assumed that $z(\xi,\eta)$ has limited support, namely,
part of the image is occupied by zeros.  In practice, it is
usually assumed that, in each direction, half or more pixels of
$z(\xi,\eta)$ are zeros. To capture this information in the Fourier
domain one should ``over-sample'' by a factor of two (or more) in
each dimension. Hence, if the known (not necessarily tight)
support area of $z$ is $n\times n$ pixels, then in the Fourier
plane it must be sampled with a sensor of size $2n\times 2n$
pixels. Such ``over-sampling'' usually guarantees unique (up to
trivial transformations: shifts, constant phase factor, and axis
reversal) reconstruction in the case of the classical phase
retrieval problem, where only $|\hat{z}|$ is available
\shortcite{hayes82reconstruction}. It is not known whether this
two-fold oversampling is absolutely necessary in our case where
two measurements are available. However, our experiments indicate
that for a general complex-valued signal it still seems to be
necessary to over-sample by a factor of two. Hence, the
reconstruction problem reads: find $z(\xi,\eta)$ such that $|\hat{z}|$
is known, $\angle(\hat{z})=\psi\pm\alpha$, and $z(\xi_{O},\eta_{O}) =
0$. Here, $\alpha$ is the known phase difference between $\hat{z}$
and $\hat{R}$ as defined by Equations~(\ref{eq:phase-holo-9})
and~(\ref{eq:phase-holo-10}); $(x_{O},y_{O})$ denotes the known off-support
area in the $(\xi,\eta)$ plane.

The problem is combinatorial in nature, and many different methods
can be applied to find a solution. Our method is based on
replacing the equality $\angle(\hat{z})=\psi\pm\alpha$ with the
less strict inequality
\begin{equation}
  \label{eq:phase-holo-11}
  \psi - \alpha \leq \angle(\hat{z}) \leq \psi + \alpha\ .
\end{equation}
By this relaxation we reduce the original problem into the phase
retrieval problem with approximately known phase.  For this
situation we have developed an efficient Quasi-Newton optimization
method based on convex relaxation. Note that the problem we are facing
here:
\begin{equation}
  \label{eq:phase-holo-12}
  \begin{split}
    \min_{x} &\quad \||\mathcal{F}\{x\}| -  |\hat{z}|\|^{2}\\
    \mathrm{subject\ to} &\quad \psi - \alpha \leq \angle(\mathcal{F}\{x\}) \leq
    \psi+\alpha \, , \\
    &\quad x(x_{O}, y_{O}) = 0 \, ,
  \end{split}
\end{equation}
is exactly the same as the one we solved in
Chapters~\ref{cha:appr-four-phase-first},
and~\ref{cha:appr-four-phase-explanation}. Hence, we use the same
convex relaxation as we did before. Likewise, the solution of the
above minimization problem $x$ is not guaranteed to be equal to the
sought signal $z$. Due to the relaxation we performed, the phase of $\hat{x}$ is
allowed to assume the continuum of values in the interval
$[\psi-\alpha, \psi+\alpha]$ instead of the two discrete values
$\psi\pm\alpha$.  However, due to the uniqueness of the phase
retrieval problem, the phase of $x$ may differ from the phase of $z$
only by a constant.  That is, $x(\xi,\eta)=z(\xi,\eta)\exp[jc]$ for
some scalar $c$ (see~\shortcite{hayes82reconstruction} for details). This
does not pose any problems, as only the relative phase distribution inside
the support area of $z(\xi,\eta)$ is usually of interest. Moreover, in
the case where the absolute phase is required, $c$ can be recovered by
adding a post-processing step that solves the one-dimensional
optimization problem:
\begin{equation}
  \label{eq:phase-holo-13}
  \displaystyle\min_{c} \||\hat{x}\exp[jc] + \hat{R}| -
  I_{2}^{1/2}\|^{2}\ .
\end{equation}
Note that we intentionally do not add a penalty term like
$\||\hat{x} + \hat{R}| - I_{2}^{1/2}\|^{2}$ into the main
minimization scheme as defined by Equation~(\ref{eq:phase-holo-12}).  Adding
such a term would introduce a strong connection between $x$ and
the reference beam $\hat{R}$. This connection will inevitably
deteriorate the quality of reconstruction when the reference beam
$\hat{R}$ contains errors.

There is, however, an important dissimilarity between the current
situation and the one we considered in the previous chapter: the phase
uncertainty interval can be as large as $\pi$ radians. Nevertheless,
our experiments indicate that the reconstruction is stable and its
speed is very fast (see Figure~\ref{fig:phase-holo-we-reconstruction-speed}).
Moreover, it can be further accelerated, and our experiments indicate
that more aggressive oversampling (zero padding in the object domain)
results in faster convergence in terms of the number of iterations. In
fact, use of a special reference beam can, in theory, result in a
trivial non-iterative reconstruction in a way similar to
holography. However, such special reference beams may not be easily
realizable in physical systems and the quality of the reconstructed
signal is strongly influenced by the quality of the reference beam. We
discuss this setup in the next section and compare its sensitivity to
possible errors in $\hat{R}$ against our method in
Section~\ref{sec:phase-holo-results}.

\section{Relation to holography}
\label{sec:relation-holography}

Our method was initially developed for the phase retrieval
problem. However, the use of interference patterns creates a
strong connection with holography. Therefore, it may be pertinent
to discuss the advantages that our method provides over the classical
holographic reconstruction. Note that our discussion is limited to
basic holography only, and no attempt is made to cover all possible
setups and techniques that can be used in digital holography. We
nonetheless believe that this novel approach can compete with or
improve upon existing algorithms used in digital holography.

In classical holography one uses a specially designed reference
beam so as to allow easy non-iterative recovery of the sought
image. This has an obvious advantage over iterative methods,
especially when the speed of the reconstruction is of high
importance. However, reliance on the reference beam means that
reconstruction quality may deteriorate badly when the reference
beam contains errors, that is, when it differs from the ``known''
values. To review the non-iterative reconstruction method used in
holography, recall that the recorded intensity $I_{2}$ is the
result of superimposition of $\hat{z}$ and $\hat{R}$ as defined by
Equation~(\ref{eq:phase-holo-7}). In optical Fourier holography, this
intensity is recorded on optical material. The recorded image is
used then as an amplitude modulator for an illuminating beam
$\hat{A}(\xi',\eta')$, which then undergoes a Fourier transform to
form a new signal $B(x',y')$. In a digital computer we may use the
same technique. Moreover, we are free to use either the forward or
the inverse Fourier transform, as it makes no practical difference
(the resulting images will be reversed conjugate copies of each
other). Here we use the inverse Fourier transform:
\begin{equation}
  \label{eq:phase-holo-14}
  \begin{aligned}
    B(\xi,\eta)
    =&\,  \mathcal{F}^{-1}\{\hat{A}(\xi',\eta')\,I_{2}(\xi',\eta')\}\\
    =&\,   \mathcal{F}^{-1}
    \left\{
      \hat{A}
      \left[
        |\hat{z}|^{2} + |\hat{R}|^{2}
        + \hat{z}^{*}\hat{R} + \hat{z}\hat{R}^{*}
      \right]
    \right\}\\
    = &\,  
    A(\xi,\eta) \otimes z(\xi,\eta) \otimes z^{*}(-x,-y) +
    A(\xi,\eta) \otimes R(\xi,\eta) \otimes R^{*}(-x,-y) +\\
    &\,
    A(\xi,\eta) \otimes z^{*}(-x,-y) \otimes R(\xi,\eta) +
    A(\xi,\eta) \otimes z(\xi,\eta) \otimes R^{*}(-x,-y) \, ,
  \end{aligned}
\end{equation}
where $\otimes$ denotes convolution. Note that in this case the
fourth and the third terms are equal to the sought wavefront
$z(\xi,\eta)$ and its Hermitian counterpart $z^{*}(-x,-y)$ convolved
with $A(\xi,\eta)\otimes R(\xi,\eta)$ and $A(\xi,\eta)\otimes R^{*}(-x,-y)$,
respectively.  The best possible choice is $A(\xi,\eta) = \delta(\xi,\eta)$
and $R(\xi,\eta)= \delta(\xi-\xi_{0},\eta-\eta_{0})$, where $\delta(\xi,\eta)$ is the
Dirac delta function.  In this case the obtained wave becomes:
\begin{equation}
  \label{eq:phase-holo-15}
  B(\xi,\eta) =
  z(\xi,\eta) \otimes z^{*}(-\xi,-\eta) + \delta(\xi,\eta) +
  z^{*}(-\xi+\xi_{0}, -\eta+\eta_{0}) + z(\xi-\xi_{0}, \eta-\eta_{0})\ .
\end{equation}
Hence, if $\xi_{0}$ and/or $\eta_{0}$ are large enough, the four terms
in the above sum will not overlap in the $(\xi,\eta)$ plane. Thus, we
can easily obtain the sought signal $z(\xi,\eta)$, albeit shifted by
$(\xi_{0}, \eta_{0})$, provided that the spatial extent of $z(\xi,\eta)$ is
limited by the box $\xi\in[-L_{\xi}, L_{\xi}]$, $\eta\in[-L_{\eta},L_{\eta}]$.
The spatial extent of the autocorrelation $z(\xi,\eta) \otimes
z^{*}(-\xi,-\eta)$ is twice as large, that is, limited by the box
$\xi\in[-2L_{\xi}, 2L_{\xi}]$, $\eta\in[-2L_{\eta},2L_{\eta}]$. Hence, to avoid
overlapping we must have $\xi_{0} > 3L_{\xi}$, or $\eta_{0}>3L_{\eta}$.
Thus, in theory, one can generate an ideal delta function in the
$(\xi,\eta)$ plane located at a sufficient distance from the support area
of $z(\xi,\eta)$. In the Fourier domain, this delta function
corresponds to a plane wave arriving at a certain angle at the
plane of measurements. If such a construction is possible, then a
simple inverse transform of the intensity obtained in the Fourier
plane is sufficient to obtain the sought signal $z(\xi,\eta)$. However,
as mentioned earlier, this approach has some drawbacks. First, it
is impossible to create an ideal delta function. Any physical
realization will necessarily have a finite spatial extent, and
this will result in a ``blurred'' reconstructed image. Note that
the term ``blurring'' describes well the resulting image in the
case where $z(\xi,\eta)$ is real-valued or has constant phase. In the
more general case, where the phase of $z(\xi,\eta)$ varies at
non-negligible speed, the result appears more distorted (see
Figure~\ref{fig:holography-reconstruction-intensity}).
The second
drawback is the sensitivity of this method to errors in $\hat{R}$.
In Section~\ref{sec:phase-holo-results} we demonstrate how the quality of
reconstruction depends on the error in $\hat{R}$ (see
Figures~\ref{fig:objectdomain-error}, 
\ref{fig:objectdomain-error-corrected}, and
\ref{fig:visual-results-phase-error}). Our method, on the other
hand shows very little sensitivity to the reference beam shape.
Moreover, its modification described in the next section allows
the reference beam to contain severe errors without deteriorating
significantly the quality of reconstruction.

\section{Reconstruction method for imprecise reference beam}
\label{sec:reconstr-meth-impr} Here we consider the situation
where the reference beam is not known precisely, that is, we
assume that the phase of $\hat{R}$ contains some unknown error. It
is easy to verify that if the reference beam phase
$\psi(\xi',\eta')$ has error $\epsilon(\xi',\eta')$ then the sought
phase $\phi(\xi',\eta')$ becomes
\begin{equation}
  \label{eq:phase-holo-15}
  \phi(\xi',\eta') = \psi(\xi',\eta') + \epsilon(\xi',\eta') \pm
  \alpha(\xi',\eta')\ ,
\end{equation}
in a manner similar to Equation~(\ref{eq:phase-holo-9}). We do not consider
errors in the magnitude $|\hat{R}|$ for several reasons. Many
aberrations manifest themselves through phase
distortion~\shortcite{goodman05introduction}. Also, the magnitude of
$\hat{R}$ can
be measured. Moreover, looking at the above equation, it is
obvious that any error in $\psi$ can be viewed as an error in
$\alpha$. That is, the situation would be the same if the reference
beam phase $\psi$ were known precisely, while the difference
$\alpha$ would contain some errors. This observation is relevant
because errors in the phase $\alpha$ can arise from many different
sources, including imperfect measurements and errors in the
reference beam magnitude.

The true error $\epsilon(\xi',\eta')$ is, of course, unknown. Hence,
we assume just an upper bound (assumed known) on the absolute
phase error:
\begin{equation}
  \label{eq:phase-holo-17}
  \psi-\epsilon-\alpha
  \leq\angle(\hat{z})\leq\psi+\epsilon+\alpha\ ,
\end{equation}
as in Equation~(\ref{eq:phase-holo-11}). This time, however, the phase
uncertainty interval may be larger than $\pi$ radians which makes our
method inapplicable. On the other hand, limiting the phase uncertainty
interval by $\pi$ radians will prevent us from reconstructing the
precise image, because the true phase may lie outside this interval. A
possible solution is to measure the intensity of the reference beam
and then to reconstruct its phase using the method presented in
Chapter~\ref{cha:appr-four-phase-explanation}, because this problem
itself can be seen as a phase retrieval with approximately known
phase. However, taking another measurement may be undesirable, and
therefore we developed the following reconstruction method:
\begin{description}
\item[Step 1:] Set the phase uncertainty interval as defined by
  Equation~(\ref{eq:phase-holo-11}) (as if there were no errors in the reference
  beam phase).

\item[Step 2:] Solve the resulting minimization problem, obtaining a
solution $x(\xi,\eta)$).

\item[Step 3:] If not converged, set the phase uncertainty interval to
  $[\angle(\hat{x}) - \pi/2, \angle(\hat{x}) + \pi/2]$. Clip it, if
  applicable, to the limits defined by Equation~(\ref{eq:phase-holo-17}) and go to
  Step 2.
\end{description}

In this algorithm we perform a number of outer iterations, each
time updating the phase uncertainty interval. This approach leads
to a successful reconstruction method even in cases where the
reference beam contains severe errors. The results are much better
than those of non-iterative holographic reconstruction (see
Figures~\ref{fig:objectdomain-error},
and~\ref{fig:objectdomain-error-corrected}).  This improvement is
achieved by decoupling the reconstruction problem (which becomes
the pure phase retrieval with approximately known phase) and the
erroneous interferometric measurements.

\section{Experimental results}
\label{sec:phase-holo-results} 
The method was tested on a variety of images with similar
results. Here we present numerical experiments conducted on one
natural image so as to allow easy perception of the reconstruction
quality under various conditions. The image intensity (squared
magnitude) is shown in Figure~\ref{fig:experimenal-image}.
\begin{figure}[H]
  \centering
  \includegraphics{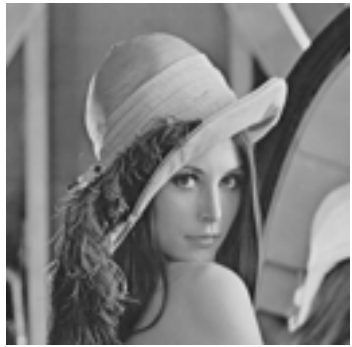}
  \caption[Test image]{Original image (intensity).}
  \label{fig:experimenal-image}
\end{figure}

The technical
details of the image are as follows: the size is $128\times 128$
pixels, and pixel values (amplitude) vary from $0.2915$ to $0.9634$ on
a scale of $0$ to $1$.
with mean value of $0.6835$. These parameters will become important
later, when we will consider the reference beam design, and when we will
assess the reconstruction quality. Since the original image is a
photograph, it does not have any phase information. Hence, we
generated three different phase distributions to account for the
assortment of possible real-world problems where our method can be
applied. The first distribution assumes that the image is non-negative
real-valued, that is, the phase is zero everywhere. The second
distribution is designed to mimic a relatively smooth phase. To this
end, the phase is set to be proportional to the image values (scaled
to the interval $[-\pi,\pi]$). Finally, in the third distribution the
phase is chosen at random, uniformly spread over the interval
$[-\pi,\pi]$. This distribution is designed to show the behavior of
our reconstruction method in cases where the true phase varies
rapidly. We also consider three possible reference beams, again, to
demonstrate the robustness of our method. The first reference beam is
an ideal delta-function in the $(\xi,\eta)$ plane, located at the
coordinate $(256, 256)$ so that the holographic condition is
satisfied. With this reference beam exact reconstruction is obtained
as long as the sampling in the Fourier domain is sufficiently dense
($512\times512$ pixels, or more). We do not present the visual results
of reconstruction for this reference beam as both methods produce
images that are indistinguishable from the true image.  The speed of
convergence of our method is shown in
Figure~\ref{fig:phase-holo-we-reconstruction-speed}.  Later, we show also how the
reconstruction quality of both methods is affected by Fourier phase
errors in the reference beam. Before this, we demonstrate the effect
of departure from the ideal delta function: the second reference beam
is a small square of size $3\times3$ pixels, located at the
coordinates $(256:258,256:258)$.  In this setup the reconstruction
quality of the holographic method is degraded, as evident from
Figure~\ref{fig:holography-reconstruction-intensity}.
\begin{figure}[H]
  \centering
  \subfloat[]{
    \label{fig:holography-real-square}
    \includegraphics{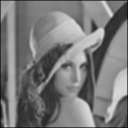}
  }
  \quad{}
  \subfloat[]{
    \label{fig:holography-complexSmooth-square}
    \includegraphics{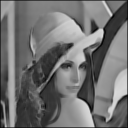}
  }
  \quad{}
  \subfloat[]{
    \label{fig:holography-complexRandom-square}
    \includegraphics{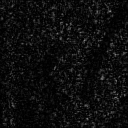}
  }
  \caption[Image reconstructed by the holographic
    technique using a small square as the reference beam]{Image (intensity) reconstructed by the holographic
    technique using a small square as the reference beam: (a) the
    image is real-valued, (b) image phase varies slowly, (c)
    image phase is random (varies rapidly).}
  \label{fig:holography-reconstruction-intensity}
\end{figure}

It is also evident
that faster variations in the object phase result in greater
deterioration in the reconstruction, in agreement with our
expectations. Our method, on the other hand, is insensitive to the
reference beam form. In Figure~\ref{fig:we-reconstruction-intensity} we
demonstrate our reconstruction results for the aforementioned small
square as the reference beam (the first row), and for another
reference beam that was formed in the Fourier plane by combining unit
magnitude with random phase (in the interval $[-\pi, \pi]$). This
beam, of course, is not suitable for holography, as its extent in the
object plane occupies the whole space.
\begin{figure}[H]
  \centering
  \subfloat[]{
    \label{fig:we-real-square}
    \includegraphics{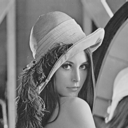}
  }
  \quad{}
  \subfloat[]{
    \label{fig:we-complexSmooth-square}
    \includegraphics{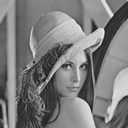}
  }
  \quad{}
  \subfloat[]{
    \label{fig:we-complexRandom-square}
    \includegraphics{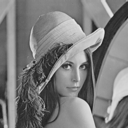}
  }\\
   \subfloat[]{
    \label{fig:we-real-random}
    \includegraphics{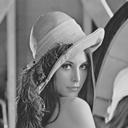}
  }
  \quad{}
  \subfloat[]{
    \label{fig:we-complexSmooth-random}
    \includegraphics{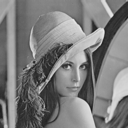}
  }
  \quad{}
  \subfloat[]{
    \label{fig:we-complexRandom-random}
    \includegraphics{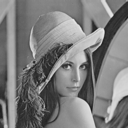}
  }
  \caption[Image reconstructed by our method]{Image reconstructed (intensity) by our method:
    (a), (b), and (c) --- reference beam is a small square and object
    phase is zero (a), smooth (b), random (c).
    (d), (e), and (f) --- reference beam is random and object phase
    is zero (d), smooth (e), random (f).}
  \label{fig:we-reconstruction-intensity}
\end{figure}

Reconstruction is very fast
and, in fact, is almost independent of the sought image and reference
beam type. Figure~\ref{fig:phase-holo-we-reconstruction-speed} demonstrates that
less than 20 iterations are required to solve the minimization problem
as defined by Equation~(\ref{eq:phase-holo-12}).
\begin{figure}[H]
  \centering
  \subfloat[]{
    \label{fig:we-speed-real}
    \includegraphics[height=0.25\textheight]{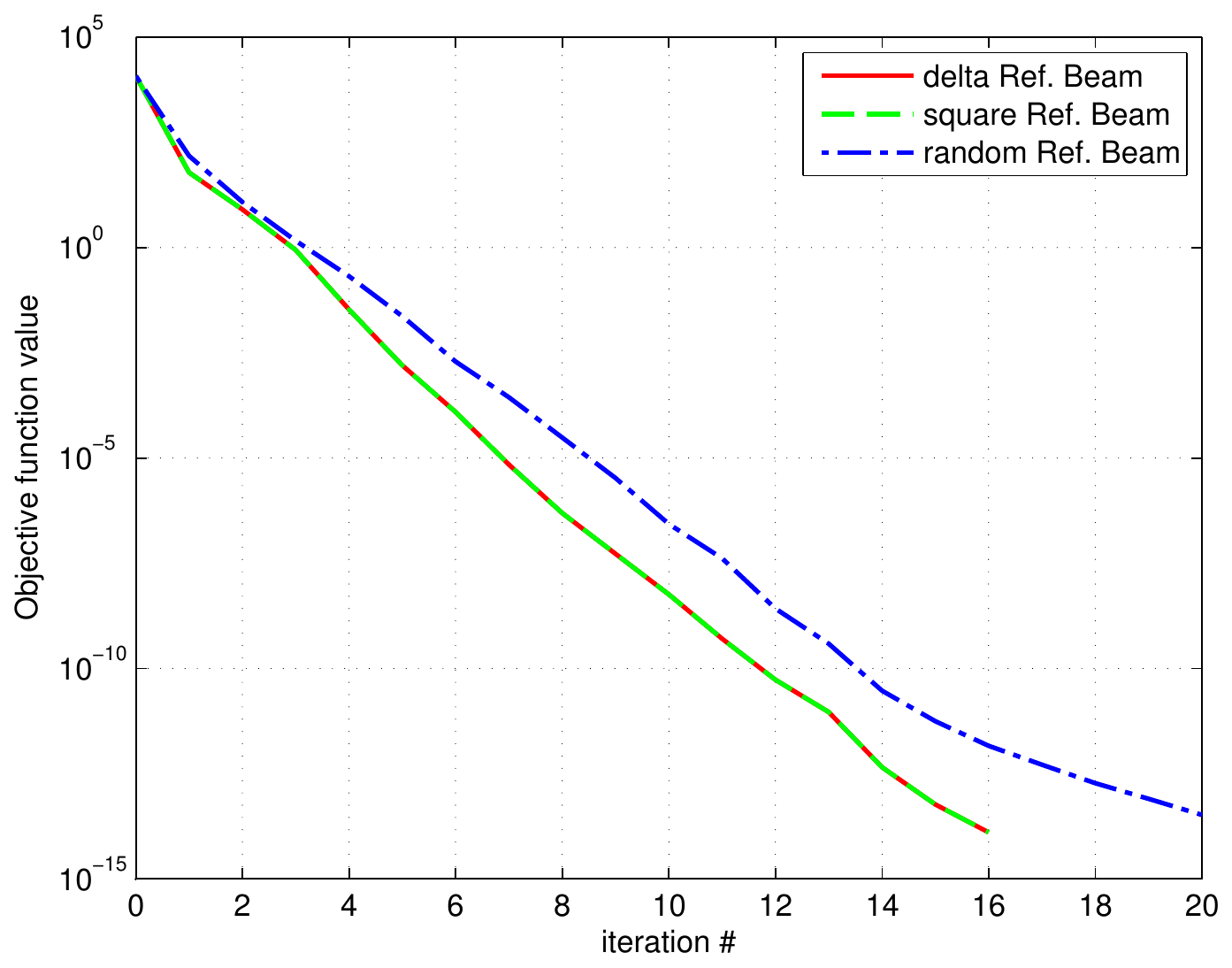}
  }
  \subfloat[]{
    \label{fig:we-speed-complexSmooth}
    \includegraphics[height=0.25\textheight]{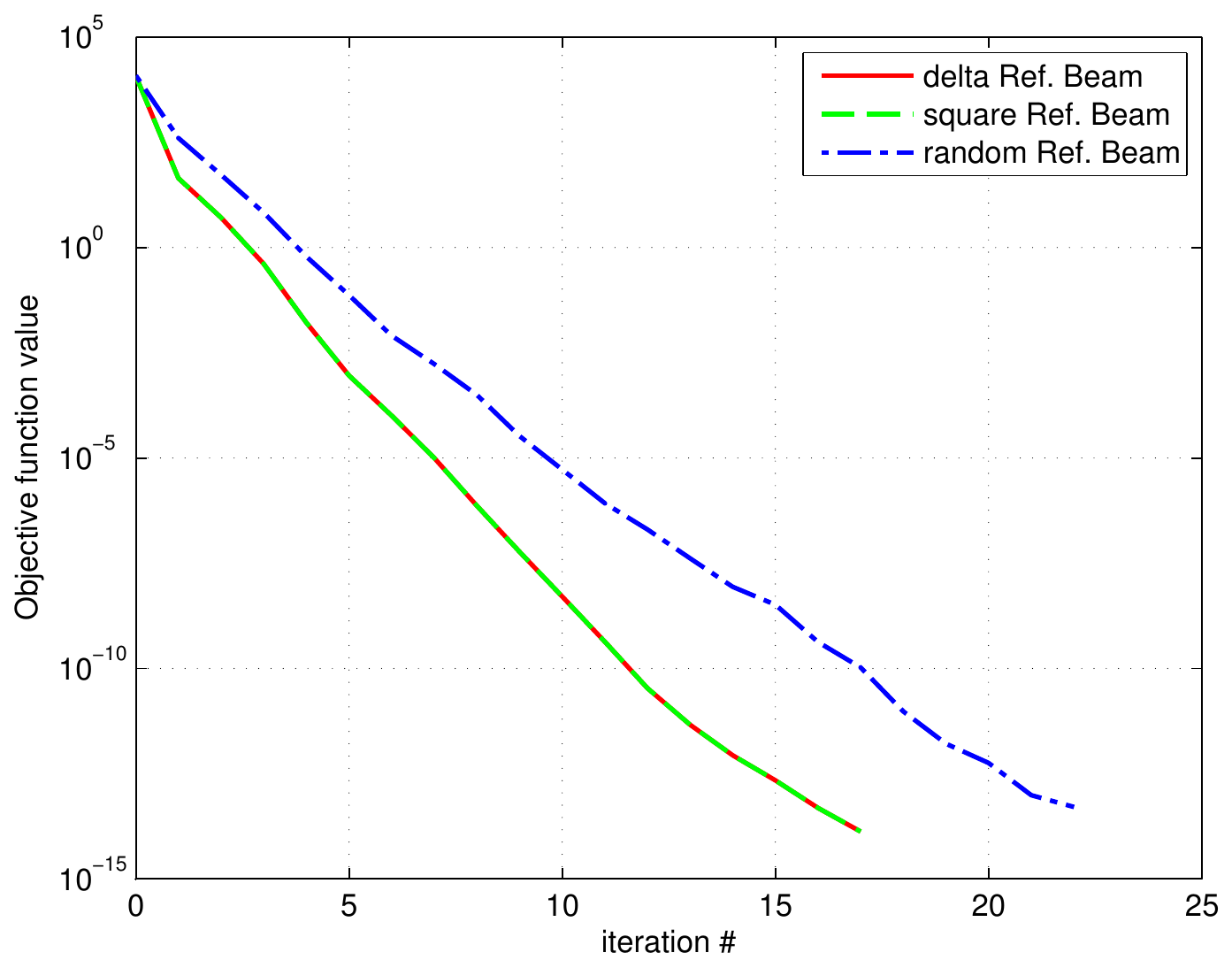}
  }
  \qquad{}
  \subfloat[]{
    \label{fig:we-speed-complexRandom}
    \includegraphics[height=0.25\textheight]{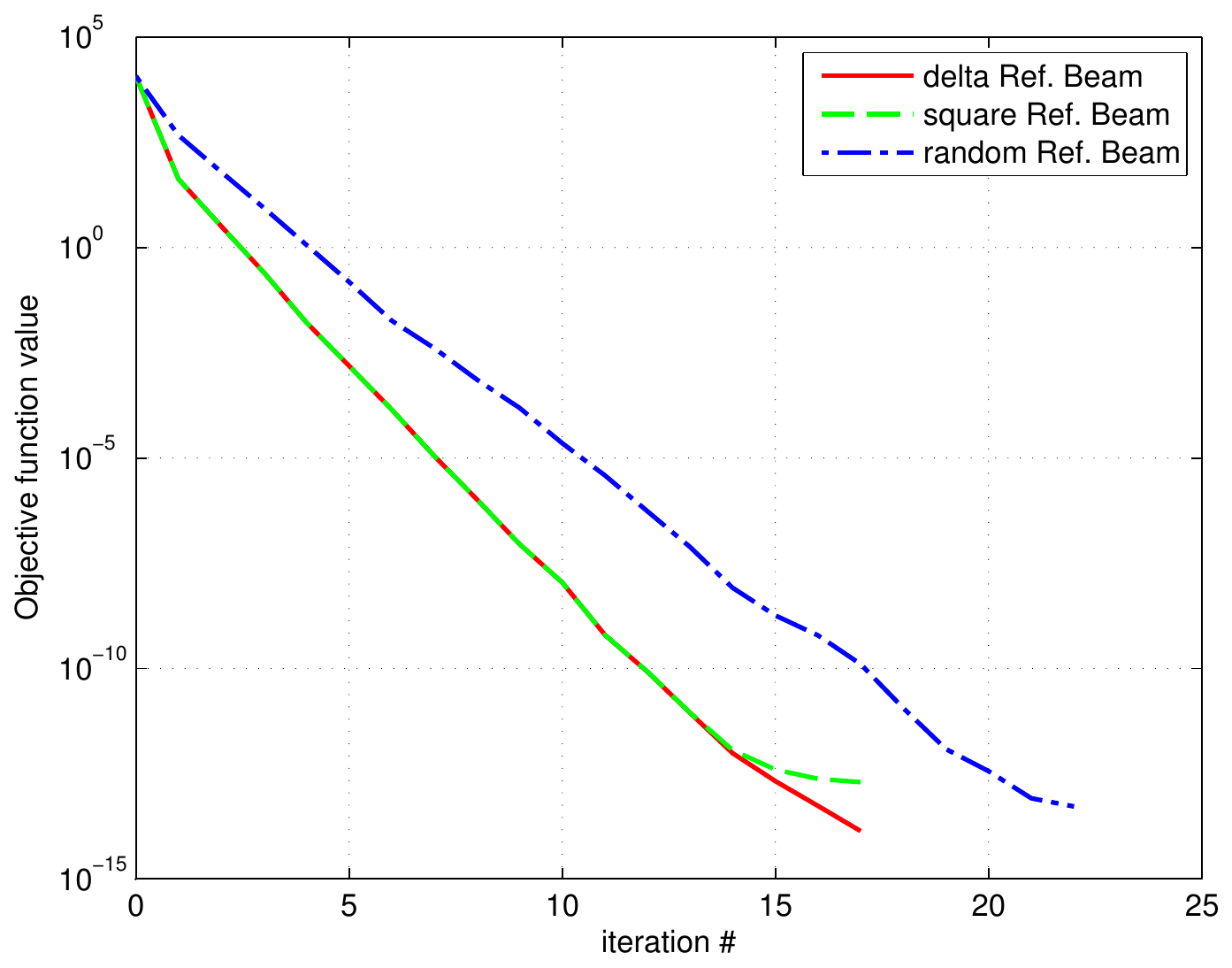}
  }
  \caption[Reconstruction speed of our method]{Reconstruction speed of our method: (a) real valued image,
    (b) image phase is smooth, (c) image phase is random.}
  \label{fig:phase-holo-we-reconstruction-speed}
\end{figure}

In all these examples we assume perfect knowledge of the reference
beam. Next we consider the situation where the actual reference
beam does not match the expected signal in the Fourier plane.
Following our discussion in Section~\ref{sec:reconstr-meth-impr},
we evaluate how the reconstruction quality of the holographic
approach and our method are affected by errors in the reference
beam Fourier phase. We use again the three aforementioned models
of the sought image (real-valued, smooth phase, and random phase)
and the three reference beams (delta function, small square, and
random). From Figure~\ref{fig:we-successrate} it is evident that in
all these cases we were able to solve the minimization problem to
sufficient accuracy as long as the phase error was below 25\%.
That is, our method can tolerate reference
beam Fourier phase errors of up to $\pi/2$ radians. The sharp
discontinuity that happens at this value has a
simple explanation: phase error greater than $\pi/2$ radians can
result in the phase uncertainty interval greater than $2\pi$ radians
(see Equation~(\ref{eq:phase-holo-17})).  Hence, all phase information is
lost.

\begin{figure}[H]
  \centering
  \subfloat[]{
    \label{fig:we-success-real}
    \includegraphics[height=0.25\textheight]{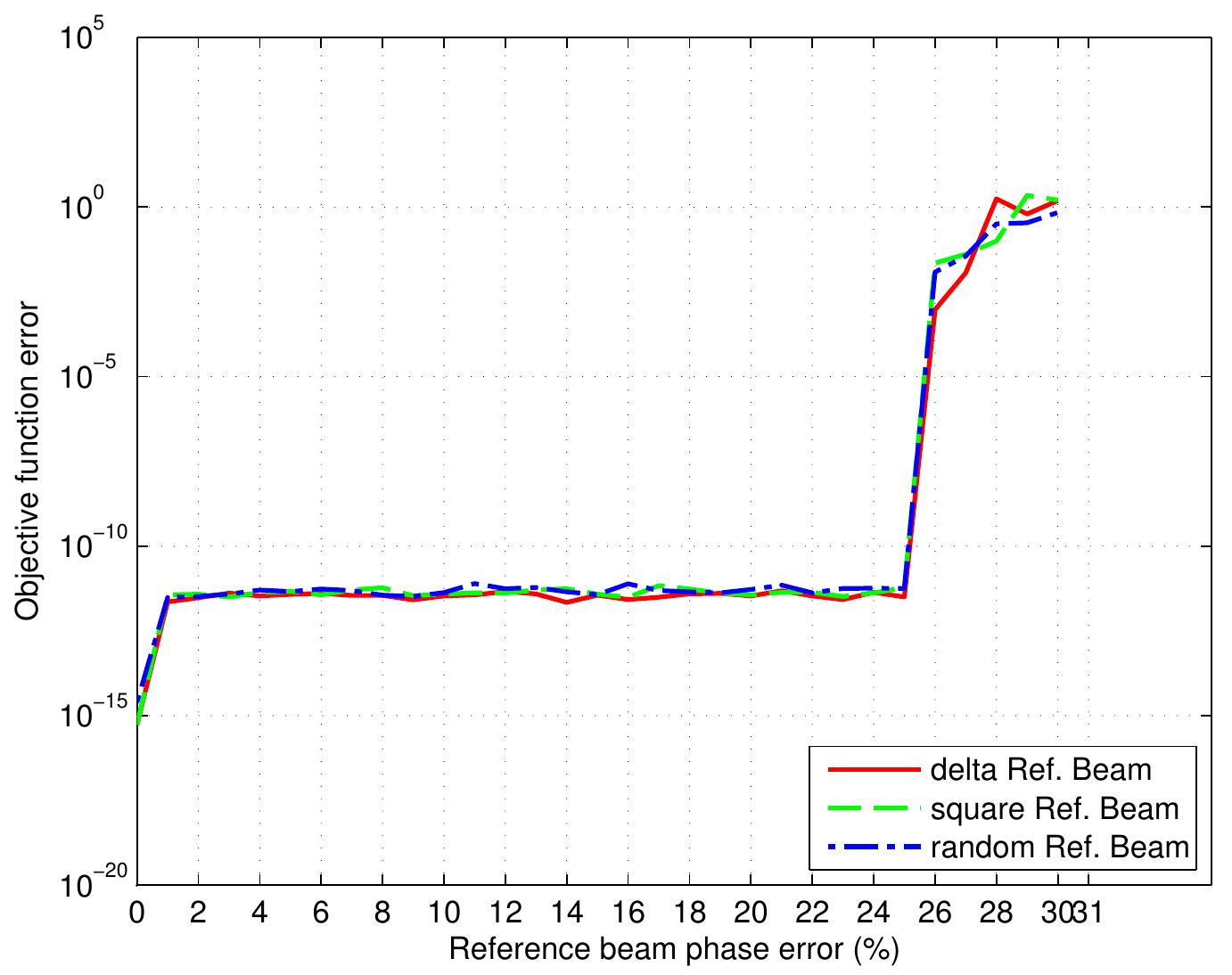}
  }
  \subfloat[]{
    \label{fig:we-success-complexSmooth}
    \includegraphics[height=0.25\textheight]{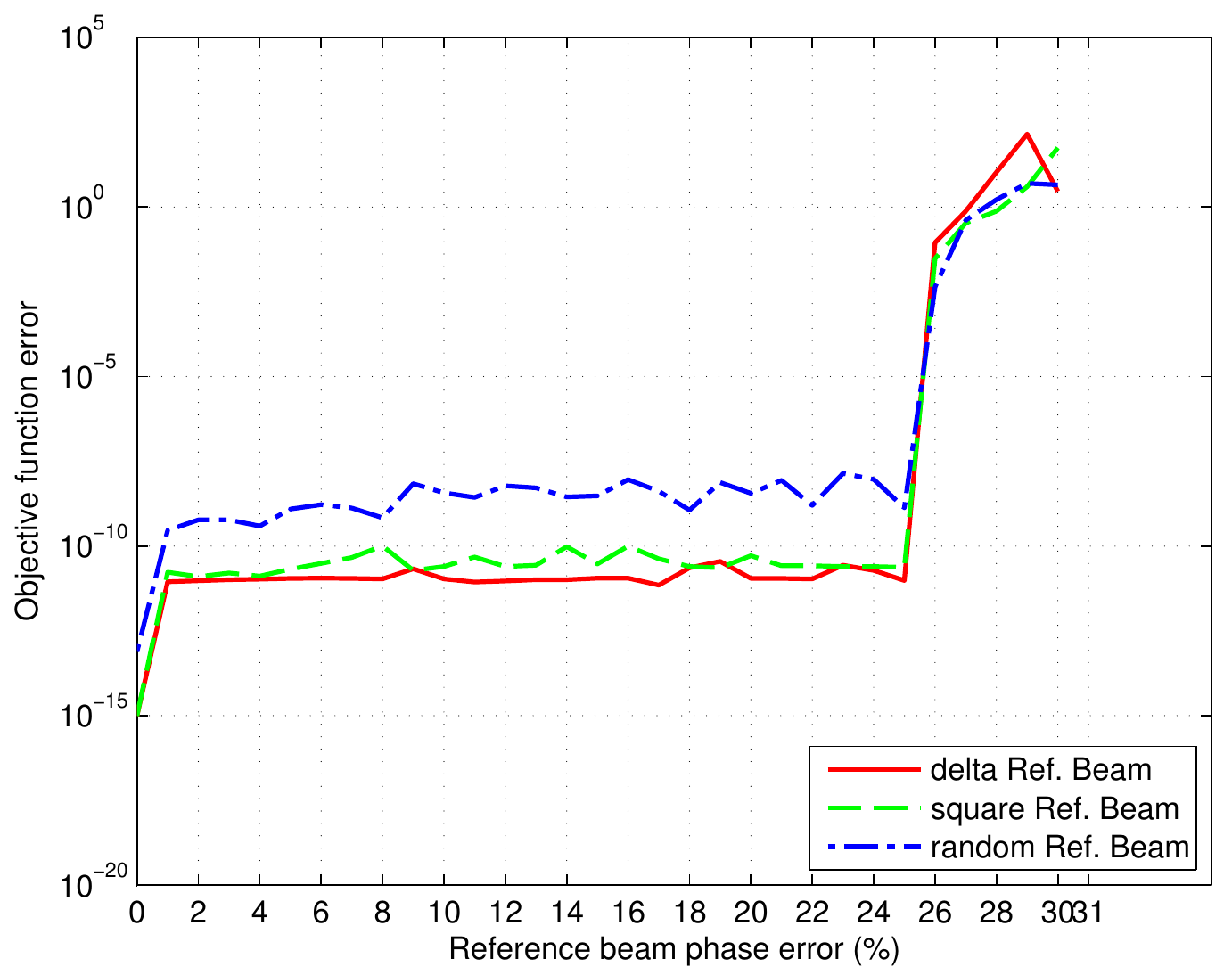}
  }
  \qquad{}
  \subfloat[]{
    \label{fig:we-success-complexRandom}
    \includegraphics[height=0.25\textheight]{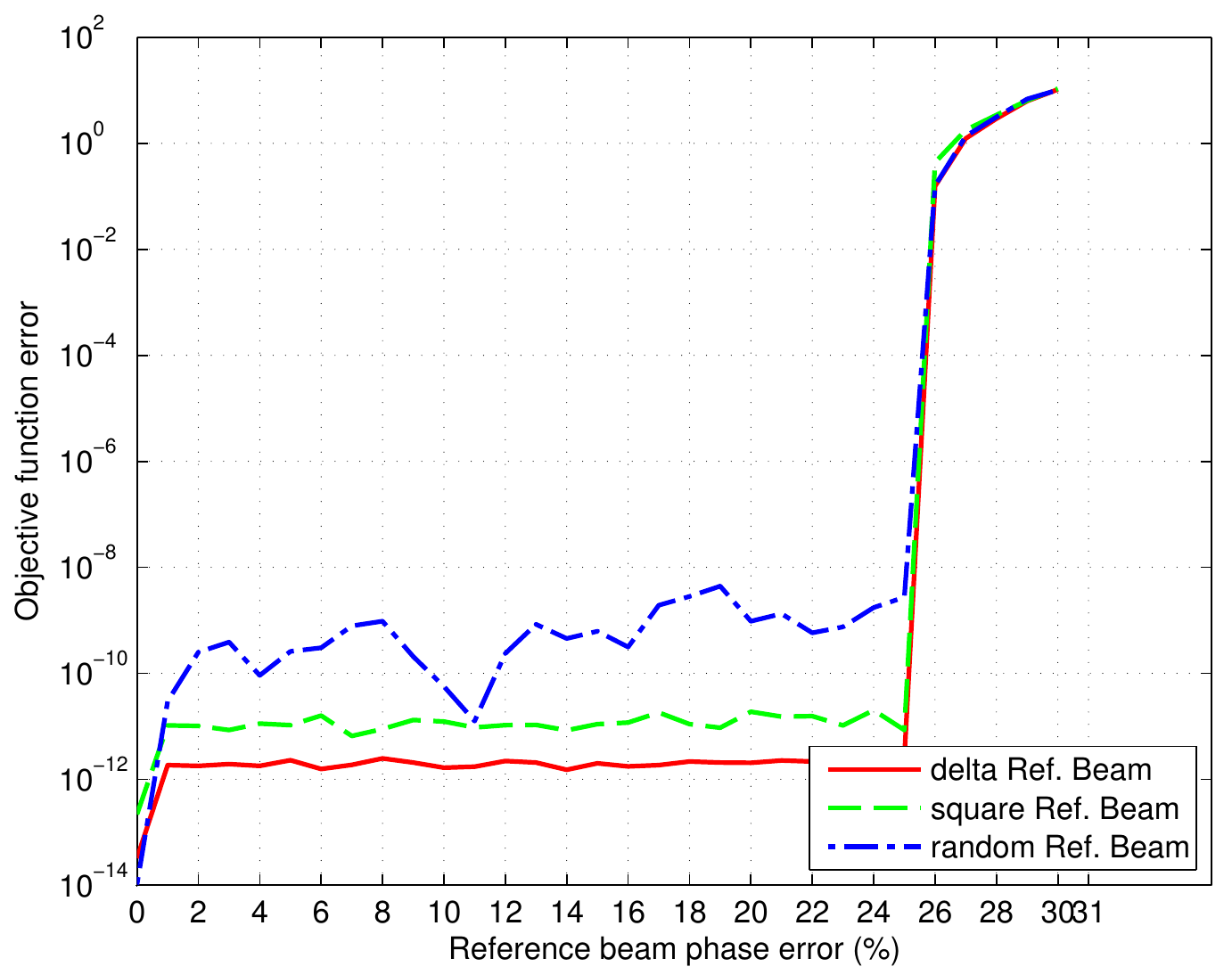}
  }
  \caption[Fourier domain error vs. phase
    error in the reference beam]{Fourier domain ($\||\hat{x}| - |\hat{z}|\|^{2}$) error vs. phase
    error in the reference beam: (a) real valued image,
    (b) image phase is smooth, (c) image phase is random. }
  \label{fig:we-successrate}
\end{figure}

A comparison
with the holographic reconstruction is given in
Figure~\ref{fig:objectdomain-error} where the error norm in the
object domain is depicted. Note that the objective function values
are about $10^{-10}$, hence, one would expect the object domain
error norm to be of order $10^{-5}$ (the difference stems from the
fact that the objective function uses \textit{squared} norm). This
is not so in the case of complex-valued images and the random
reference beam. This effect is due to the relaxation we perform in
the Fourier phase, as discussed in
Section~\ref{sec:basic-reconstr-algor}. It does not change the
relative phase distribution, but all phases can get a constant
addition. This can be corrected by solving the one dimensional
minimization defined by Equation~\eqref{eq:phase-holo-13}.
Figure~\ref{fig:objectdomain-error-corrected} depicts the corrected
values yielded by this process. Visual results comparing the
holographic reconstruction with our method are provided in
Figure~\ref{fig:visual-results-phase-error}.

\begin{figure}[H]
  \centering
  \subfloat[]{
    \label{fig:we-quality-real}
    \includegraphics[height=0.25\textheight]{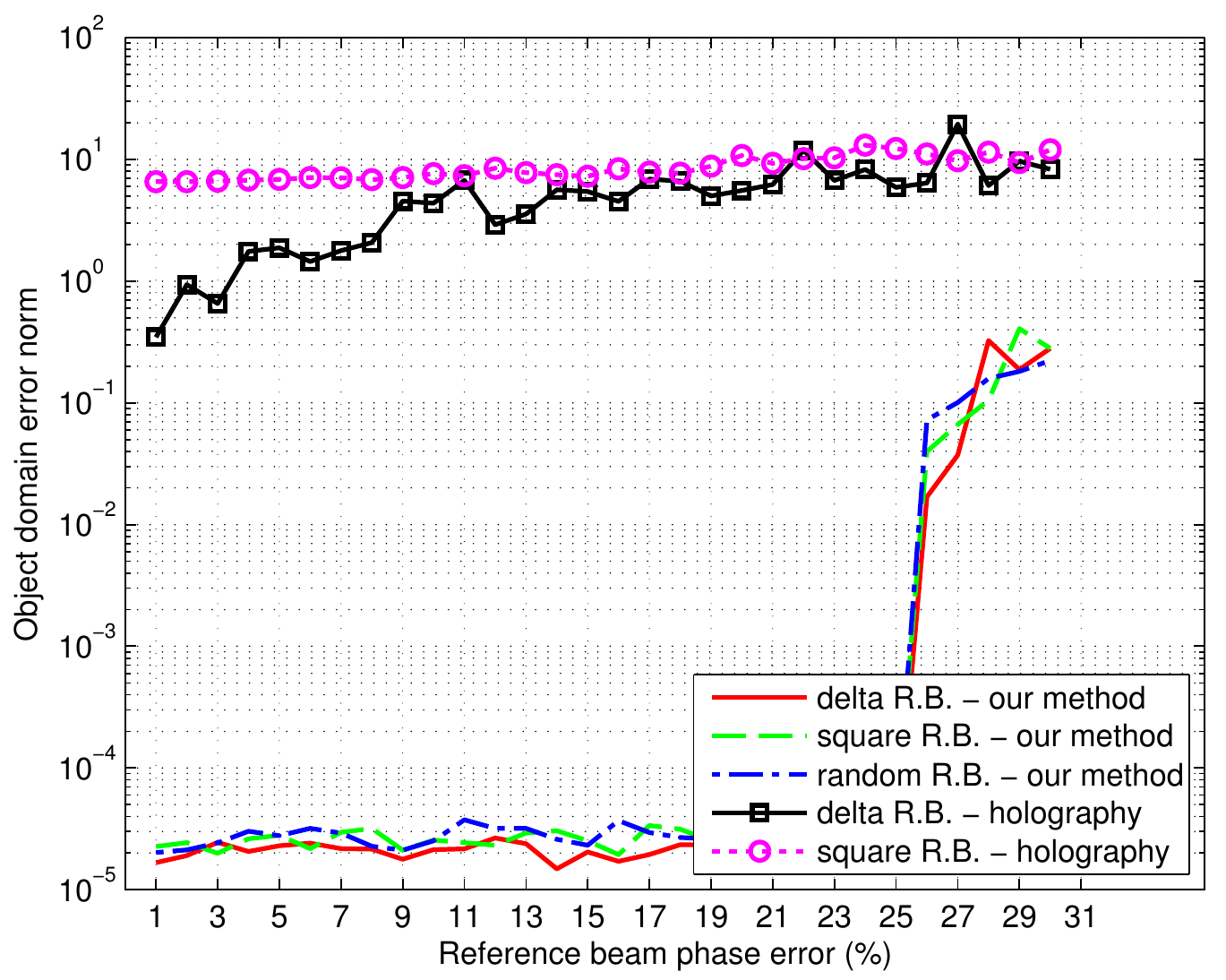}
  }
  \subfloat[]{
    \label{fig:we-quality-complexSmooth}
    \includegraphics[height=0.25\textheight]{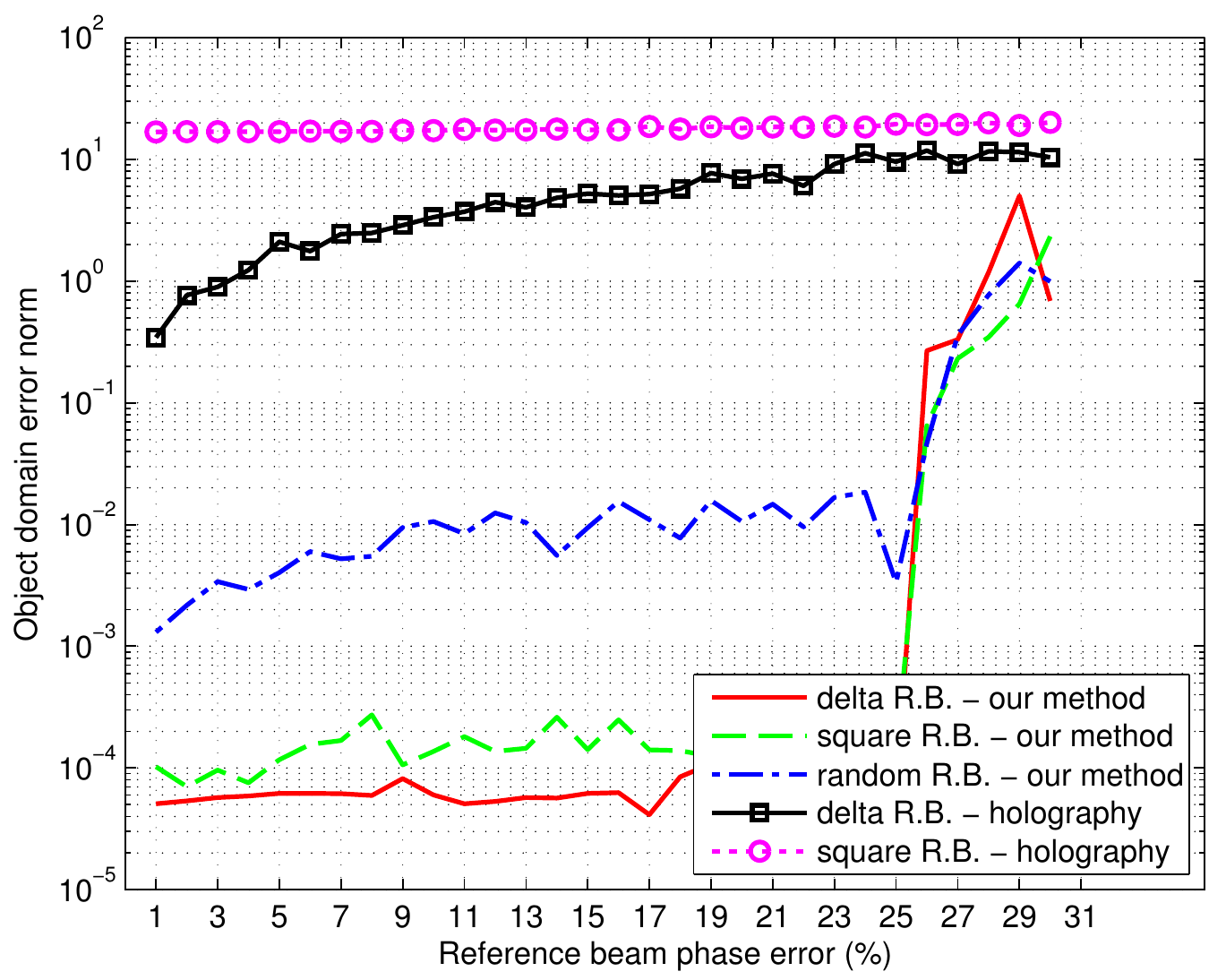}
  }
  \qquad{}
  \subfloat[]{
    \label{fig:we-quality-complexRandom}
    \includegraphics[height=0.25\textheight]{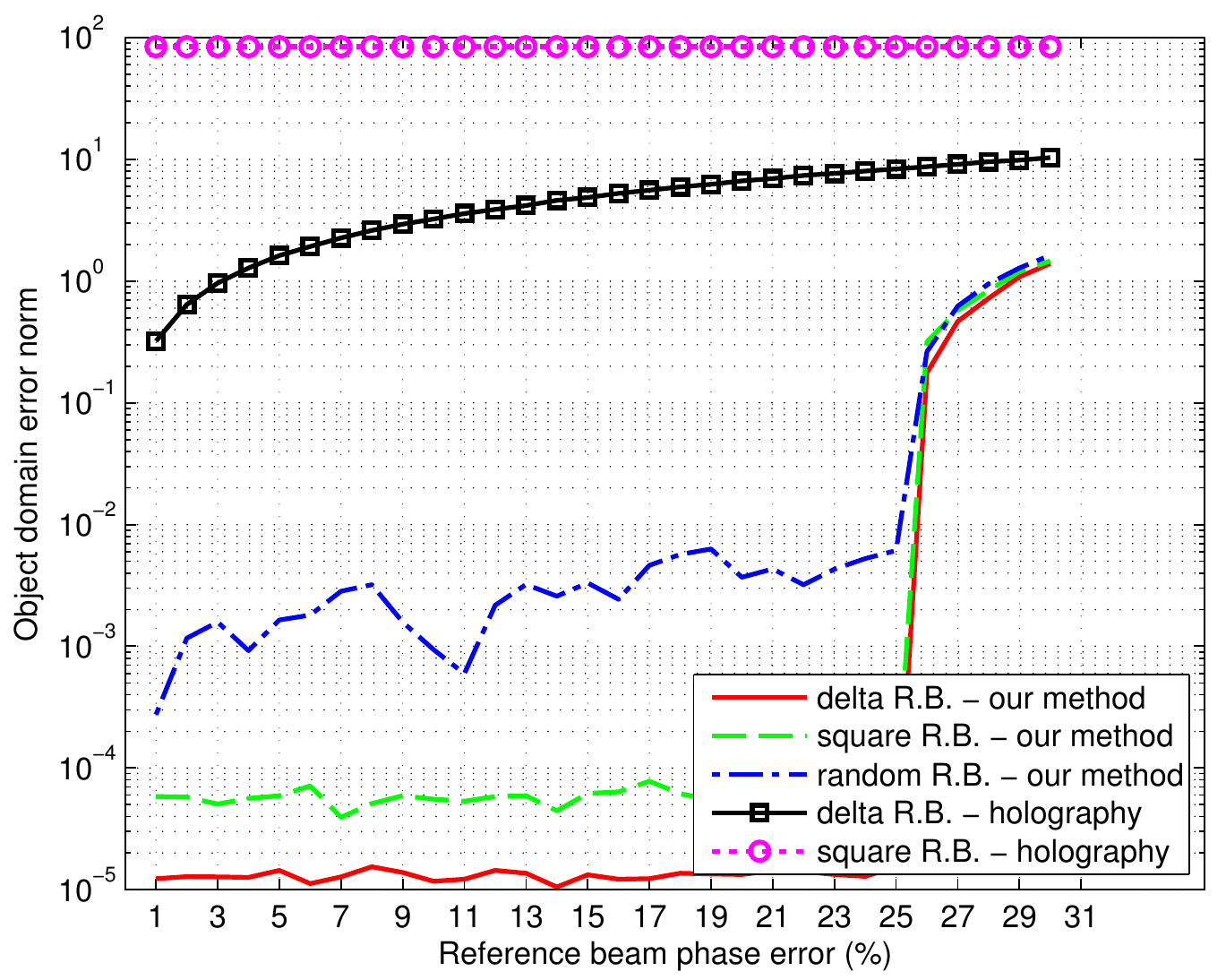}
  }
  \caption[Object domain error vs. phase error in the
    reference beam]{Object domain error ($\|x-z\|$) vs. phase error in the
    reference beam: (a) real valued image,
    (b) image phase is smooth, (c) image phase is random.}
  \label{fig:objectdomain-error}
\end{figure}

\begin{figure}[H]
  \centering
  \subfloat[]{
    \label{fig:we-qualitycorrected-real}
    \includegraphics[height=0.25\textheight]{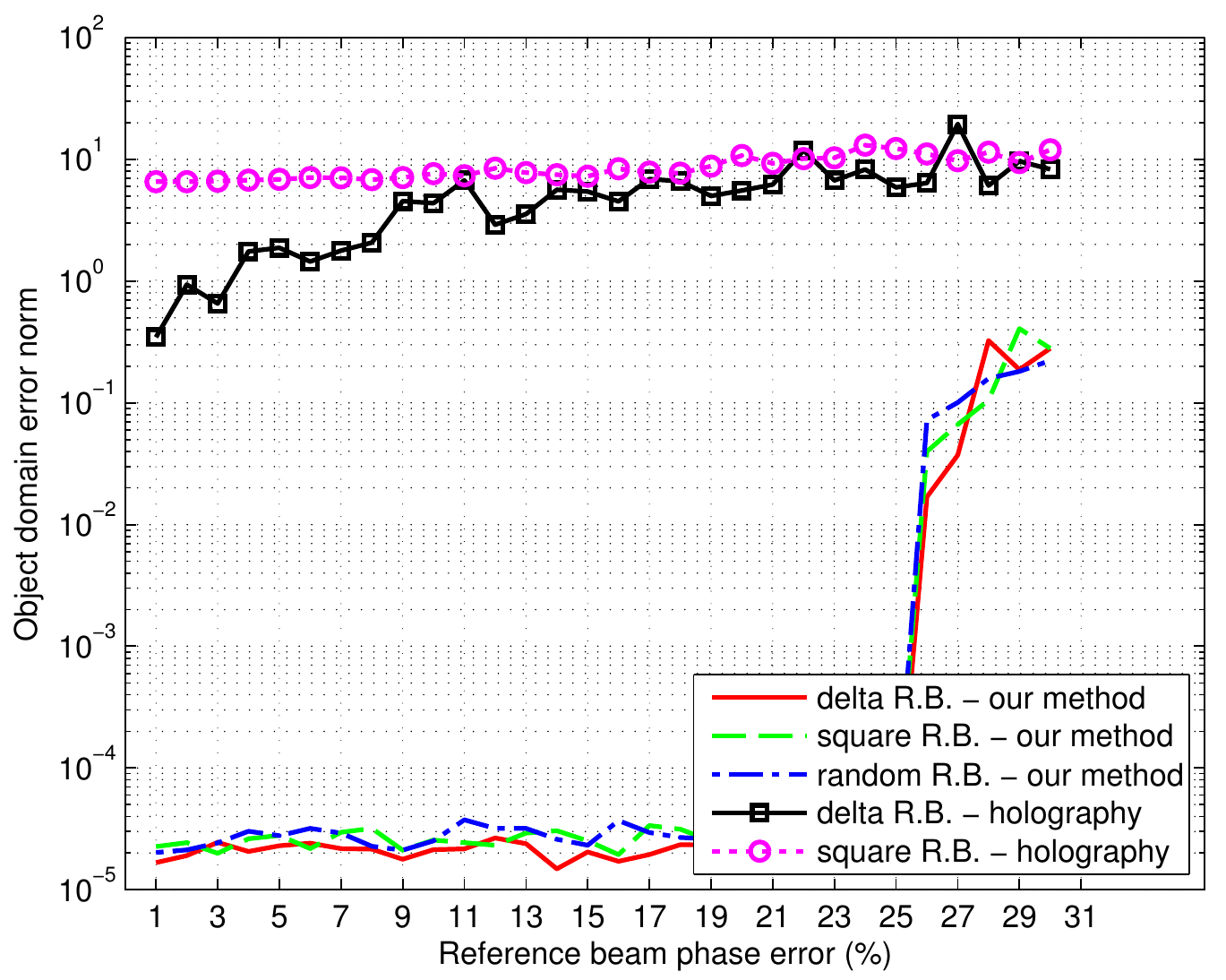}
  }
  \subfloat[]{
    \label{fig:we-qualitycorrected-complexSmooth}
    \includegraphics[height=0.25\textheight]{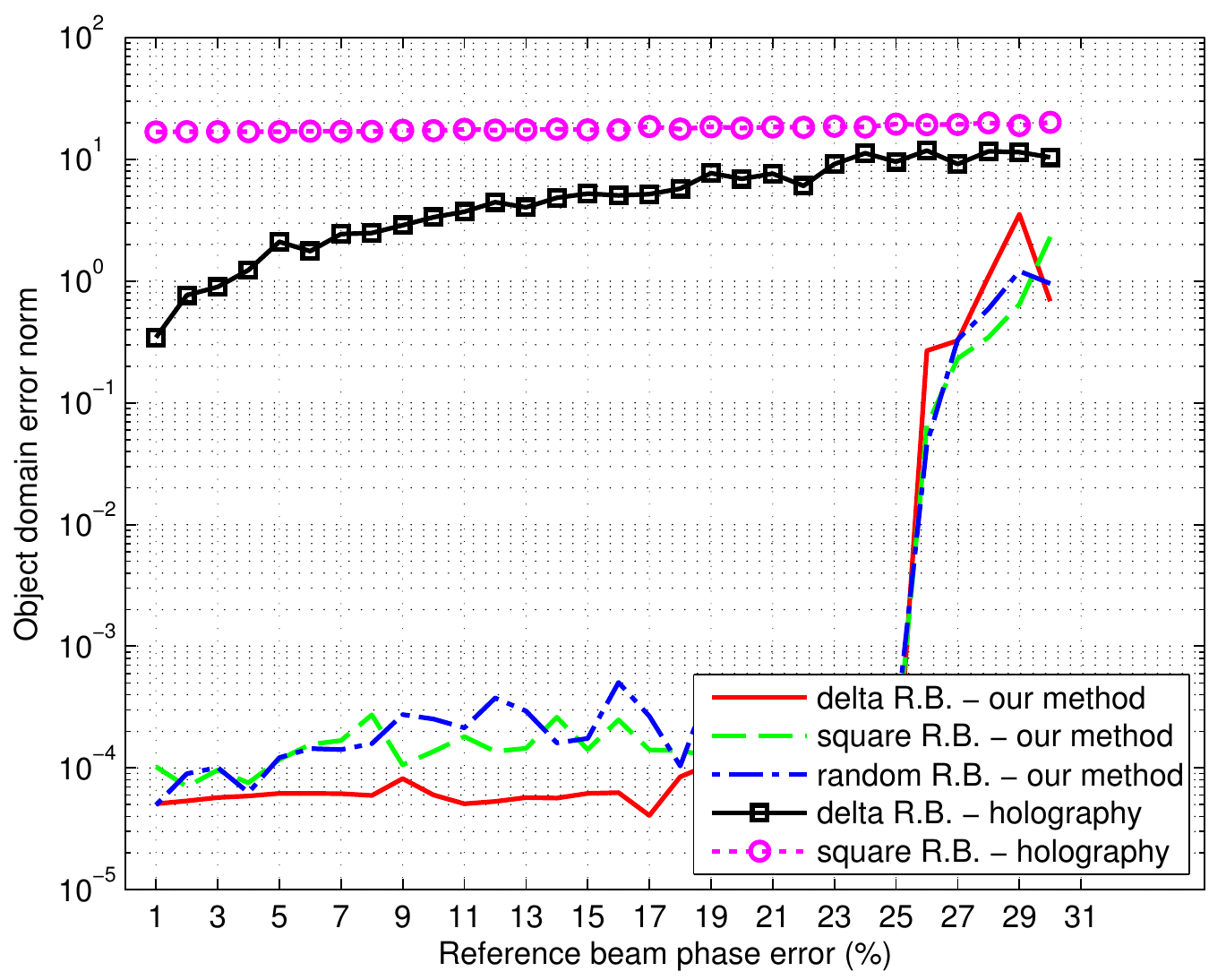}
  }
  \qquad{}
  \subfloat[]{
    \label{fig:we-qualitycorrected-complexRandom}
    \includegraphics[height=0.25\textheight]{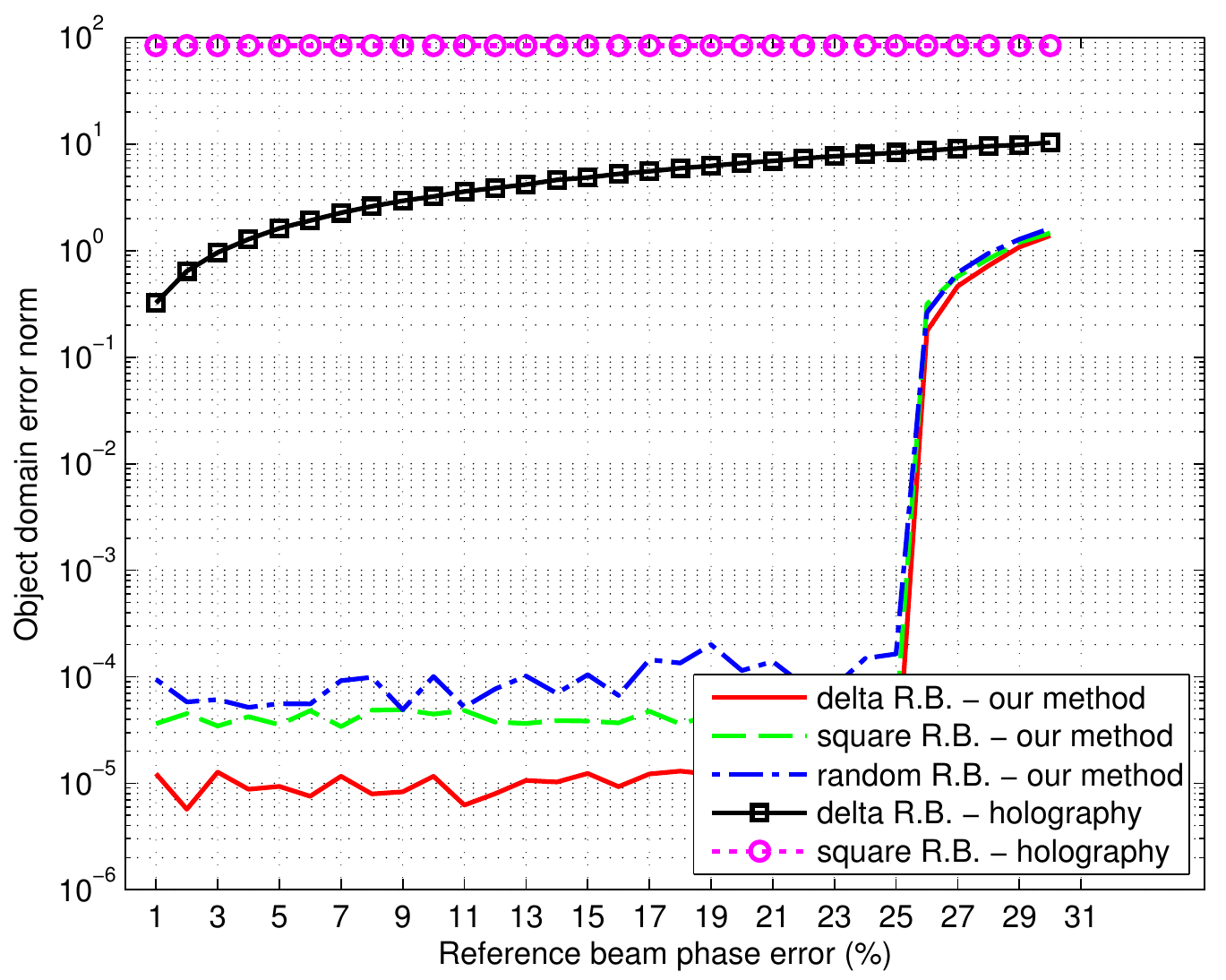}
  }
  \caption[Corrected object domain error  vs. phase error in the
    reference beam]{Corrected object domain error ($\|x-z\|$) vs. phase error in the
    reference beam: (a) real valued image,
    (b) image phase is smooth, (c) image phase is random.}
  \label{fig:objectdomain-error-corrected}
\end{figure}

\begin{figure}[H]
  \centering
  \subfloat[]{
    \includegraphics{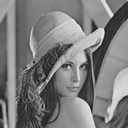}
  }
  \quad{}
  \subfloat[]{
    \includegraphics{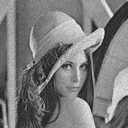}
  }
  \quad{}
  \subfloat[]{
    \includegraphics{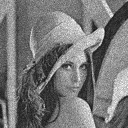}
  }\\
   \subfloat[]{
    \includegraphics{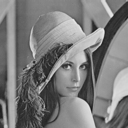}
  }
  \quad{}
  \subfloat[]{
    \includegraphics{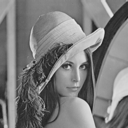}
  }
  \quad{}
  \subfloat[]{
    \includegraphics{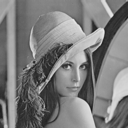}
  }
  \caption[Image reconstructed by the holographic method
    and our method]{Image (intensity) reconstructed by the holographic method
    (upper row) and our method (lower row). Object phase is random,
    and the reference beam is a delta function with Fourier phase
    errors: 
    (a), (b), and (c) --- phase error is 1\%, 10\%, and 20\% respectively
    (d), (e), and (f) --- phase error is 1\%, 10\%, and 20\% respectively.}
  \label{fig:visual-results-phase-error}
\end{figure}

As the results show, our method demonstrates a substantial
advantage over ordinary holographic reconstruction. It is
remarkable that even when the minimization is not particularly
successful (in cases of very large phase errors) our
reconstruction is still closer to the true signal than the
holographic method. This success is due to our approach of
decoupling the phase retrieval from the interferometric
measurements. As mentioned earlier, we deliberately avoid strong
dependence on the reference beam. The interference pattern is only
used to \textit{estimate} the Fourier phase bounds. The results
indicate that this approach is well justified.

\section{Concluding remarks}
\label{sec:phase-holo-conclusions}
In this chapter we presented a new reconstruction method from two
intensities measured in the Fourier plane. One is the magnitude of the
sought signal's Fourier transform, and the other is the intensity
resulting from the superimposition of the original image and an
approximately known reference beam. While the method was originally
developed for the phase retrieval problem, it can be useful in digital
holography, because it poses less stringent requirements on the
reference beam. The method is designed specifically to allow severe
errors in the reference beam without compromising the quality of
reconstruction. Numerical simulations justify our approach, exhibiting
reconstruction that is superior to that of holographic techniques.

It is also important to note that this method can, potentially,  be used in the
classical phase retrieval problem without knowing the Fourier phase
withing the error of $\pi/2$ radians as was done in
Chapters~\ref{cha:appr-four-phase-first},
and~\ref{cha:appr-four-phase-explanation} as it uses some sort of
bootstrapping  technique that allows much rougher phase estimation.

%%% Local Variables: 
%%% mode: latex
%%% TeX-master: "../thesis"
%%% End: 

%% file: boundary/boundary.tex
\chapter{Designing boundaries for phase retrieval\footnotemark}
\label{cha:design-bound-phase}

\footnotetext{The material presented in this section is currently in
  preparation for submission to a journal.}

In this chapter we develop an algorithm for signal reconstruction from
the magnitude of its Fourier transform in a situation where some
(non-zero) parts of the sought signal are known. Although our method
does not assume that the known part comprises the boundary of the
sought signal, this is often the case in microscopy: a specimen is
placed inside a known mask, which can be thought of as a known light source
that surrounds the unknown signal. Therefore, in the past, several
algorithms were suggested that solve the phase retrieval problem
assuming known boundary values
\shortcite{hayes82importance,hayes83recursive,fienup83reconstruction,fienup86phase}.
Unlike
our method, these methods do rely on the fact that the known part in
on the boundary.

Besides the reconstruction method we give an explanation of the
phenomena observed in previous work: the reconstruction is much faster
when there is more energy concentrated in the known part
\shortcite{fiddy83enforcing,fienup86phase}. Quite surprisingly, this can be
explained using our previous results on phase retrieval with
approximately known Fourier phase.

\section{Review of existing methods}
\label{sec:revi-exist-meth}
Of course, it is possible to use the methods reviewed/developed in the
previous sections. One can apply the HIO algorithm and hope not to be
entrapped in a situation where the method stagnates. However,
this approach is not optimal, as it does not use the additional
information available in this case.
% Thus, even if converges, the
% method will be very slow. It is interesting to check whether the GS
% algorithm is capable of the unknown part reconstruction when part of
% the signal is known. We will run such tests later.
Note that we cannot  use the
reconstruction method developed in the previous chapter because
in the current setup we assume that only one measurement is
available. Our method will be presented later. At the time being let us start by
reviewing other methods suggested in literature.

In accordance with the convention adopted in the previous chapters,
the unknown signal is denoted by $z(m,n)$\footnote{Here we explicitly
  assume that the signal is discrete.}. The signal is assumed to have
a finite support, specifically, it vanishes outside the box
$[0,M-1]\times[0,N-1]$, and our goal is, as before, to
reconstruct it from the (oversampled) magnitude of its Fourier
transform $|\hat{z}(p,q)|$, where
\begin{equation}
  \label{eq:boundary-1}
  \hat{z}(p,q) =
  |\hat{z}(p,q)|\exp(j\phi(p,q))\equiv\mathcal{F}[z(m,n)]\ .
\end{equation}
In our implementation, $F$ the \textit{unitary} Fourier transform. This
regularization constant is chosen with only one purpose: to make the
distance (norm) in the Fourier domain equal that in the object
domain. Thus, by the discrepancy (error) in the Fourier domain one can
easily estimate the error in the object domain. Hence,
\begin{equation}
  \label{eq:boundary-2}
  \mathcal{F}[z(m,n)]=(PQ)^{-1/2}\sum_{m=0}^{P-1}
  \sum_{n=0}^{Q-1}z(m,n)\exp[-j2\pi (mp/P + nq/Q)]\ ,
\end{equation}
where $m,p=0,1,\ldots,P-1$ and $n,q=0,1,\ldots,Q-1$. In addition, we
require the ``two-fold oversampling'' in the Fourier domain, that is
$P=2M-2$, $Q=2N-2$. The purpose of this oversampling is to capture the
information that $z(m,n)$ has only limited support. With these
definitions, there exists one well-known relation between the squared
magnitude of the Fourier transform and the linear (as opposed to
cyclic) auto-correlation (denoted by $\star$) function of $z(m,n)$:
\begin{equation}
  \label{eq:boundary-3}
  \begin{split}
    r(i,j)
    & \equiv z(m,n) \star z(m,n) \\
    & \equiv \sum_{m=0}^{M-1}\sum_{n=0}^{N-1}z(m,n)z^{*}(m-i,n-j)\\
    & = \mathcal{F}^{-1}(|\hat{z}(p,q)|^{2})\ ,
  \end{split}
\end{equation}
where $r(i,j)$ is of size $(2M-1)\times(2N-1)$ pixels defined over the
region $1-M\leq i \leq M-1$, $1-N\leq j \leq N-1$. This relation
between the magnitude of the Fourier transform and the
auto-correlation function of the sought signal was used by Hayes and
Quatieri to develop an elegant algorithm for finding
$x$\footnote{Again, $x$ is not necessarily equal to $z$ due to
  possible non-uniqueness.} when its boundaries are
known~\shortcite{hayes83recursive,hayes82importance}. The algorithm assumes
that the boundaries, that is, the first row $z(0,:)$\footnote{Here we
  use MATLAB notation, where a colon (:) denotes the entire vector of
  indexes of the corresponding dimension.}, the last row $z(M-1,:)$,
as well as the first and the last columns: $z(:,1)$, and $z(:,N-1)$,
are known. The algorithm is iterative and at each iteration it
recovers two new unknown rows (or columns) of $x$. The authors
demonstrated that the $k$-th iteration of the algorithm reduces to a
simple matrix (pseudo) inverse to solve the following system of
equations
\begin{equation}
  \label{eq:boundary-4}
  \left[
    \begin{array}{c}
      F \ \vert\  L
    \end{array}
  \right] \,
  \left[
    \begin{array}{c}
      x(N-k, :)\\
      x(k, :)
    \end{array}
  \right] =
  \left[    
   \begin{array}{c}
      \tilde{r}(N-k, :)
    \end{array}
  \right] \,, 
\end{equation}
where the matrices $F$ and $L$ correspond to the cross-correlation
with the first and the last rows, respectively. The right-hand side
$\tilde{r}(N-k)$ is obtained from the $(N-k)$-th row of the
auto-correlation function $r(i,j)$ from which the contribution of
already recovered rows $1,2,\ldots,k-1$ and $M-2,M-3,\ldots, M-k+1$
has been subtracted. For a more detailed description we address the
reader to the original articles.  The most appealing property of this
algorithm is that it requires only a small, known in advance, number
of iterations until the whole signal $x(m,n)$ is recovered. The authors
also provided conditions that they claimed were sufficient to
guarantee uniqueness of the reconstruction. The latter, however, were
proven to be incorrect (see \shortcite{fienup86phase}). Nonetheless, the biggest
problem with this algorithm is not the non-uniqueness of
reconstruction, because, as we already mentioned, the two-dimensional phase
retrieval solution is usually unique (up to trivial transformations)
to start with. The main difficulty  
that makes the algorithm impractical for all but tiny problems is its
numerical instability. It can easily be shown that the error grows
exponentially due to the recurrent nature of the algorithm. 
Even if we assume that the measurements are ideal, containing
absolutely no error, each iteration of the algorithm will introduce a
small error due to the finite computer precision. In the next
iteration, the norm of this error will be increased by a factor
proportional to the condition number of the matrix
$[F\,\vert\,L]$. The new error will be further amplified (by the same
factor) in the next iteration, and so on. This will result in
extremely fast (exponential) error growth. This is demonstrated in
Figure~\ref{fig:hayes-quatieri-reconstruction} where a small
($128\times128$ pixels) image was reconstructed by the HQ algorithm in
the horizontal direction, that is, reconstructing column after
column. This exponential growth is observed whenever the condition
number of the matrix $[F\,\vert\,L]$ is greater than one. It can
equal unity in some very special cases, for example, when the known
boundaries contain a single delta function. This situation was
considered in \shortcite{fienup83reconstruction,fienup86phase}, and in
\shortcite{fiddy83enforcing}, although in the latter it was not used
directly for the reconstruction---the authors added this condition to
guarantee the uniqueness of the reconstruction. Moreover, all of them
observed empirically that the reconstruction was faster when the known
part contained more energy. This observation is common, even though
the authors use different reconstruction methods. None of them,
however, provided an explanation for this phenomenon. In the next
section we present our reconstruction routine and explain why a
``strong'' known part leads to a fast and stable reconstruction.
Additionally, we will consider the influence of noise in the
measurement---another issue that has been largely overlooked in 
previous works despite its enormous importance.

\section{Our reconstruction method}
\label{sec:our-reconstr-meth}
First, we must consider the source of the known boundary. In
microscopy, it is often natural to create a mask (for transparencies)
or a bed (for light reflecting objects) whose boundaries are known and
designed in a way that leads to easy image reconstruction. An example of
such a mask is shown in Figure~\ref{fig:mask}.

\begin{figure}[H]
  \centering
  \includegraphics[width=8cm]{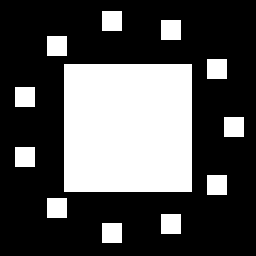}
  \caption{Artificial mask design.}
  \label{fig:mask}
\end{figure}

The object here is
assumed to be transparent, so the white areas correspond to simple
windows in an opaque material (shown in black). The big window in the
center is where the object is placed. The whole
construction is then illuminated by a coherent plane wave such that
the small windows in the mask can be assumed to be of known intensity.
There is nothing special about the plane wave here, the
most important point is that part of the  image is
known. This setup allows us to formulate the following minimization
problem to find the unknown signal $x(m,n)$:
\begin{equation}
  \label{eq:boundary-5}
  \begin{split}
    \min_{x} &\quad \||\mathcal{F}[x + b]| - r\|^{2}\\
    \mathrm{subject\ to} &\quad x(m\in \mathcal{O}_{m}, n\in
    \mathcal{O}_{n})=0 \,, 
  \end{split}
\end{equation}
where $r$ denotes the measured Fourier magnitude of the entire signal,
$b$ represents the known part (boundary) and $(\mathcal{O}_{m},
\mathcal{O}_{n})$ designate the 
location of the off-support parts of $x$ (basically, these are the
locations occupied by the mask except the central window where the
object is located). Note that we, again, use $x$ to denote the
reconstruction result because, in general, it may not be equal the
sought signal $z$.

Following the discussion in Chapter~\ref{cha:phase-retr-holography},
the mask can, in principle, be constructed in a way that makes the reconstruction
trivial: it is sufficient to place only one infinitesimally small
window (a delta function) at a sufficient distance from the object. In
this case the mask would satisfy the holography conditions and the
reconstruction could be as easy as applying a single Fourier
transform.  However, as was mentioned before, generating a delta
function is not possible in practice. Practical mask design must
balance between the production costs/difficulties and how helpful it
is for the reconstruction process. Hence, here we use simple square
windows of relatively  large size, located close to the object. Hence,
a non-iterative reconstruction is not possible. However, using a
quasi-Newton optimization method to solve (\ref{eq:boundary-5}) we
were able to get good results, as demonstrated in
Section~\ref{sec:numerical-results}.

Before we proceed to the
numerical results, it is pertinent to discuss briefly the design of
the mask. It is not known at the moment what is the best way to design
the mask so that the reconstruction would be fast and robust. Our
experience shows that the mask should have a strong ``presence'' in
all frequencies. More precisely, we can explain why this situation
leads to a fast reconstruction. Consider a single element in the
Fourier space. The contribution of the known part (denoted by
$\hat{b}_{i}$) is fully determined, that is, its magnitude and phase
are known. Hence, the full Fourier domain data at this location is the
sum $\hat{x}_{i}+\hat{b}_{i}$ that is located somewhere on the red circle
shown in Figure~\ref{fig:strong-boundary}. Note that the phase
$\angle(\hat{b}_{i} + \hat{x}_{i})$ is known up to $\pi/2$ radians
when $\alpha \leq \pi/4$ (see Figure~\ref{fig:strong-boundary}). That
is, when
\begin{equation}
  \label{eq:boundary-6}
  \arcsin
  \left (
    \frac{|\hat{x}_{i}|}{|\hat{b}_{i}|}
  \right ) \leq \frac{\pi}{4}
  \Longleftrightarrow |\hat{b}_{i}| \geq \sqrt{2} |\hat{x}_{i}| \,. 
\end{equation}

\begin{figure}[H]
  \centering
  \includegraphics{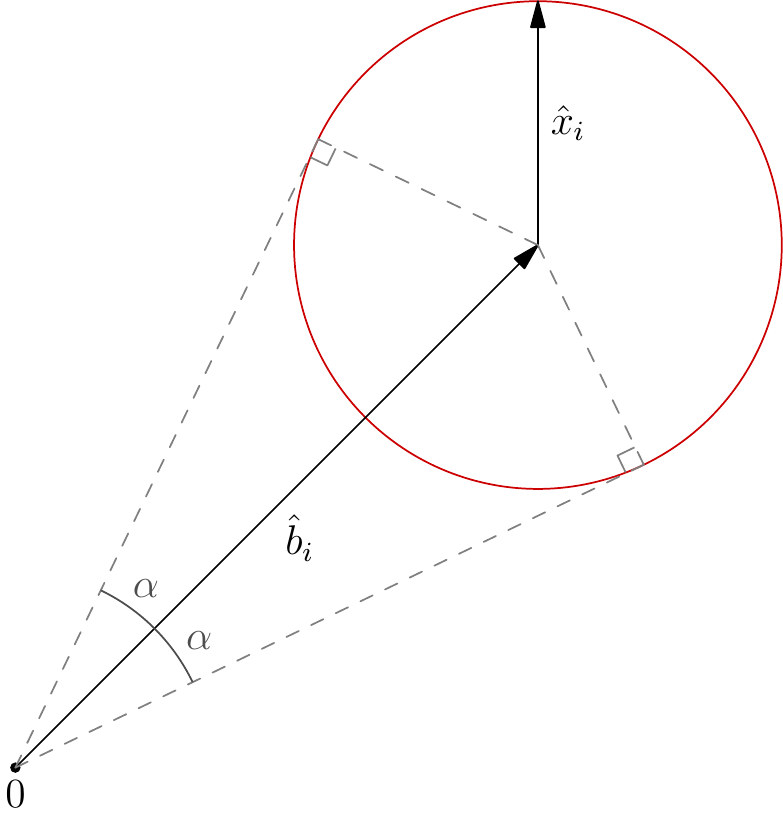}
  \caption{``Strong'' known part leads to the phase retrieval problem
    with approximately known Fourier phase.}
  \label{fig:strong-boundary}
\end{figure}
The relation in Equation~\eqref{eq:boundary-6} is quite interesting:
if the energy concentrated in the known part is at least twice as
large as the energy in the unknown part\footnote{The energy must also
  be distributed ``well'' in the frequency domain.}, the problem
reduces to phase retrieval with Fourier phase known within the limit
of $\pi/2$ radians. Hence, according to our results in Chapter
\ref{cha:appr-four-phase-explanation}, any reasonable algorithm will be
able to reconstruct the unknown part. Moreover, the larger the ratio
$|\hat{b}_{i}|/|\hat{x}_{i}|$, the smaller the Fourier phase
uncertainty is.  However, there is another aspect that is important in
practice---in all physical experiments, the measurements
inevitably contain some noise. In our case, where the measurements are
light intensity, the noise is well approximated by the Poisson
distribution. This means that stronger intensity will result in more
noise (though, the signal to noise ratio (SNR) usually increases with
the intensity growth). Hence, making the known part too strong
compared with the sought signal will make the measurement noise more
dominant with the respect to the unknown part and will eventually
arrive at the level where the reconstruction is not possible. Hence,
one should not just increase the intensity of the known part, as this leads
to poor reconstruction quality in case of noisy
measurements\footnote{Here we assume that the noise grows with the
  signal, as indeed happens with Poissonian noise.}. The most
appealing approach would be to design a mask that would use the
limited power in an efficient way. Without any a priori knowledge
about the sought signal, it seems that the optimal way would be to
create a mask whose Fourier domain power is spread evenly over all
frequencies. Note the relation to the holographic reconstruction where
a single small window (delta function) is used, because the Fourier domain
power of a delta function is the same over all frequencies. Instead of
using a delta function we can obtain a very good approximation to the
uniform power spectrum if some randomness is added to the mask
windows. It can be random shape of the windows or random
values/phases across them. Making random shapes may be more difficult
than adding a diffuser into a square window. Hence, in
Section~\ref{sec:numerical-results} we demonstrate the reconstruction
with mask windows of constant intensity and random intensity. As
expected, adding randomness improves the reconstruction speed and
noise-immunity.

\section{Numerical results}
\label{sec:numerical-results}
We experimented with several images, however, the results provided
below are limited to two images of size $128\times 128$ pixels. These
were chosen to represent two different classes of objects. One is a
natural image ``Lena''  already used in our previous
experiments. Another is the Shepp-Logan phantom.  These images are
very popular in other fields. ``Lena'' is a classical benchmark in the
image processing community, because it has a lot of features and
delicate details. The second image, ``phantom'', is often used as a
benchmark in MRI related algorithms. However, its piece-wise constant
nature can approximate well objects that are often investigated in
microscopy, for example, cells. Besides the above differences these two
images differ by their support. Lena's support is tight, that is, it
occupies all the space and no shifts are possible. The phantom, on the
other hand, has non-tight support. As we saw earlier, this is an
important property for phase retrieval algorithms---complex valued
images with non-tight support are much more difficult for the current
reconstruction methods like HIO. The image intensities (squared
magnitude) is shown in Figure~\ref{fig:boundary-test-images}.
\begin{figure}[H]
  \centering
  \subfloat[]{\includegraphics{lena_128}}
  \qquad{}
  \subfloat[]{\includegraphics{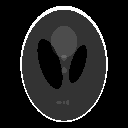}}
  \caption[Test images]{Test images (intensity).}
  \label{fig:boundary-test-images}
\end{figure}
The images were tested in two different scenarios: first, they were
assumed to be real-valued and non-negative; second, they were made
complex-valued by adding a phase distribution. Here we present the
results for the case where the objects' phases were chosen to be
proportional to their intensity (scaled to the interval
$[0,2\pi]$). This choice corresponds to the case where the phase
changes in a relatively smooth manner.  However, all our experiments
indicate that the particular phase distribution has little effect on
the reconstruction, except the case where the object can be assumed
real non-negative. In this case the non-negativity prior can be used
to speed-up the reconstruction and improve its quality in the case of
noisy measurements, as is evident from
Figures~\ref{fig:lena-reconstruction-speed-flat-mask},
and~\ref{fig:phantom-reconstruction-speed-flat-mask}. 

Let us begin with a demonstration of the HQ algorithm
results. Following our discussion in Section~\ref{sec:revi-exist-meth}
we expect the error to grow exponentially as we progress from the
boundaries toward the image center. This expectation is confirmed by
the results shown in Figure~\ref{fig:hayes-quatieri-reconstruction}. Note
that the reconstructed intensity in
Figure~\ref{fig:hayes-quatieri-rec-intensity} cannot convey the true
error because the image storage format clips all values outside the interval
$[0,1]$. The true error is evident from
Figure~\ref{fig:hayes-quatieri-rec-error}.
\begin{figure}[H]
  \centering
  \subfloat[]{
    \includegraphics{lena_128}
  }\qquad{}
  \subfloat[]{
    \label{fig:hayes-quatieri-rec-intensity}
    \includegraphics{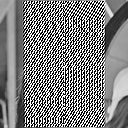}
  }\\
  \subfloat[]{
    \label{fig:hayes-quatieri-rec-error}
    \includegraphics[]{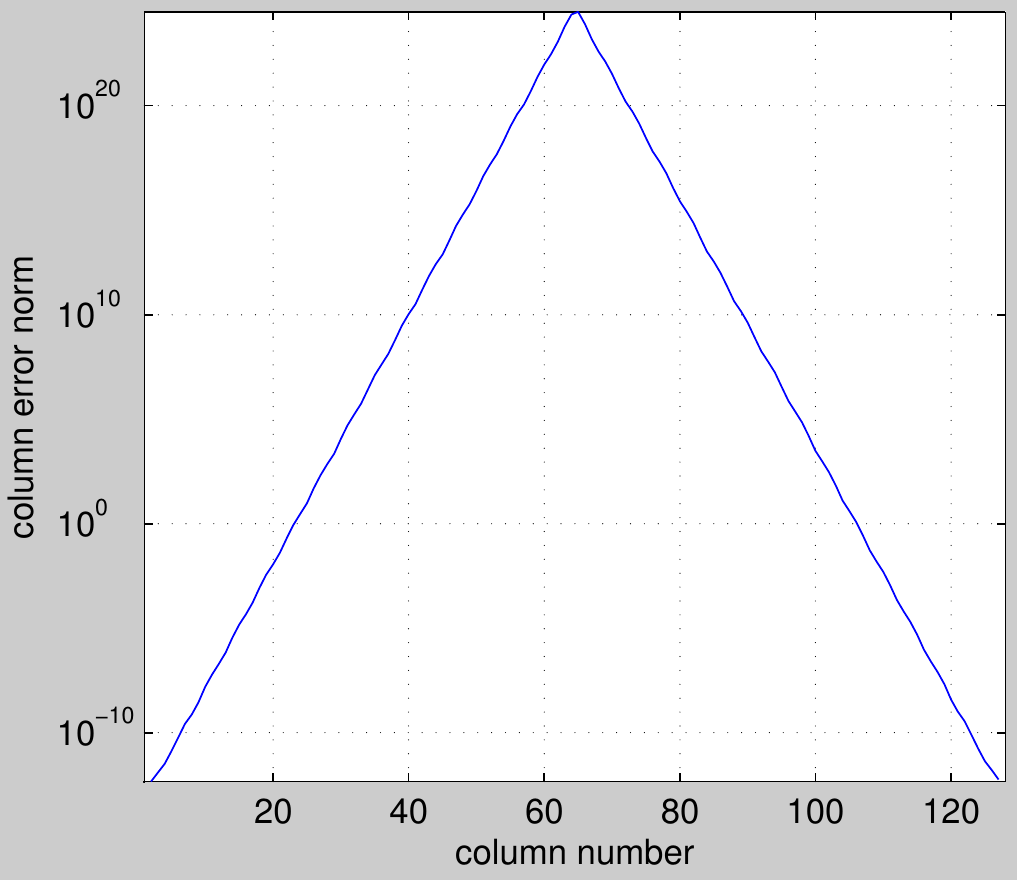}
  }
  \caption[Error growth in the Hayes-Quatieri recursive algorithm]{Error growth in the Hayes-Quatieri recursive algorithm:
  (a) original image, (b) reconstructed by the HQ algorithm, and (c)
  actual error growth in the HQ algorithm.}
  \label{fig:hayes-quatieri-reconstruction}
\end{figure}

In theory, we can use the same boundary conditions as the HQ
algorithm, however, our experiments indicate that the optimization
routine used in our method is prone to stagnation when the known
boundary (or image part in general) carries little energy. This is,
of course, in agreement with our discussion in the previous section:
when the boundary caries little energy it does not provide enough
information about the Fourier phase. This, in turn, causes
line-search optimization algorithms to stagnate (see
Chapter~\ref{cha:appr-four-phase-explanation}). This stagnation could
be addressed by a more sophisticated reconstruction routine similar,
for example, to that in Chapter~\ref{cha:phase-retr-holography}
combined with some techniques from global optimization.  However, this
approach can be expensive in the terms of computational
effort. Moreover, the stagnation can be alleviated by designing a mask
that makes the reconstruction much easier. In our approach the
emphasis is put on the simplicity of the mask design. Hence, we used
the simple mask shown in Figure~\ref{fig:mask}.  The mask is of size
$256\times 256$ pixels with 11 square windows of size $20\times 20$
located approximately on a circle of radius 100 pixels (see
Figure~\ref{fig:mask}). The objects are placed in the middle of the
mask where a special window of size $128\times128$ pixels is
provided. The area outside this special window comprises the known
boundary. When placed in the mask, the test images look as shown in
Figure~\ref{fig:mask-withimages}.
\begin{figure}[H]
  \centering
  \subfloat[]{
    \includegraphics[width=8cm]{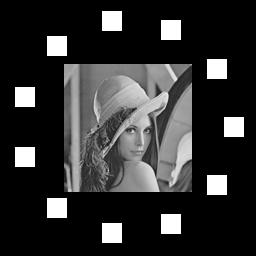}
  }
  \\
  \subfloat[]{
    \includegraphics[width=8cm]{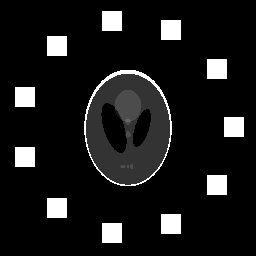}
  }
  \caption{Artificial mask with test images.}
  \label{fig:mask-withimages}
\end{figure}

Note also that the algorithm based on the minimization problem defined
in Equation~\eqref{eq:intro-5} is very naive and does
not try to use the ``approximately known Fourier phase'' in the case
were the energy in the known part is sufficiently large. However, our
main goal is to show the influence of the mask design. This influence
is essentially independent of the algorithm used for reconstruction. 

Let us demonstrate how the energy contained in the known part affects
the reconstruction. Recall that our expectation is that a mask that
has a strong ``presence'' in all Fourier frequencies will better suit
the reconstruction process in the noise-less case. The results shown
in Figures~\ref{fig:lena-reconstruction-speed-flat-mask},
and~\ref{fig:phantom-reconstruction-speed-flat-mask} fully support
this conjecture. These figures present the reconstruction speed (the
objective function minimization rate) for three different magnitudes
of the mask's windows values: 5, 10, and 100. The lowest value (5) was
chosen so as to be close to the theoretical ratio between the energy in
the known and the unknown parts that is required for guaranteed
reconstruction, as defined in Equation~\eqref{eq:boundary-6}. However, as
we can see the value of 5, and even the value of 10 did not result in
perfect reconstruction. This phenomenon has two reasons: first, 
doubling the amount of energy in the known part compared to the unknown part
is not sufficient---this energy must be distributed properly in the
Fourier domain; second, it may be attributed to the simplicity of the
chosen reconstruction algorithm. However, the second reason is less
likely in view of the following experiment. In an attempt to create a
``better'' distribution of the mask's energy in the Fourier domain we
added some ``randomness'' to the mask by modulating (multiplying) the
flat values across the mask's windows with random values in the
interval $[0,1]$. This step improved the speed of the reconstruction
and its robustness, as is evident from 
Figures~\ref{fig:lena-reconstruction-speed-random-mask},
and~\ref{fig:phantom-reconstruction-speed-random-mask}.

\begin{figure}[H]
  \centering
   \subfloat[]{
    \includegraphics[width=0.45\textwidth{}]{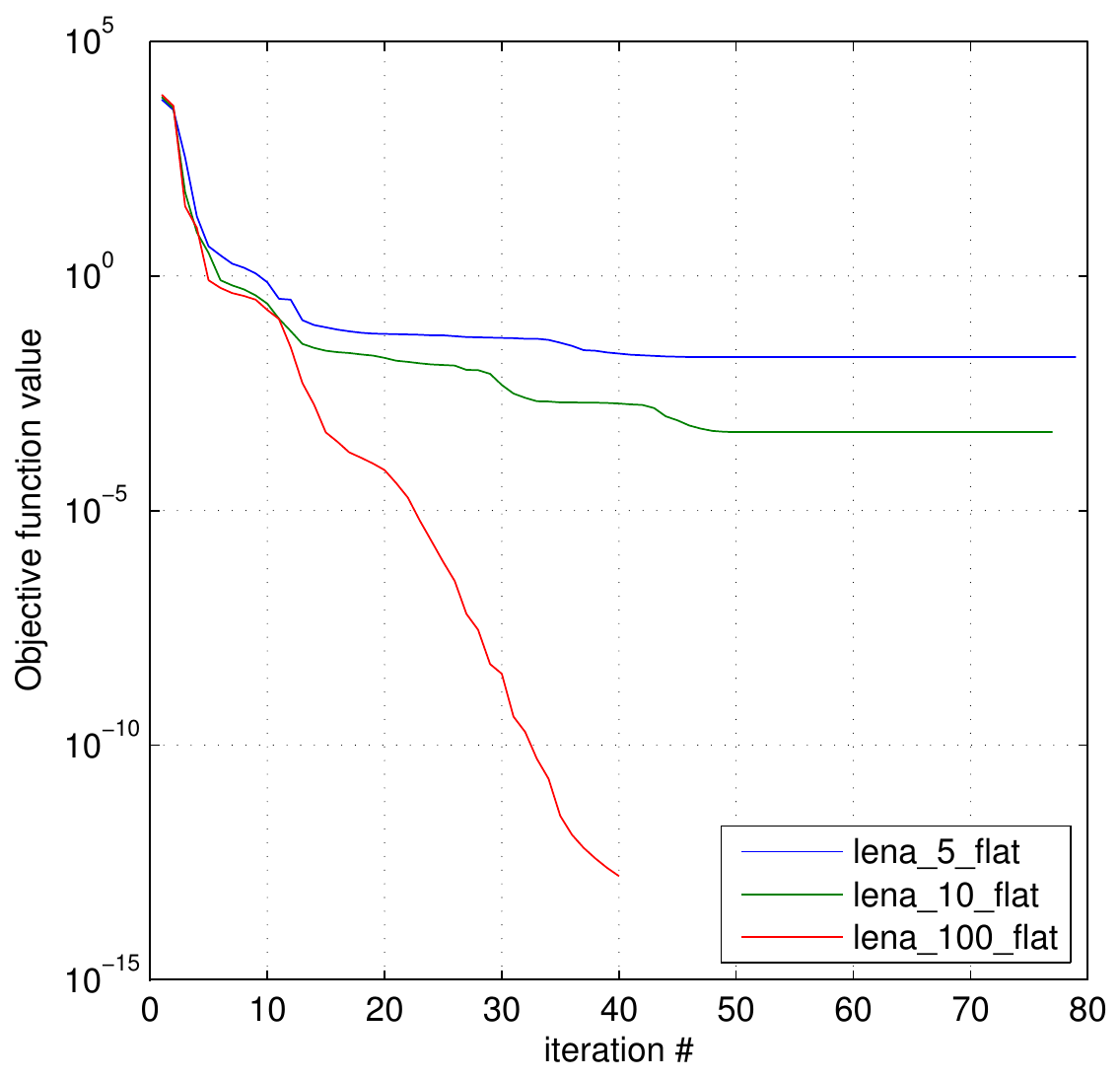}
  }
  \quad{}
  \subfloat[]{
    \includegraphics[width=0.45\textwidth{}]{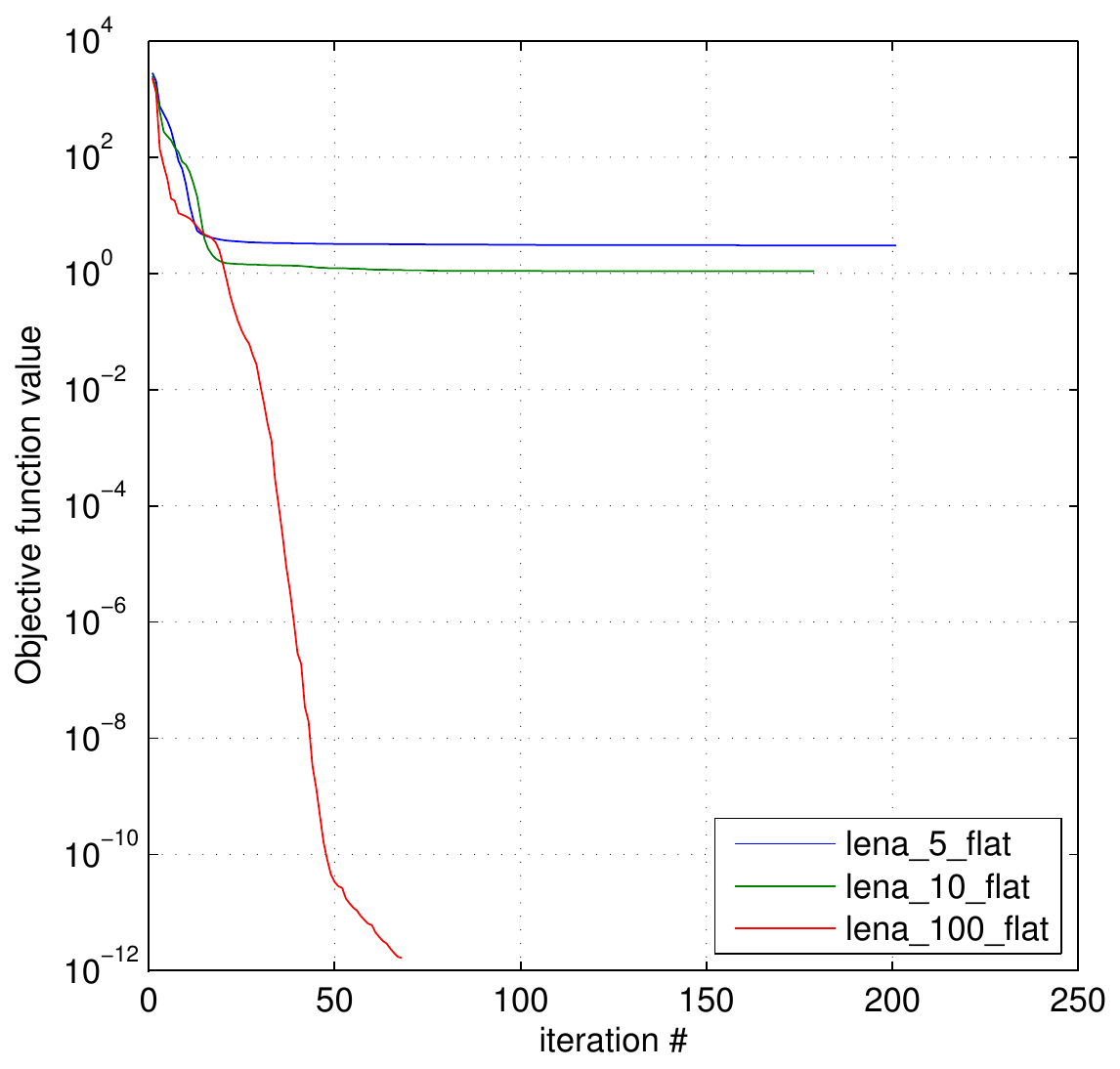}
  }
  \caption[Lena's reconstruction speed with flat mask
  values.]{Lena's reconstruction speed with flat mask values: (a)
    real-valued, (b) complex valued.}
  \label{fig:lena-reconstruction-speed-flat-mask}
\end{figure}

\begin{figure}[H]
  \centering
   \subfloat[]{
    \includegraphics[width=0.45\textwidth{}]{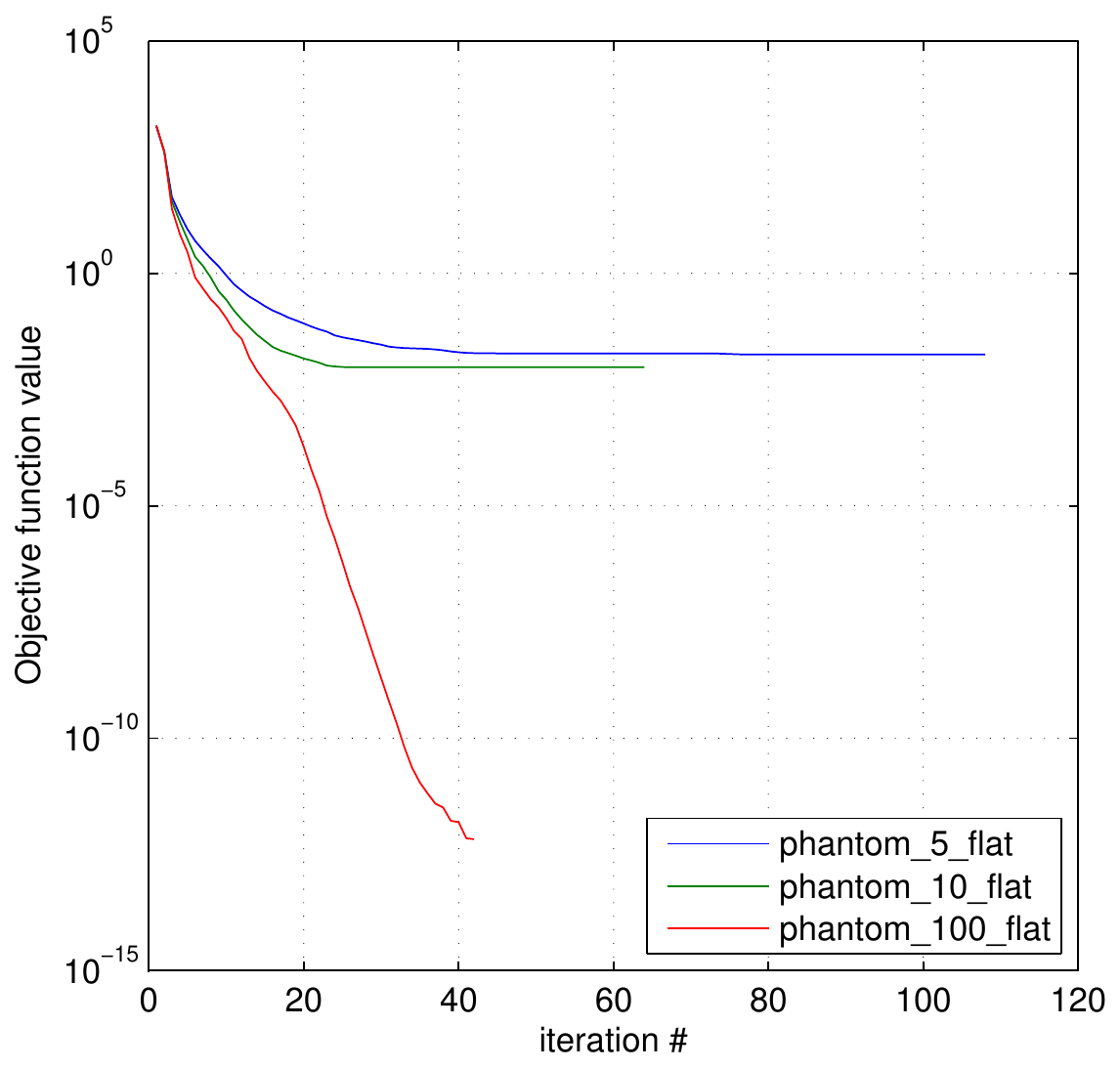}
  }
  \quad{}
  \subfloat[]{
    \includegraphics[width=0.45\textwidth{}]{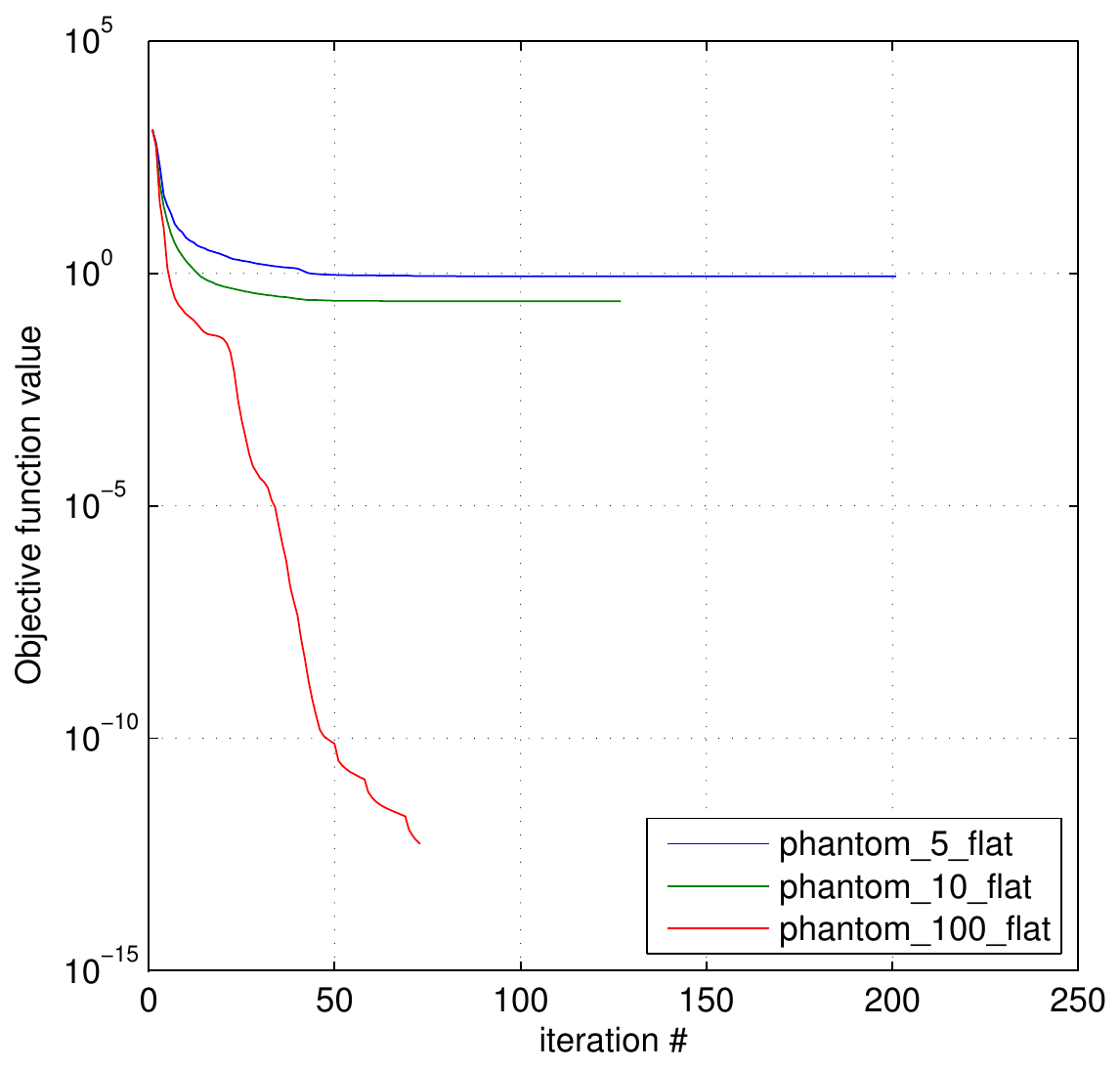}
  }
  \caption[Phantom's reconstruction speed with flat mask
  values.]{Phantom's reconstruction speed with flat mask values: (a)
    real-valued, (b) complex valued.}
  \label{fig:phantom-reconstruction-speed-flat-mask}
\end{figure}

\begin{figure}[H]
  \centering
   \subfloat[]{
     \includegraphics[width=0.45\textwidth{}]{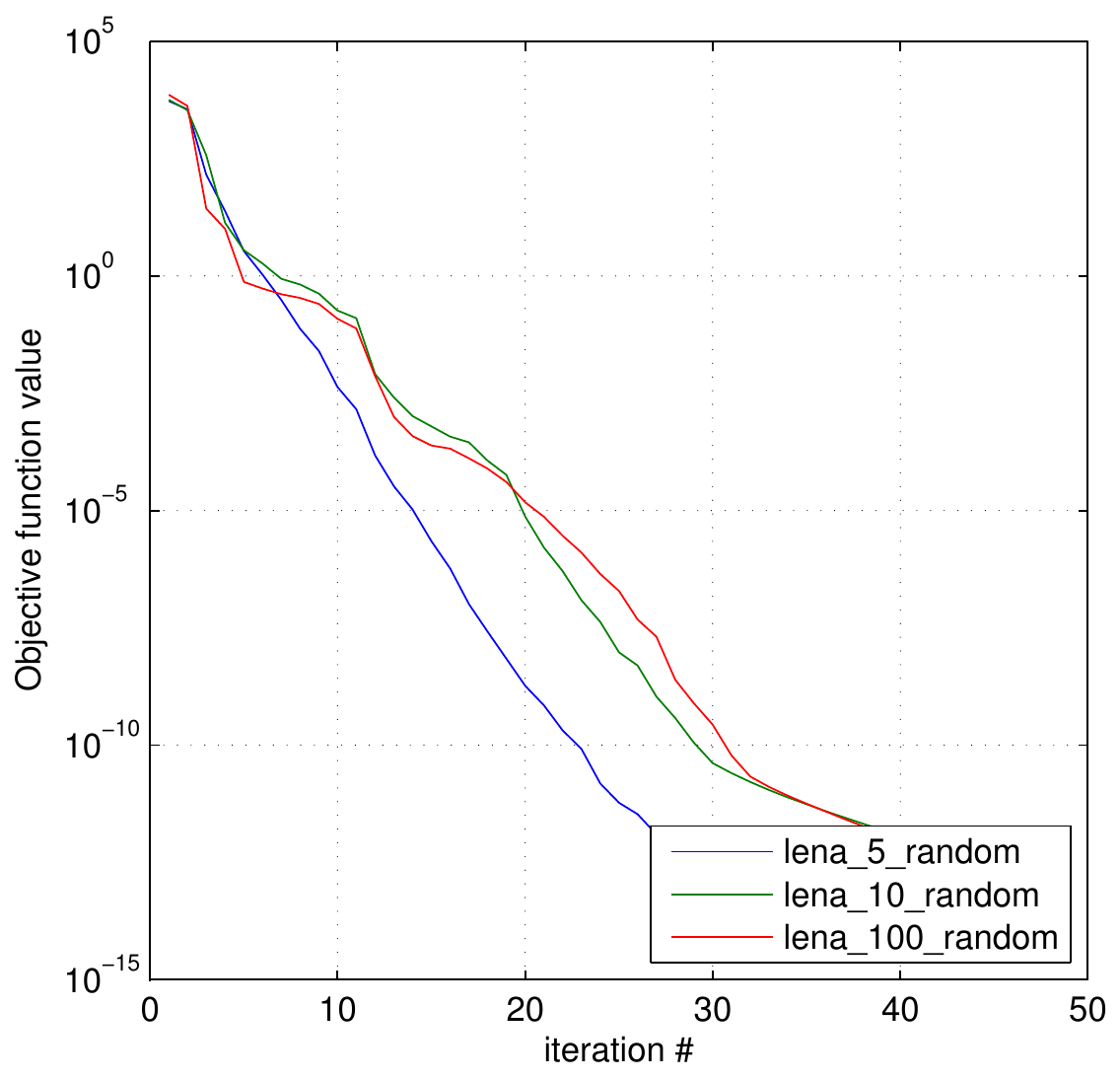}
  }
  \quad
  \subfloat[]{
    \includegraphics[width=0.45\textwidth{}]{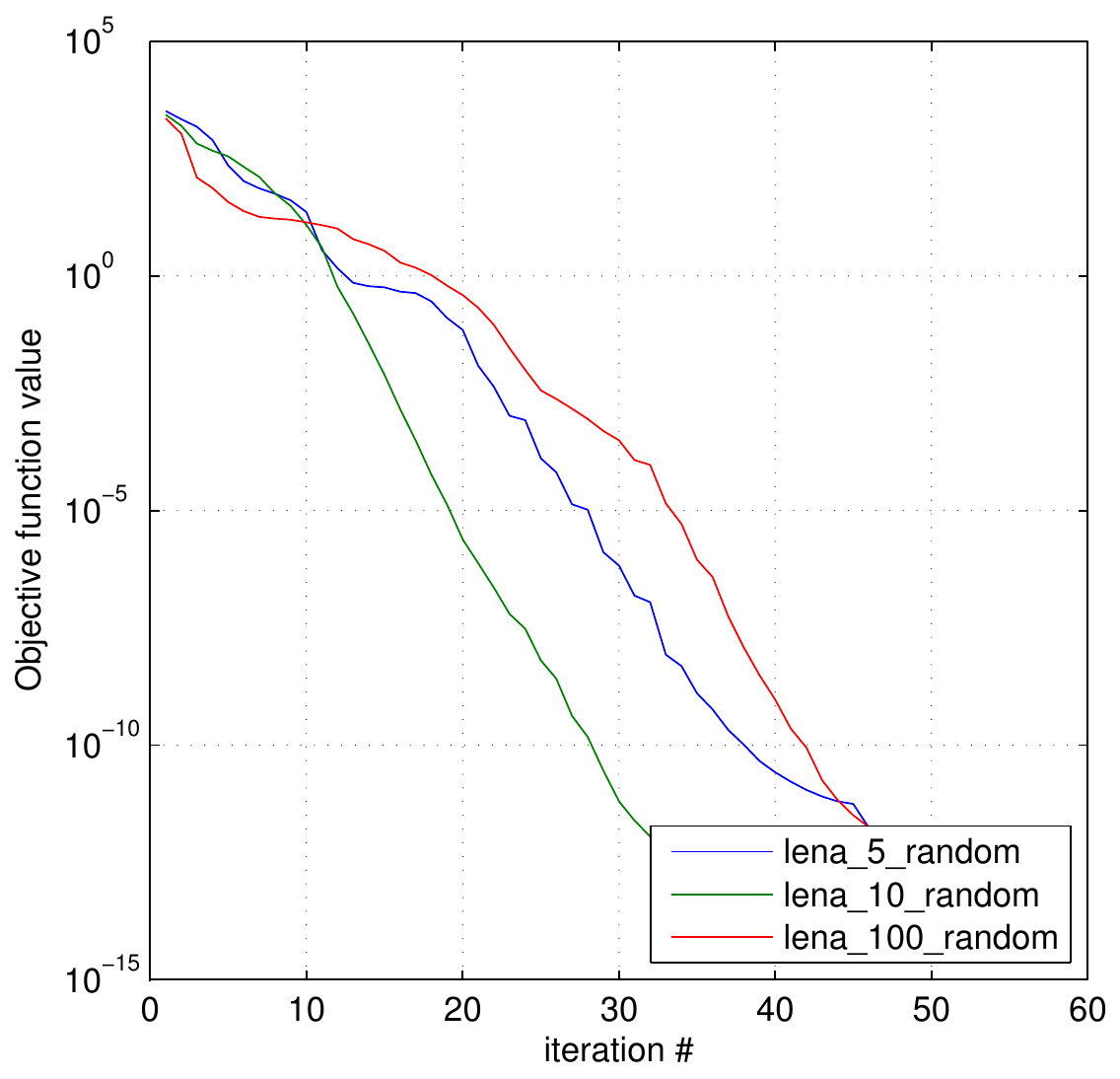}
  }
  \caption[Lena's reconstruction speed with random mask values]{Lena's
    reconstruction speed with random mask values: (a) real-valued, (b)
  complex-valued.}
  \label{fig:lena-reconstruction-speed-random-mask}
\end{figure}

\begin{figure}[H]
  \centering
   \subfloat[]{
     \includegraphics[width=0.45\textwidth{}]{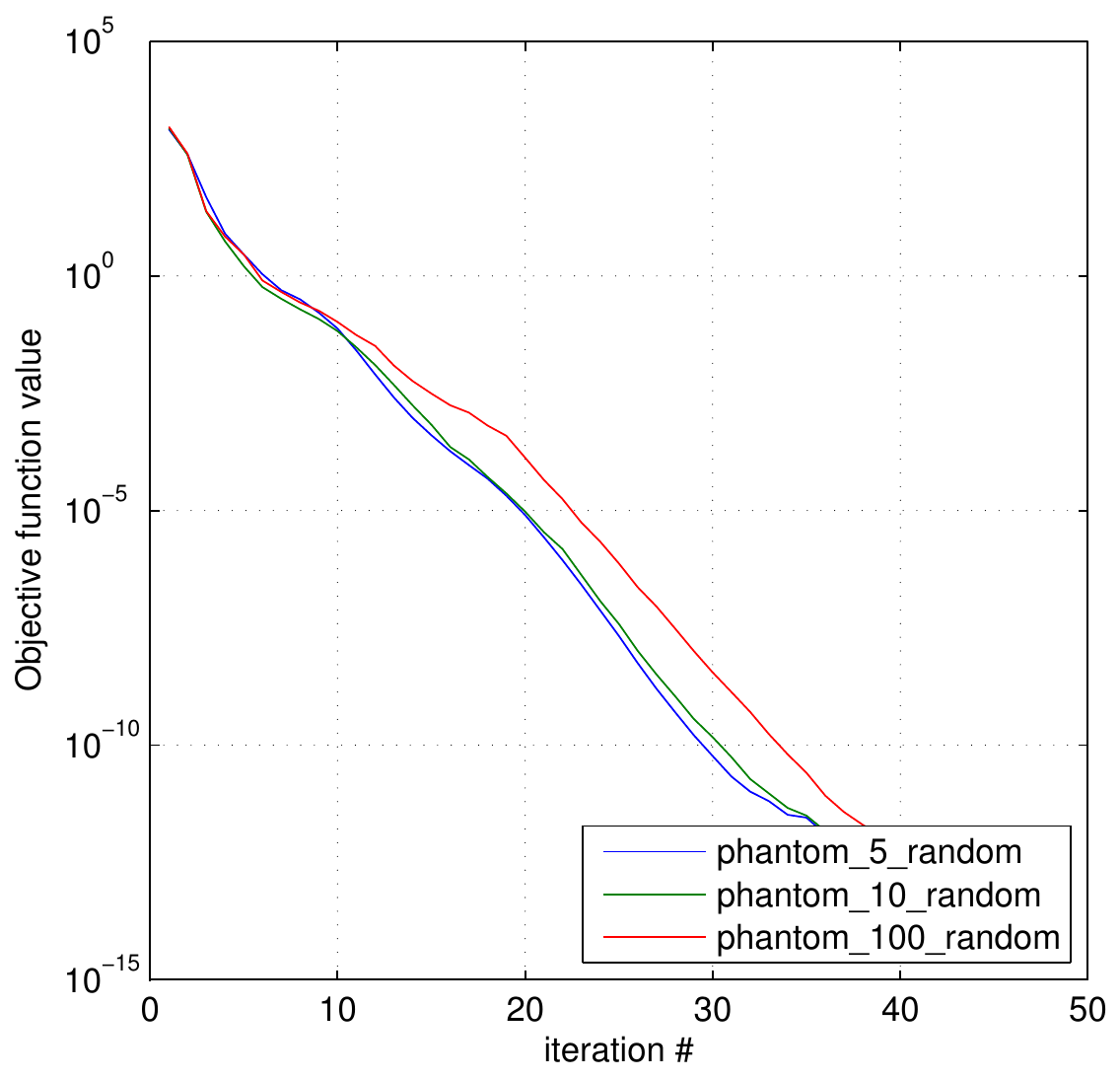}
  }
  \quad
  \subfloat[]{
    \includegraphics[width=0.45\textwidth{}]{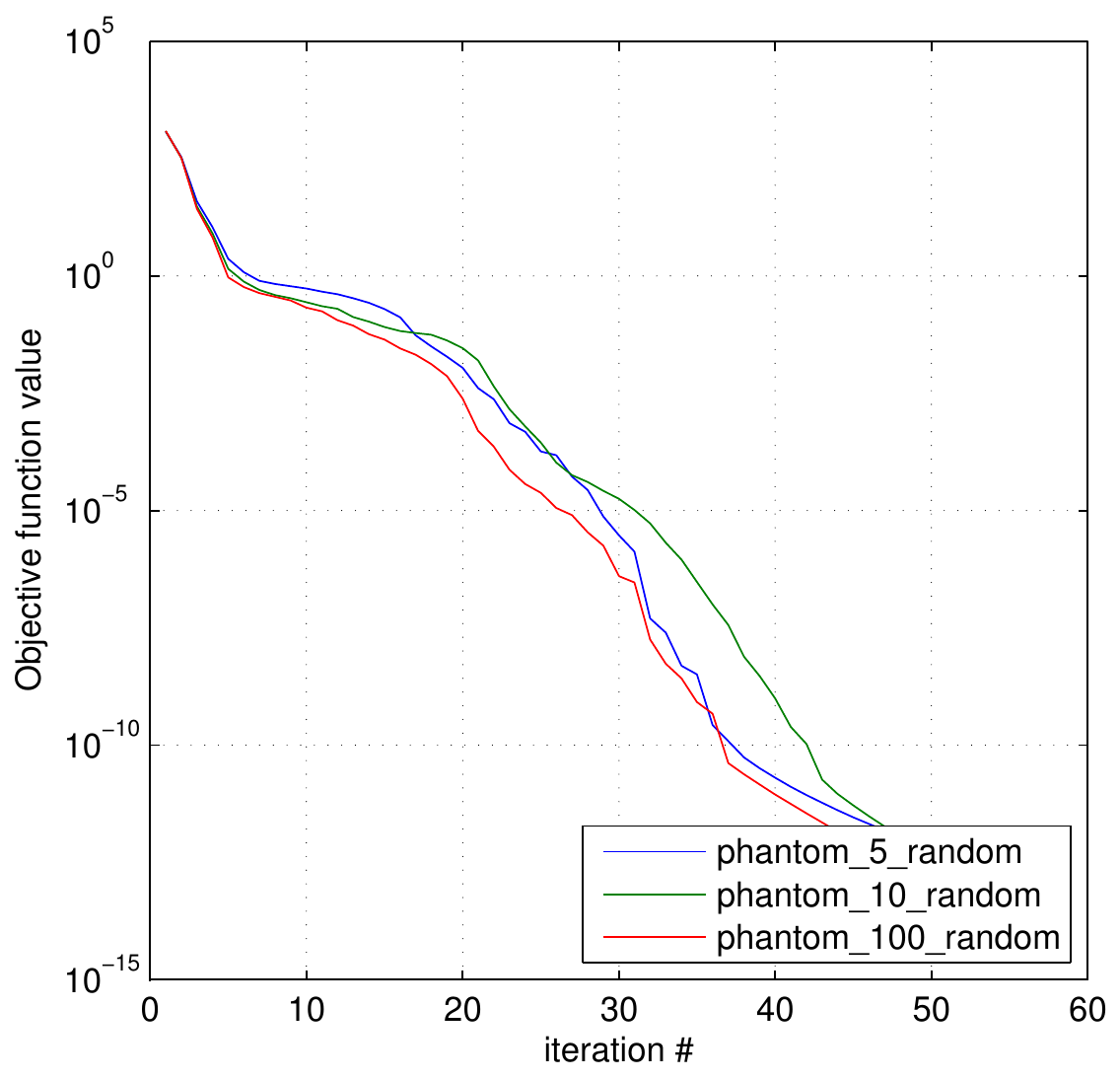}
  }
  \caption[Phantom's reconstruction speed with random mask
  values]{Phantom's reconstruction speed with random mask values: (a)
    real-valued, (b) complex-valued.}
  \label{fig:phantom-reconstruction-speed-random-mask}
\end{figure}

Note that now the reconstruction is successful for all mask values and
is very fast. However, despite this fast convergence, one must be
careful not to put too much energy into the known part. This approach
may harm the reconstruction quality of the unknown part when there is
noise in the measurements, as we demonstrate next.  In these
experiments the measurements (intensity values) were contaminated with
Poisson noise with different SNR ranging from 10 to 60 decibels. As is
evident from Figures~\ref{fig:lena-reconstruction-quality}
and~\ref{fig:phantom-reconstruction-quality}, the more energy is
concentrated in the known part the worse is the reconstruction quality
of the unknown part. This phenomenon is, of course, expected. The
Poisson noise is signal dependent: higher intensity results in more
noise. However, the intensity (energy) of the unknown part remains
constant, hence the noise becomes more and more significant compared
to it. To obtain the best result one would like to design a mask whose
power spectrum will correlate well with the power spectrum of the
sought signal. Unfortunately, this approach cannot be implemented,
because designing such a mask requires a priori
knowledge of the sought signal's Fourier magnitude, which is
unavailable in our case. However, it may be a good approach when the
Fourier magnitude of the sought signal is known \emph{approximately}. 
\begin{figure}[H]
  \centering{}
  \subfloat[]{
    \includegraphics[width=0.45\textwidth{}]{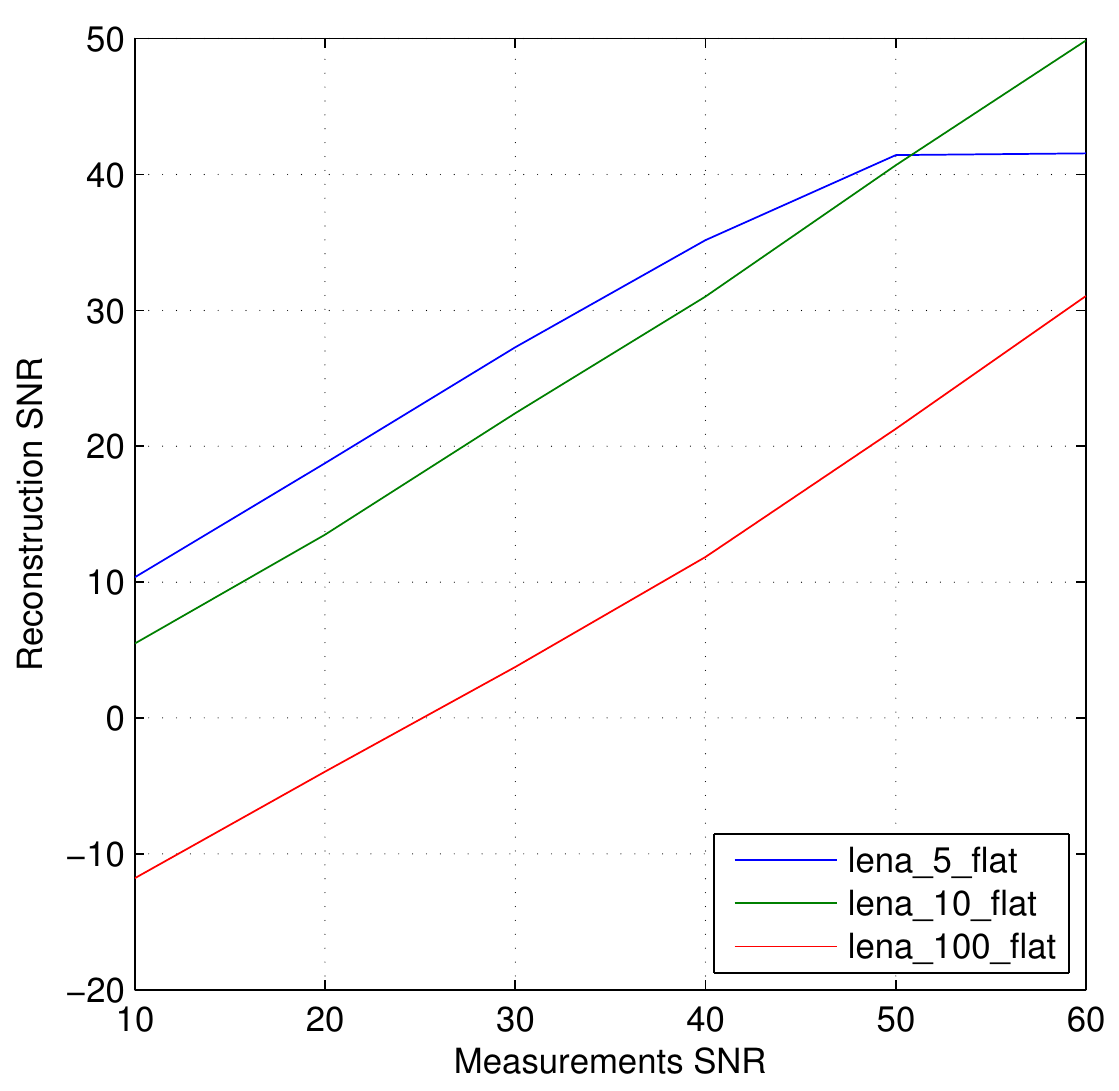}
  }
  \quad{}
  \subfloat[]{
    \includegraphics[width=0.45\textwidth{}]{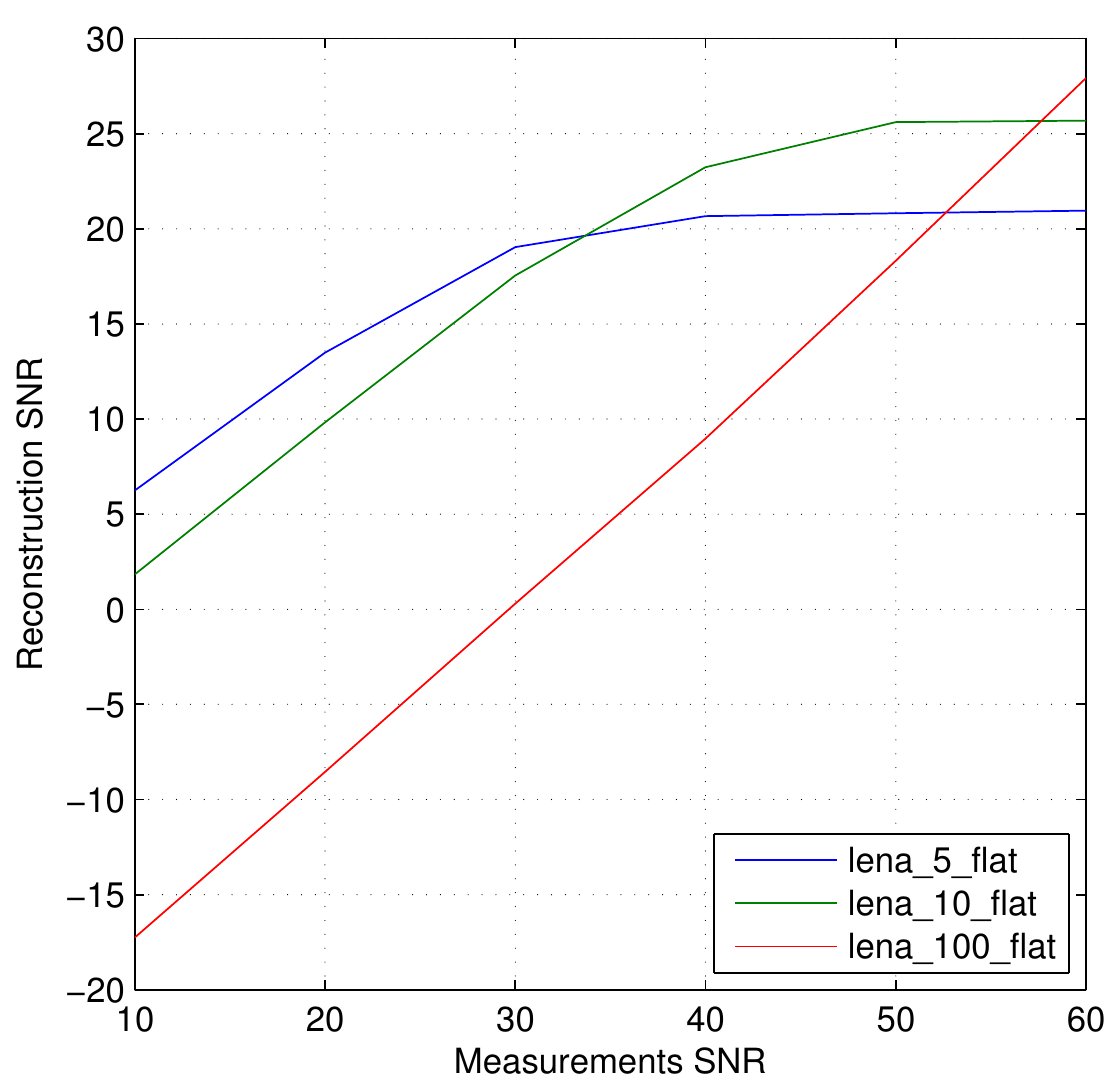}
  }
  \\
  \subfloat[]{
    \includegraphics[width=0.45\textwidth{}]{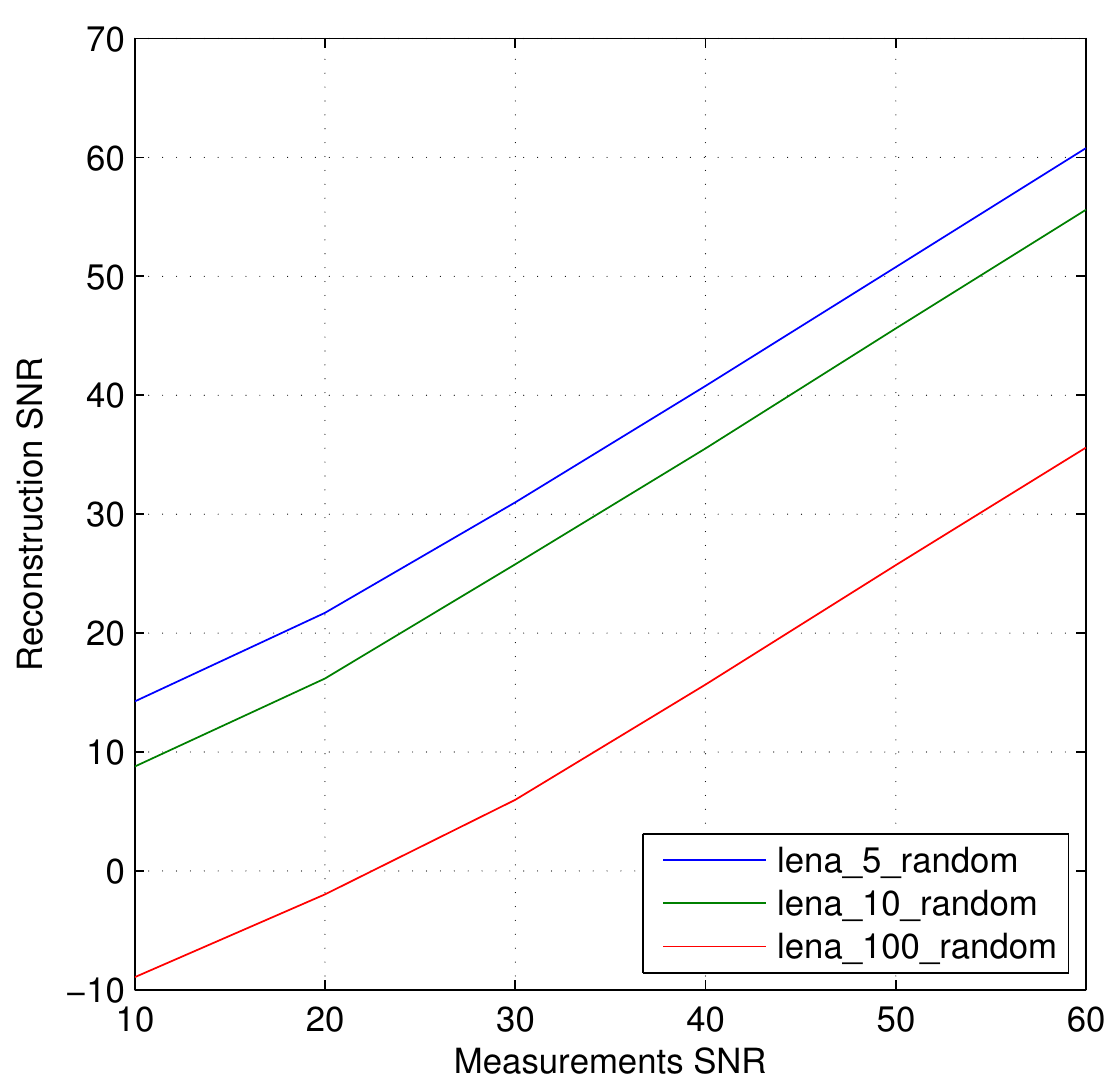}
  }
  \quad{}
  \subfloat[]{
    \includegraphics[width=0.45\textwidth{}]{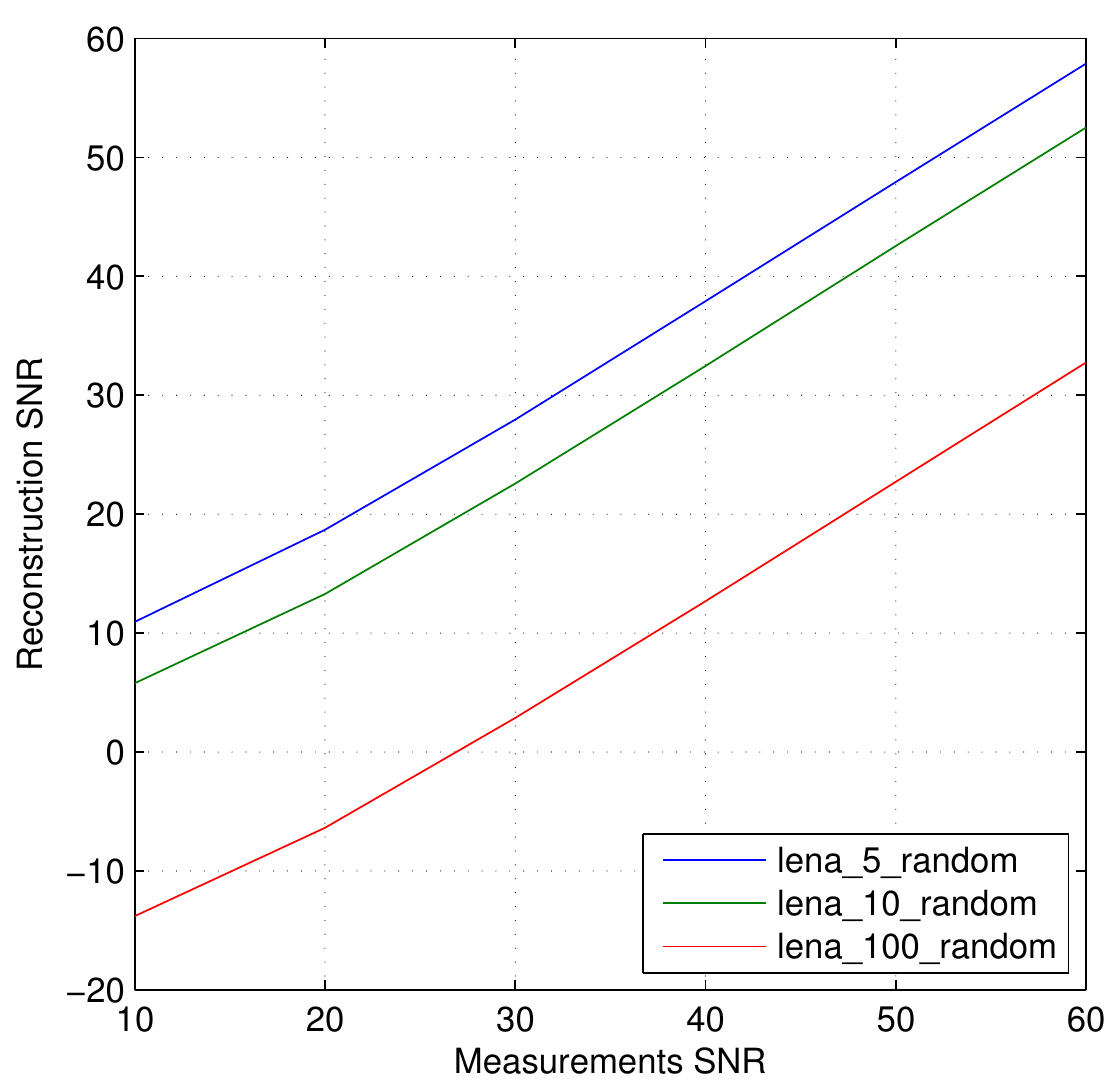}
  }
  \caption[Lena's reconstruction quality]{Lena's reconstruction
    quality: (a) real-valued with flat mask value, (b) complex-valued
    with flat mask value, (c) real-valued with random mask values, (d)
    complex-valued with random mask values.}
  \label{fig:lena-reconstruction-quality}
\end{figure}
\begin{figure}[H]
  \centering{}
  \subfloat[]{
    \includegraphics[width=0.45\textwidth{}]{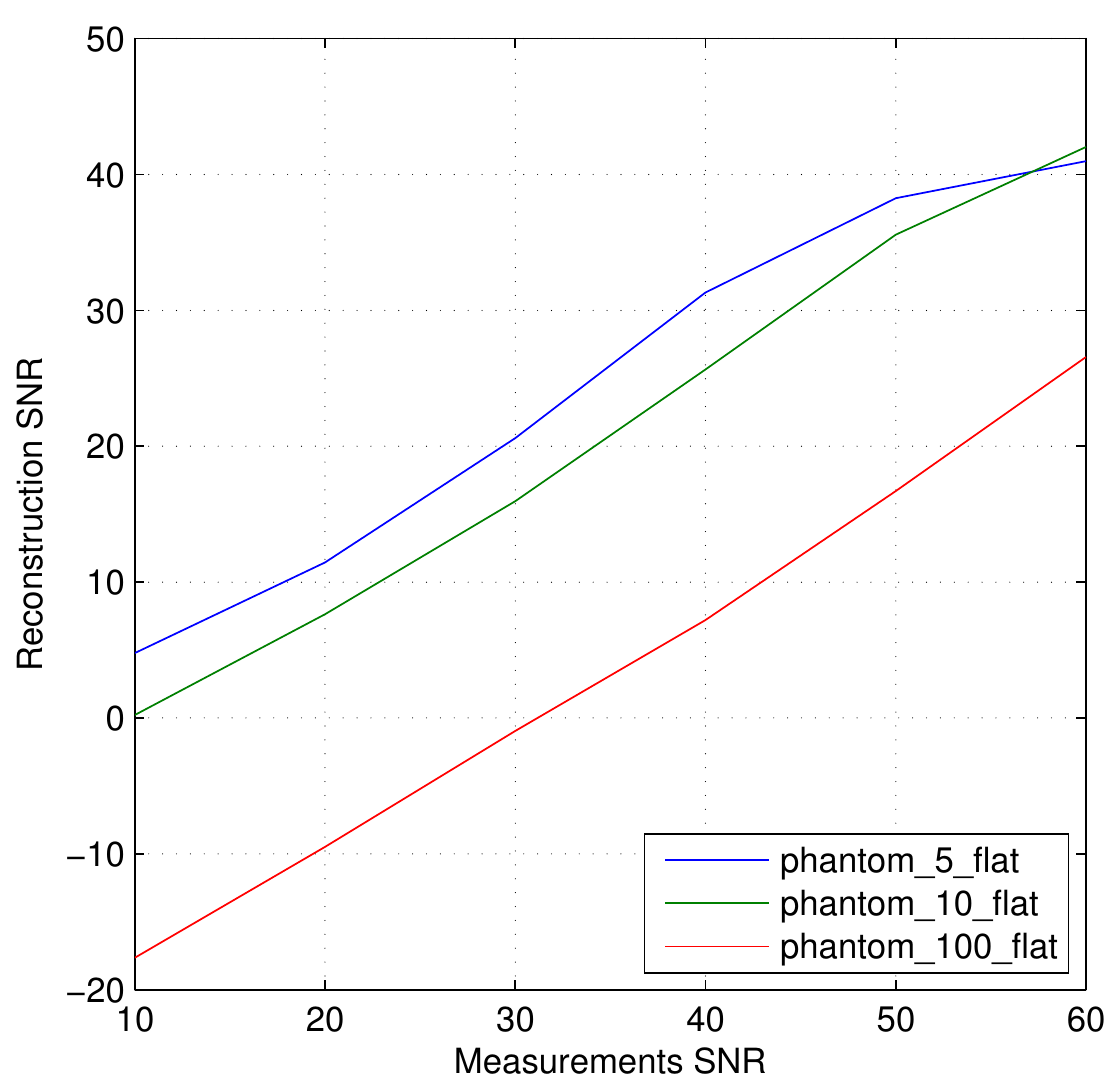}
  }
  \quad{}
  \subfloat[]{
    \includegraphics[width=0.45\textwidth{}]{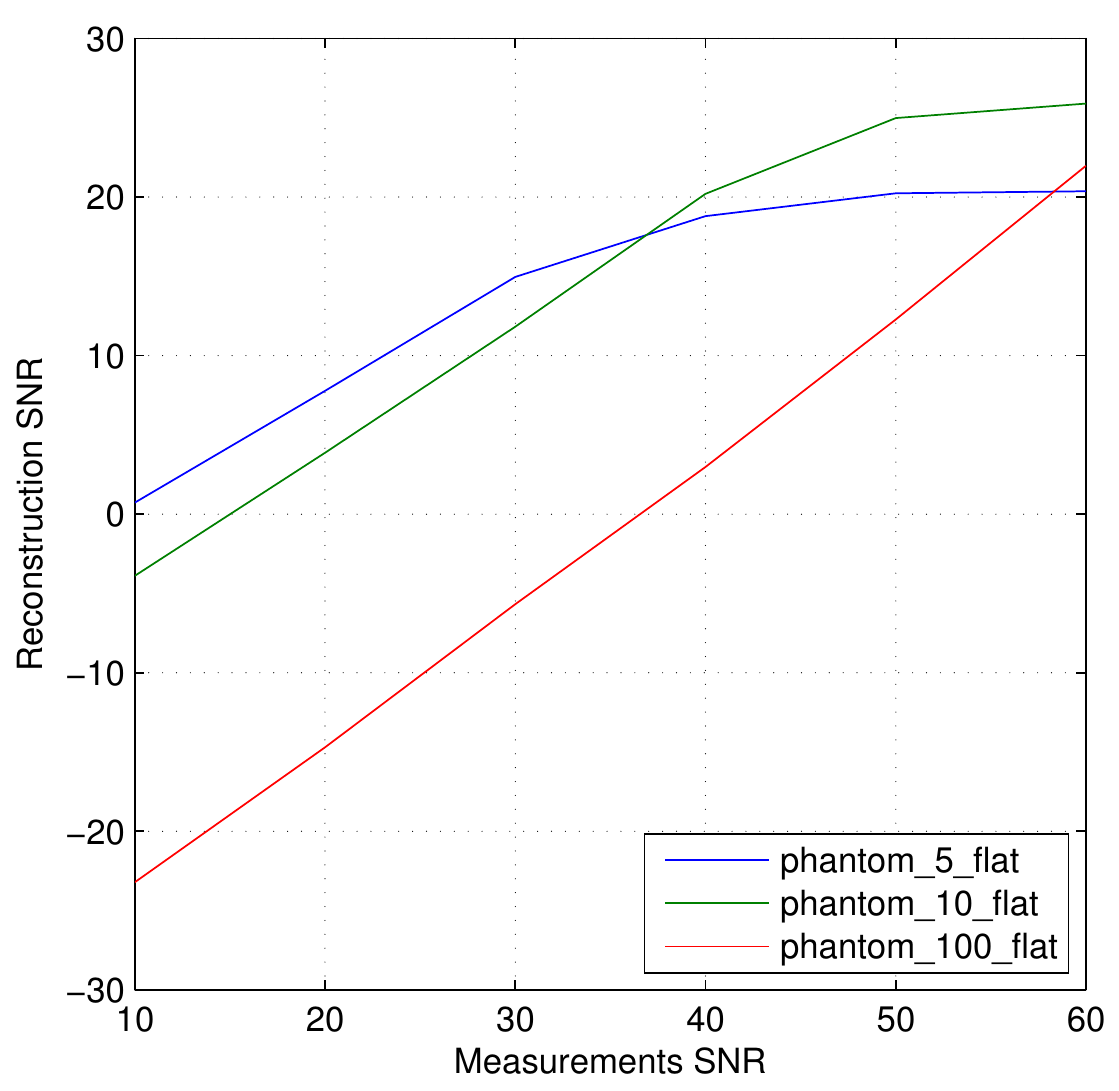}
  }
  \\
  \subfloat[]{
    \includegraphics[width=0.45\textwidth{}]{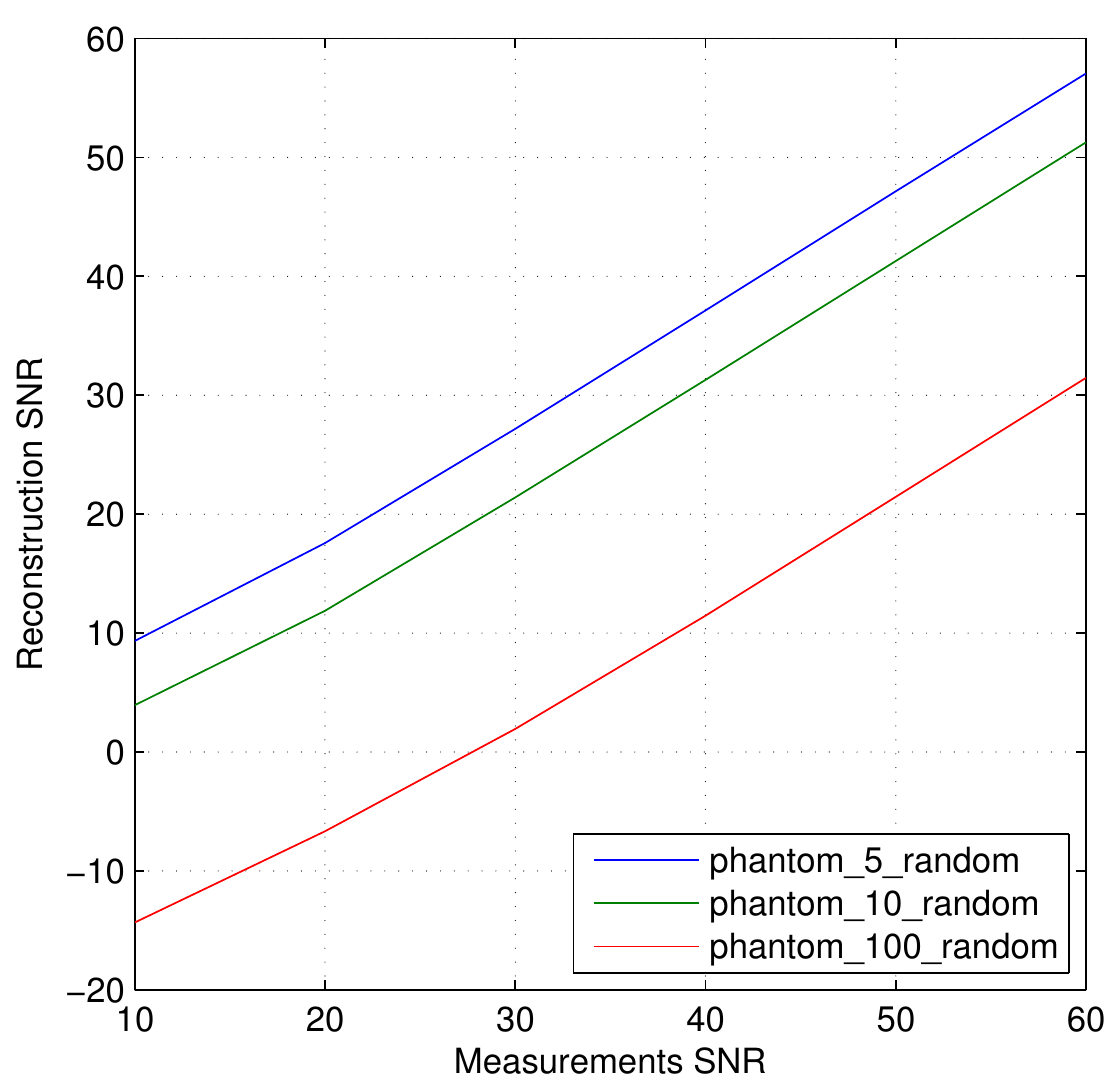}
  }
  \quad{}
  \subfloat[]{
    \includegraphics[width=0.45\textwidth{}]{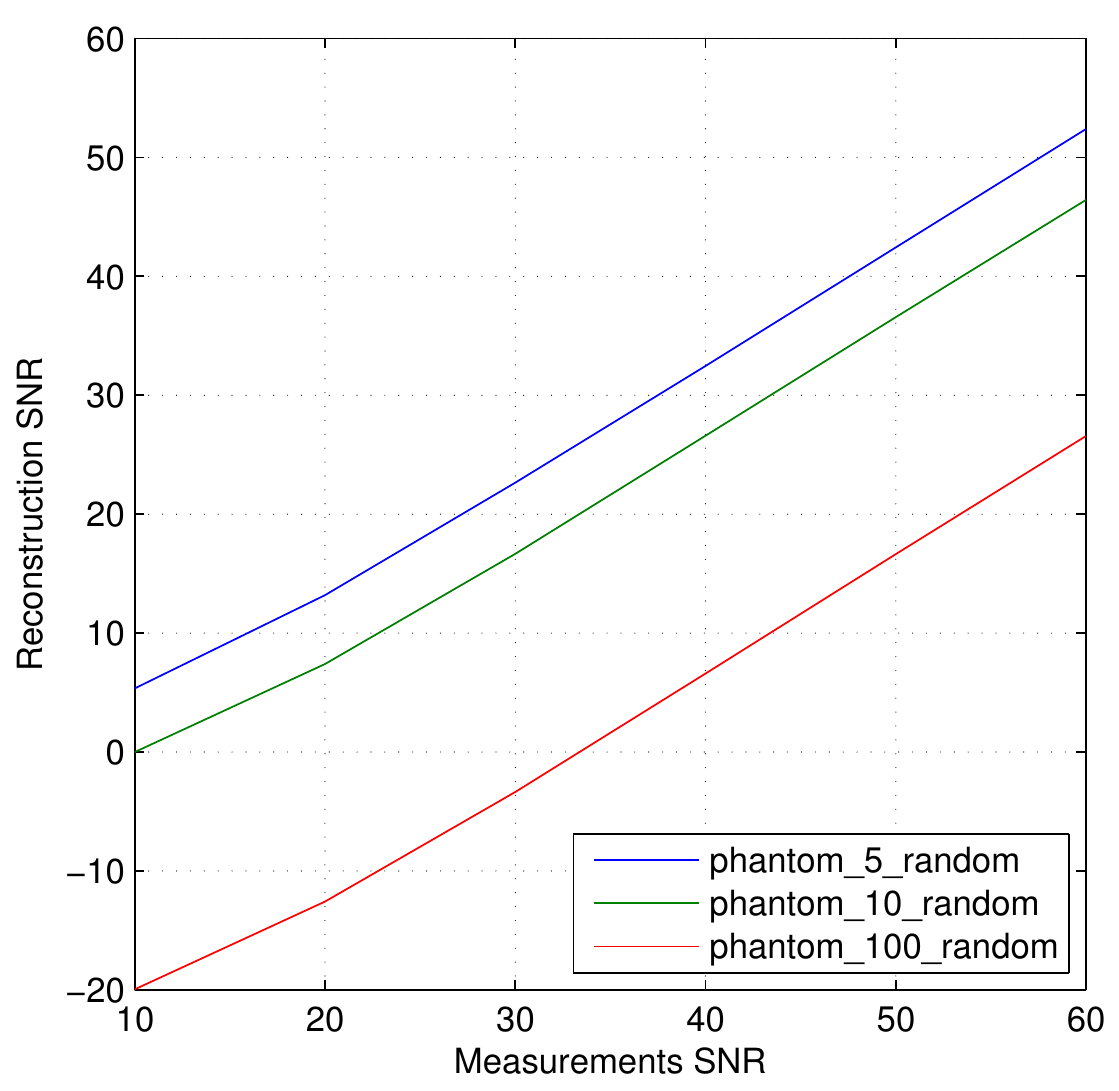}
  }
  \caption[Phantom's reconstruction quality]{Phantom's reconstruction
    quality: (a) real-valued with flat mask value, (b) complex-valued
    with flat mask value, (c) real-valued with random mask values, (d)
    complex-valued with random mask values.}
  \label{fig:phantom-reconstruction-quality}
\end{figure}

\section{Concluding remarks}
\label{sec:boundary-conclusions}

In this chapter we considered the problem often met in the Fourier
domain holography: signal reconstruction form the Fourier magnitude of
the sum of the sought signal and a reference beam. We provided an
explanation to the fact observed in practice: why a strong reference
beam leads to a faster reconstruction for a variety of reconstruction
methods. Based on this explanation we suggested a ``good''  boundary
(reference beam) design. The latter problem (reference beam design)
requires more research, as the optimal reference beam must satisfy at
least two requirements in the presence of signal dependent noise. For
example, in the case of Poissonian noise (or any other noise model in
which the noise level grows with the signal intensity) the optimal
reference beam must be simultaneously ``strong'' (to aid the reconstruction process)
and ``weak'' (to alleviate) the destructive influence of the noise.

In general, when the Fourier magnitude of the sought image is known
approximately, the best mask should have the power spectrum that is
about two times larger than that of the sought signal (in each
frequency). If the power spectrum of the sought images is unknown, the
mask should have a strong presence in all frequencies. In this case,
it seems that the best design would be based on some ``randomness'' in
the mask values or shapes.

%%% Local Variables: 
%%% mode: latex
%%% TeX-master: "../thesis"
%%% End: 

%% file: sparsity/sparsity.tex
\chapter{Bandwidth extrapolation using sparsity
  constraints\footnotemark\footnotemark}
\label{cha:bandw-extr-using}

\footnotetext[1]{The work presented in this chapter was done in
  collaboration with Prof. Segev's group form the Technion Physics
  Department, Solid State Institute.}
\footnotetext[2]{The material presented in this chapter was submitted
  to the Nature Materials journal. Part of it was presented at the
  \textit{Frontiers in Optics 2011} conference. Also, part of it was
  submitted to the \textit{CLEO 2012} conference.}

In this chapter we present our work on the bandwidth extrapolation
(super-resolution) problem with application to Coherent Diffracting
Imaging (CDI). CDI is an algorithmic imaging technique where intricate
features are reconstructed from measurements of the freely-diffracting
intensity pattern 
\shortcite{sayre52implications,miao99extending,miao01approach,quiney10coherent,chapman10coherent}. An important goal of such lensless-imaging methods
is to study the structure of molecules (including many proteins) that
cannot be crystallized
\shortcite{sandberg07lensless,chapman07femtosecond,neutze00potential}. Clearly,
high spatial resolution and very fast measurement are key features for
many applications of CDI. Ideally, one would want to perform CDI at
the highest possible spatial resolution and in a single-shot
measurement---such that the techniques could be applied to imaging at
ultra-fast rates. Undoubtedly, such capabilities would give rise to
unprecedented possibilities. For example, observing molecules while
they dissociate or undergo chemical reactions will considerably expand
the knowledge in physics, chemistry and biology. However, the
resolution of all current CDI techniques is limited by the diffraction
limit, and therefore cannot resolve features smaller than one half the
wavelength of the illuminating light \shortcite{lipson10optical}, which is
considered a fundamental limit in diffractive imaging
\shortcite{abbe73betrage}. Moreover, combining CDI with current
sub-wavelength imaging techniques would not allow for rapid
single-shot measurements that are able to follow ultra-fast dynamics,
because such techniques rely on multiple exposures, either through
mechanical scanning (e.g., Scanning Near-Field Microscope
\shortcite{lewis84development,betzig91breaking}, scanning a sub-wavelength
``hot spot''
\shortcite{di_francia52super-gain,lezec02beaming,huang09super-resolution}),
or by using ensemble-averaging over multiple experiments with
fluorescent particles
\shortcite{yildiz03myosin,hell09diffraction-unlimited}.  Here, we present
sparsity-based single-shot sub-wavelength resolution in coherent
diffraction microscopy: algorithmic reconstruction of sub-wavelength
features from far-field intensity patterns of sparse optical
objects. We experimentally demonstrate imaging of irregular and
ordered arrangements of \unit[100]{nm} features with illumination wavelength of
\unit[532]{nm} (green light), thereby obtaining resolutions several times
better than the diffraction limit. The sparsity-based sub-wavelength
imaging concept relies on minimization of the number of degrees of
freedom, and operates on a single-shot basis
\shortcite{gazit09super-resolution,szameit10sparsity-based,gazit10super-resolution}.
Hence, it is suitable for capturing a series of ultrafast
single-exposure images, and subsequently improving their resolution
considerably beyond the diffraction limit. This work paves the way for
ultrafast sub-wavelength CDI, via phase retrieval at the
sub-wavelength scale. For example, sparsity-based methods could
considerably improve the CDI resolution with x-ray free electron laser
\shortcite{chapman11femtosecond}, without hardware
modification. Conceptually, sparsity-based methods can enhance the
resolution in all imaging systems, optical and non-optical.

\section{Background information}
\label{sec:sparse-background}
Improving the resolution in imaging and microscopy has been a driving
force in natural sciences for centuries. Fundamentally, the
propagation of an electromagnetic field in a linear medium can be
fully described through the propagation of its eigen-modes (a complete
and orthogonal set of functions which do not exchange power during
propagation). In homogeneous, linear and isotropic media, the most
convenient set of eigen-modes are simply plane-waves, each
characterized by its spatial frequency and associated propagation
constant (see~\shortcite[Section~3.10]{goodman05introduction}). However,
it is also known that when light of
wavelength $\lambda$ propagates in media with refractive index $n$,
only spatial frequencies below $n/\lambda$ can propagate, whereas all
frequencies above $n/\lambda$ are rendered evanescent and decay
exponentially (see~\shortcite[Section~6.6]{goodman05introduction}). Hence, for all propagation distances larger than
$\lambda$, diffraction in a homogeneous medium acts as a low-pass
filter. Consequently, optical features of sub-wavelength resolution
appear highly blurred in a microscope, due to the loss of information
carried by their high spatial-frequencies. Over the years, numerous
``hardware'' methods for sub-wavelength imaging have been demonstrated
\shortcite{lewis84development,betzig91breaking,di_francia52super-gain,lezec02beaming,huang09super-resolution,yildiz03myosin,hell09diffraction-unlimited};
however, all of them rely on multiple exposures. Apart from hardware
solutions, several algorithmic approaches for sub-wavelength imaging
have been suggested (see,
e.g. \shortcite{harris64diffraction,papoulis75new,gerchberg74super-resolution}). Basically,
algorithmic sub-wavelength imaging aims to reconstruct the extended
spatial frequency range (amplitudes and phases) of the information
(``signal''), from measurements which are fundamentally limited to the
range $[-n/\lambda, n/\lambda]$ in the plane-wave spectrum. However,
as summarized by~\citeauthor{goodman05introduction} in his
book~\citeyear{goodman05introduction}, ``all methods for extrapolating
bandwidth beyond the diffraction limit are known to be extremely
sensitive to both noise in the measured data and the accuracy of the
assumed a priori knowledge'', such that ``it is generally agreed that
the Rayleigh diffraction limit represents a practical frontier that
cannot be overcome with a conventional imaging system.''

In spite of this commonly held opinion that algorithmic methods for
sub-wavelength imaging are impractical \shortcite{goodman05introduction}, a
recent work proposed a method for reconstructing sub-wavelength
features from the far-field (and/or blurred images) of sparse optical
information \shortcite{gazit09super-resolution}. The concept of
sparsity-based sub-wavelength imaging is related to Compressed Sensing
(CS), which is a relatively new area in information processing
\shortcite{candes06robust,candes06near-optimal,donoho06compressed,candes08introduction,duarte11structured}.
It has been shown that sparsity-based methods work for both coherent
\shortcite{gazit09super-resolution,szameit10sparsity-based} and incoherent
\shortcite{shechtman10super-resolution,shechtman11sparsity} light. An
experimental proof-of-concept was presented in
\shortcite{gazit09super-resolution}: the recovery of fine features that
were cut off by a spatial low-pass filter. Subsequently, these
concepts were taken into the true sub-wavelength domain and
demonstrated experimentally resolutions several times better than the
diffraction limit: the recovery of \unit[100]{nm} features illuminated
by \unit[532]{nm}
wavelength light \shortcite{szameit10far-field}. These ideas were followed
by several groups, most notably the recent demonstration of
sparsity-based super-resolution of biological specimens
\shortcite{babacan11cell}.

Here, we take the sparsity-based concepts into a new domain, and
present the first experimental demonstration of sub-wavelength CDI:
single-shot recovery of sub-wavelength images from far-field intensity
measurements.  That is, we demonstrate sparsity-based sub-wavelength
imaging combined with phase-retrieval at the sub-wavelength level. We
recover the sub-wavelength features without measuring (or assuming)
any phase information whatsoever; the only measured data at our
disposal is the intensity of the diffraction pattern (Fourier power
spectrum) and the support structure of the blurred image. Our
processing scheme combines bandwidth extrapolation and phase
retrieval, considerably departing from classical CS. We therefore
devise a new sparsity-based algorithmic technique which facilitates
robust sub-wavelength CDI under typical experimental conditions.

\section{Sparsity based super-resolution}
\label{sec:sparsity-based-super}

In mathematical terms, the bandwidth extrapolation problem underlying
sub-wavelength imaging corresponds to a non-invertible system of
equations which has an infinite number of solutions, all producing the
same (blurred) image carried by the propagating spatial
frequencies. That is, after measuring the far field, one can add any
information in the evanescent part of the spectrum while still being
consistent with the measured image. Of course, only one choice
corresponds to the correct sub-wavelength information that was cut off
by the diffraction limit. The crucial task is therefore to extract the
one correct solution out of the infinite number of possibilities for
bandwidth extension. This is where sparsity comes into play. Sparsity
presents us with prior information that can be exploited to resolve
the ambiguity resulting from our partial measurements, and identify
the correct bandwidth extrapolation which will yield the correct
recovery of the sub-wavelength image.

Information is said to be sparse when most of its projections onto a
complete set of base functions are zero (or negligibly small). For
example, an optical image is sparse in the near-field when the number
of non-zero pixels is small compared to the entire field of
view. However, sparsity need not necessarily be in a near-field basis;
rather, it can be in any mathematical basis. Many images are indeed
sparse in an appropriate basis. In fact, this is the logic behind many
popular image compression techniques, such as JPEG. In the fields of
signal processing and coding theory, it is known for some time that a
sparse signal can be precisely reconstructed from a subset of
measurements in the Fourier domain, even if the sampling is carried
out entirely in the low-frequency range
\shortcite{vetterli02sampling}. This basic result was extended to the case
of random sampling in the Fourier plane and initiated the area of CS
\shortcite{candes06robust}. An essential result of CS is that, in the
absence of noise, if the ``signal'' (information to be recovered) is
sparse in a basis that is sufficiently uncorrelated with the
measurement basis, then searching for the sparsest solution (that
conforms to the measurements) yields the correct solution. In the
presence of noise (that is not too severe), the error is bounded, and
many existing CS algorithms can recover the signal in a robust fashion
under the same assumptions.

The concept underlying sparsity-based super-resolution imaging and
sparsity-based CDI relies on the advance knowledge that the optical
object is sparse in a known basis. The concept yields a method for
bandwidth extrapolation. Namely, sparsity makes it possible to
identify the continuation of the truncated spatial spectrum that
yields the correct image. As was shown in
\shortcite{gazit09super-resolution}, sparsity-based super-resolution
imaging departs from standard CS, since the measurements are forced to
be strictly in the low-pass regime, and therefore cannot be taken in a
more stable fashion, as generally required by CS. Therefore, a
specialized algorithm was developed, Non Local Hard Thresholding
(NLHT), to reconstruct both amplitude and phase from low-frequency
measurements \shortcite{gazit09super-resolution}. However, NLHT, as well as
other CS techniques necessitate the measurement of the phase in the
spectral domain. In contrast, the current problem of sub-wavelength
CDI combines phase-retrieval with sub-wavelength imaging, aiming to
extrapolate the bandwidth from amplitude measurements
only. Mathematically, this problem can be viewed in principle as a
special case of quadratic CS, introduced in
\shortcite{shechtman11sparsity}. However, the algorithm suggested there
is designed for a more general problem resulting in high computational
complexity. Here we devise a specified algorithm that directly treats
the problem at hand.

\section{Sparsity based sub-wavelength CDI}
\label{sec:sparsity-based-cdi}
For the current case of
sub-wavelength CDI, the phase information in the spectral domain is
not available. Hence, fundamentally, sub-wavelength CDI involves both
bandwidth extrapolation and phase retrieval. However, despite the
missing phase that carries extremely important information, we show
that sparsity-based ideas can still make it possible to identify the
correct extrapolation. Namely, if we know that our signal is
sufficiently sparse in an appropriate basis, then from all the
possible solutions which could create the truncated spectrum the
correct extrapolation is often the one yielding the maximally sparse
signal. Moreover, even under real experimental conditions, i.e., in
the presence of noise, searching for the sparsest solution that is
consistent with the measured data often yields a reconstruction that
is very close to the ideal one.

Our algorithm iteratively reveals the support of the sought image by
sequentially rejecting less likely areas (circles, in the experiments
shown below). Thus, the sparsity of the reconstructed image increases
with each iteration. This process continues as long as the
reconstructed image yields a power-spectrum that remains in good
agreement with the measurements. The process stops when the
reconstructed power spectrum deviates from the measurements by some
threshold value. However, it is important to emphasize that the exact
threshold value and the degree of sparseness of the sought image need
not be known a priori, as our method provides a natural termination
criterion. Namely, the correct reconstruction is identified
automatically. A detailed description of the reconstruction method, as
well as comparison with other methods (that do not exploit sparsity),
are provided in Section~\ref{sec:reconstr-meth}.

\section{Finding suitable basis}
\label{sec:find-suit-basis}
As explained above, sparsity-based CDI relies on the advance knowledge
that the object is sparse in a known basis. In some cases, the
``optimal'' basis---the basis in which the object is represented both
well and sparsely---is known from physical arguments. For example, the
features in Very Large Scale Integration (VLSI) chips are best
described by pixels on a grid, because they obey certain design
rules. In some cases, however, the prior knowledge about the optimal
basis is more loose, namely, it may be known that the object is well
and sparsely described in a basis that belongs to a certain family of
bases. For example, one may know in advance that the object is sparse
in the near field using a rectangular grid, yet the optimal grid
spacing is not known a priori. We address this issue in
Section~\ref{sec:choosing-grid-basis},  where we describe a sparsity-based
method that uses the experimental data to algorithmically find the
optimal grid size (optimal basis) for our sub-wavelength CDI
technique. That section also shows that the choice of basis functions
is not particularly significant in our procedure: we obtain very
reasonable reconstruction with almost any choice of basis functions,
as long as they conform to the optimal grid. Finally, we note that
recent work has shown that it is often possible to find the basis from
a set of low-resolution images, using ``blind
CS''\cite{gleichman10blind}.
Likewise, in
situations where a sufficient number of images of a similar type is
available at high resolution, one can reconstruct the optimal basis
through dictionary learning algorithms~\shortcite{aharon06k-svd:}.

\section{Experiments}
\label{sec:sparse-experiments}
We demonstrate sub-wavelength CDI technique experimentally on
two-dimensional (2D) structures. The optical information is generated
by passing a $\lambda=\unit[532]{nm}$ laser beam through an arrangement of
nano-holes of diameter $100nm$ each. The sample is made of a
\unit[100]{nm} thick
layer of chromium on glass; this thickness is larger than the skin
depth at optical frequencies, such that the sample is opaque except
for the holes. We use a custom microscope (numerical aperture $NA=1$,
magnification $\times 26$)
and a camera to obtain the blurred image. The optical Fourier
transform of the optical information is obtained by translating the
camera to the focal plane of the same microscope.

The optical information is generated by passing a collimated laser
beam ($\lambda = \unit[532]{nm}$) through a mask, whose transmission function
corresponds to the optical information superimposed on the laser
beam. The mask is fabricated as follows: As substrate material we
chose fused silica, because it is a high quality transparent material
at optical frequencies, and because its processing technology is well
developed. In order to to create a mask containing the optical image,
we deposit opaque material on the substrate and make several patterned
holes in it, such that the holes pass the light while the opaque
material blocks it. For this purpose, we sputter a chromium layer onto
the surface of the substrate. Chromium is a metal, which absorbs light
at optical frequencies. Nevertheless, the thickness of the chromium
layer has to be larger than the skin depth at optical frequencies, to
avoid undesired transmission through that layer. Thus we select a
thickness of \unit[100]{nm} as suitable compromise between high quality
optical behavior and fabrication considerations. The structures in the
chromium layer are nano-holes, drilled in the chromium by a beam of
focused gallium ions from a liquid metal ion
source~\shortcite{krohn75ion,prewett80characteristics} (Zeiss Neon
60). With this technology, it is feasible to mill the desired
structures into the chromium layer directly and efficiently, without
any additional lithography process. Utilizing a convenient set of
parameters, it is possible to imprint the designed structures into the
metal layer, without significantly affecting the substrate material,
and with high spatial accuracy.  We fabricated two different samples
yielding a two-dimensional sub-wavelength optical structure: (a) a
Star of David (SOD) image, consisting of 30 holes, with \unit[100]{nm}
diameter each, spaced by \unit[100]{nm}; and (b) a ``random'' image comprised of 12
circular holes of \unit[100]{nm} diameter each, placed in a random order. The
Scanning Electronic Microscope (SEM) images of the samples are shown
in Figure~\ref{fig:sem-images}. Note that the SEM images are not in
proportion as, in reality, the holes are of the same size and their
diameter is equal to the spacing between holes. Generally, we
use this approach throughout the paper: all images are shown in some
abstract units that are, however, proportional to the corresponding
physical quantities. The correspondence can be established using the
fact that all holes are of diameter of \unit[100]{nm}.
\begin{figure}[H]
  \centering
  \subfloat[]{
    \includegraphics[width=0.9\textwidth{}]{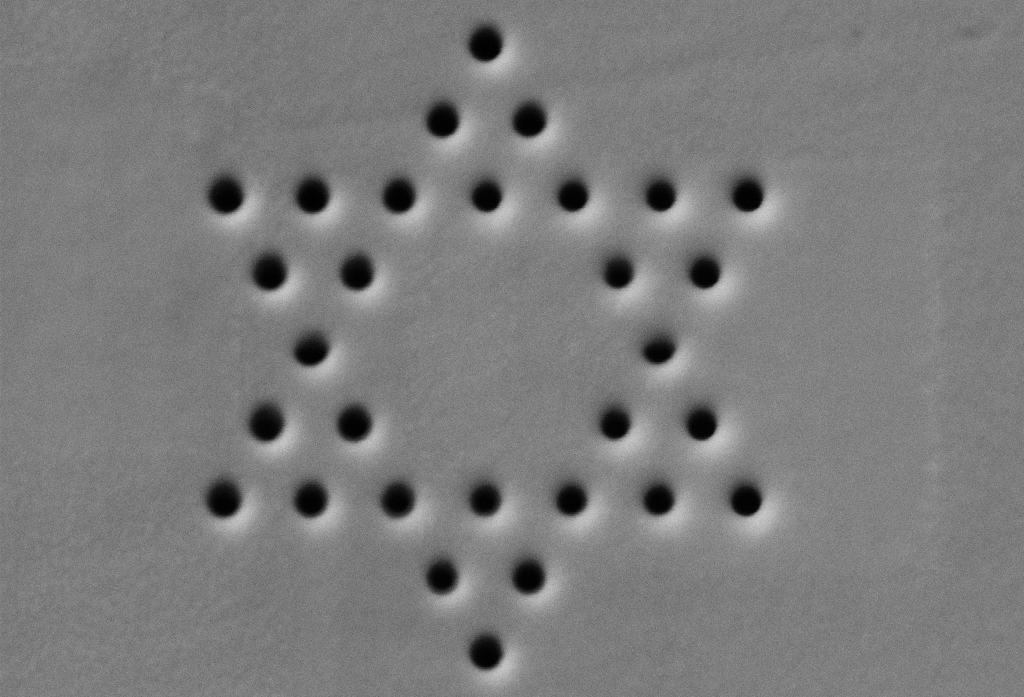}
  }\\
  \subfloat[]{
    \includegraphics[width=0.9\textwidth{}]{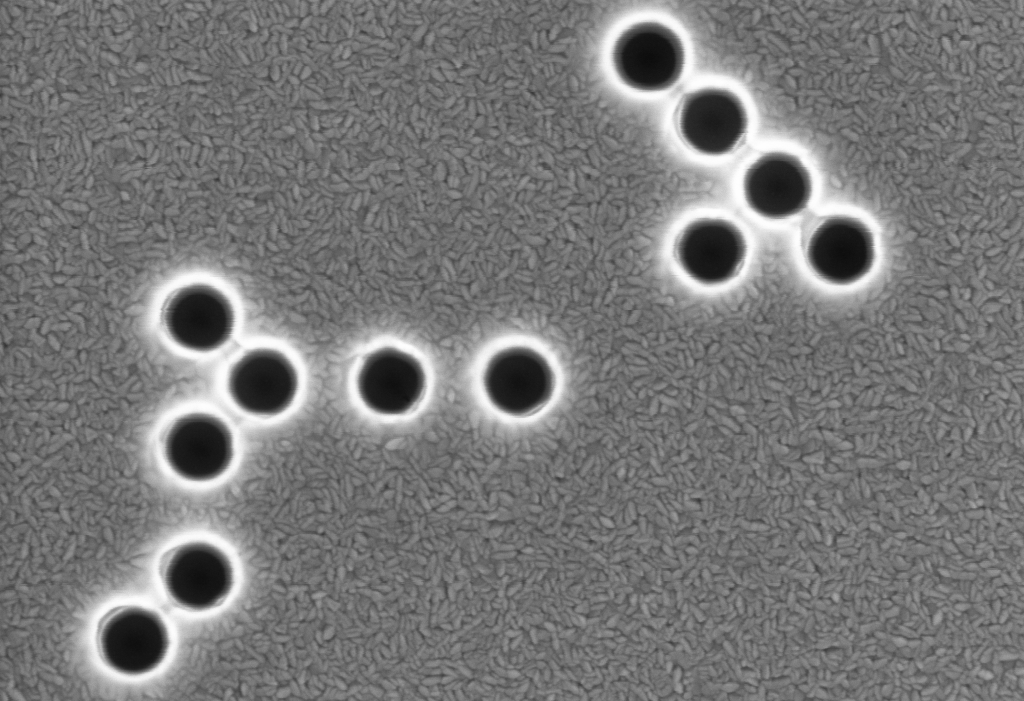}
  }
  \caption[Test images]{SEM images of the samples: (a) Star of David, (b)
    ``random''.}
  \label{fig:sem-images}
\end{figure}

Let us present the reconstruction results first, followed by a full
description of our method in Section~\ref{sec:reconstr-meth}.  We
begin with an ordered structure: a Star of David, consisting of 30
nanoholes. Figure~\ref{fig:sod-main-results}a shows an SEM image of
this sample.  Figure~\ref{fig:sod-main-results}b depicts the image
seen in the microscope. As expected, the image is small and severely
blurred. The spatial power spectrum (absolute value squared of the
Fourier transform) of the image is shown in
Figure~\ref{fig:sod-main-results}c. This truncated power spectrum
covers a larger area on the camera detector, therefore facilitating a
much higher number of meaningful measurements (each pixel corresponds
to one measurement). We emphasize that only intensity measurements are
used, in both the (blurred) image plane and in the (truncated) Fourier
plane (Figures~\ref{fig:sod-main-results}b,
\ref{fig:sod-main-results}c, respectively), without measuring (or
assuming) the phase anywhere. The recovered image, using our
sparsity-based algorithm, is shown in
Figure~\ref{fig:sod-main-results}d.  Clearly, we recover the correct
number of circles, their positions, their amplitudes, and the entire
spectrum (amplitude and phase), including the large evanescent part of
the spectrum. This demonstrates sub-wavelength Coherent Diffractive
Imaging: image reconstruction combined with phase-retrieval at the
sub-wavelength scale. Moreover, as explained in
Section~\ref{sec:reconstr-meth}, the intensity of the blurred image
(Figure~\ref{fig:sod-main-results}b) is used only for rough estimation
of the image support. Our reconstruction method yields better results
than other phase-retrieval algorithms (see comparisons in
Section~\ref{sec:comp-with-other}), because it exploits the sparsity
of the signal (the image to be recovered), as prior information. As
mentioned earlier, the underlying logic is to minimize the number of
degrees of freedom, while always conforming to the measured data,
which in this case is the truncated power spectrum (intensity in
Fourier space). In the example presented in
Figure~\ref{fig:sod-main-results}, we take the data from
Figures~\ref{fig:sod-main-results}b and~\ref{fig:sod-main-results}c,
search for the sparsest solution in the basis of circles of \unit[100]{nm}
diameter on a grid, and reconstruct a perfect Star of David, as shown
in Figure~\ref{fig:sod-main-results}d. The grid is rectangular with
\unit[100]{nm}  spacing
(Section~\ref{sec:choosing-grid-basis} describes how this parameter is
found automatically), while the exact position of the grid with
respect to the reconstructed information is unimportant (see
Section~\ref{sec:reconstr-meth}).
\begin{figure}[H]
  \centering
  \includegraphics[width=\textwidth{}]{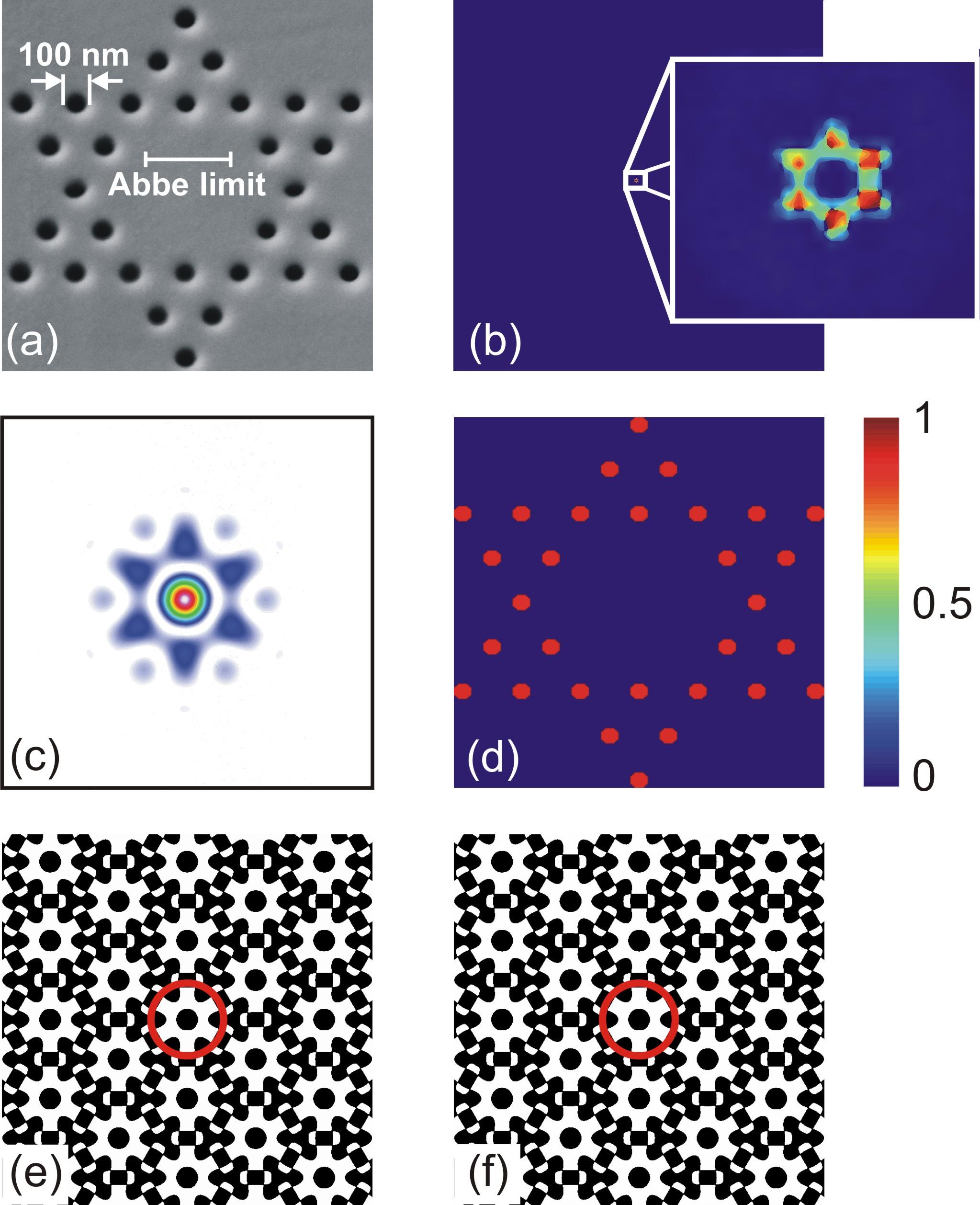}
  \caption{Reconstruction results for the SOD image.}
  \label{fig:sod-main-results}
\end{figure}

We emphasize that our reconstruction algorithm is able to reconstruct
the phase in the spatial spectrum domain (the Fourier transform), from
the intensity measurement in Fourier space and some rough estimation
of the image support. In addition, we use the knowledge that the holes
are illuminated by a plane wave, implying non-negativity of the image
in real space. In this Star of David example, our algorithm
reconstructs the phase in the spectral plane, as presented in Figure
1e. For comparison, Figure 1f shows the phase distribution in Fourier
space, as obtained numerically from the ideal model of the
subwavelength optical information (calculated from the SEM image of
Figure~\ref{fig:sod-main-results}a). The reconstruction in
Figure~\ref{fig:sod-main-results} therefore constitutes the first
demonstration of subwavelength CDI.

Interestingly, when comparing the Fourier transform of the sample with
the measured spatial power spectrum, one finds that more than 90\% of
the power spectrum is truncated by the diffraction limit, acting as a
low-pass filter (see Figure~\ref{fig:sod-powerspectrum}).  That is, we
use the remaining 10\% of the power spectrum and the blurred image, to
successfully reconstruct the sub-wavelength features with high
accuracy. In other words, the prior knowledge of sparsity and the
basis is overcoming the loss of information in 90\% of the power
spectrum. As demonstrated in Section~\ref{sec:reconstr-meth}, it is
the sparsity prior that makes it happen: without assuming the sparsity
prior the reconstruction suffers from large errors.
\begin{figure}[H]
  \centering
  \includegraphics[width=\textwidth{}]{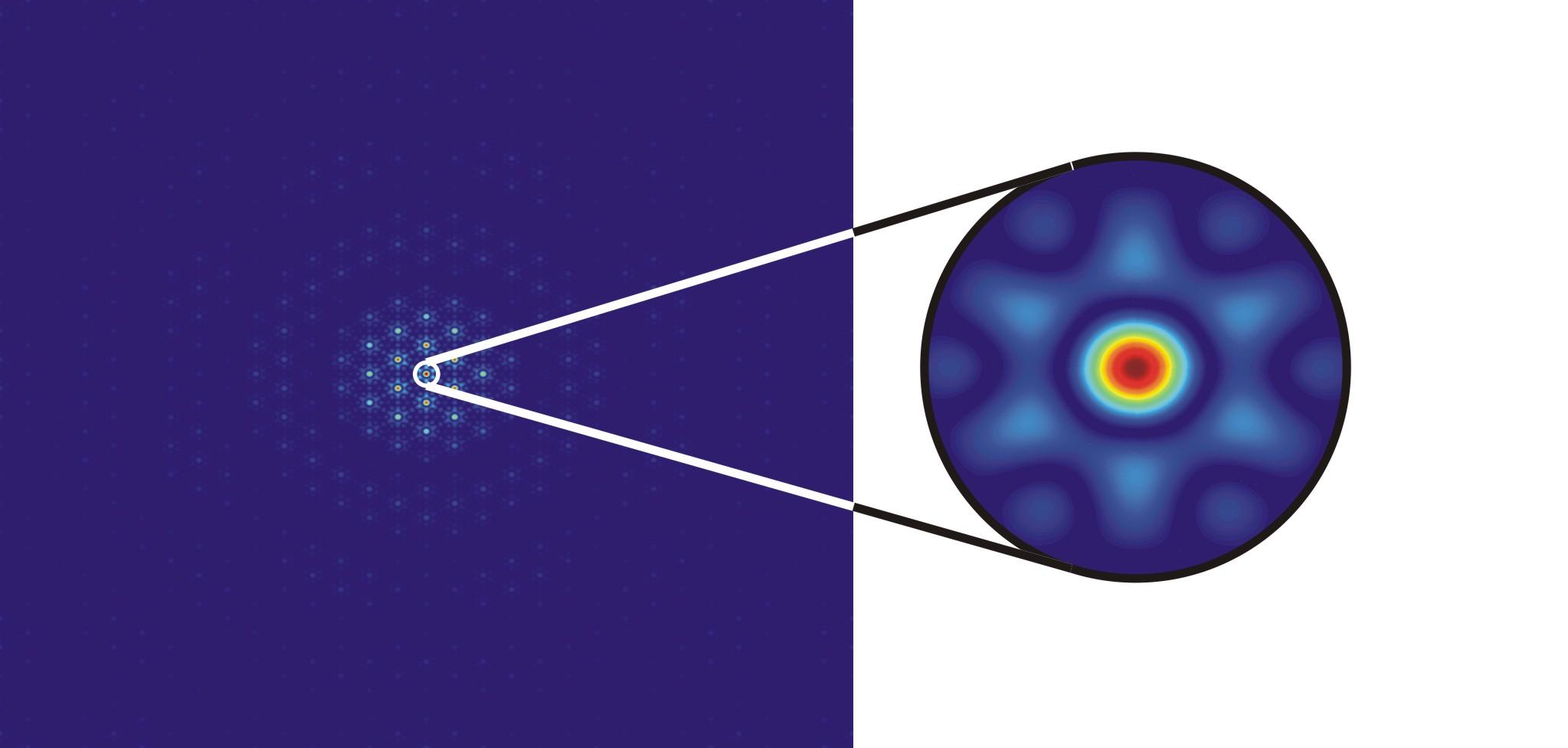}
  \caption{SOD image: available power spectrum.}
  \label{fig:sod-powerspectrum}
\end{figure}

The Star of David exhibits certain symmetries which could in principle
assist the phase retrieval, had these symmetries been known. However,
symmetry was not used for reconstruction of sub-wavelength features of
Figure~\ref{fig:sod-main-results}. Nevertheless, it is illustrative to
present another example with no spatial symmetry at all: an irregular
arrangement of sub-wavelength holes on the assumed
grid. Figure~\ref{fig:random-main-results}a shows the blurred image of
an unknown number of sub-wavelength circles, distributed in a random
manner. The respective Fourier power spectrum, as observed in the
microscope, is shown in Figure~\ref{fig:random-main-results}b. This
sample is clearly not symmetric in real space, hence it does not
exhibit a real Fourier transform. Still, we are able to reconstruct
the sub-wavelength information, as shown in
Figure~\ref{fig:random-main-results}c, where all features of the
original sample are retrieved, despite the inevitable noise in the
experimental system. Figure~\ref{fig:random-main-results}c shows the
SEM image of the sample, displaying the random arrangement of \unit[100]{nm}
holes. The electromagnetic (EM) field passing through these nano-holes
has roughly the
same amplitude for all the holes. The reconstructed amplitudes at the
hole sites are represented by the colors in
Figure~\ref{fig:random-main-results}c, highlighting the fact that the
reconstructed field has similar amplitude at all the holes. The
reconstructed phase in the spectral plane is presented in Figure 3e,
where the white circle marks the cutoff imposed by the diffraction
limit. As shown there, our algorithm recovers the phase throughout the
entire Fourier plane, including the region of evanescent waves far
away from the cutoff frequency. For comparison,
Figure~\ref{fig:random-main-results}f shows the phase distribution in
Fourier space, as obtained numerically from the ideal model of the
sub-wavelength optical information (calculated from the SEM image of
Figure~\ref{fig:random-main-results}d). Clearly, the correspondence
between the original spectral phase and the reconstructed one is
excellent, including in the deep evanescent regions. Interestingly,
Figure~\ref{fig:random-main-results}e also displays the correct
reconstruction of the phase around the faint high-frequency circle (of
radius approximately 4 times the diffraction limit) where the phase
jumps by $\pi$ Physically, this ``phase-jump circle'' is located at the
first zero of the Fourier transform of a circular aperture, which in
Fourier space multiplies the phase distribution generated by the
irregular positions of the holes. The excellent agreement between
Figures~\ref{fig:random-main-results}e
and~\ref{fig:random-main-results}f highlights the strength of the
sparsity-based algorithmic technique.

\begin{figure}[H]
  \centering
\includegraphics[width=\textwidth]{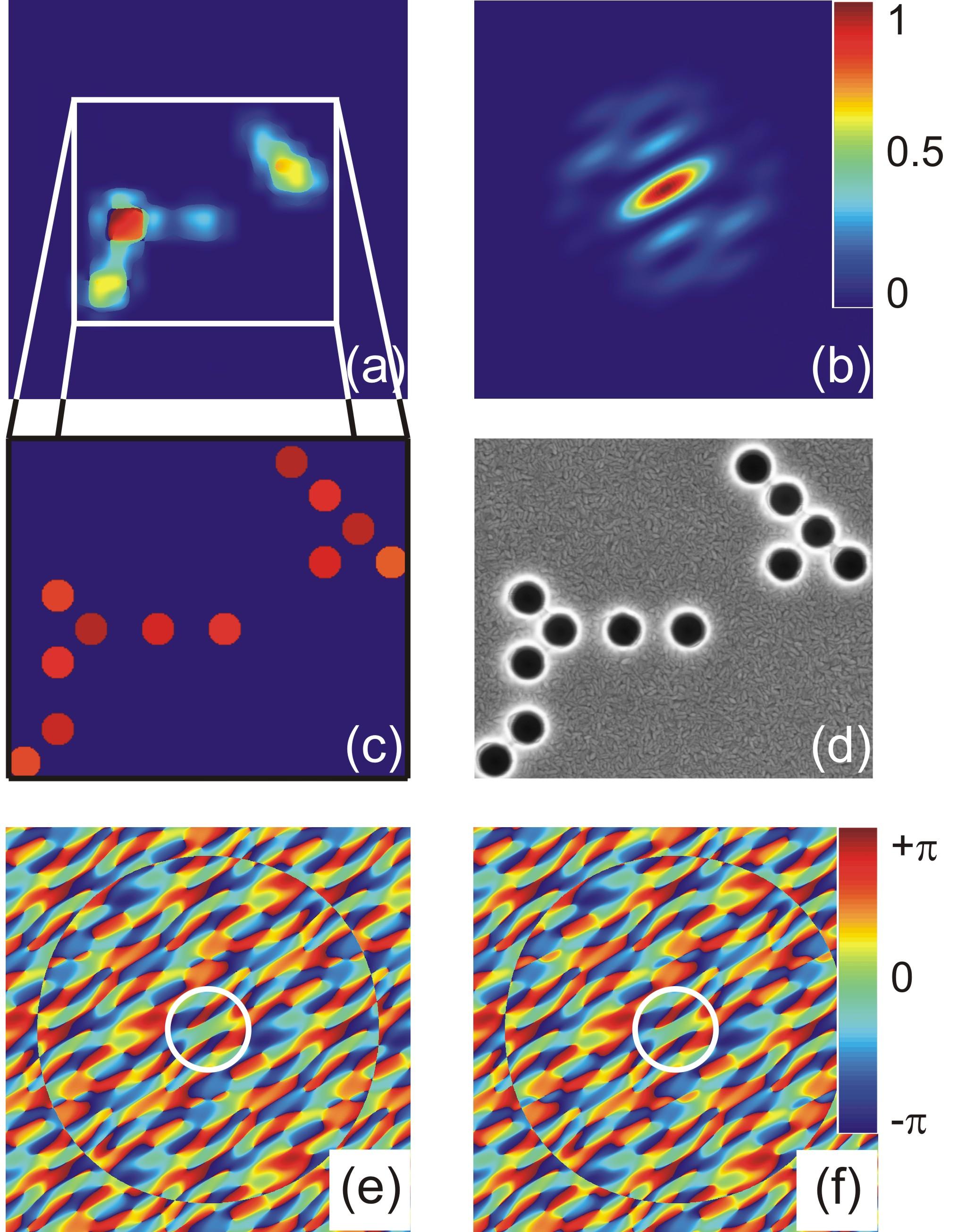}
  \caption{Reconstruction results for the "random" image.}
  \label{fig:random-main-results}
\end{figure}

\section{Our sparsity-based reconstruction method for CDI}
\label{sec:reconstr-meth}
Under the experimental conditions described in the previous section,
our problem amounts to the reconstruction of a signal from the
magnitude of its Fourier transform, assuming furthermore that this
information is known only for a small interval of low frequencies as
shown in Figure~\ref{fig:fourier-data}. The discussion below is
general and applies to both examples given in the paper (and, of
course, to a very large class of optical images). However, in order to
make the explanation more succinct, we demonstrate most of the results
on the ``random'' image, because it has no implicit symmetries. The
SOD image exhibits very similar behavior and its main results will be
presented
below.% We would like to remind that all the images are shown in
% some abstract units that are, nevertheless, proportional to the
% corresponding physical quantities.
\begin{figure}[H]
  \centering
  \subfloat[]{
    \includegraphics[width=0.4\textwidth]{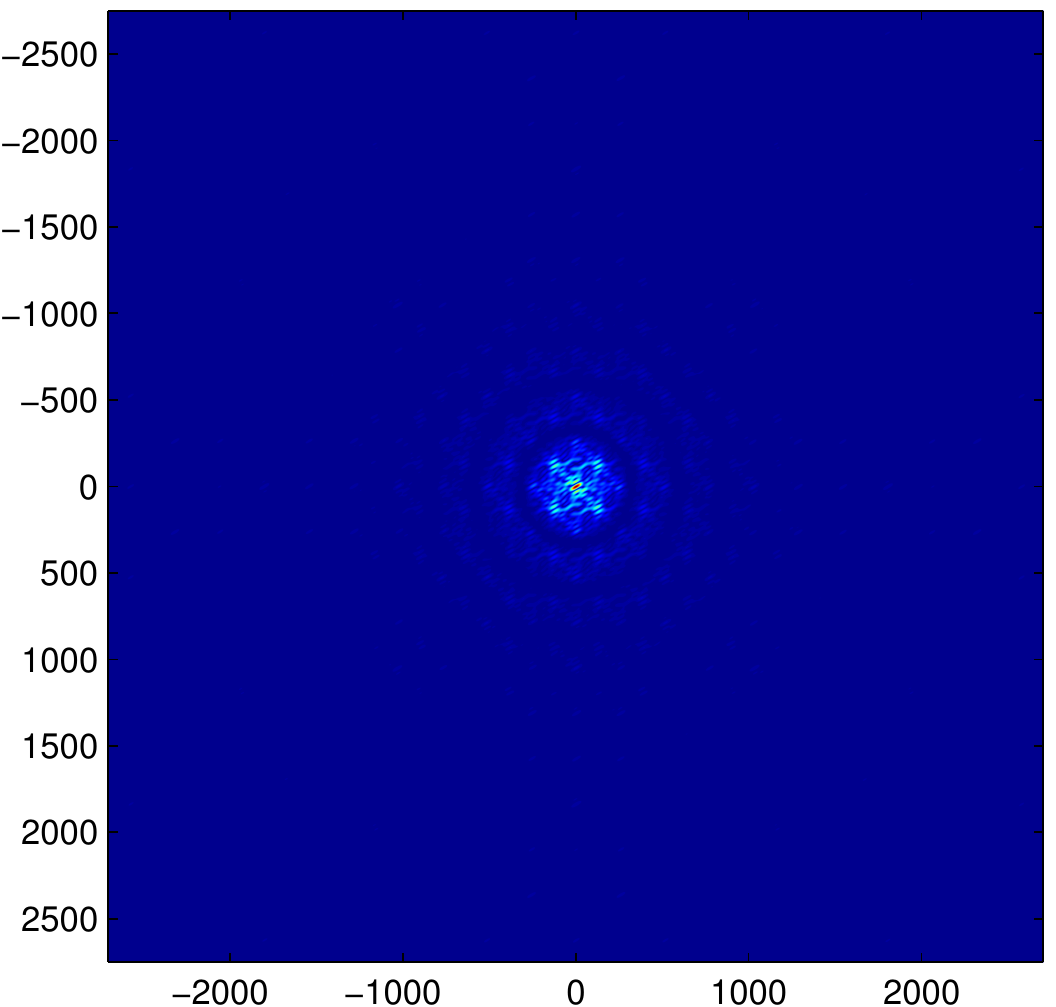}
  }%
  \qquad{}%
  \subfloat[]{
    \includegraphics[width=0.4\textwidth]{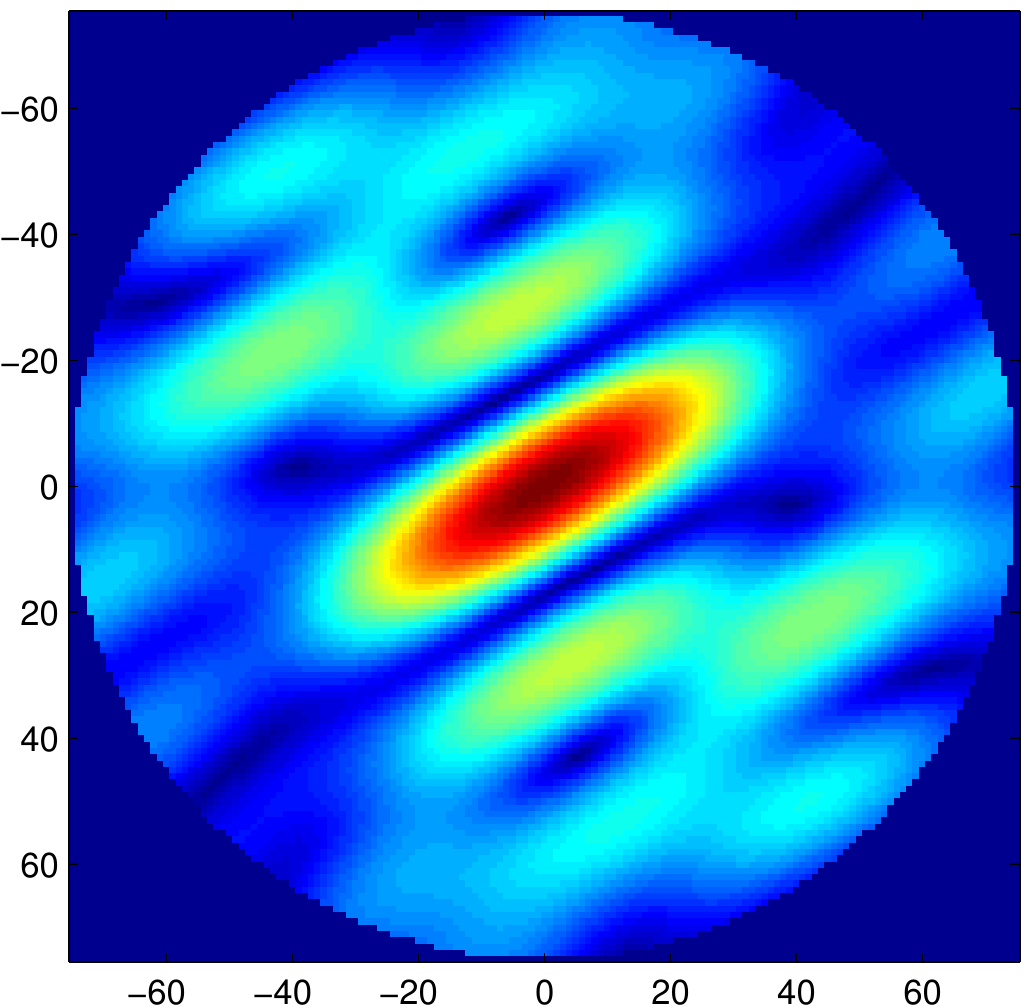}
  }
  \caption[Fourier domain magnitude of the ``random'' image]{Fourier domain magnitude of the ``random'' image: (a) the
    full spectrum (simulated, without noise) needed to reconstruct the
    image precisely (by a simple application of the inverse Fourier
    transform), (b) the low-frequency part (actual measurements, in
    the presence of experimental noise).}
  \label{fig:fourier-data}
\end{figure}
Of course, when the majority of the frequencies are lost, precise
reconstruction is not possible, unless we have, or may assume, some
additional information about the sought signal. In fact, the problem
is even more difficult, because the measurements contain non-negligible
noise.  In a manner similar to~\shortcite{gazit09super-resolution}, we
assume that the EM field in the object domain $(u,v)$ can be
represented precisely, or approximated adequately (hereinafter, this
relation is denoted by $\approxeq$) by means of a known generating
function $g(u,v)$. That is
\begin{equation}
  \label{eq:sparse-1}
  E(u,v) \approxeq \sum_{m}\sum_{n} x_{mn}g(u-m\Delta_{u}, v - n\Delta_{v})\ ,
\end{equation}
where ${x_{mm}}$ are unknown signal coefficients in the basis defined
by the shifted versions of $g(u,v)$. Note that the set
$\left\{(m\Delta_{u}, n\Delta_{v})\right\}$ defines a rectangular grid
where the shifted versions of the generating function are
located. Hence, for example, by choosing $g(u,v)$ to be the Dirac
delta function we can obtain the sampled version of the continuous EM
field distribution, where $\Delta_{u}$, and $\Delta_{v}$ define the
sampling interval.  Another classical example: all bandwidth limited
signals can be represented precisely in this form when $g$ is chosen
to be the $\mathrm{jinc(\rho)}$ function. For more examples
see~\shortcite{eldar09beyond} and references therein. Of course, the
generator must be chosen in a way that corresponds to the signal in
question (although, in the most general case of 2D information, the
generator could simply be rectangular pixels).  In this section and in
Section~\ref{sec:comp-with-other}, where we compare our algorithm with
other methods, we assume that the basis function is chosen in a way
that allows a perfect reconstruction of the sought signal, namely $g$
represents a circle of a priori known diameter (\unit[100]{nm}). We assume
also that $\Delta_{u}=\Delta_{v}=\unit[100]{nm}$. That is, we assume that the
sought signal is comprised of non-overlapping circles of known
diameter.  The grid $\left\{(m\Delta_{u}, n\Delta_{v})\right\}$
containing all possible locations (144) is shown in
Figure~\ref{fig:full-grid}. Note that the exact placement of the grid
is unimportant as our measurements are insensitive to shifts.  A more
detailed explanation of this property is presented in
Section~\ref{sec:choosing-grid-basis}, where we discuss the
implications of the grid assumption along with the impact of the basis
function on the reconstructed signal.

Before moving on, there are two points we would like to stress. First,
the assumption of an underlying grid is natural in many situations
arising in digital signal processing. A prominent example are digital
images that are comprised of pixels located on a rectangular
grid. Just like in digital images, the grid in our case defines the
resolution of a digitized version of the sought signal (see
Section~\ref{sec:choosing-grid-basis} for details). Second, it is
important to note that all our comparisons with other methods are done
under exactly the same assumptions, including a grid, basis functions,
etc. As is evident from the experiments presented in
Section~\ref{sec:comp-with-other}, our algorithm outperforms other
methods.
\begin{figure}[H]
  \centering
  \includegraphics[width=0.4\textwidth]{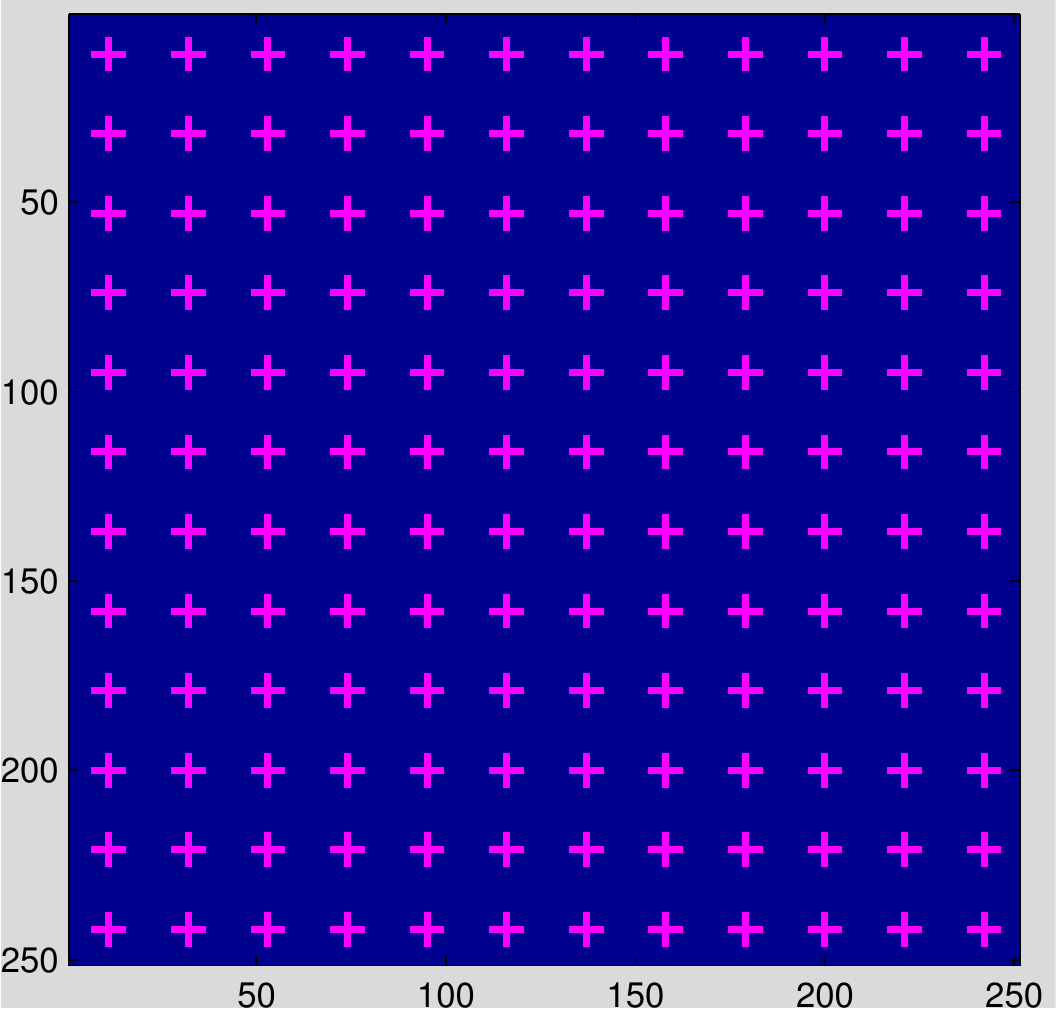}
  \caption{The full grid.}
  \label{fig:full-grid}
\end{figure}
We emphasize  that even if the correct number of circles were known
(12 circles, in this example) there
would be ${144 \choose 12} > 10^{17}$ possible variants to choose from for
the signal support. To limit the search space, we use the blurred version
of the signal as shown in Figure~\ref{fig:grid-restriction}.
\begin{figure}[H]
  \centering
  \subfloat[]{
    \includegraphics[width=0.4\textwidth]{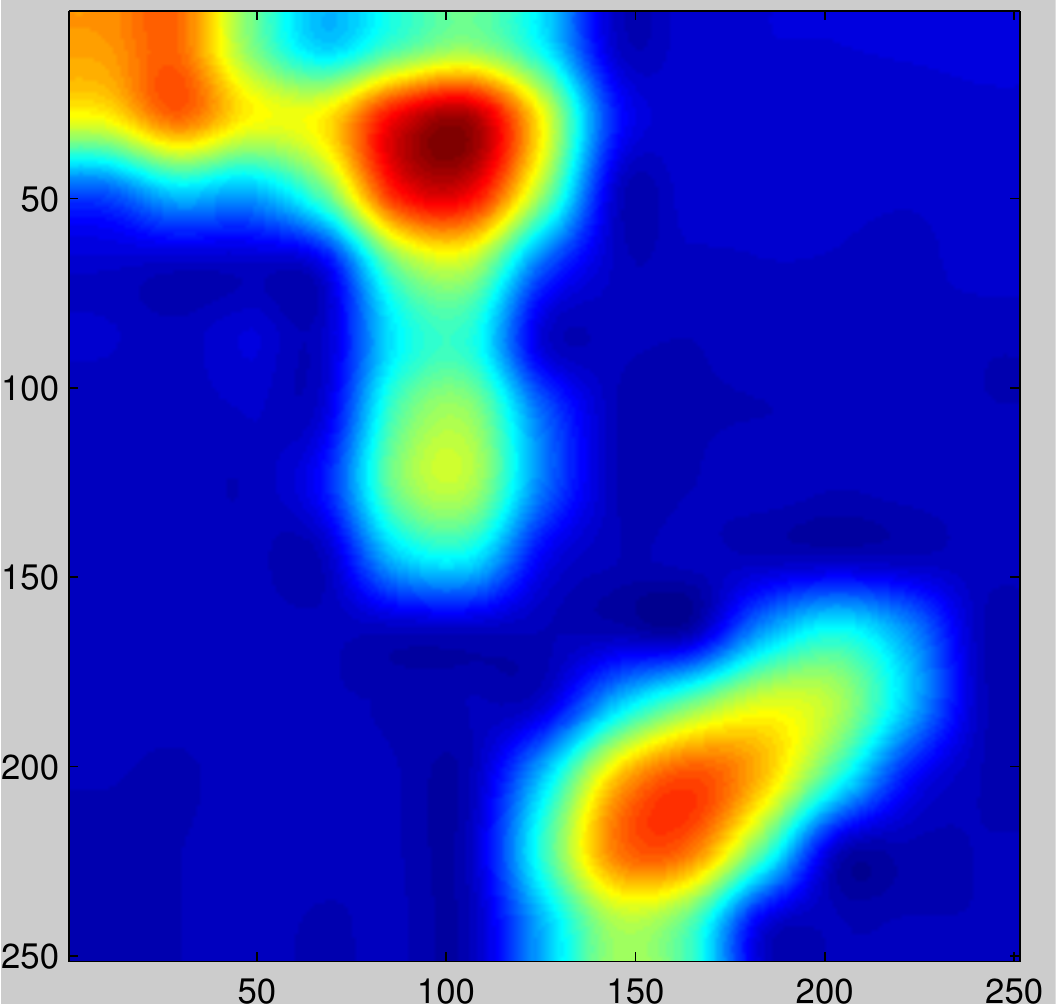}
  }
  \qquad{}
  \subfloat[]{
    \includegraphics[width=0.4\textwidth]{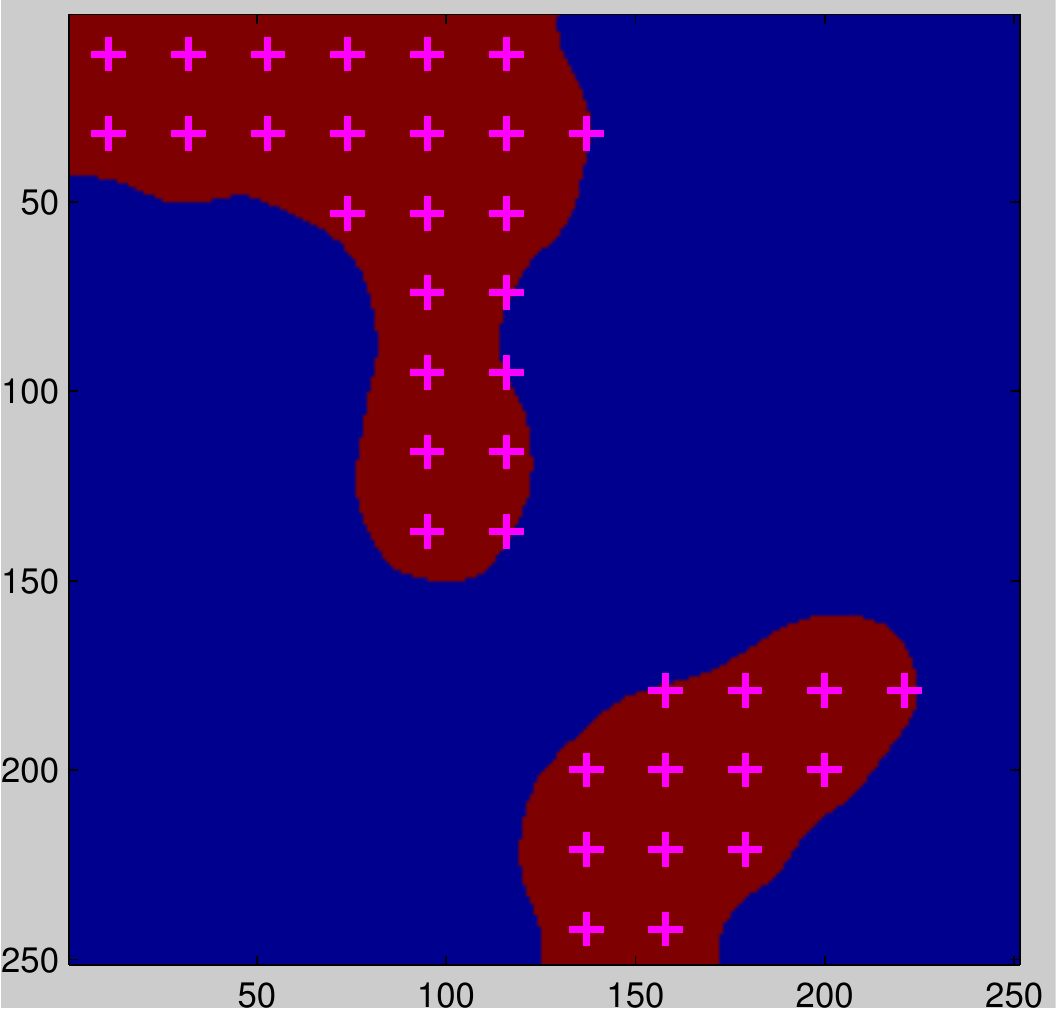}
  }
  \caption[Support restriction by the low-resolution image]{Support restriction by the low-resolution image: (a)
    blurred image magnitude, (b) grid restricted by the blurred
    image.}
  \label{fig:grid-restriction}
\end{figure}
However, even after this restriction, there still remain ${37 \choose 12}
> 1.85\times 10^{9}$ variants. More importantly,
even after this restriction, the image cannot be reconstructed
precisely unless additional information is available (see
Section~\ref{sec:comp-with-other}). Below we present
our method that provides excellent reconstruction results based on the
knowledge that the total number of circles in the image is small,
that is, the image is \emph{sparse} in the basis associated with the
circles, as defined by Equation~\eqref{eq:sparse-1}.

In our method, we reconstruct the support and the magnitude of the
circles in the sought signal simultaneously. To this end, we seek
the sparsest $x$ ($x$ being a column vector comprised of the
image coefficients $x_{mn}$ as defined by Equation~\eqref{eq:sparse-1}), that
yields a good agreement with the measurements. Mathematically, we
try to solve the following optimization problem
\begin{equation}
  \label{eq:sparse-2}
  \begin{split}
    \min &\quad\|x\|_{0}\\
    \mathrm{subject\ to} &\quad \||LFCx| - r\|_{2}^{2}\leq\epsilon\, , \\
    &\quad x \geq 0\ .
  \end{split}
\end{equation}
Here $\|\cdot\|_{0}$ denotes the $l_{0}$ norm: $\|x\|_{0} =
\sum_{i}|x_{i}|^{0}$,  that is, $\|x\|_{0}$
equals to the number of elements of $x$ that are not zero. The
measured (noisy) magnitude in the Fourier domain is denoted by
$r$. Note that the operators and inequalities, like $|\cdot|$, and
$\geq$ are applied element-wise. The matrix $C$ represents all
possible shifts of the generator function (a circle); hence, $Cx$ is
the actual image that we reconstruct; $F$ stands for the Fourier
transform operator; $L$ represents the low-pass filter. That is, $L$
is obtained from the identity matrix of appropriate size by removing
most of its rows while keeping only those that correspond to the low
frequencies of its operand, as shown in
Figure~\ref{fig:fourier-data}. Physically, $L$ is the low-pass filter
associated with the cutoff spatial frequency of the optical system,
which, for microscopes with NA=1, corresponds to the diffraction
limit. Note that, due to errors in the
measurements, the discrepancy in the Fourier domain is allowed to be
up to some small value $\epsilon $ ($>0$). A short discussion about the precise value of
$\epsilon$ and whether it must be known a priori will follow. Note
also that the last requirement $x\geq 0$ is valid because the optical
information is generated by illuminating the sample with a plane wave,
that is, a plane of equal amplitude and phase. Hence, the
phase is the same across the whole image. Therefore, without loss of
generality, we may assume that the phase is zero everywhere, since the
absolute phase is unimportant. We do not
assume that all circles have the same magnitude, however---they can
have any value.

To solve \eqref{eq:sparse-2} we developed an iterative method whose basic
iteration contains the following two steps:
\begin{description}
\item[Step 1:] Solve the minimization problem:
  \begin{equation}
    \label{eq:sparse-3}
    \begin{split}
      \min &\quad \||LFCx| - r\|_{2}^{2}\\
      \mathrm{subject\ to} &\quad x \geq 0 \ .
    \end{split}
  \end{equation}
  (in practice, we use an unconstrained formulation that is solved by
  the L-BFGS method \shortcite{liu89limited}).
\item[Step 2:]
  After a solution $x$ to Step 1 is found, set to zero the entry of
  $x$ with minimal value. Once set to zero the entry remains so
  forever.
\end{description}
In theory, the iterations should be repeated so long as the constraint
$\||LFCx| - r\|^{2}\leq\epsilon$ is satisfied. It is often argued that
the value of $\epsilon$ is known a priori or can be estimated from
physical constraints (as a matter of fact, in the case of the
``random'' image, the difference between the measured Fourier
magnitude $r$ and its ideal variant $r^{*}$ is $\|r - r^{*}\|^{2} =
1.7434$, which corresponds to a signal-to-noise ratio of $\|r^{*}\|/\|r
- r^{*}\| = 1/0.041$). However, it is an important question whether the
best value of $\|x\|_{0}$ (the true number of circles in the image)
can be determined \emph{automatically}. Consider the different stages
of our method as shown in Figure~\ref{fig:randomimg-rec-stages}. Is
there any way to recognize that the correct number of circles is 12
without knowing $\epsilon$?  It turns out that the answer to the above
question is affirmative. As is evident from
Figure~\ref{fig:randomimg-objfunction}, there is a big jump in the
objective function value when the number of circles dips below the
correct value of 12. Hence, even without knowing the noise bound
$\epsilon$ one can easily identify that the smallest number of circles
that ``explains'' well the measurements is 12 (this is, of
course, correct as long as the circles have large enough amplitude).
The result of our
reconstruction and the true image are shown in
Figure~\ref{fig:randimg-rec-and-true}.  Note that some circles have low
magnitude so they are invisible in the color images. We therefore, place
the '+' sign at the center of all circles in the image ($x$'s entries
that are not zeros).
\begin{figure}[H]
  \centering
  \subfloat[]{
    \includegraphics[width=0.3\textwidth]{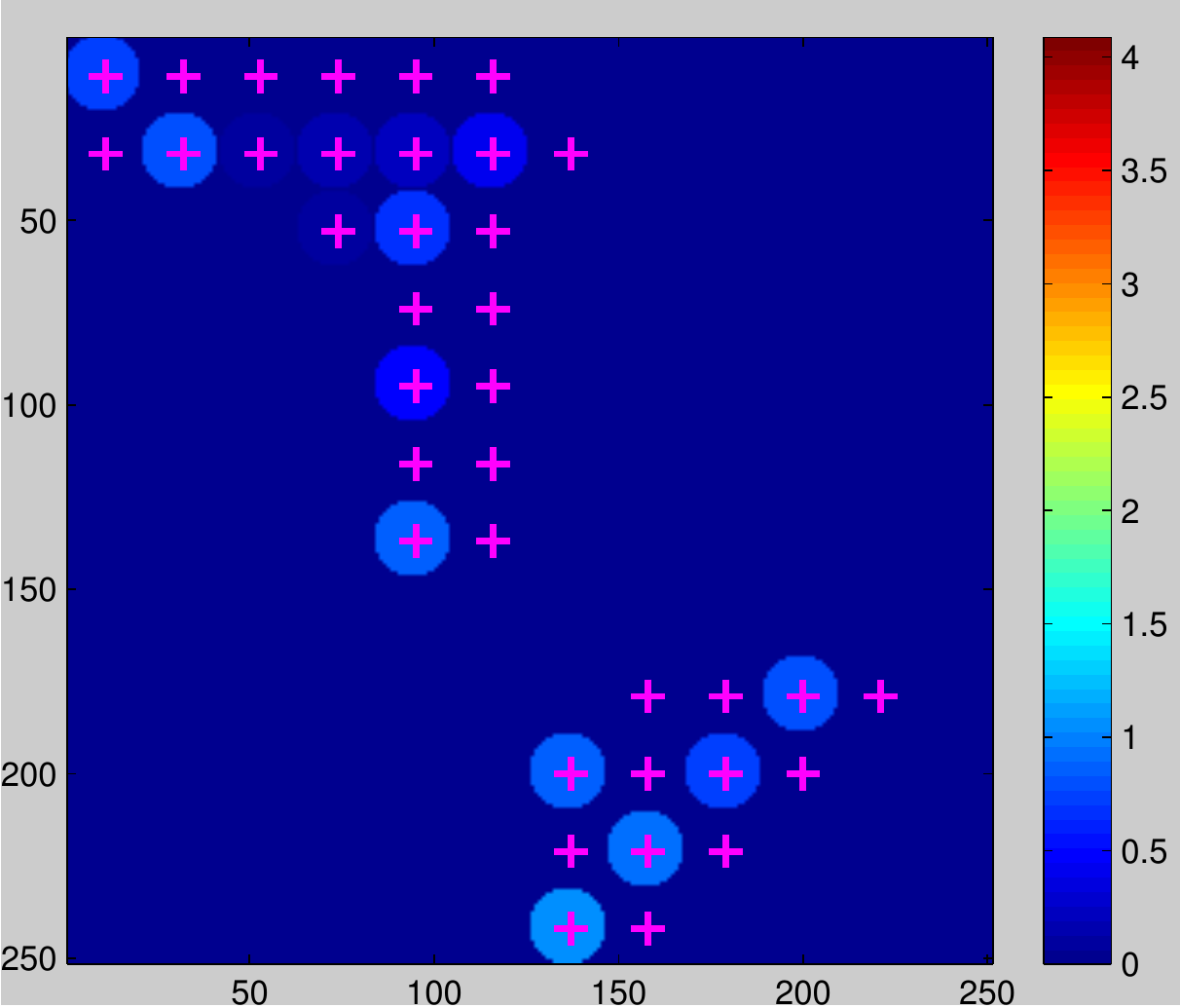}
  }
  \subfloat[]{
    \includegraphics[width=0.3\textwidth]{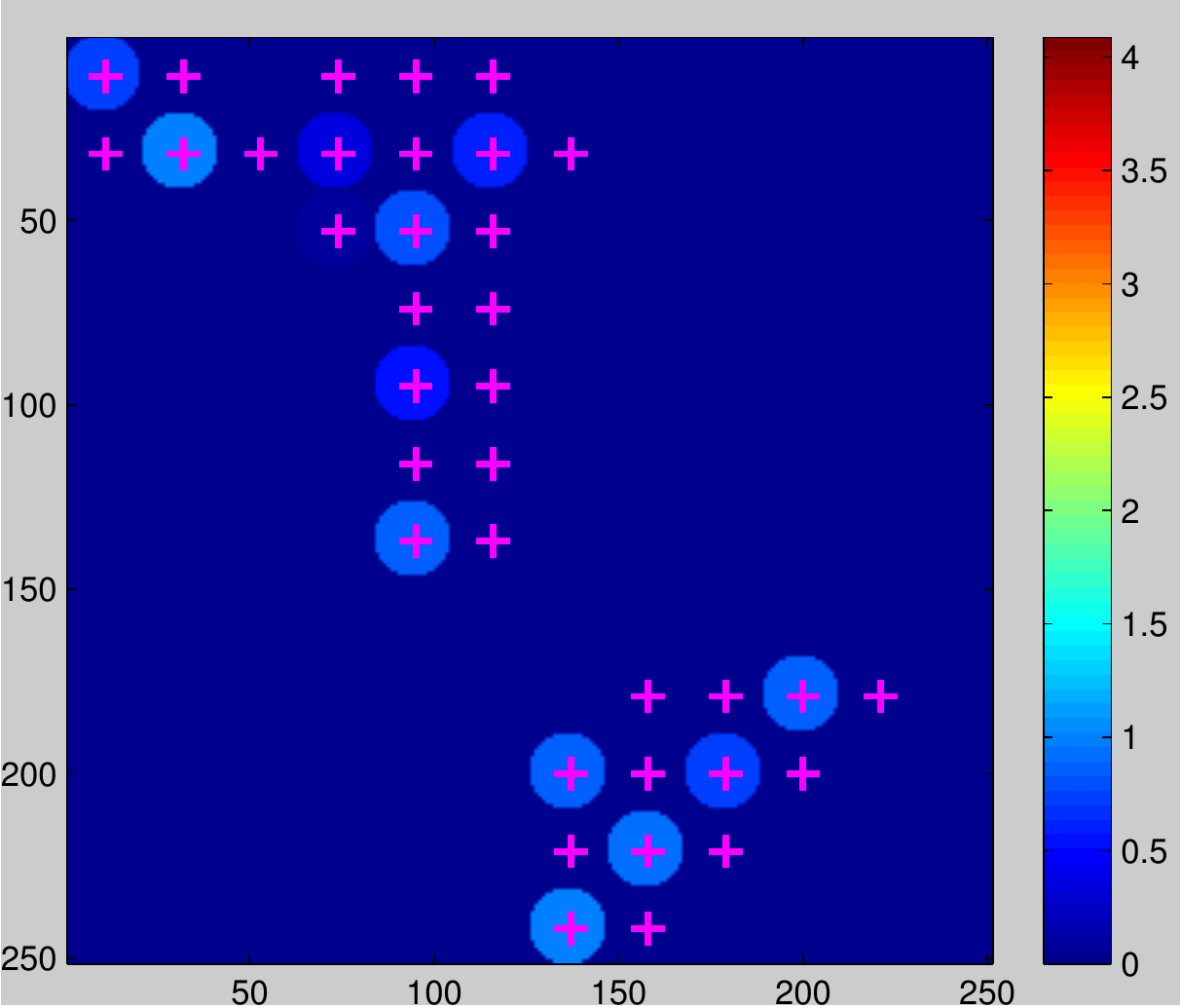}
  }
  \subfloat[]{
    \includegraphics[width=0.3\textwidth]{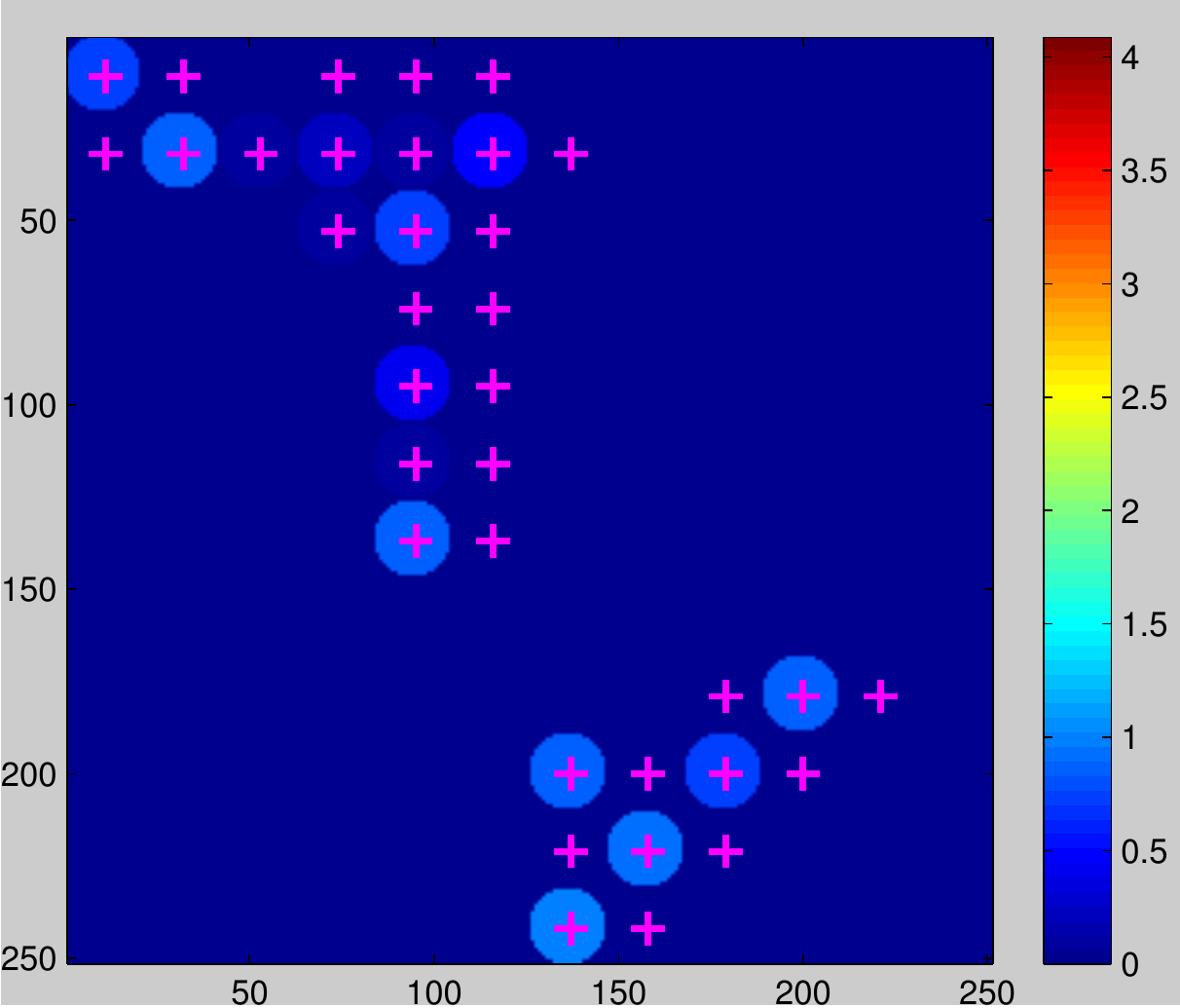}
  }\\
  \subfloat[]{
  \includegraphics[width=0.3\textwidth]{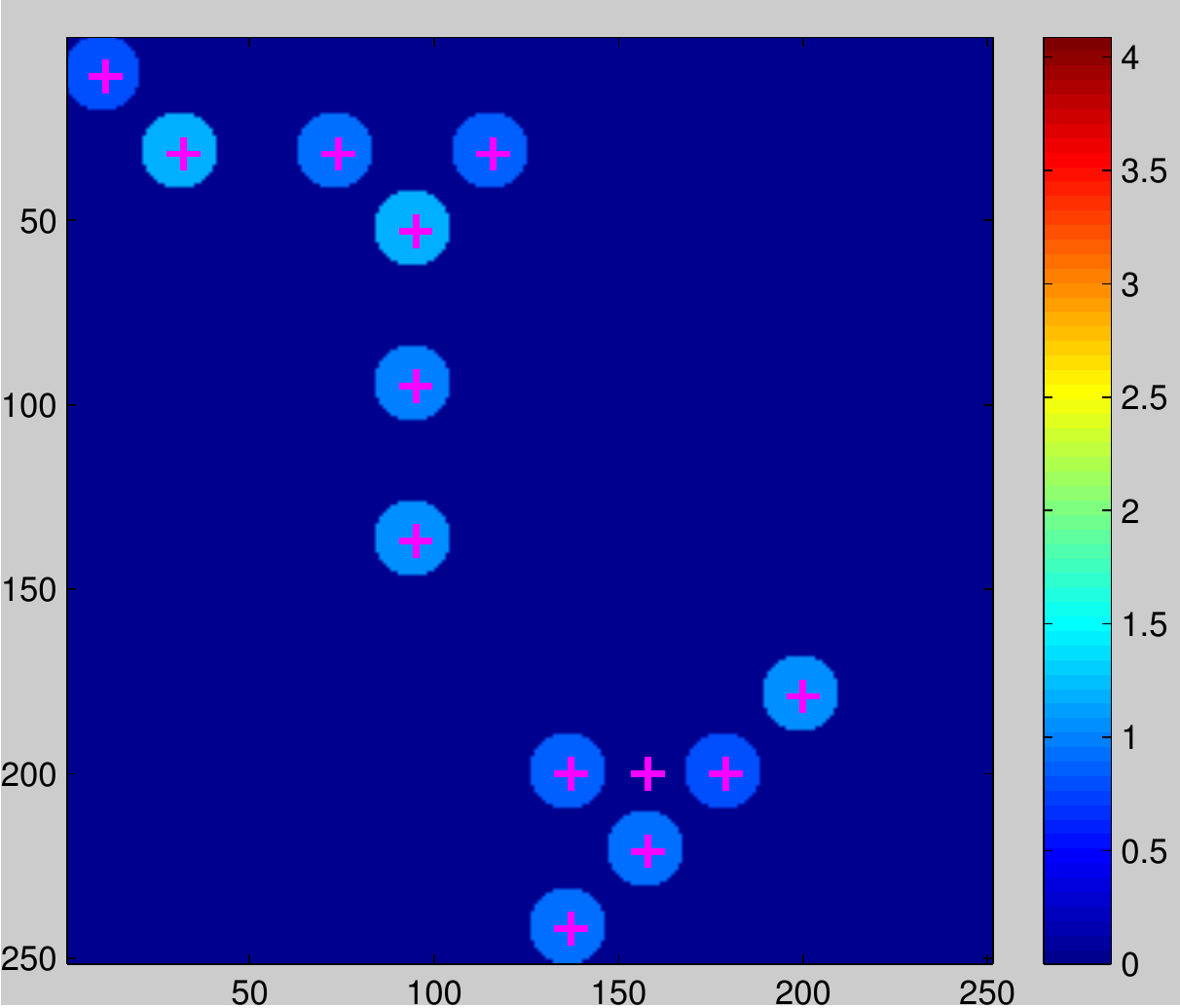}
  }
  \subfloat[]{
    \includegraphics[width=0.3\textwidth]{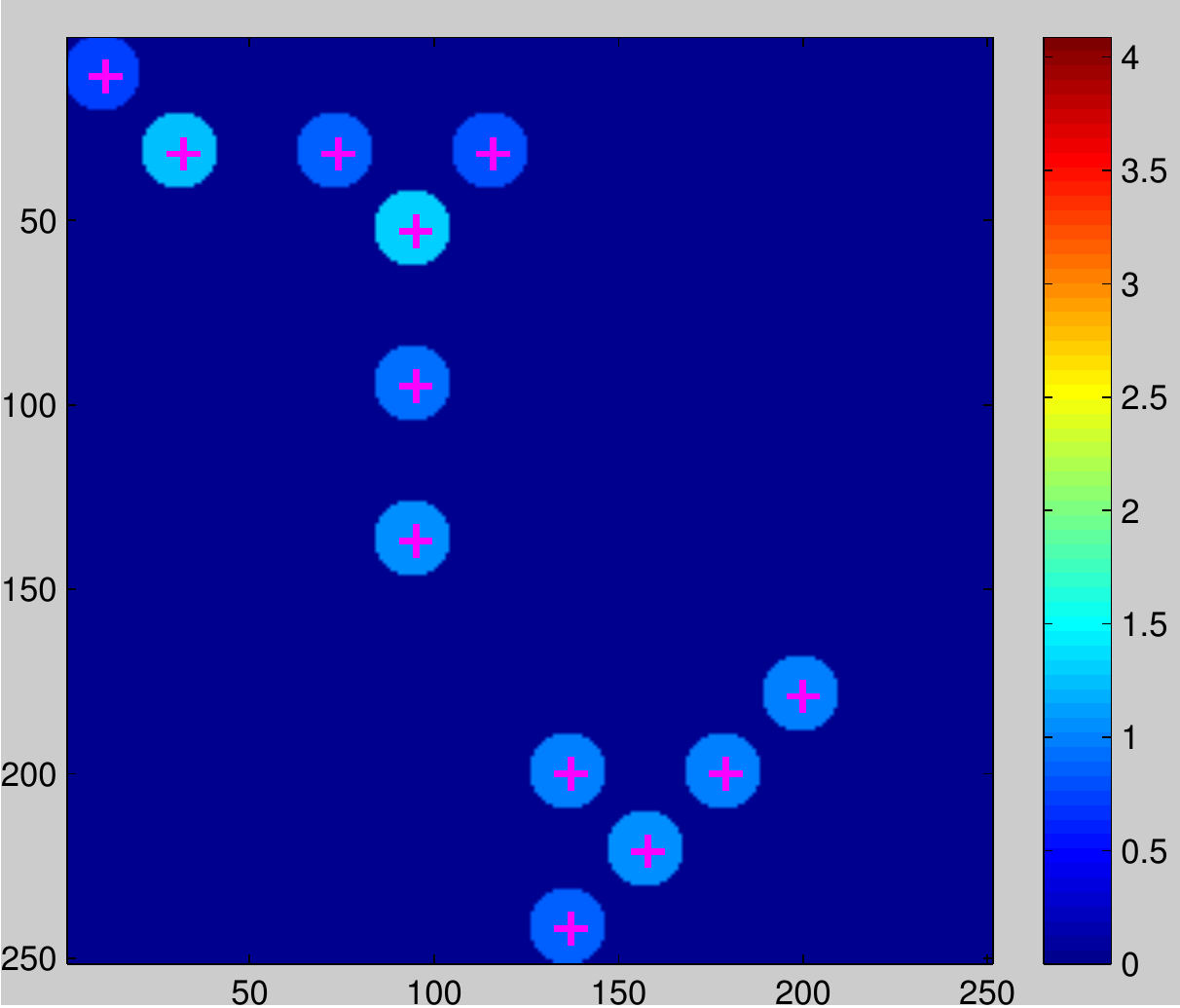}
  }
  \subfloat[]{
    \includegraphics[width=0.3\textwidth]{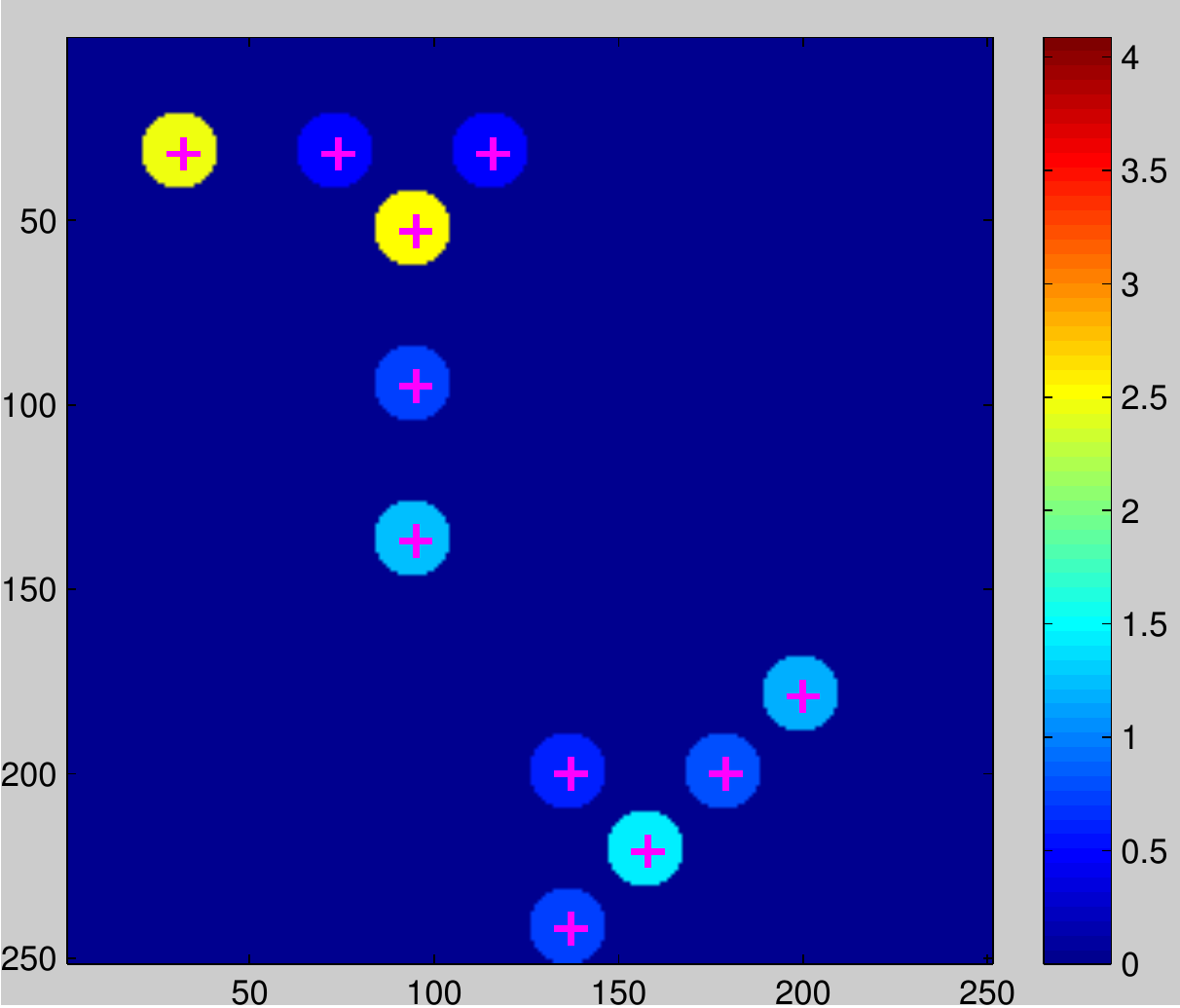}
  }\\
    \subfloat[]{
    \includegraphics[width=0.3\textwidth]{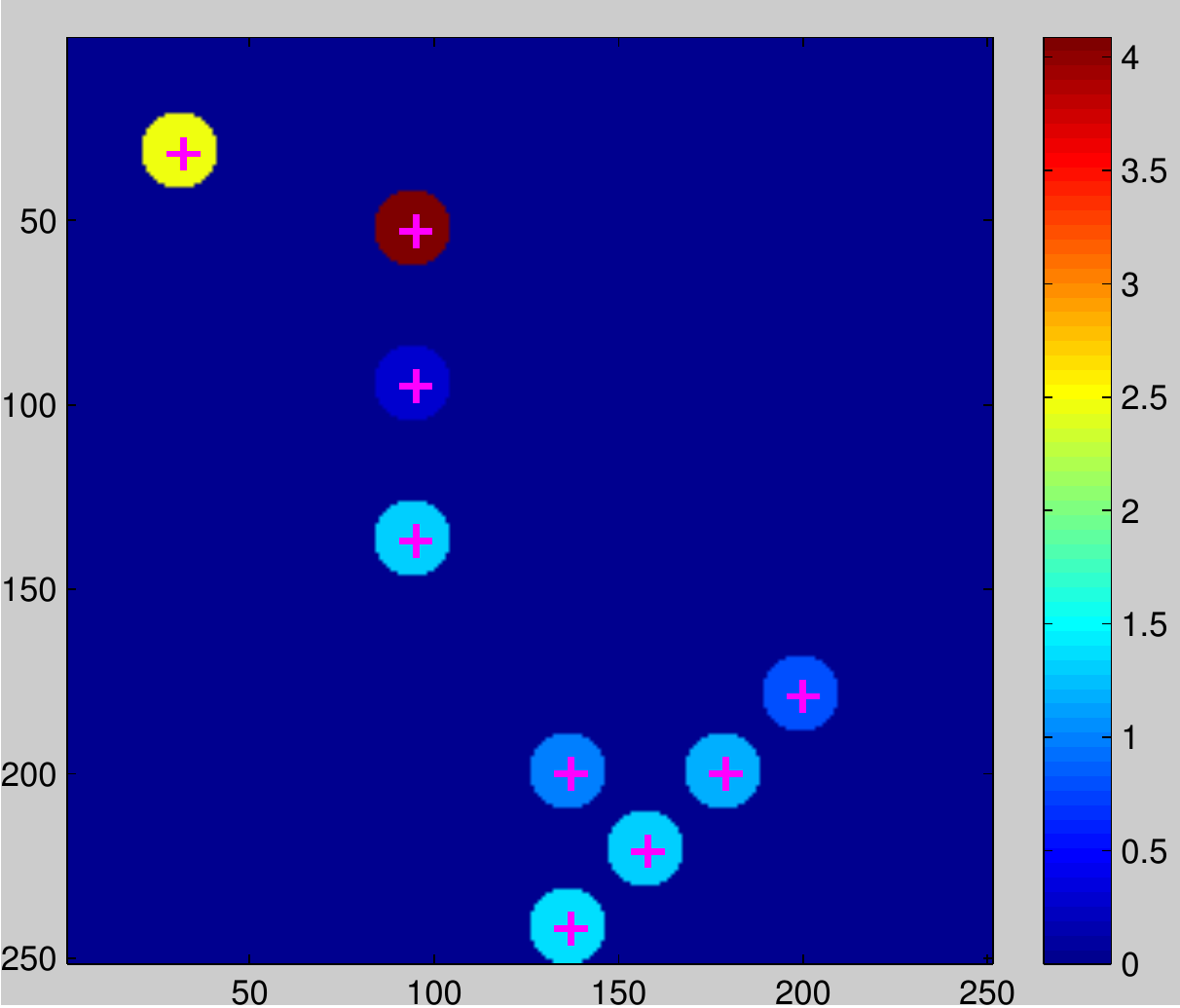}
  }
  \subfloat[]{
    \includegraphics[width=0.3\textwidth]{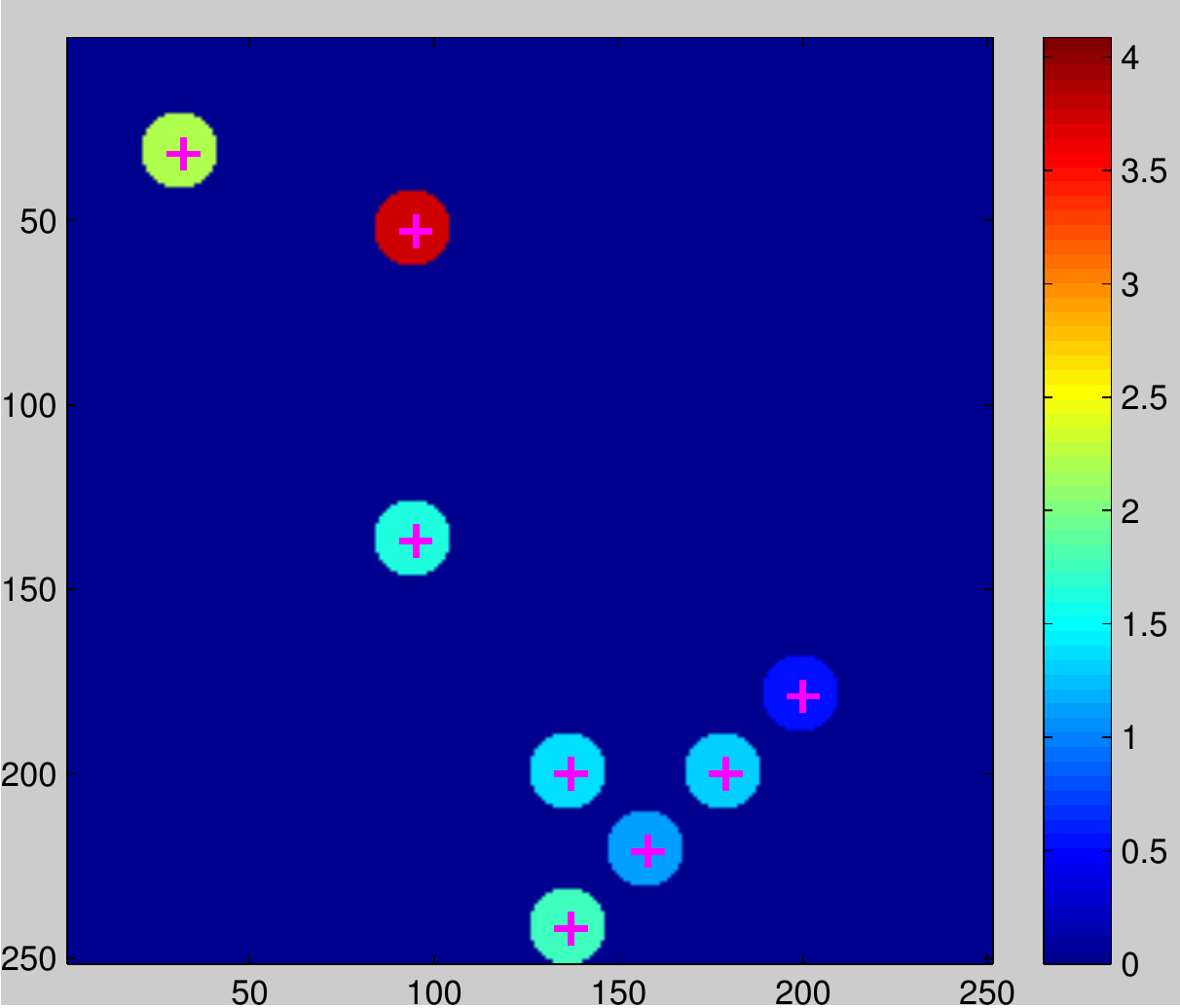}
  }
  \subfloat[]{
    \includegraphics[width=0.3\textwidth]{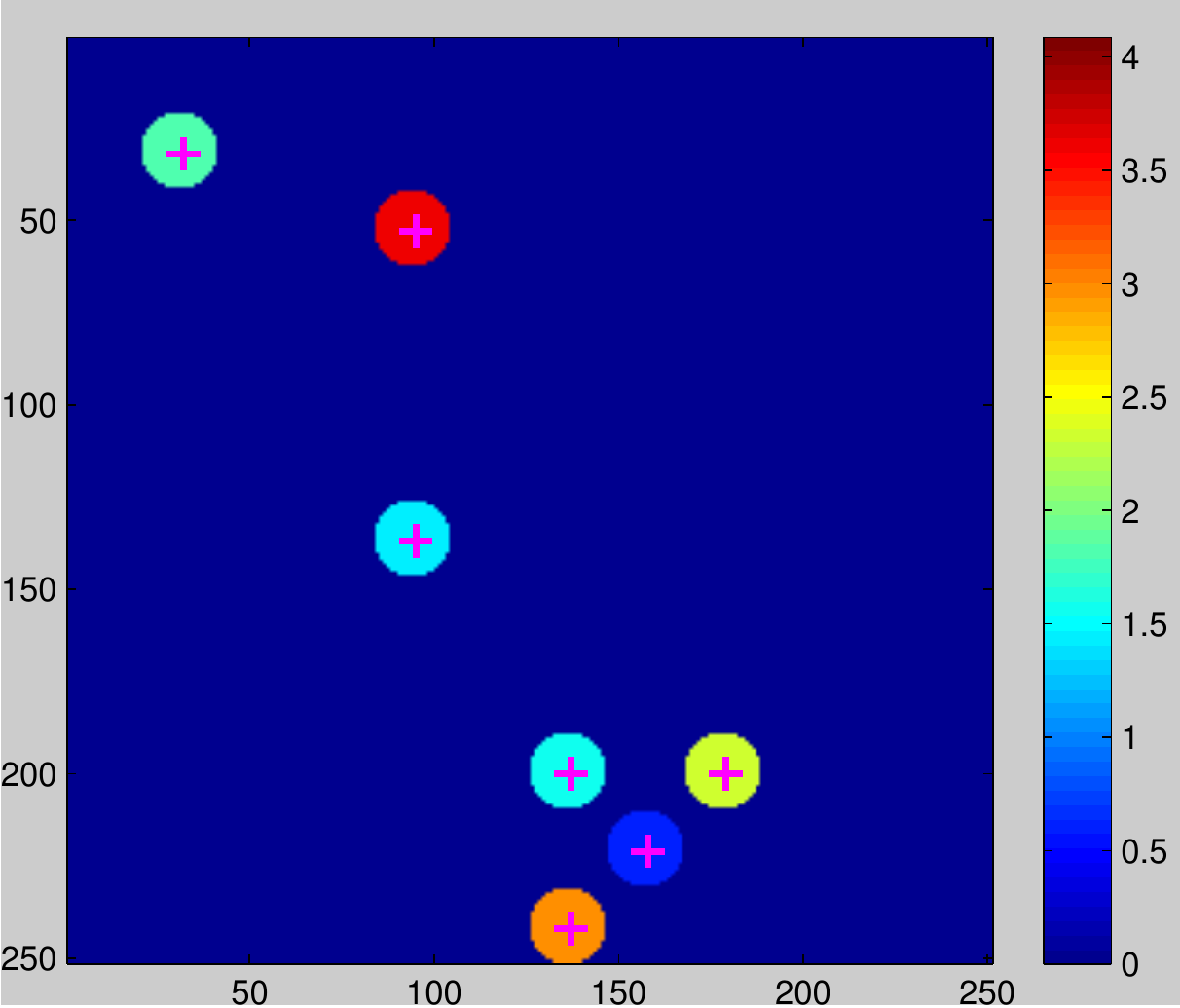}
  }
  \caption[Reconstruction stages for the example of the ``random''
    image]{Reconstruction stages for the example of the ``random''
    image. Each stage (iteration) corresponds to a certain number of
    circles (non-zero entries in $x$): (a) 37 circles, (b) 36, (c) 35,
    (d) 13, (e) 12, (f) 11, (g) 9, (h) 8, (i) 7.}
  \label{fig:randomimg-rec-stages}
\end{figure}
%\clearpage{}
\begin{figure}[H]
  \centering
  \includegraphics[width=0.9\textwidth]{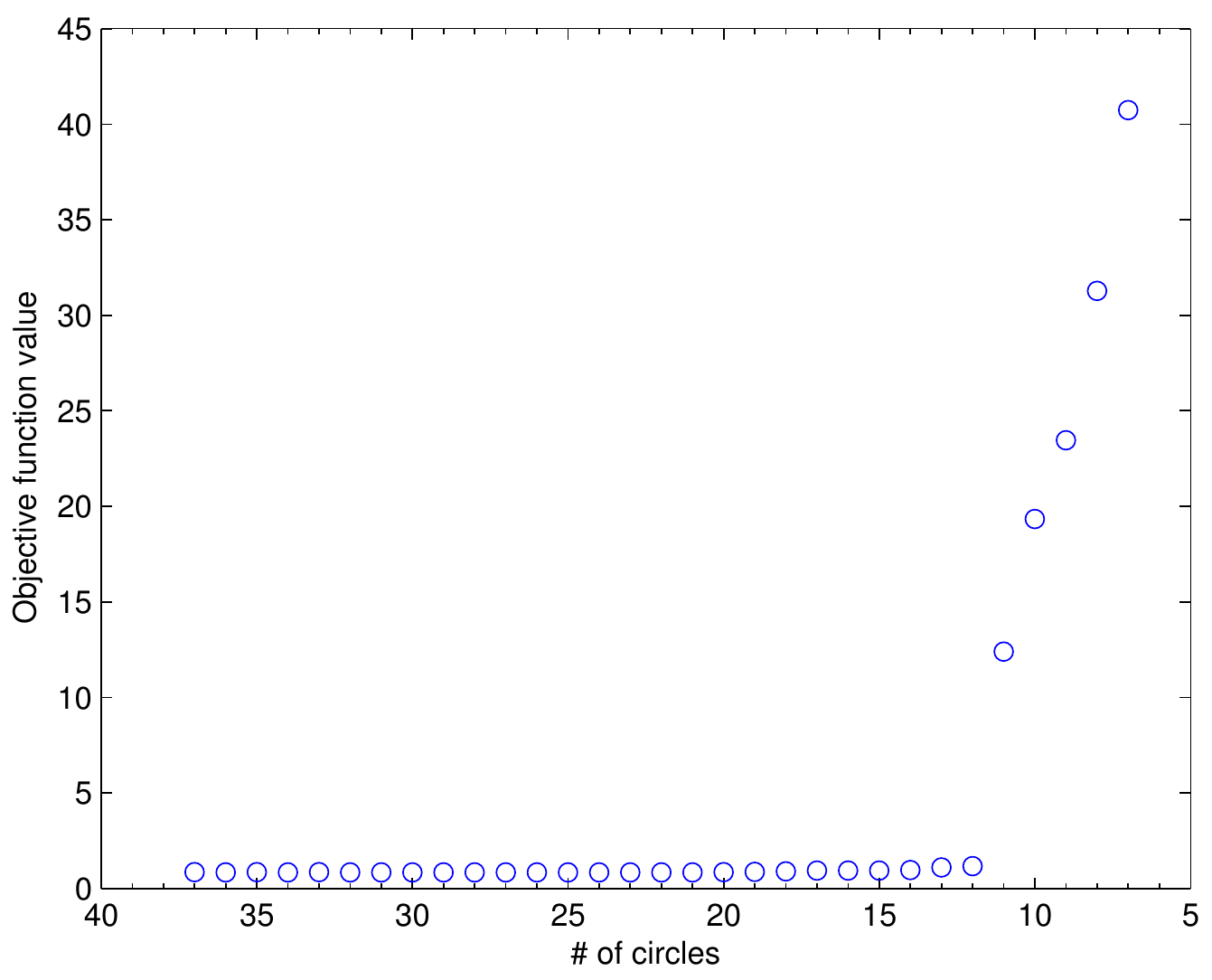}
  \caption[``Random'' image: objective function value]{``Random'' image: objective function value (Fourier domain
    discrepancy) versus the number of circles in the solution.}
  \label{fig:randomimg-objfunction}
\end{figure}

\begin{figure}[H]
  \centering
  \subfloat[]{
    \includegraphics[width=0.4\textwidth]{sparse/randImg_2011-08-25-16-42-47-iter26-circles12}
  }
  \qquad{}
  \subfloat[]{
    \includegraphics[width=0.4\textwidth]{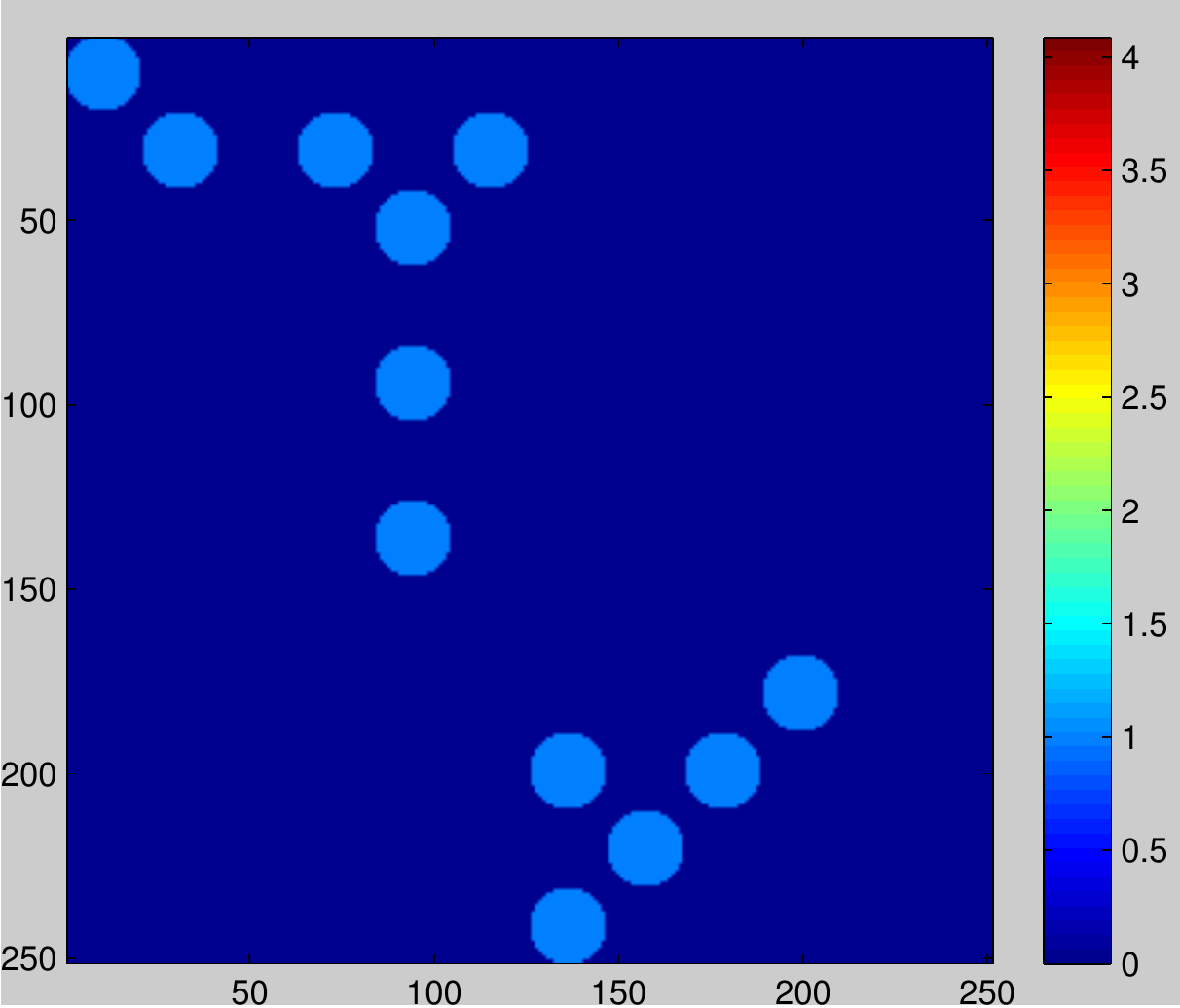}
  }
  \caption[Reconstruction result for the ``random'' image]{Reconstruction result for the ``random'' image (a), and the
    true (original) image (b).}
  \label{fig:randimg-rec-and-true}
\end{figure}
Very similar behavior is
observed for the second image (SOD) whose results are shown in
Figures~\ref{fig:sodimg-objfunction}
and~\ref{fig:sodimg-rec-and-true}.
\begin{figure}[H]
  \centering
  \includegraphics[width=0.9\textwidth]{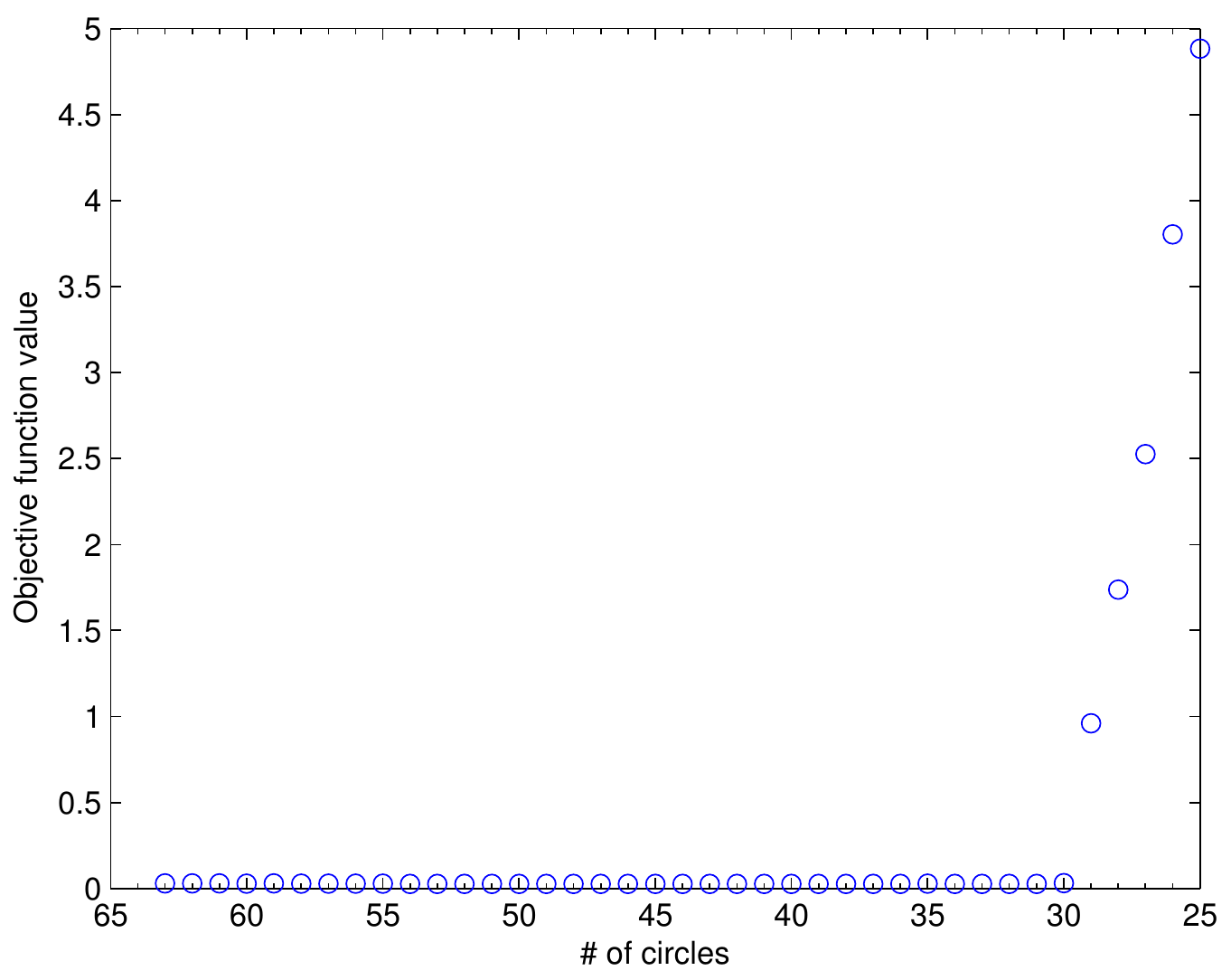}
  \caption[SOD image: objective function value]{SOD image: objective function value (Fourier domain
    discrepancy) versus the number of circles in the solution.}
  \label{fig:sodimg-objfunction}
\end{figure}

\begin{figure}[H]
  \centering
  \subfloat[]{
    \includegraphics[width=0.4\textwidth]{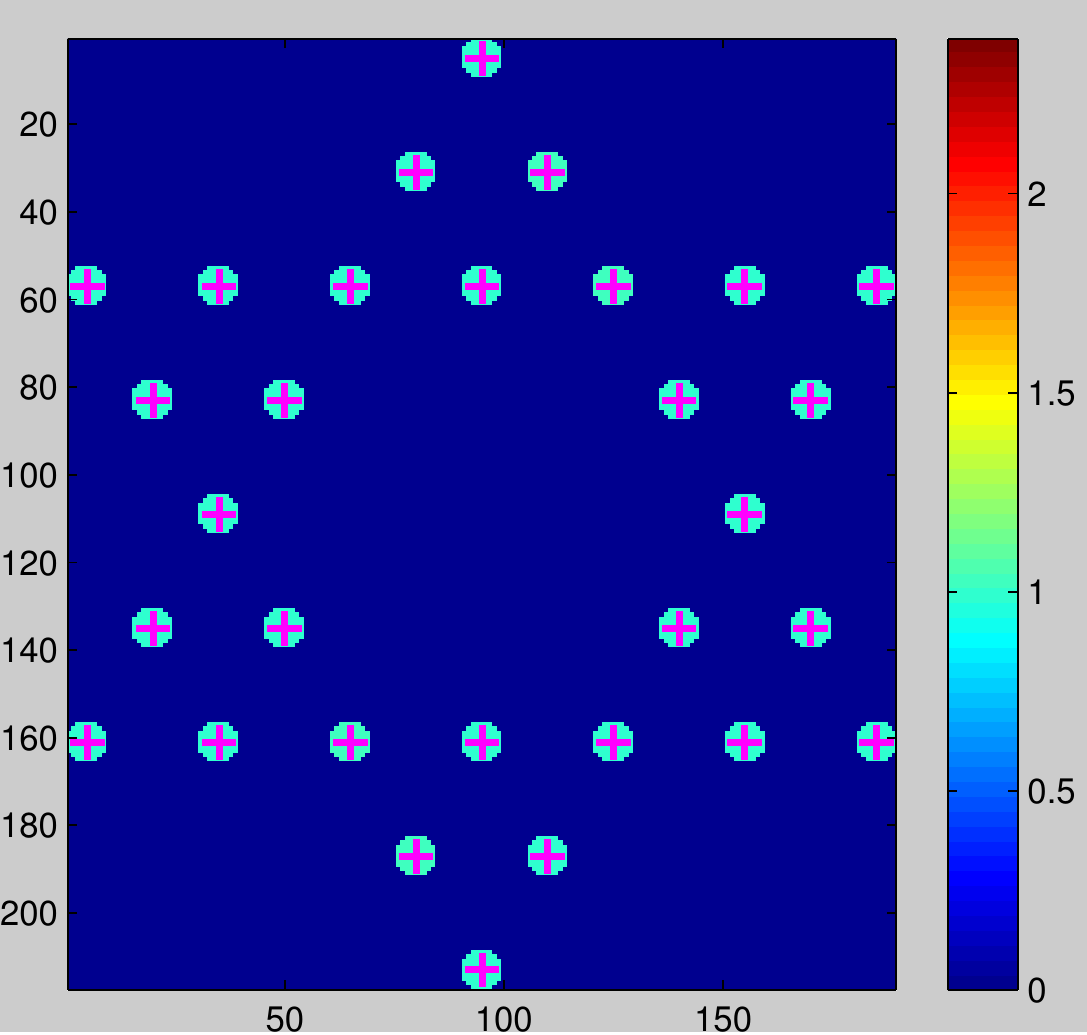}
  }
  \qquad{}
  \subfloat[]{
    \includegraphics[width=0.4\textwidth]{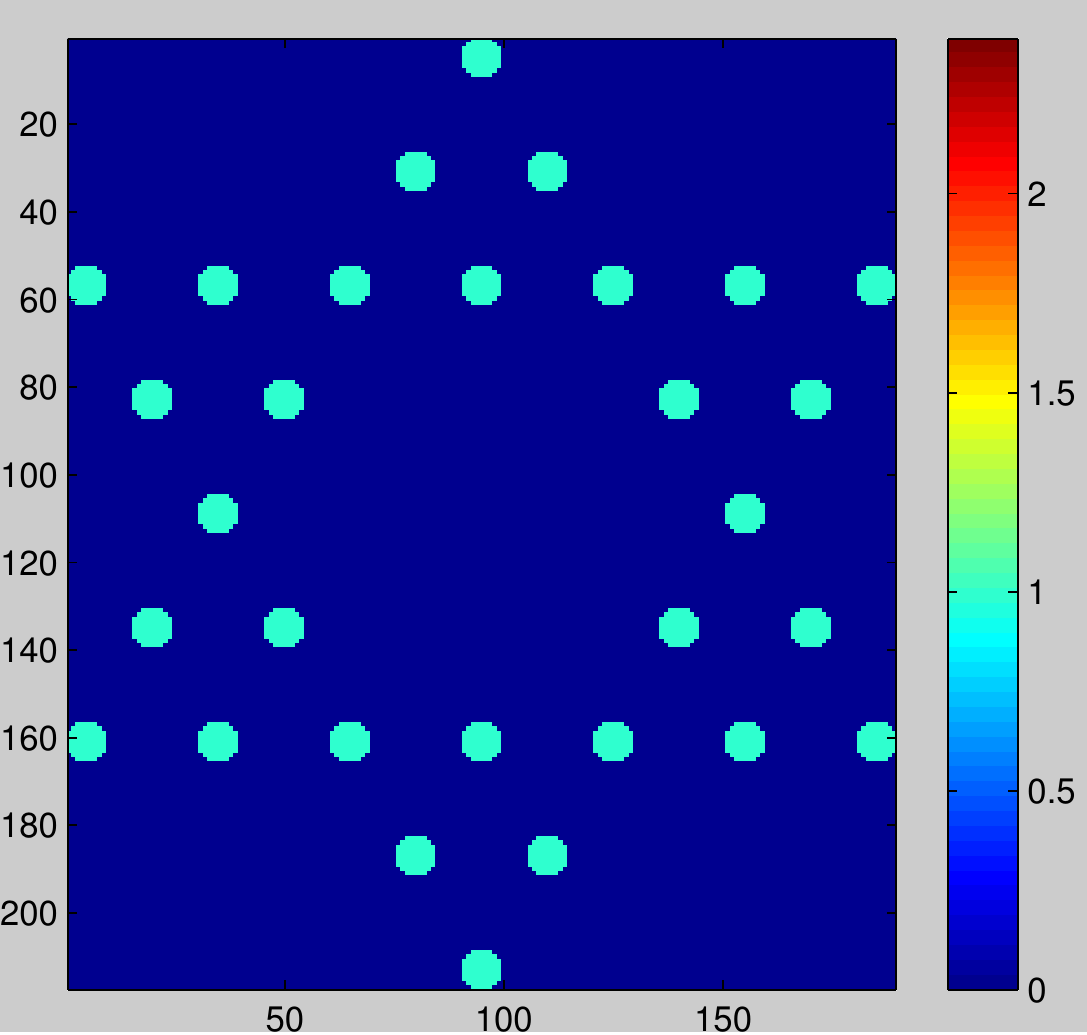}
  }
  \caption[Reconstruction result for the SOD image]{Reconstruction result for the SOD image (a), and the true
    (original) image (b).}
  \label{fig:sodimg-rec-and-true}
\end{figure}
In
Section~\ref{sec:choosing-grid-basis} we demonstrate that choosing an
``incorrect'' basis function, even one whose shape does not allow 
perfect representation of 
the sought signal, results, nevertheless, in a reasonable
reconstruction. Furthermore, we also demonstrate  that the  grid's cell
size can be determined automatically.

\section{Comparison with other methods}
\label{sec:comp-with-other}
We would like to stress again that our method is successful because we
exploit the \emph{sparsity} of the sought signal. To demonstrate this,
we present a comparison with some classical reconstruction methods, and
discuss the relation between our setup and  classical compressed
sensing.

\subsection{Without a regularization}
\label{sec:with-regul}
Our sparsity-based technique  minimizes the $l_{0}$ norm
subject to additional constraints. This formulation resembles closely
a regularization imposed on $x$. Hence, the most naive approach would
be to abandon the regularization altogether and to try to find $x$ that
minimizes the discrepancy in the measurements. That is, we might
solve the following problem
\begin{equation}
  \label{eq:sparse-4}
    \begin{split}
      \min &\quad \||LFCx| - r\|^{2} \\
      \mathrm{subject\ to} &\quad x \geq 0 \ .
    \end{split}
\end{equation}
Note that this is exactly the problem we solve in the first iteration
of our method. However, using this approach as the full reconstruction
process has a number of drawbacks. First,
the problem of image reconstruction from the magnitude of its Fourier
transform (also called phase retrieval) is known to be particularly
tough for continuous optimization techniques (for explanation and
further details see~\shortcite{osherovich11approximate}). To the best of our
knowledge, the
most widely used method for phase retrieval without additional information
is the Hybrid Input-Output method~\shortcite{fienup82phase}.  A more
detailed investigation of this method will follow in
Section~\ref{sec:hybrid-input-output}. Here, we present the results
obtained by our optimization routine. As mentioned earlier, this
formulation is equivalent to performing only one iteration of our
method. Hence, the result is as shown in
Figure~\ref{fig:rec-without-regularization}. Note that the
reconstruction contains many superfluous circles, and even if the correct
number of the circles were known, a simple thresholding would yield an
incorrect reconstruction.
\begin{figure}[H]
  \centering
  \subfloat[]{
    \includegraphics[width=0.4\textwidth]{sparse/randImg_2011-08-25-16-42-47-iter1-circles37}
  }
  \qquad{}
  \subfloat[]{
    \includegraphics[width=0.4\textwidth]{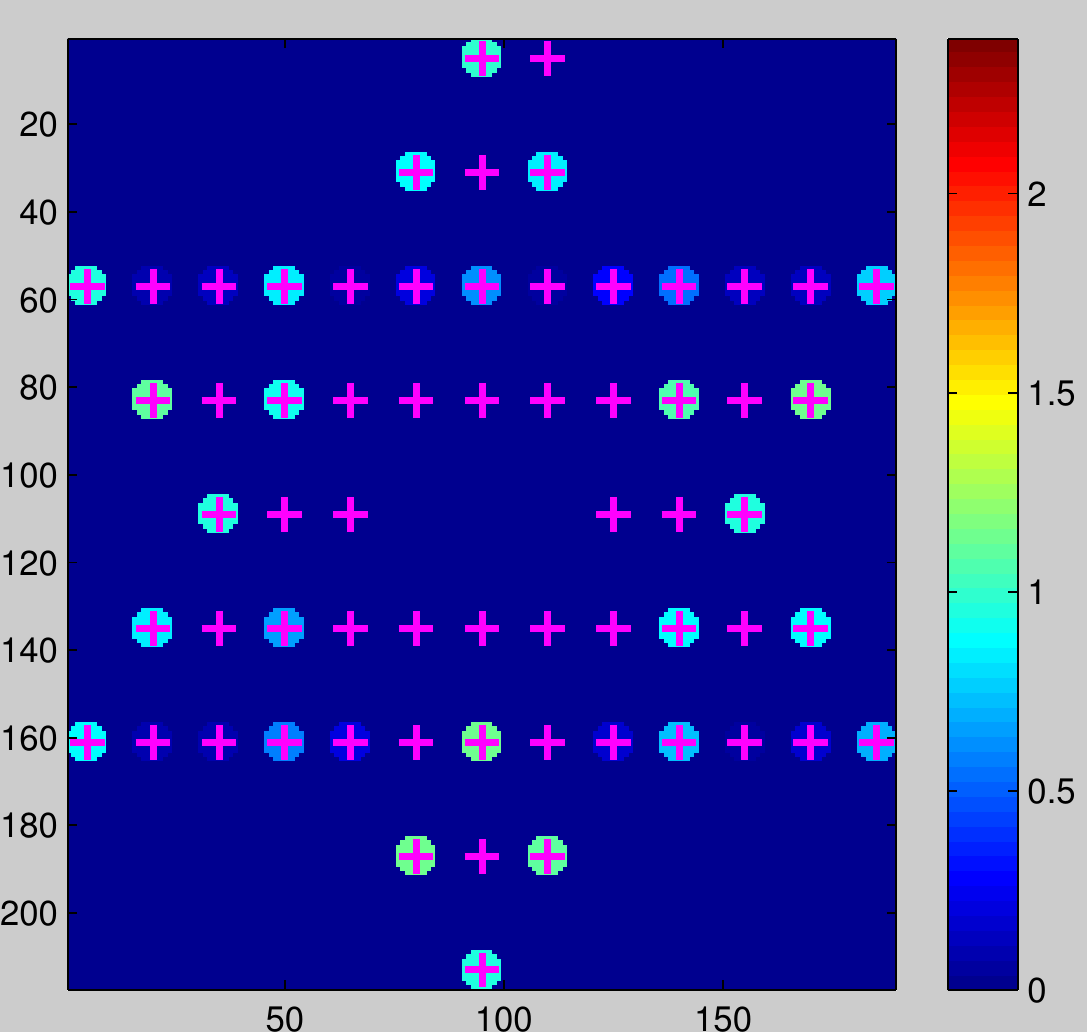}
  }
  \caption[Reconstruction without a regularization on $x$]{Reconstruction without a regularization on $x$: (a)
    ``random'' image, (b) SOD image.}
  \label{fig:rec-without-regularization}
\end{figure}

\subsection{Replacing $l_0$ with another norm}
\label{sec:replacing-l_0-with}
Using $l_{2}$ regularization has long been a favorite among engineers
due to its simplicity and the ability to obtain closed-form solutions
in linear cases. In the non-linear case, these benefits are lost, of
course. However, for us it is more important that the $l_{2}$ norm
does not promote sparsity (actually, some papers claim that it
usually results in the most dense solution
possible~\shortcite{chen99atomic}). To demonstrate that this regularization
is not suitable for bandwidth extrapolation of sparse signals, we
solved the following problem
\begin{equation}
  \label{eq:sparse-5}
   \begin{split}
    \min &\quad\|x\| \\
    \mathrm{subject\ to} &\quad \||LFCx| - r\|^{2}\leq\epsilon\ ,\\
    &\quad x \geq 0\ .
  \end{split}
\end{equation}
The problem was solved by  transforming it into an
unconstrained optimization problem and choosing the weights of the
penalty function terms so as to get the discrepancy in the
measurements close to the true values. That is, 
assuming that the true $\epsilon$ is known ($\epsilon=1.74$ in the
case of ``random'' image, and $\epsilon=0.0329$ in the case of SOD
image). For the solution we used exactly the same routine (L-BFGS) as in our
main algorithm. The results are shown in Figure~\ref{fig:l2-rec}.
\begin{figure}[H]
  \centering
  \subfloat{
    \includegraphics[width=0.4\textwidth]{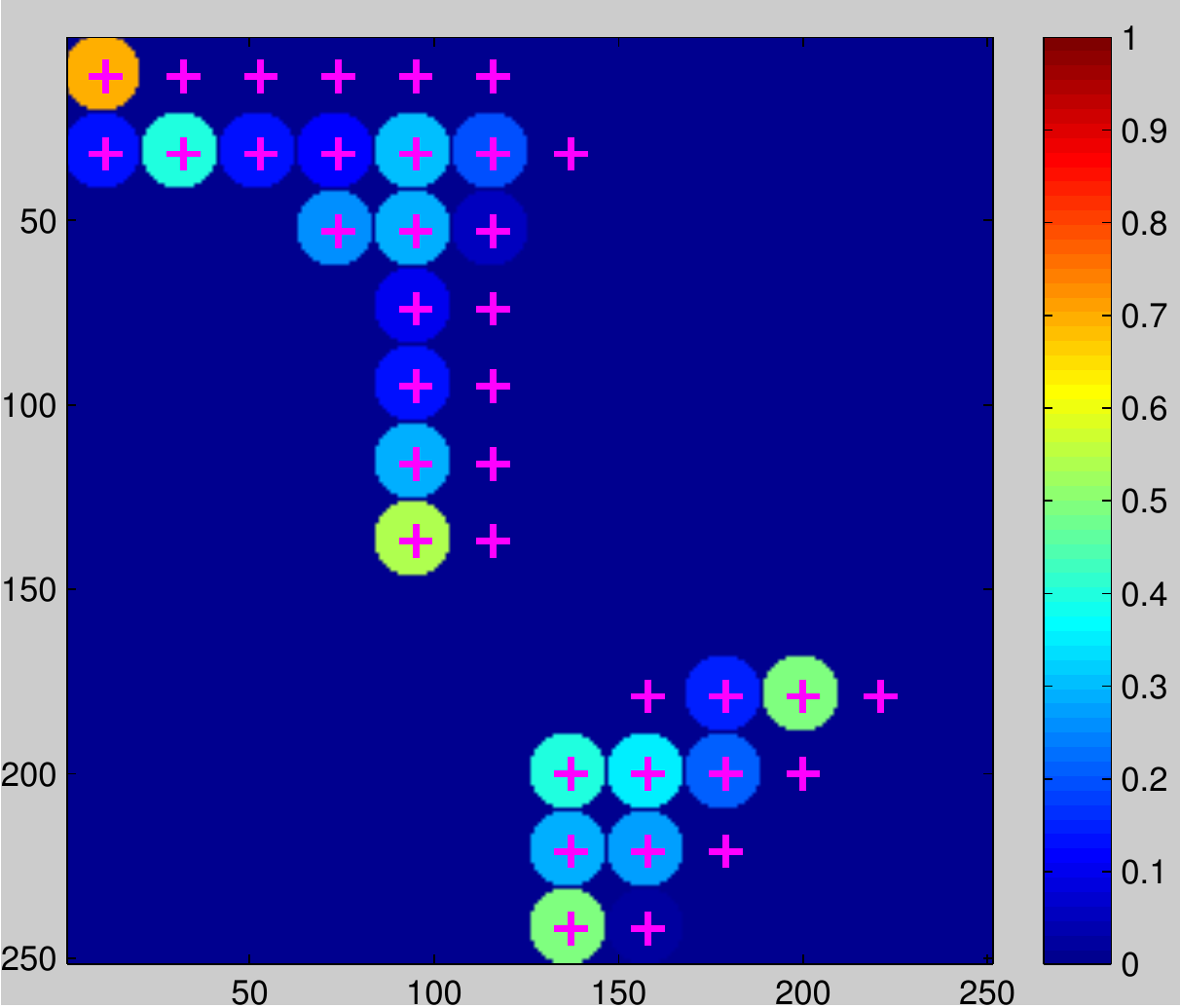}
  }
  \qquad{}
  \subfloat{
    \includegraphics[width=0.4\textwidth]{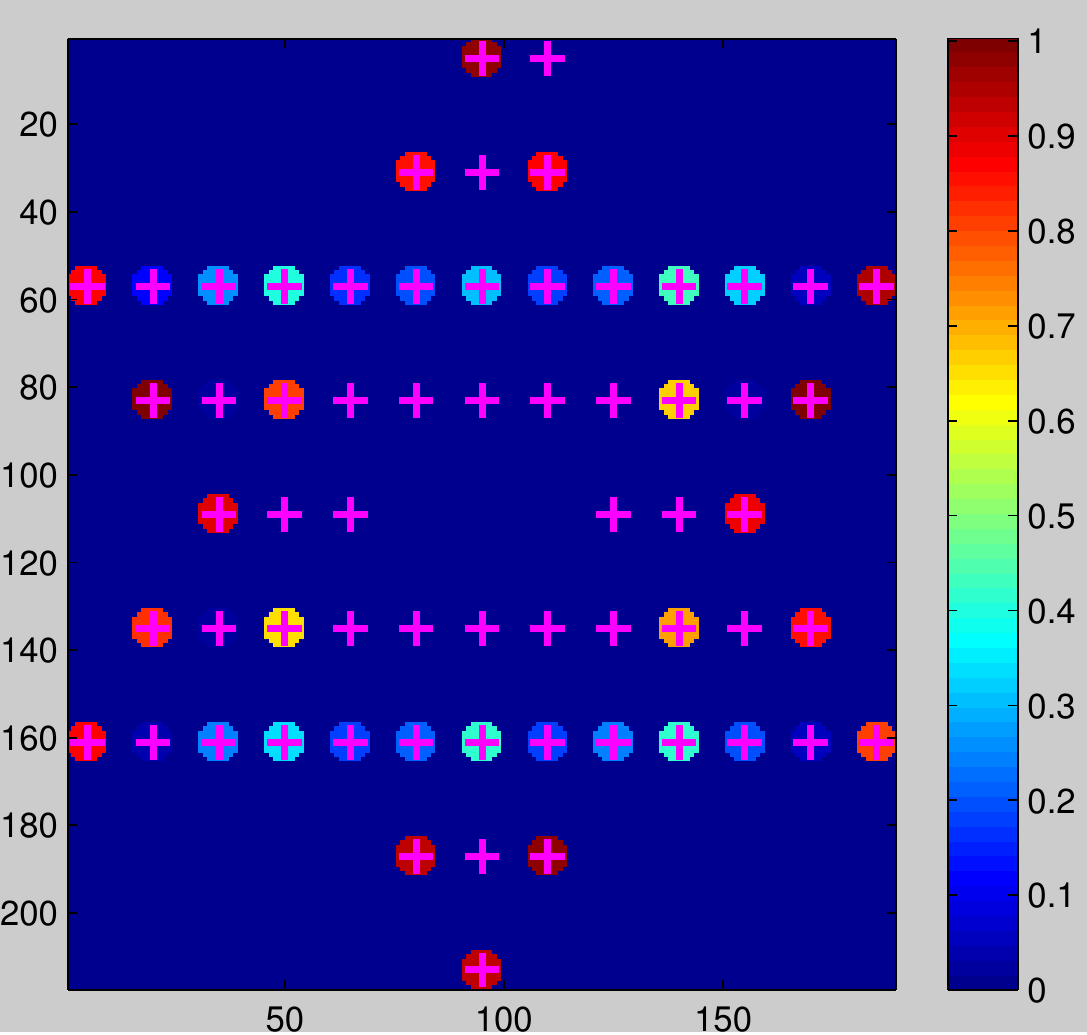}
  }
  \caption[Reconstruction using $l_2$ regularization]{Reconstruction using $l_2$ regularization: (a) ``random''
    image, and (b) SOD image.}
  \label{fig:l2-rec}
\end{figure}
It is obvious that the reconstruction result is incorrect. Moreover, even if the
correct number of circles were known, a simple thresholding would
still 
produce an incorrect result.

Another viable alternative would be using the $l_{1}$ norm. A
discussion on this norm is postponed to
Section~\ref{sec:relat-compr-sens}.

\subsection{Methods based on alternating projections}
\label{sec:meth-base-altern}
In Chapter~\ref{cha:curr-reconstr-meth} we have already seen some methods for phase
recovery~\shortcite{gerchberg72practical,fienup82phase} that are based
on a simple and elegant idea of alternating projections. Similar ideas
was applied in the field of bandwidth
extrapolation~\shortcite{gerchberg74super-resolution,papoulis75new}.
In general, the current
signal estimate is transformed back and forth between the object and
the Fourier domains. In each domain, all available information is used
to form the next estimate. Here we consider two major methods of this
type: Gerschberg type method (often referred to as Gerschberg-Saxton
or Gerschberg-Papoulis) and Fienup's Hybrid Input-Output method. As we
already know, the
former is a classical method of alternating projections where all
available information in the current domain is imposed upon the
current estimate. In the latter approach the object domain
information is not directly imposed on the current estimate; instead a
more complex update rule is used as we presented in
Section~\ref{sec:fienup-algor-phase}. 

\subsubsection{Gerschberg type methods}
\label{sec:gerschb-type-meth}
As mentioned before, Gerschberg type methods are ``pure'' projection
methods. The idea is to transform back and forth the current signal
estimate between the signal and the Fourier domain 
performing a ``projection'' in each of the domains, that is, replacing
the current estimate $x_{cur}$ with 
the nearest one that satisfies the constraints in the relevant
domain ($x_{new}$).
Hence, in each domain the
following optimization problem is solved
\begin{equation}
  \label{eq:sparse-6}
  \begin{split}
    \min_{x_{new}} & \quad \|x_{cur} - x_{new}\|_{2}^{2}\\
    \mathrm{subject\ to} & \quad x_{new} \in \mathcal{S} \ ,
  \end{split}
\end{equation}
where $\mathcal{S}$ denotes the set of all admissible signals in the
current domain. In our case the current estimate is first Fourier
transformed. Then the current (wrong) magnitude is replaced with the
measured (correct) magnitude in the low-frequency regime. The resulting
signal is back-transformed into the object domain (the result
denoted by $x'$) where it is converted into an image comprised of
circles (denoted by $x_{new}$) in the following manner. Recall that
the image model is of the following form $E = Cx$. Hence to find a
projection we must solve the following problem
\begin{equation}
  \label{eq:sparse-7}
  \begin{split}
    \min_{x_{new}} &\quad \|Cx_{new} - x'\|_{2}^{2}\ ,\\
    \mathrm{subject\ to} &\quad x_{new} \geq 0 \ .
  \end{split}
\end{equation}
The problem is convex and can be solved efficiently, however, we
used a two-steps approximation instead of a full solution.
\begin{description}
\item[Step 1] Solve $\min_{x_{new}}  \|Cx_{new} - x'\|_{2}^{2}$. Note that
  this problem has a closed form solution: $x_{new} = C^{\dagger}x'$,
  where $C^{\dagger}$ denotes the Moore-Penrose pseudo-inverse of $C$.
\item[Step 2] Set all entries of $x_{new}$ that are negative to zero.
\end{description}
In general, this is not a true projection. However, it is a
projection, if the
vector $x_{new}$ obtained after the first step is non-negative. This
is indeed the case we observe in all our experiments. The
results obtained after 5000 iterations of this method are shown in
Figure~\ref{fig:gerschberg-type-rec}. Usually, the correspondence
between these and the true image falls considerably behind our
sparsity-based reconstruction method.
\begin{figure}[H]
  \centering
  \subfloat[]{
    \includegraphics[width=0.4\textwidth]{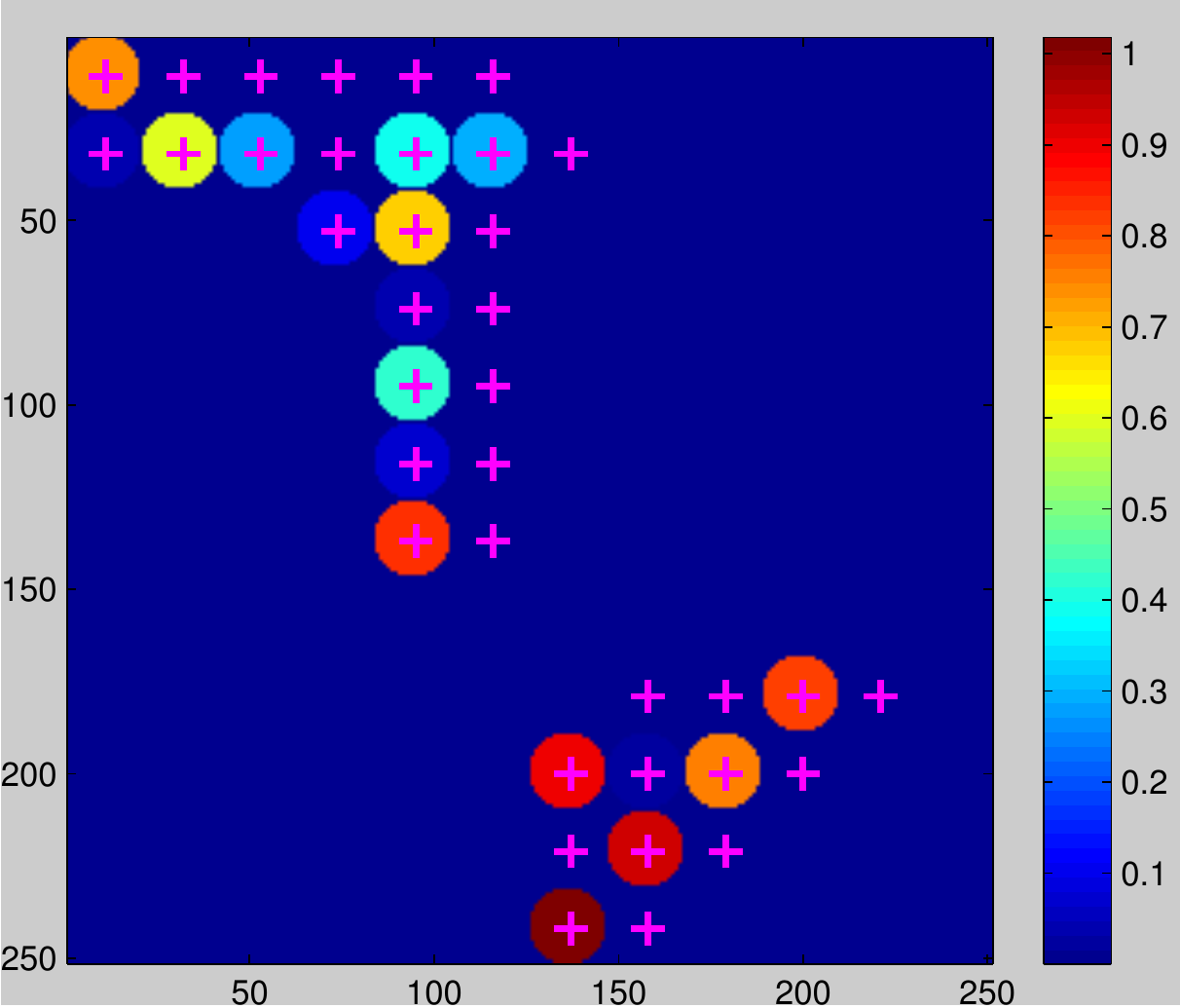}
  }
  \qquad{}
  \subfloat[]{
    \includegraphics[width=0.4\textwidth]{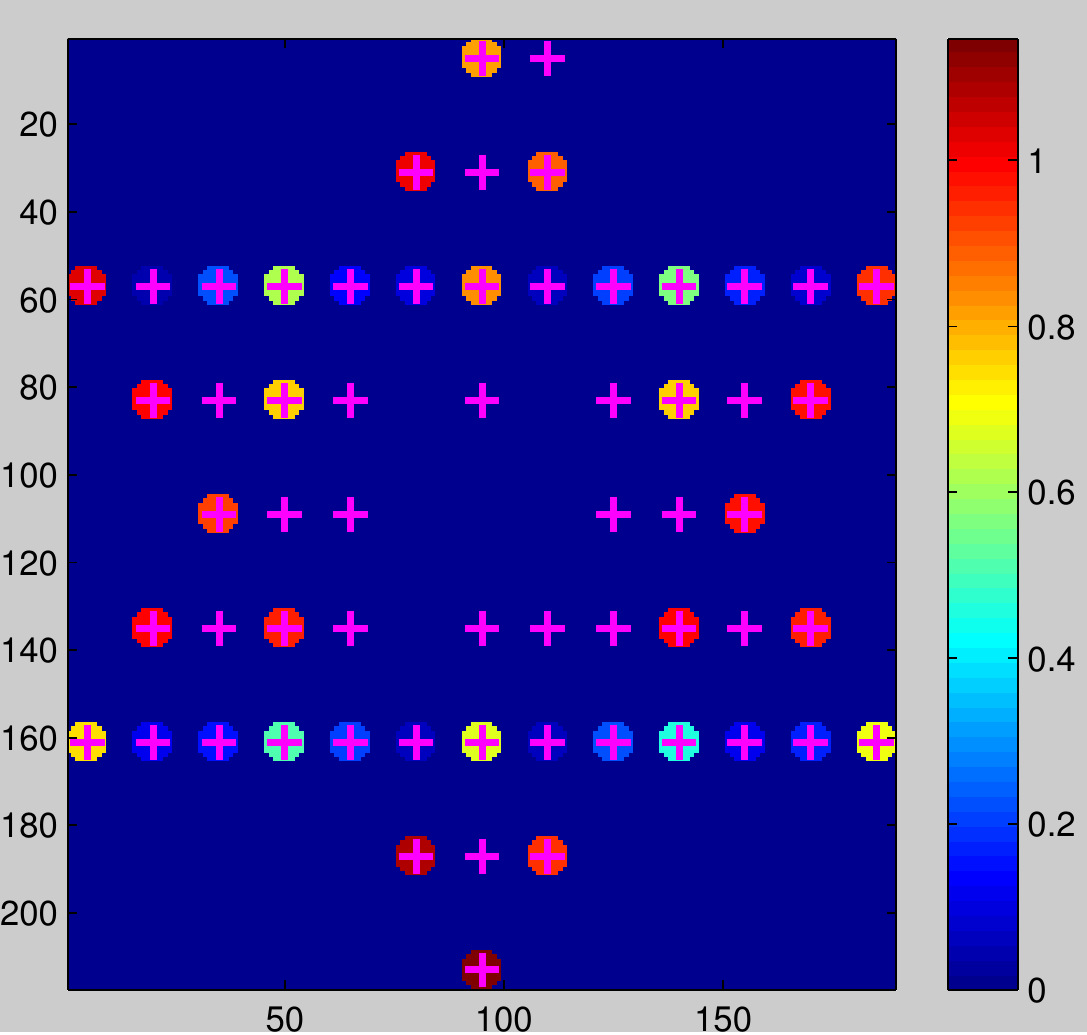}
  }
  \caption[Gerschberg type method: results of reconstruction]{Gerschberg type method: results of reconstruction (a)
    ``random'' image, (b) SOD image.}
  \label{fig:gerschberg-type-rec}
\end{figure}
From the results above, it is evident that the reconstruction is
incorrect and even if the correct number of circles were known a
simple thresholding would still result in incorrect images.

\subsubsection{Fienup's Hybrid Input-Output method}
\label{sec:hybrid-input-output}
We already have met the Hybrid Input-Output method (see
Section~\ref{sec:fienup-algor-phase}) that was developed
by Fienup for the phase retrieval
problem~\shortcite{fienup82phase}. Although based on the method of
alternating projections, HIO does not enforce the object domain
constraints, that is, the image is allowed to be non-zero in the
off-support areas and the values may be negative. To the best of our
knowledge, HIO is the most successful numerical method for signal
reconstruction from the magnitude of its Fourier transform. However,
the method only achieves good results when all or most of the Fourier
spectrum is available. Judging by the result shown below, the method
is not suitable for the situation where the Fourier magnitude is
available only for a small fraction of the frequencies. In our tests
we applied the method in its original form, using only the Fourier
domain magnitude and support information in the object domain (along
with non-negativity). We did not try to enforce a constant value
across every circle or zero values in the off-support areas, as the
original method does not do that. As a post-processing step, the
result returned by HIO was zeroed in the off-support areas (shown in
Figures~\ref{fig:hio-rec-random-orig} and~\ref{fig:hio-rec-sod-orig})
and then the values across each circle were averaged (shown in
Figures~\ref{fig:hio-rec-random-postproc}
and~\ref{fig:hio-rec-sod-postproc}). As is evident from the results,
the method is not capable of correct reconstruction of the
signals. They cannot be recovered even if the correct number of
circles is known: a simple thresholding will result in an incorrect
reconstruction.

\begin{figure}[H]
  \centering
  \subfloat[]{\label{fig:hio-rec-random-orig}
    \includegraphics[width=0.4\textwidth]{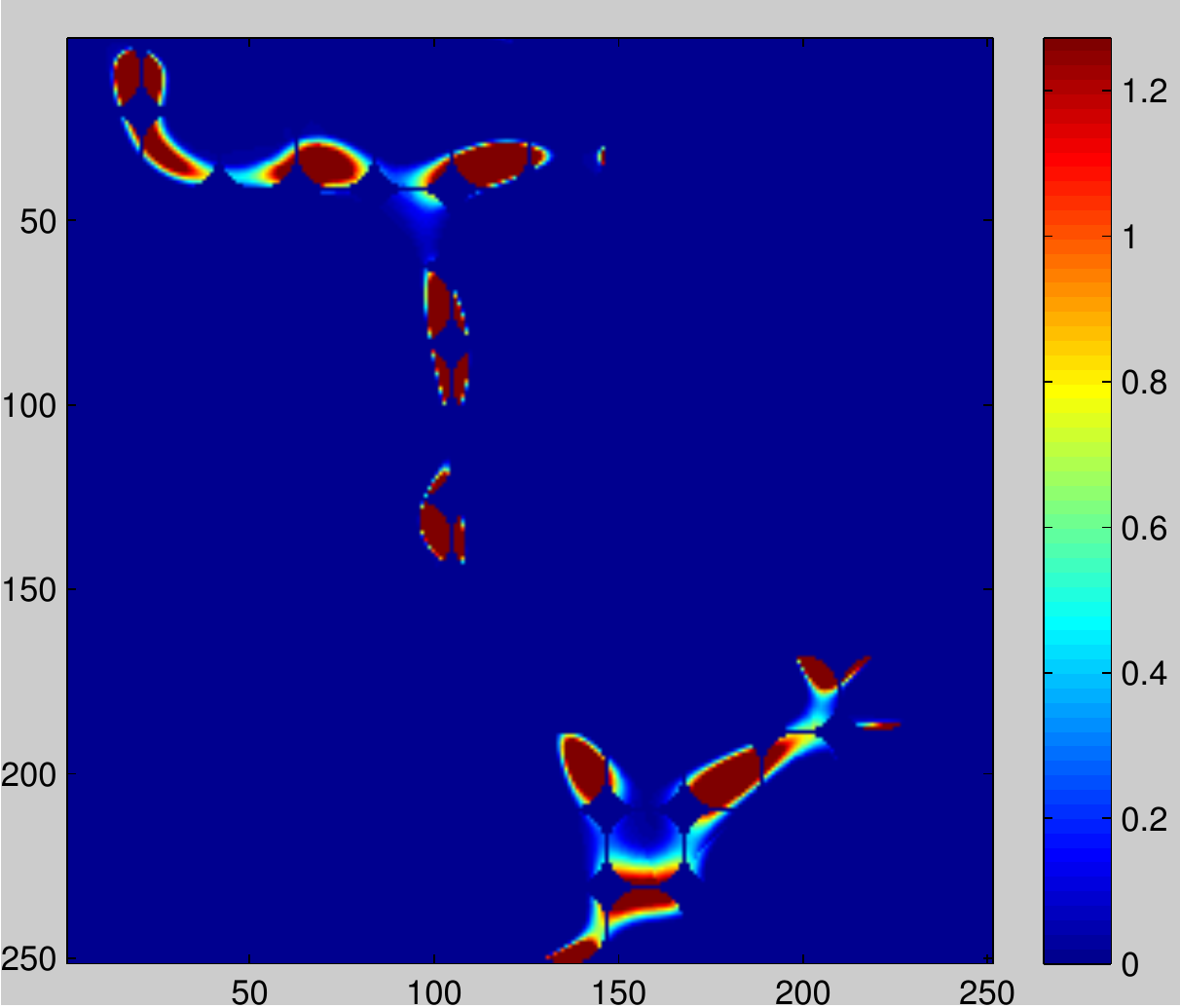}
  }
  \qquad{}
  \subfloat[]{\label{fig:hio-rec-random-postproc}
    \includegraphics[width=0.4\textwidth]{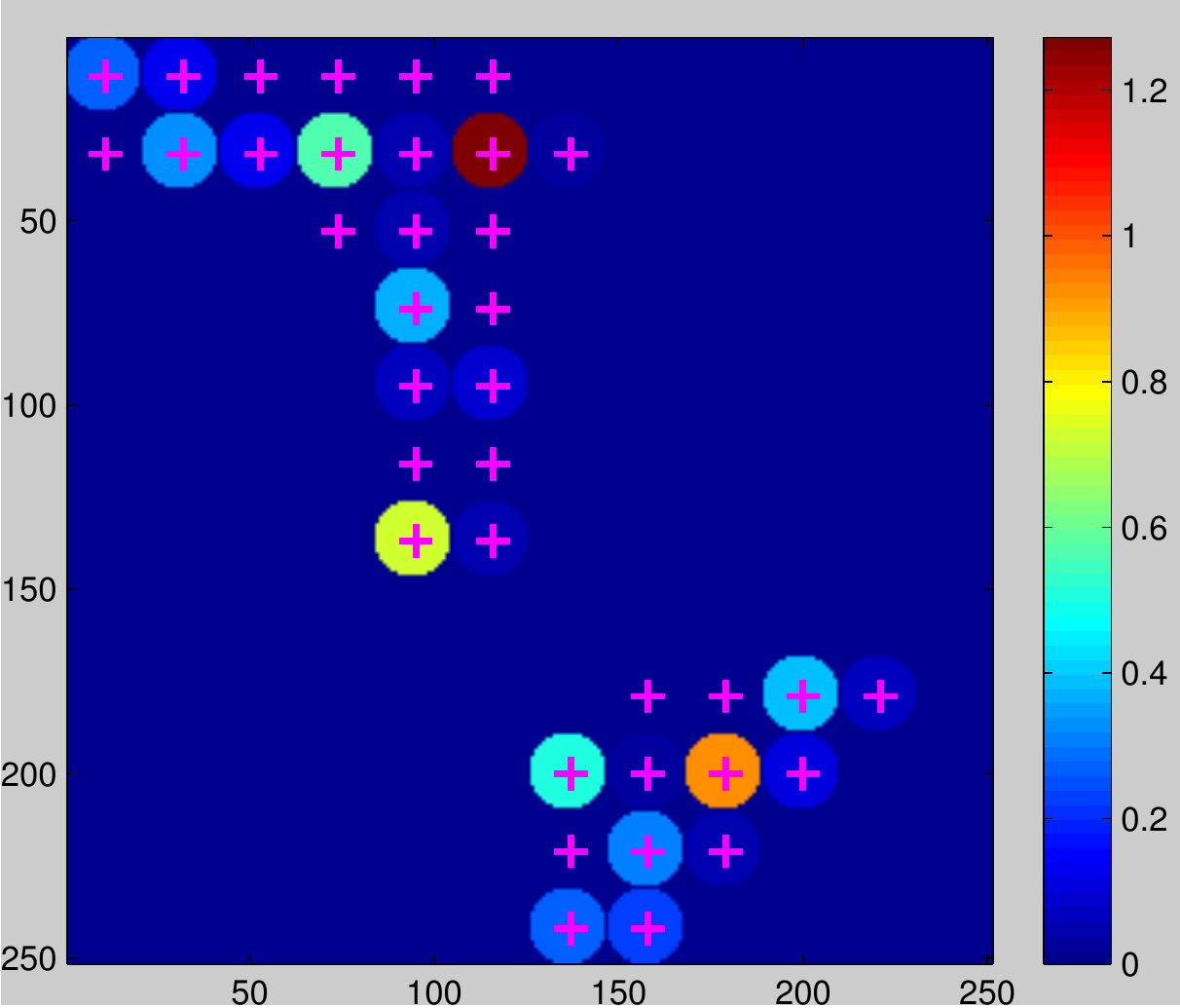}
  }
  \caption[Fienup's HIO method: ``random'' image results of
  reconstruction]{Fienup's HIO method: ``random'' image results of
    reconstruction: (a) as produced by the method, (b) after enforcing
    a constant value across every circle.}
  \label{fig:hio-rec-random}
\end{figure}

\begin{figure}[H]
  \centering
  \subfloat[]{\label{fig:hio-rec-sod-orig}
    \includegraphics[width=0.4\textwidth]{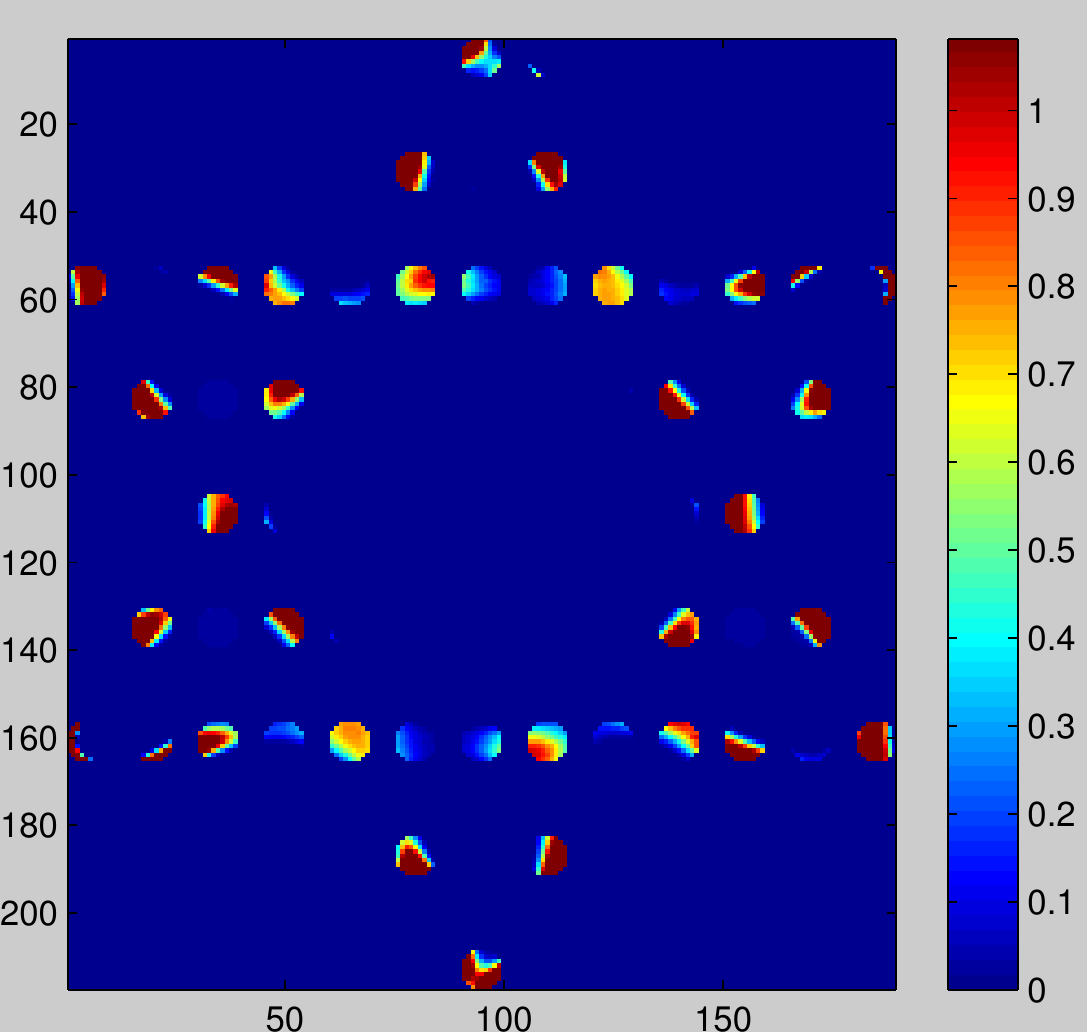}
  }
  \qquad{}
  \subfloat[]{\label{fig:hio-rec-sod-postproc}
    \includegraphics[width=0.4\textwidth]{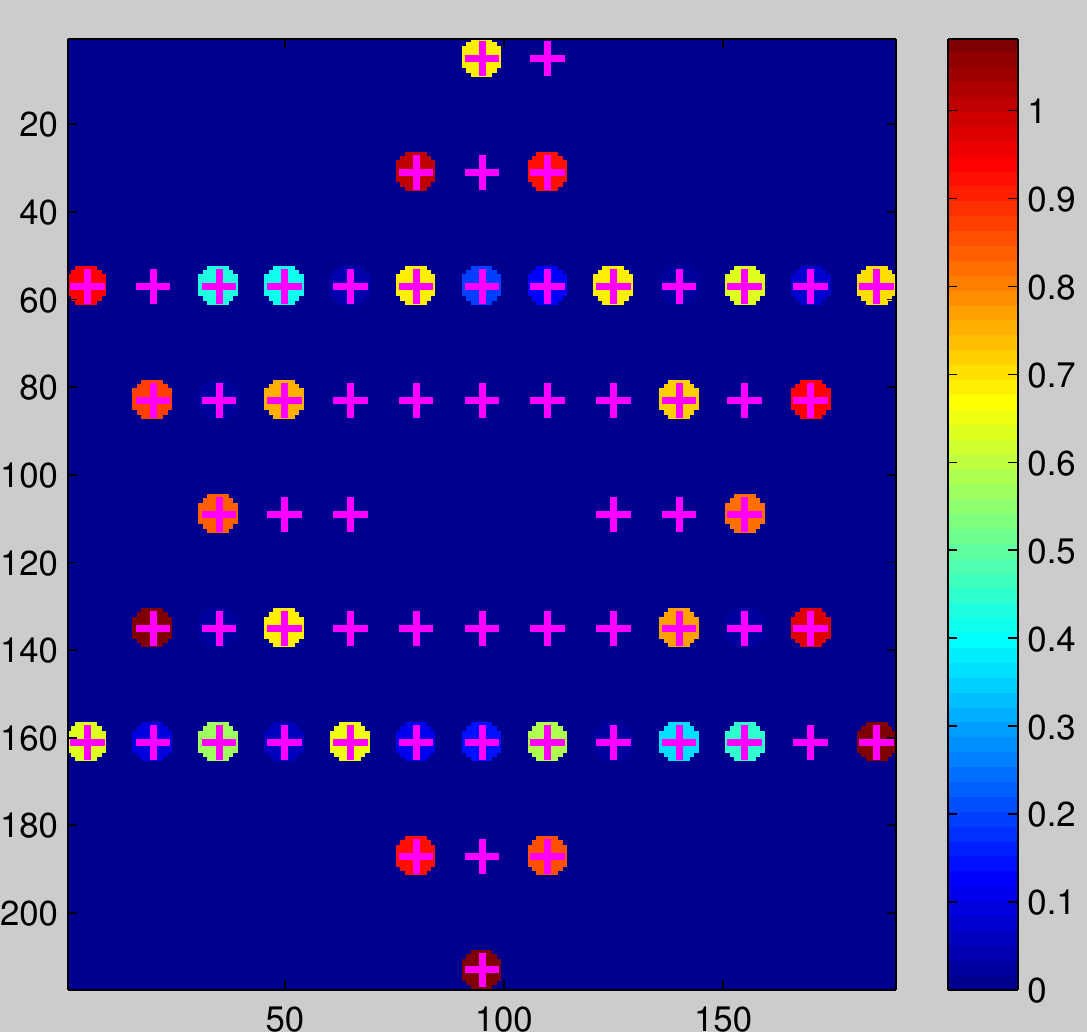}
  }
  \caption[Fienup's HIO method: SOD image results of reconstruction]{Fienup's HIO method: SOD image results of reconstruction:
    (a) as produced by the method, (b) after enforcing constant a
    value across every circle.}
  \label{fig:hio-rec-sod}
\end{figure}

\subsection{Relation to compressed sensing}
\label{sec:relat-compr-sens}
Compressed sensing (CS) is an emerging field in image processing that
performs signal reconstruction from a small number of its projections
\shortcite{donoho06compressed,candes06robust,candes06near-optimal}. Conceptually,
all CS techniques and their mathematical theory are based heavily upon
the sparsity of the sought signals. It is important to note that CS,
in its classical form, deals with measurements that are linear with
respect to the unknown signal. Likewise, CS techniques generally
assume random sampling distributed throughout the measurement
domain. By contrast, in our current case of sub-wavelength CDI the
measurements are: (1) nonlinear with respect to the sought signal, (2)
taken only in a small (low-frequency) region of the measurement
domain, where (3) they are taken in a periodic fashion (dictated by
the pixels' arrangement of a digital camera sensor). Still, our
reconstruction method relies on sparsity. As such, conceptually, our
approach can still be viewed as CS in a broader sense.

Clearly, for the reasons stated above, many theoretical results and
reconstruction methods of classical CS are not applicable to our
problem. For example, the Matching Pursuit (MP) method
\shortcite{mallat93matching} cannot be applied in its original
form. Another popular method — Basis Pursuit (BP) \shortcite{chen99atomic},
could, in principle, be applied here (considering BP as a general
approach based on replacing the $l_{0}$ with the $l_{1}$ norm, rather
than a specific algorithm). However, its benefits are not clear,
because, in contrast to the linear case, in our nonlinear
problem---using the $l_{1}$ norm still does not lead to a convex problem.

Besides the standard CS methods, which are inapplicable to the
sub-wavelength CDI problem, it is instructive to consider other
sparsity-based approaches which are related to CS, in the broader
sense.  One of these is based on division of the reconstruction
process into two stages: at the first stage the missing Fourier phase
is reconstructed using Fienup’s HIO algorithm (or Gerchberg-type
method); at the second stage this phase is combined with the measured
Fourier magnitude to form complete measurements that are linear with
respect to the unknown signal. Once these linear measurements are
available, one can use methods from classical CS (like, for example,
BP) or our previously proposed method NLHT 
\shortcite{gazit09super-resolution}, which is aimed at recovering data from
low-pass measurements. We find, however, that this approach
does not produce high quality results. This failure is, probably,
attributed to inability of the projection-based methods to reconstruct
the phase precisely, as shown in Figure~\ref{fig:fourier-phase-hio}
below.

\begin{figure}[H]
  \centering
  \subfloat[]{
    \includegraphics[width=0.4\textwidth]{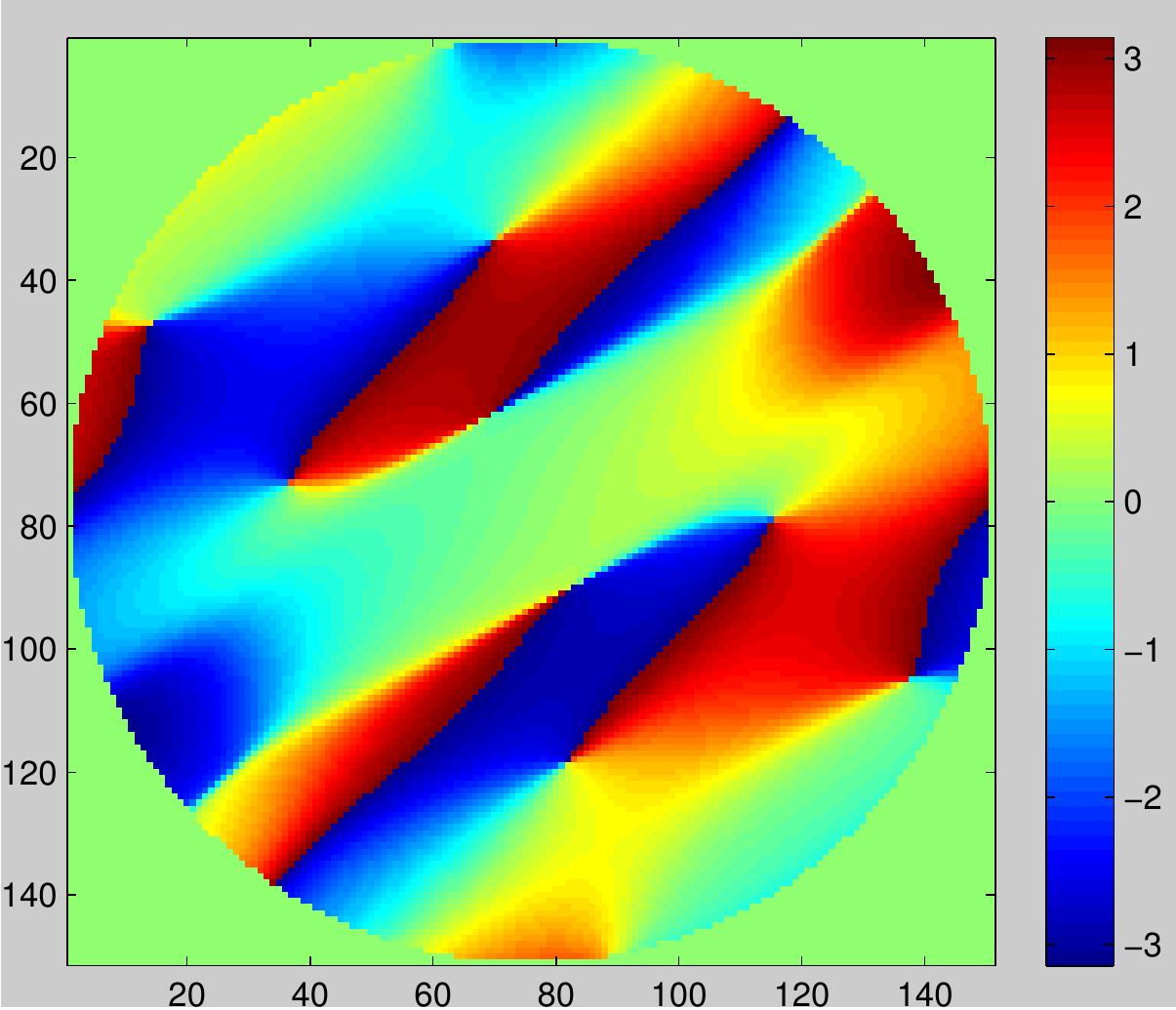}
  }
  \qquad{}
  \subfloat[]{
    \includegraphics[width=0.4\textwidth]{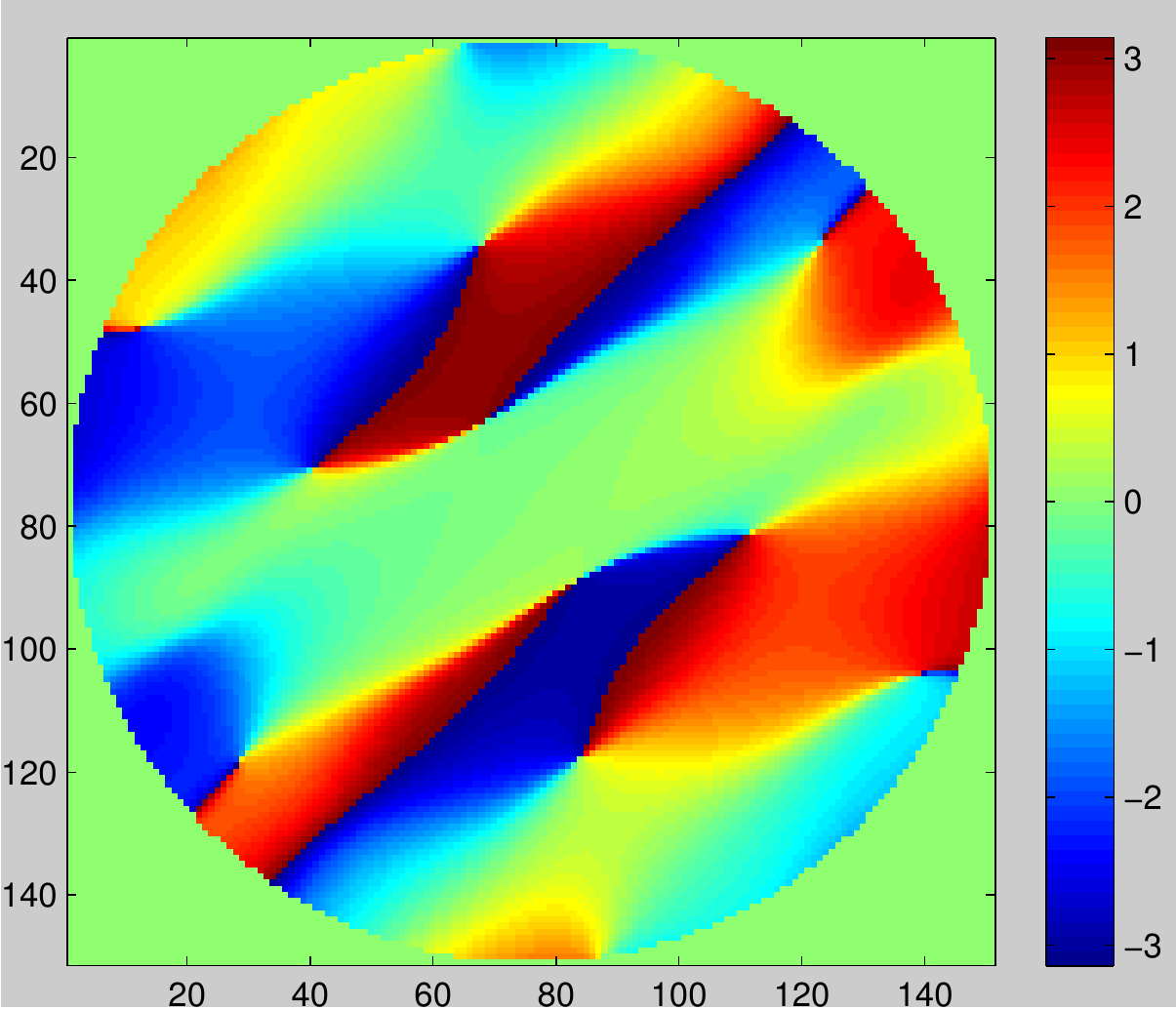}
  }
  \caption[Fourier phase of the ``random'' image]{Fourier phase of the ``random'' image: (a) the true
    phase, (b) the phase obtained after 5000 iterations of
    HIO.}
  \label{fig:fourier-phase-hio}
\end{figure}

Recently, several works have considered CS with quadratic nonlinear
measurements \shortcite{shechtman11sparsity,candes11phase}. In both papers
the resulting nonlinear constraints are relaxed to semidefinite
constraints using matrix lifting and an appropriate sparsity promoting
objective is used.  The work of \shortcite{candes11phase} considers phase
retrieval assuming the availability of several diffraction patterns
obtained from multiple structured illuminations, which is not relevant
to our problem.  In contrast, the scenario considered in
\shortcite{shechtman11sparsity} is much closer to our current case. Namely,
simultaneous phase retrieval and bandwidth extrapolation from a
single-shot power-spectrum measurement. In fact, our present problem
can be viewed as a special case of the problem addressed in
\shortcite{shechtman11sparsity}. However, the algorithm suggested in
\shortcite{shechtman11sparsity} is targeting a more general problem, hence
its computational complexity is high. With this reasoning in mind, we
devised the new sparsity-based approach and algorithm described in
the next section, which is tailored for the specific
problem of sub-wavelength CDI.

\section{A method for automatic grid determination, and the
  (un)importance of the basis function}
\label{sec:choosing-grid-basis}
In this section we would like to discuss the implications of our
assumption regarding the existence of a grid that, in fact, defines a
discrete set of allowed locations where the chosen basis function can
be placed.  In many cases, especially when the optical information
represents experimental data, introducing such a grid is highly
justified. For example, a digital image is obtained from a continuous
intensity distribution by sampling it with a sensor that physically is
an array of square pixels arranged in rows and columns. Hence,
naturally, the grid is rectangular and the basis functions are squares
whose size is equal to the grid's cell size. Likewise, our
reconstruction provides a digitized version of the true signal as if
it were performed by a sensor whose pixels' shape corresponds to the
chosen basis function (circular in our experiments above). Hence, the
grid used in our reconstruction algorithm essentially defines the
resolution of the reconstructed image. This is especially true when
the spatial extent of the basis function is smaller than (or equal to)
the grid's cell size. An example of such a sensor with circular pixels
is shown in Figure~\ref{fig:cirle-imposed}.

However, there is an important dissimilarity between our case and the
regular sampling in the object domain. Since our measurements contain
only the Fourier magnitude and no information is available about the
phase, we cannot distinguish between all the shifted versions of the
original signal. That is, if $E(u,v)$ represents the original signal,
our best hope is to reconstruct a shifted version of it, that is,
$E(u-\Delta u, v-\Delta v)$ for some $\Delta u$, and $\Delta v$.
Which version (shift) of the original signal is reconstructed depends, of
course, on the reconstruction method. Because our method seeks the
sparsest solution, we obtain the digitization that corresponds to the perfect
alignment shown in Figure~\ref{fig:circleImposed-aligned} and not the
``misaligned'' version shown in
Figure~\ref{fig:circleImposed-misaligned}. Because in the latter case
each circle in the original image ``switches on'' two pixels in the
sensor, in contrast to one pixel per circle in the aligned case. Hence,
one does not need  to manually align the grid with respect to the sought
signal as the best alignment is obtained automatically with our
reconstruction method. The only concern regarding the grid alignment
is related to the placement of the blurred image that we use for loose
support estimation. Fortunately, the solution to this problem is easy:
the blurred image must be placed in a way that guarantees maximal grid
coverage, that is, we shall keep as many allowed locations as
possible.

\begin{figure}[H]
  \centering
  \subfloat[]{\label{fig:circleImposed-aligned}
    \includegraphics{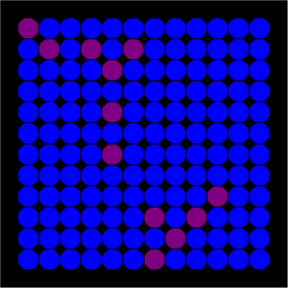}
  }\\
   \subfloat[]{\label{fig:circleImposed-misaligned}
    \includegraphics{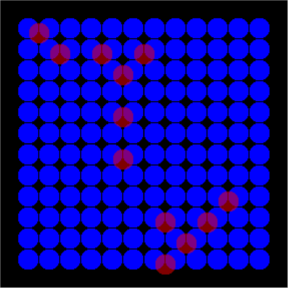}
  }
  \caption[The sought signal imposed on a sensor with circular
    pixels]{The sought signal (red) imposed on a sensor with circular
    pixels (blue). Note that the best alignment (a) is automatically
    obtained as it results is a sparser reconstruction than a bad
    alignment (b).}
  \label{fig:cirle-imposed}
\end{figure}

\subsection{The impact of the basis function}
\label{sec:impact-basis-funct}
Let us now consider the situations where the basis function is
chosen in a way that does not allow a perfect
reconstruction. Specifically, we consider basis functions in a shape
of a square and a triangle, as shown in
Figure~\ref{fig:square-triangle-imposed}.
\begin{figure}[H]
  \centering
   \subfloat[]{\label{fig:triangleImposed}
    \includegraphics[height=0.35\textheight]{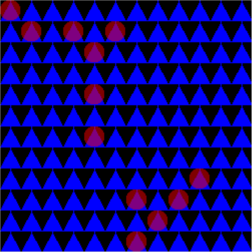}
  }\\
   \subfloat[]{\label{fig:squareImposed}
    \includegraphics[height=0.35\textheight]{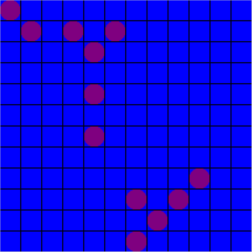}
  }
  \caption[Basis functions that do not allow a perfect
    reconstruction]{Basis functions that do not allow a perfect
    reconstruction: triangles (a), and squares (b).}
  \label{fig:square-triangle-imposed}
\end{figure}
As is evident from Figure~\ref{fig:square-triangle-reconstruction},
the reconstruction in these cases matches our expectations:
we obtain the correct ``digitized'' version of the sought signal that
corresponds to the chosen basis function and the grid. We emphasize the
fact that all experiments are done with actual data that contains a
significant amount of noise.
\begin{figure}[H]
  \centering
  \subfloat[]{\label{fig:triangleReconstruction}
    \includegraphics[height=0.35\textheight]{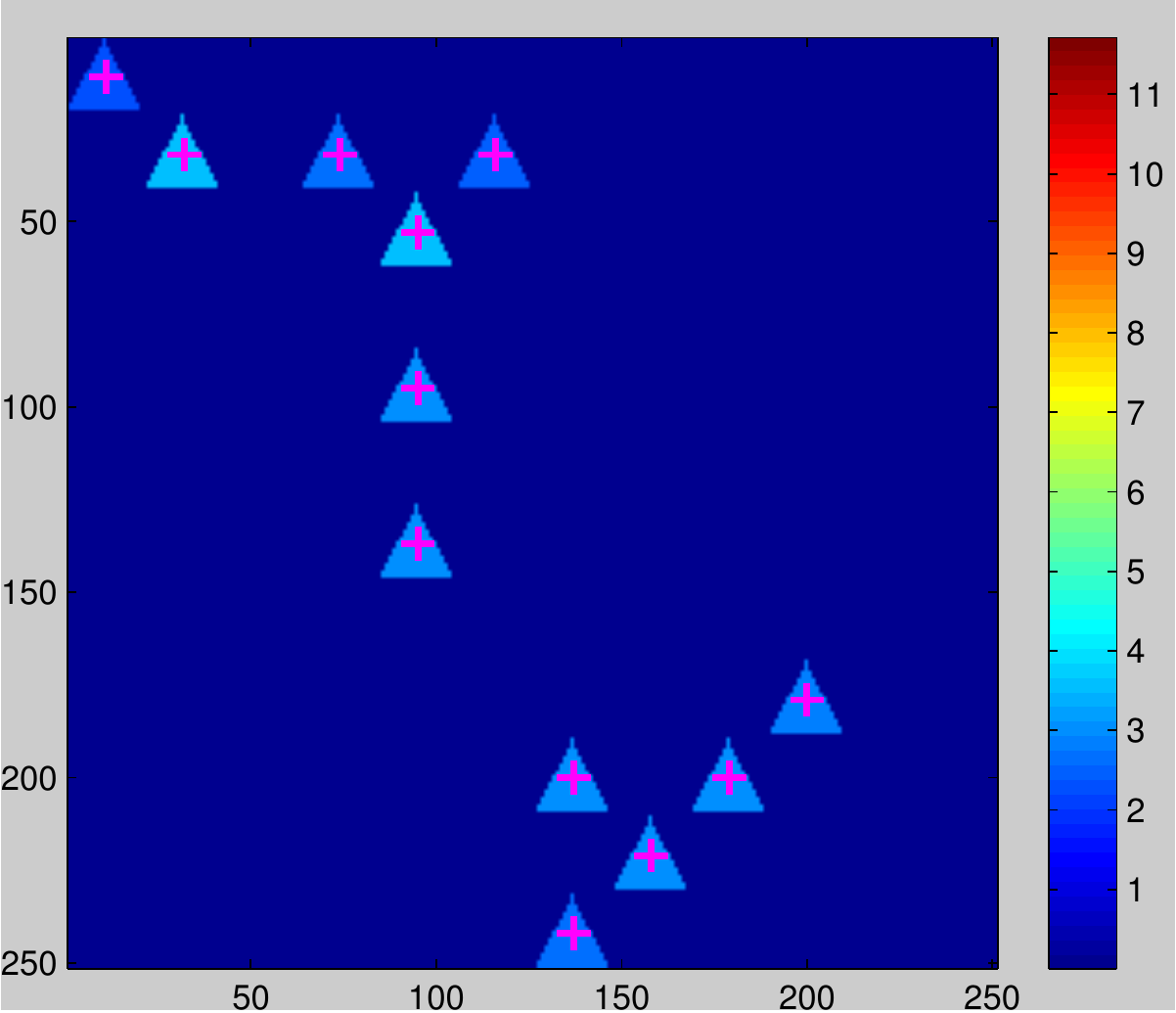}
  }\\
  \subfloat[]{\label{fig:squareReconstruction}
    \includegraphics[height=0.35\textheight]{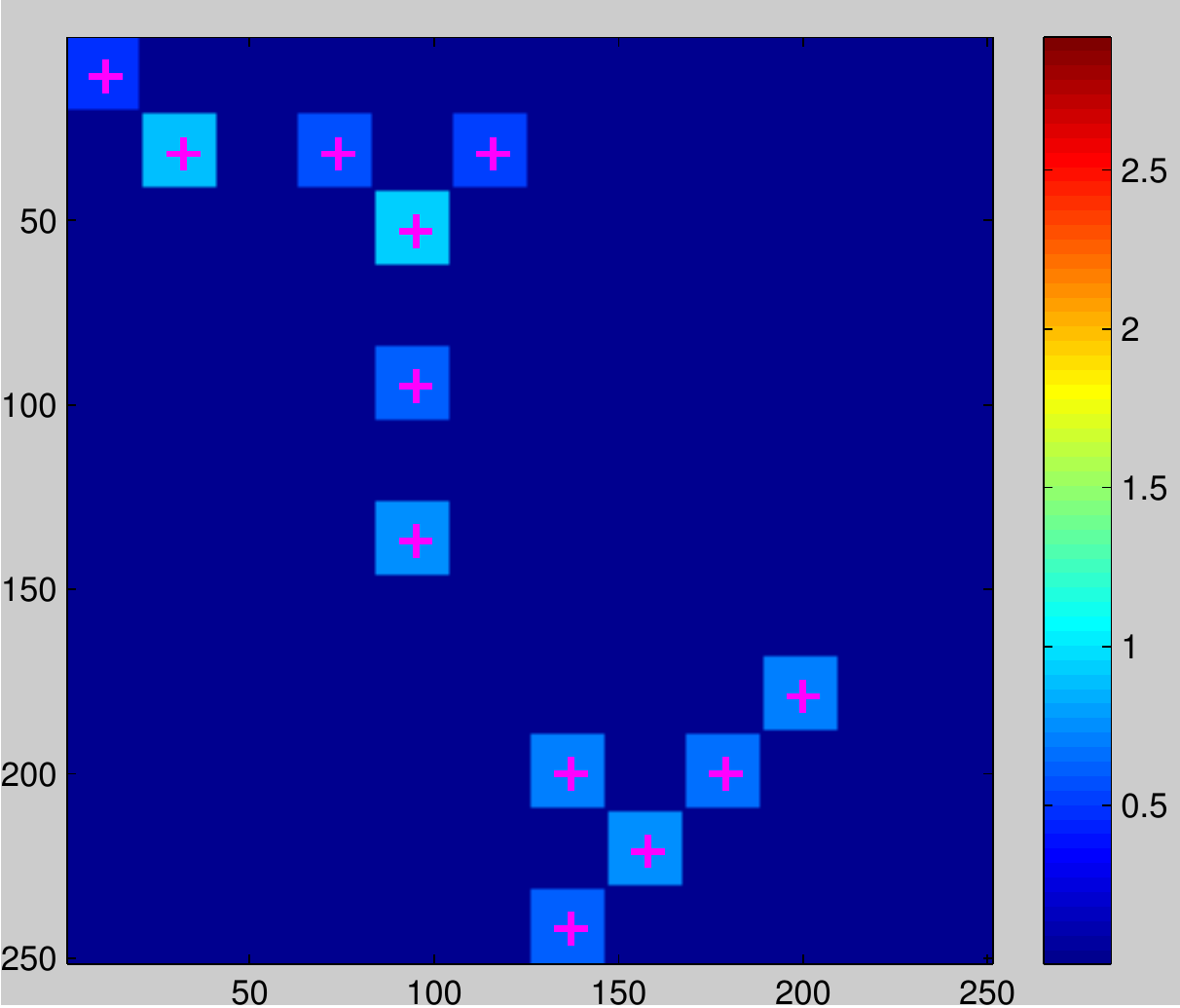}
  }
  \caption{Reconstruction in the case of basis functions that do not
    match the sought signal.}
  \label{fig:square-triangle-reconstruction}
\end{figure}
Moreover, if we consider the progress of the reconstruction process
(see Figure~\ref{fig:objective-function-different-bases}) we observe
that even an incorrect choice of the basis function has no adverse
effect on the reconstruction. This fact has a simple explanation: the
difference between a circle and a square (or a triangle) of size \unit[100]{nm}
is much smaller than \unit[100]{nm}. Hence, being able to distinguish between
these shapes would mean effective resolution that is much better than
\unit[100]{nm}. Thus, we conclude that the shape of the basis function is not of
great importance so long as its size matches the size of a typical
feature in the sought signal. In what follows, we evaluate the
possibility to discover the most appropriate grid pitch (basis
function size) automatically, without any prior information.
\begin{figure}[H]
  \centering
  \includegraphics[width=0.9\textwidth]{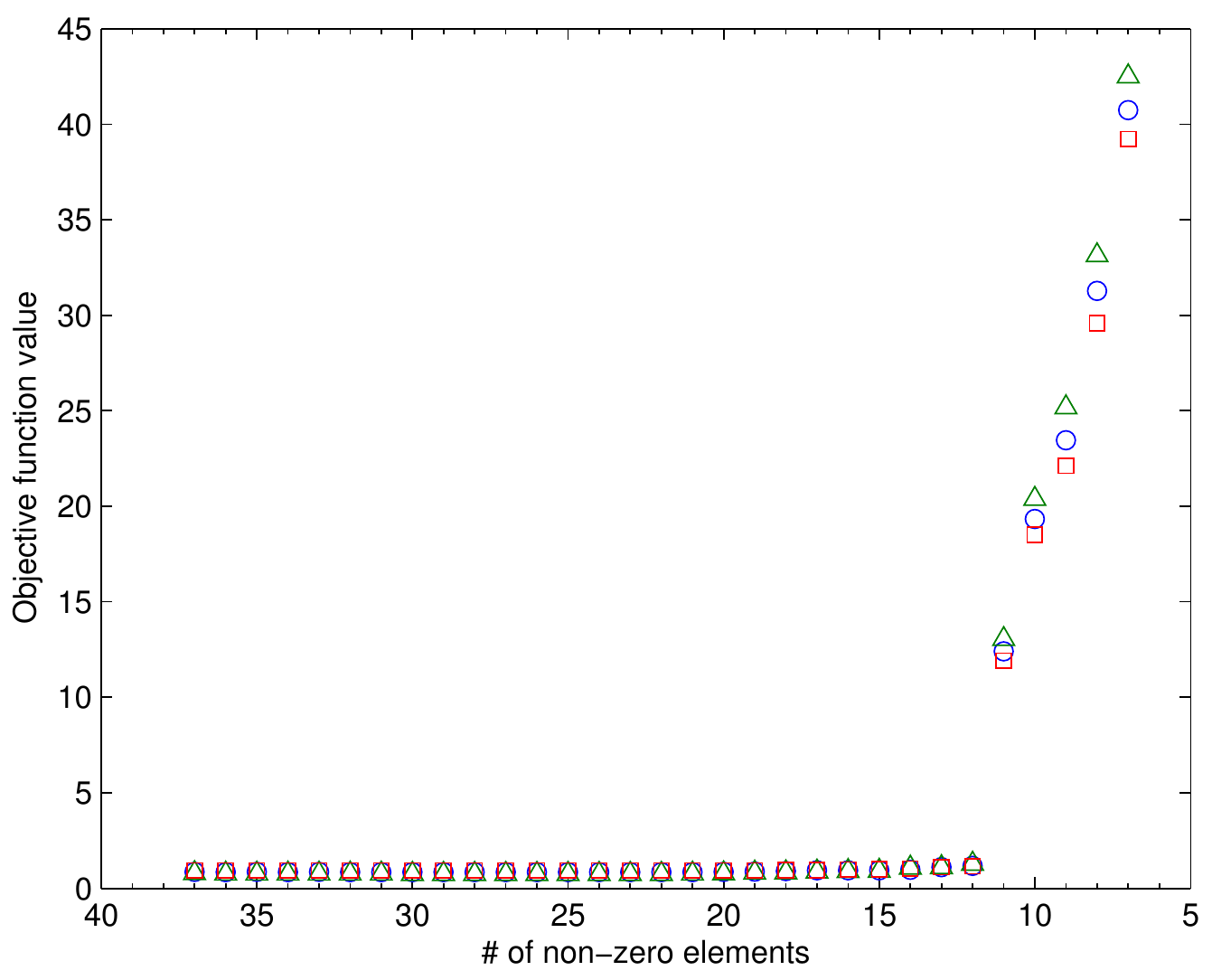}
  \caption[Reconstruction of the ``random'' image using different
    basis functions]{Reconstruction of the ``random'' image using different
    basis functions: objective function value (Fourier domain
    discrepancy) versus the number of circles/squares/triangles in the
    solution. Marker shape corresponds to the basis function shape.}
  \label{fig:objective-function-different-bases}
\end{figure}

\subsection{A method for automatic determination of an optimal grid}
\label{sec:meth-autom-grid}
So far, we have seen that the shape of the basis function has no
severe impact on the reconstruction process. Moreover, the best
possible alignment is obtained automatically due to our requirement of
maximal sparsity. These two properties can be used for automatic
determination of the optimal grid pitch. To this end we ran a series
of experiments with different grids whose pitch varies from 10 to 32
pixels (corresponds to the range of \unit[48--152]{nm}) using the square basis
function of the size that matches the grid cell. As was mentioned
earlier, the results of Section~\ref{sec:impact-basis-funct} show that
the particular choice of the basis function is not very
important. Hence, we could choose any shape of the size equal to the
grid pitch. The choice of the square basis function was stipulated by
the fact that most digital images are comprised of square
pixels. Hence, this basis function will, probably, be the first choice
in the situation where nothing is known about the sought signal. For
each grid pitch we ran a few iterations of our method keeping the
lowest discrepancy in the Fourier space as a numerical value that
corresponds to the current grid pitch . There is no need to solve the
problem completely, as our goal here is to see whether the sought
signal can be represented well by the current grid. We expect that
fine grids (small pitch) will represent well the sought signal so long
as the grid's pitch is smaller than or equal to the size of a typical
feature in the signal. However, once the grid becomes too coarse, we
expect a rapid growth of the objective function value. Hence, we
expect the graph to have the distinctive
``\rotatebox[origin=c]{90}{L}\,''-shape, similar to the graphs in
Figures~\ref{fig:objective-function-different-bases},
\ref{fig:sodimg-objfunction}, and \ref{fig:randomimg-objfunction}. As
is evident from Figure~\ref{fig:different-grids}, our expectations are
confirmed by the experimental results.
\begin{figure}[H]
  \centering
  \includegraphics[width=0.9\textwidth]{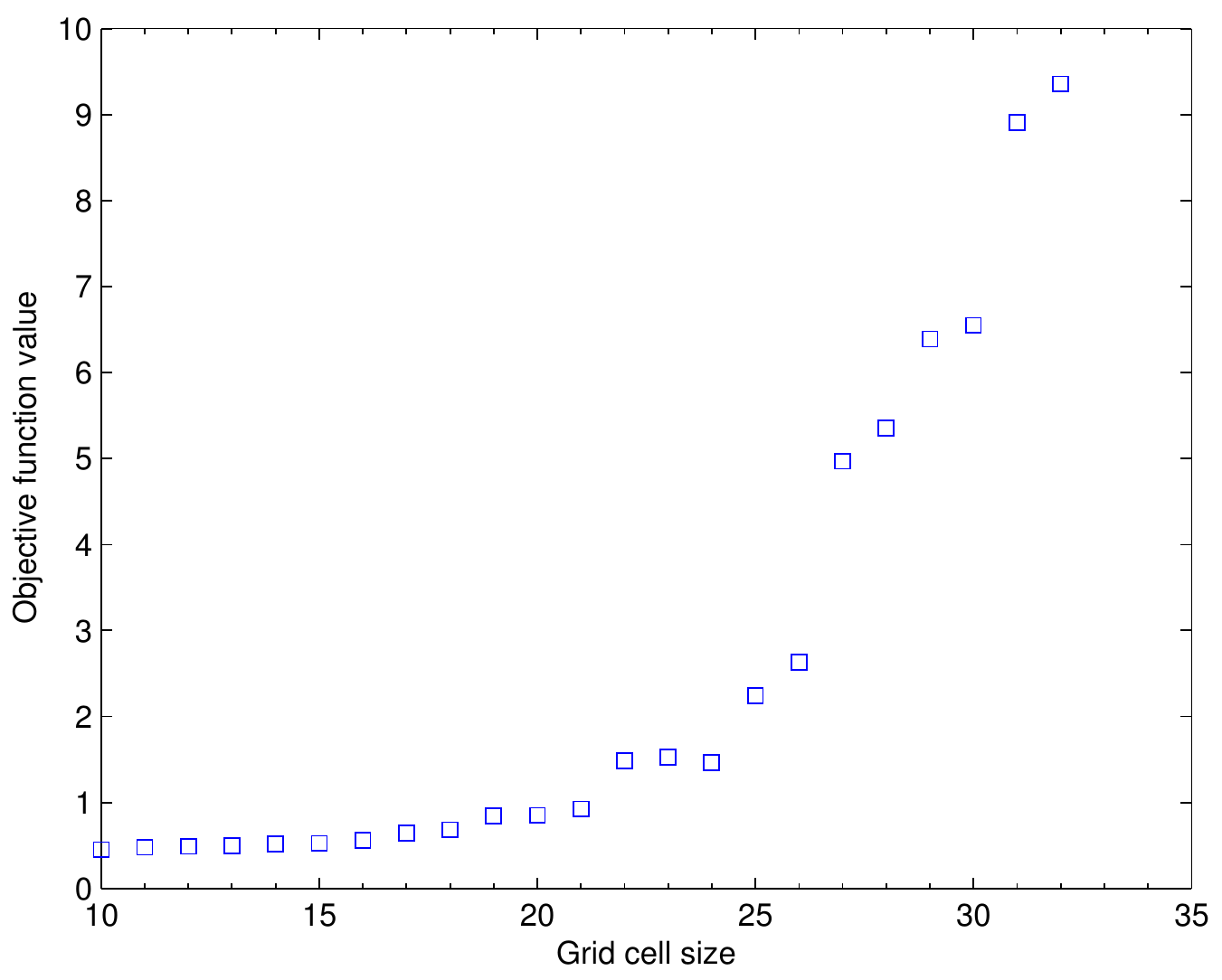}
  \caption[Objective function value versus grid pitch size]{Objective
    function value (Fourier domain discrepancy) versus grid pitch
    size.}
  \label{fig:different-grids}
\end{figure}
Note that the first sharp jump in the objective function value happens
during the transition from 21 pixels (the correct value) to 22
pixels. However, it may be argued that the transition is not
sufficiently apparent and the true value may lie in some small
interval around 21 pixels. Hence, we evaluate the behavior of our
reconstruction method for the grid pitch lying in the interval of
18--24 pixels. As is evident from
Figure~\ref{fig:objectiv-function-gridptich}, only the correct value
of 21 pixels results in a clear and sharp jump after we dip below
the correct value of squares (12). This property can be used for
pinpointing the correct pitch size. Hence, an automatic subroutine for
the optimal pitch determination is comprised of two steps: first, run
a few iterations of our reconstruction method to obtain quantitative
results indicating how well different grid sizes can represent the
sought image; second, run a full reconstruction procedure for a
limited range of pitches near the elbow in
Figure~\ref{fig:different-grids} and check what pitch results in a
clear evidence of existence of  the sparsest solution (as in
Figure~\ref{fig:objectiv-function-gridptich}).

Note that the obtained grid cell size is \emph{optimal}
in the sense that it satisfies two important properties
simultaneously: first, it allows good approximation of the sought
signal; second, it leads to a highly evident sparse solution.

The suggested method is also based on the sparsity assumption: it
works well when there are a few features in the sought signal are of
approximately the same size. This situation arises in many physical
setups. However, we currently are working on extending the algorithm
to cases where the signal features may be of various sizes.

\begin{figure}[H]
  \centering
  \subfloat[]{
    \includegraphics[height=0.19\textheight]{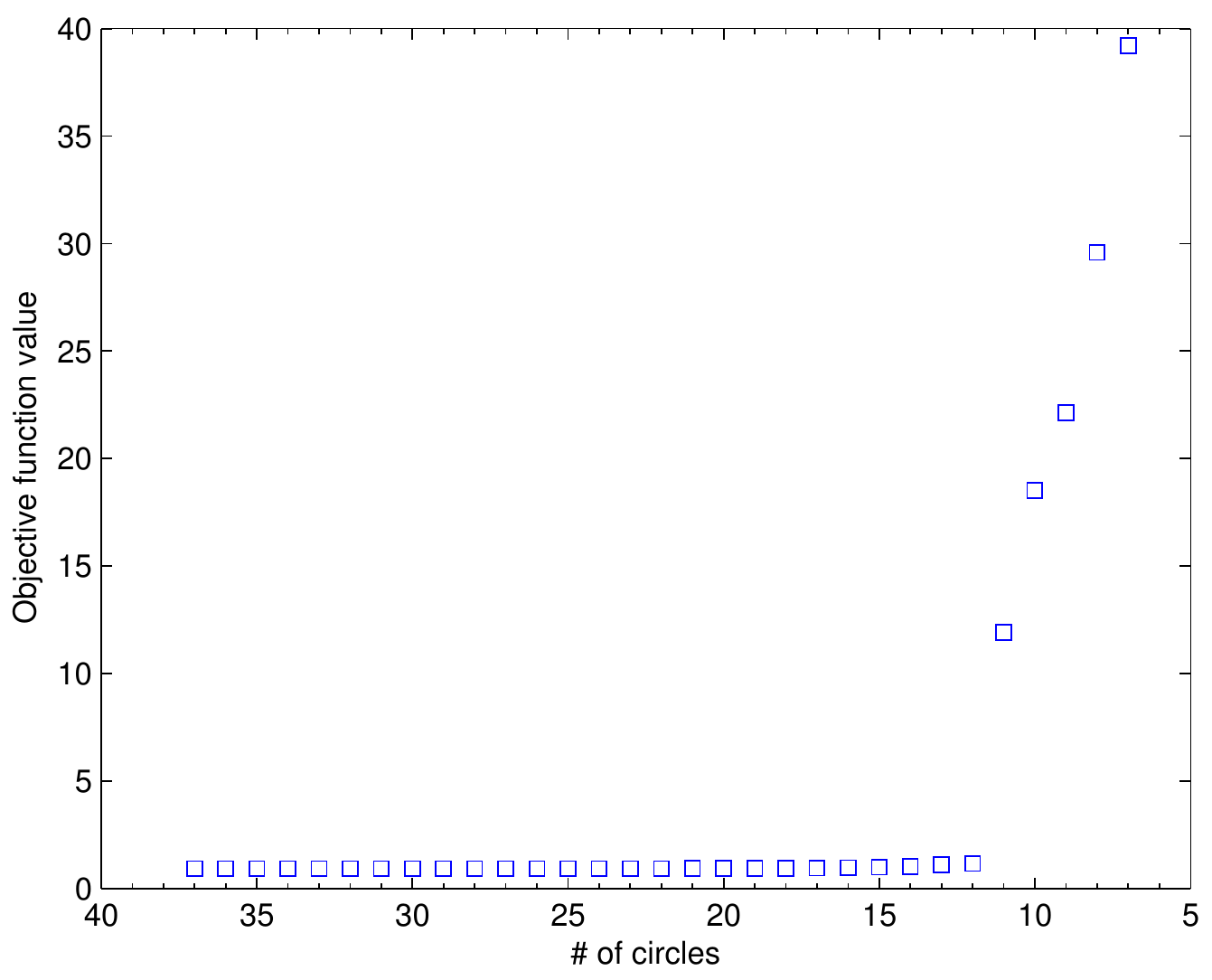}
  }\\
  \subfloat[]{
    \includegraphics[height=0.19\textheight]{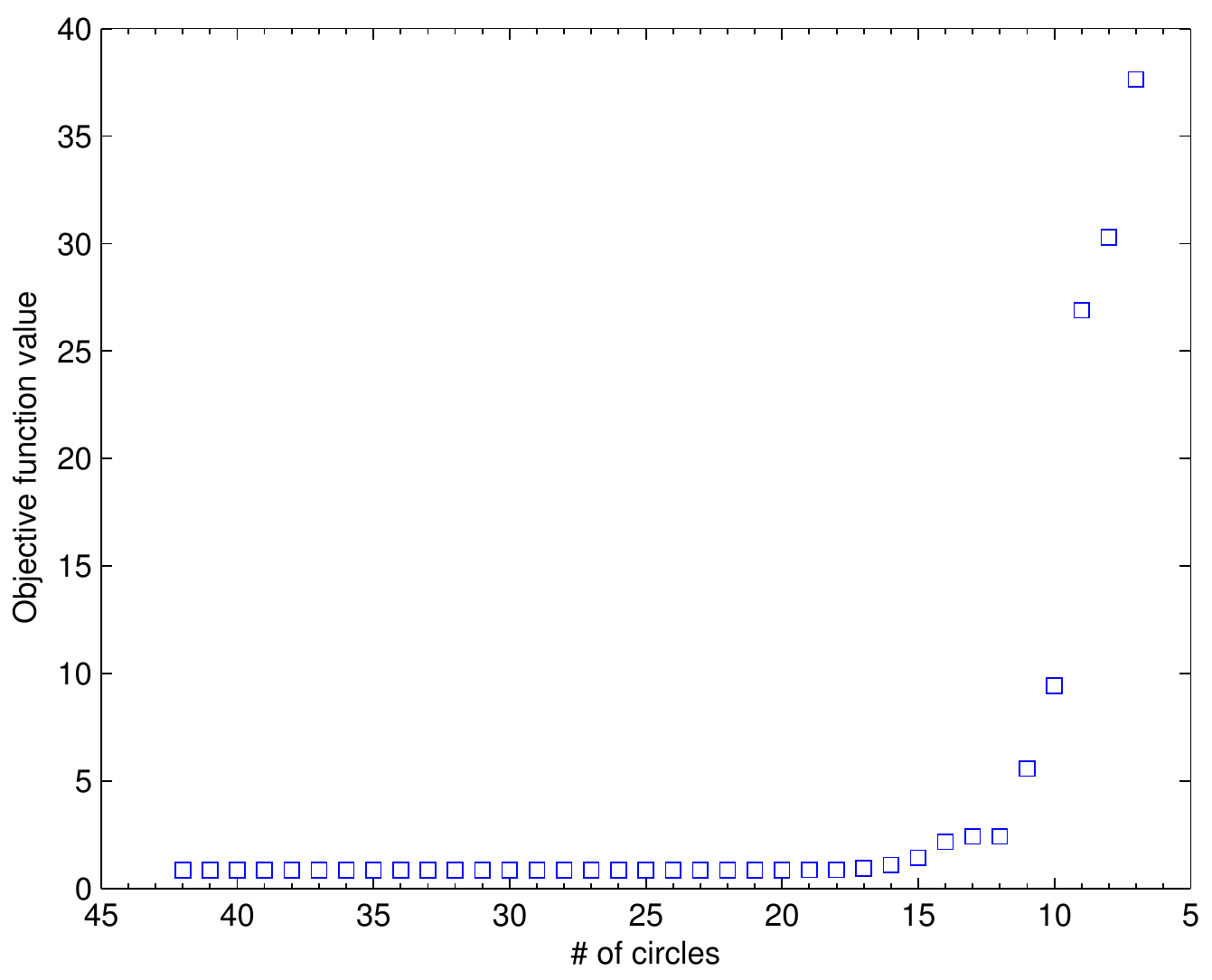}
  }
  \subfloat[]{
    \includegraphics[height=0.19\textheight]{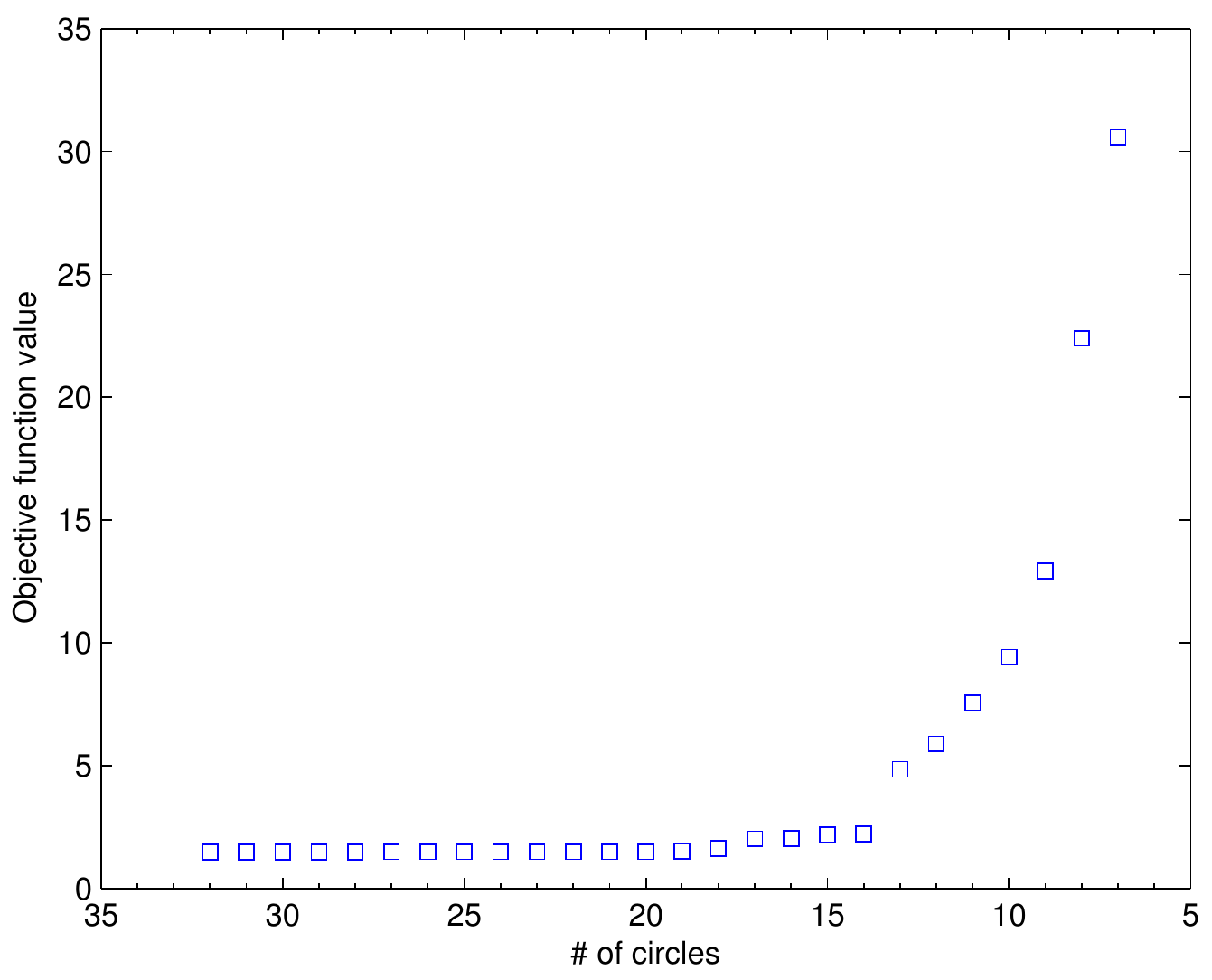}
  }\\
  \subfloat[]{
    \includegraphics[height=0.19\textheight]{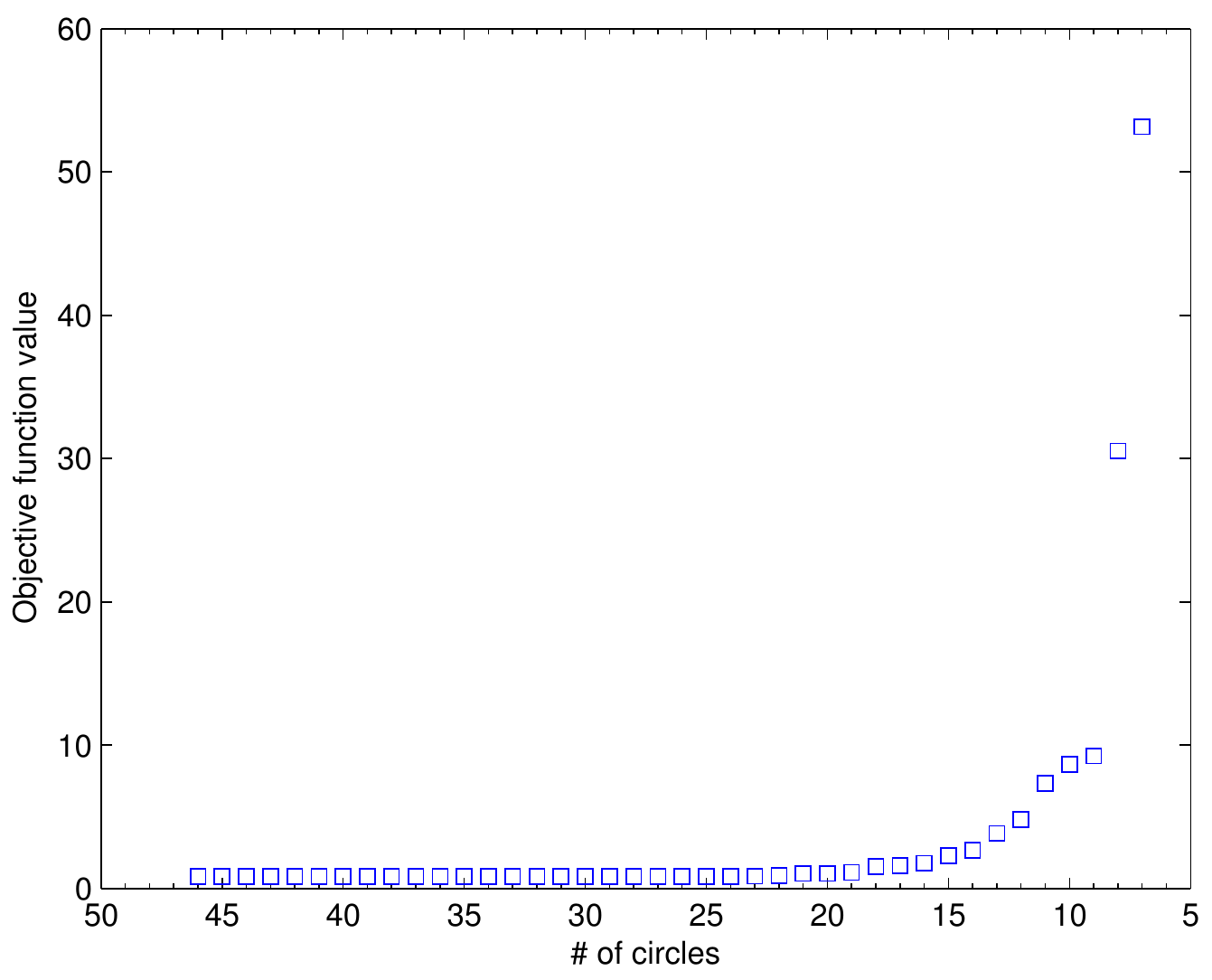}
  }
  \subfloat[]{
    \includegraphics[height=0.19\textheight]{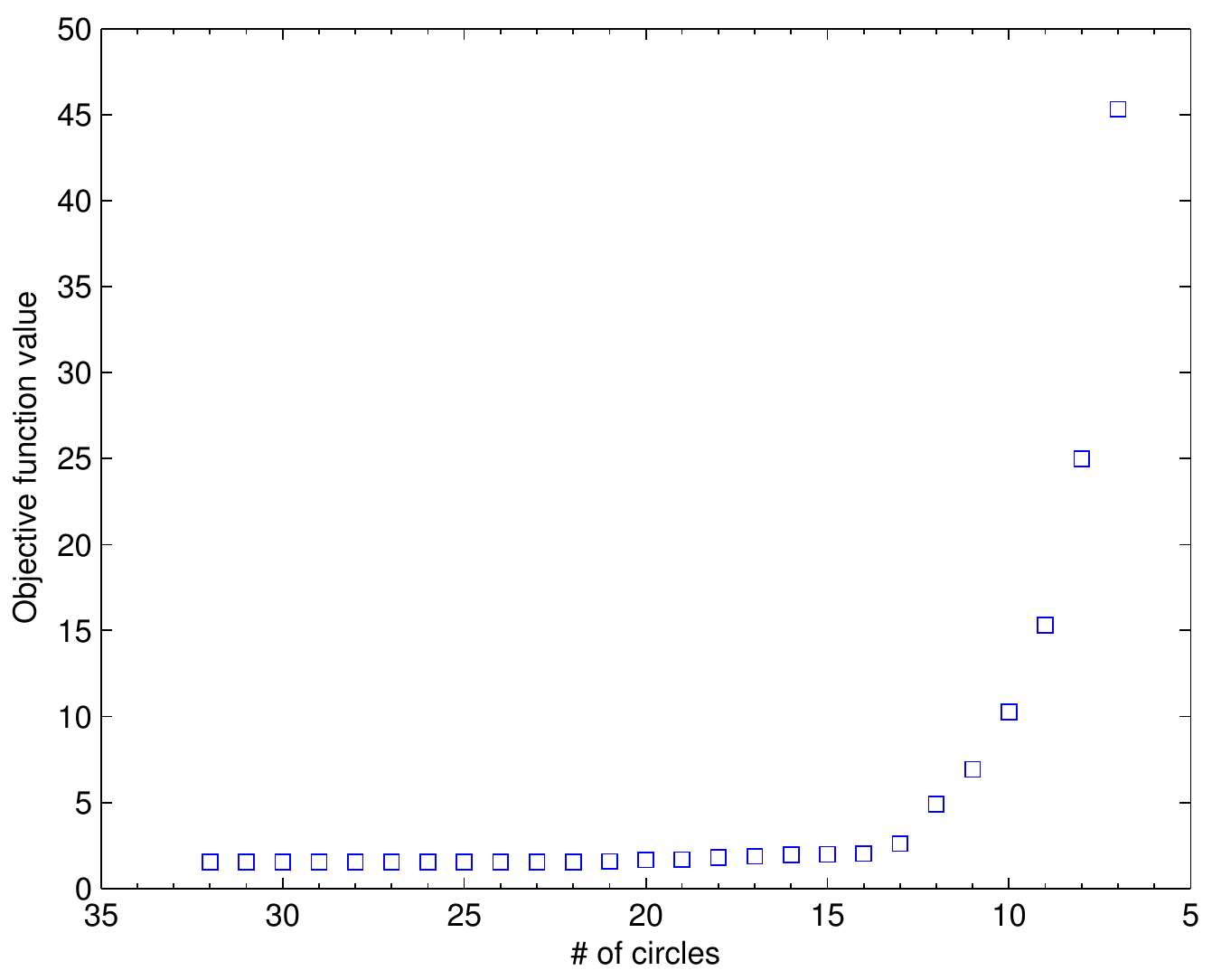}
  }\\
  \subfloat[]{
    \includegraphics[height=0.19\textheight]{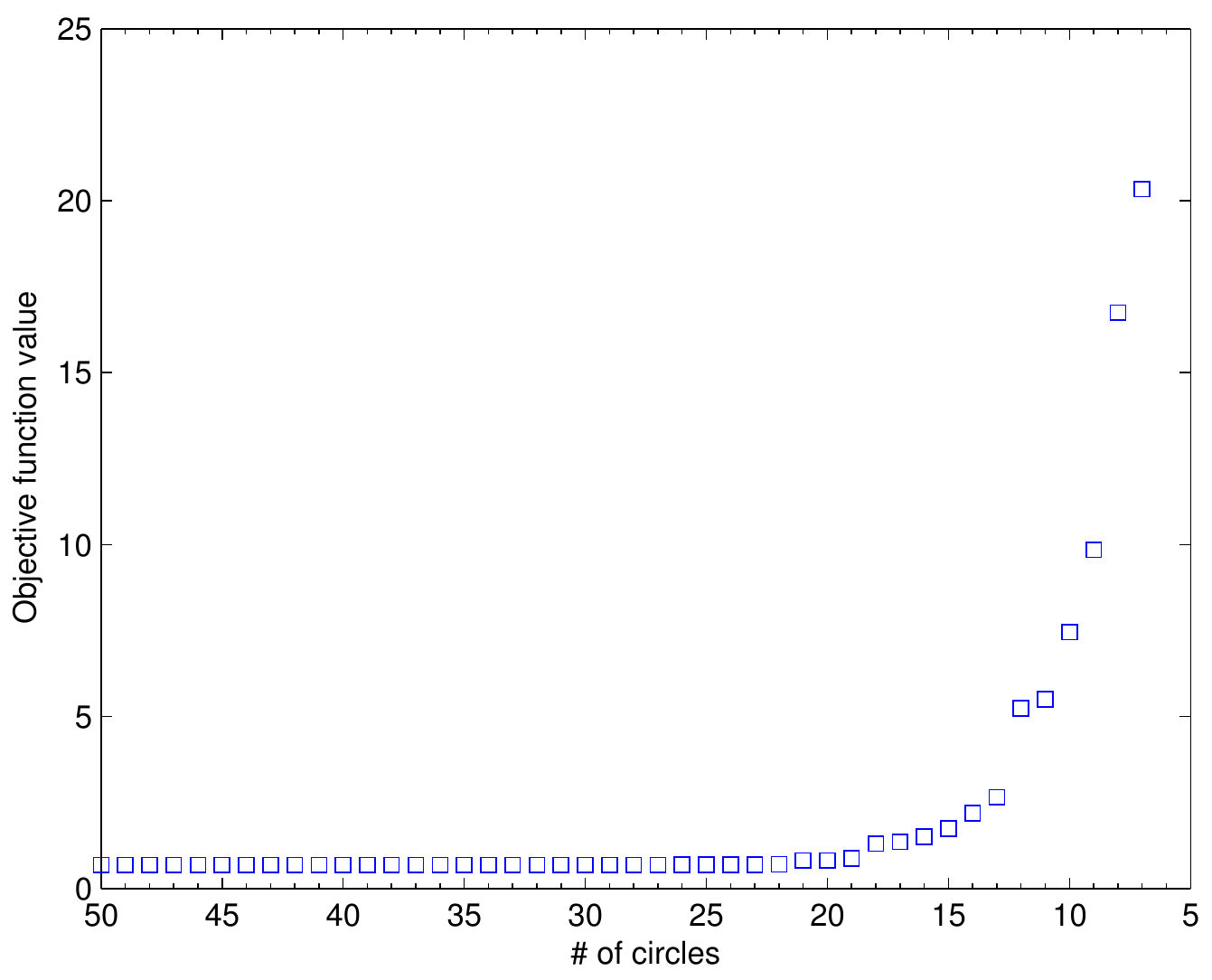}
  }
  \subfloat[]{
    \includegraphics[height=0.19\textheight]{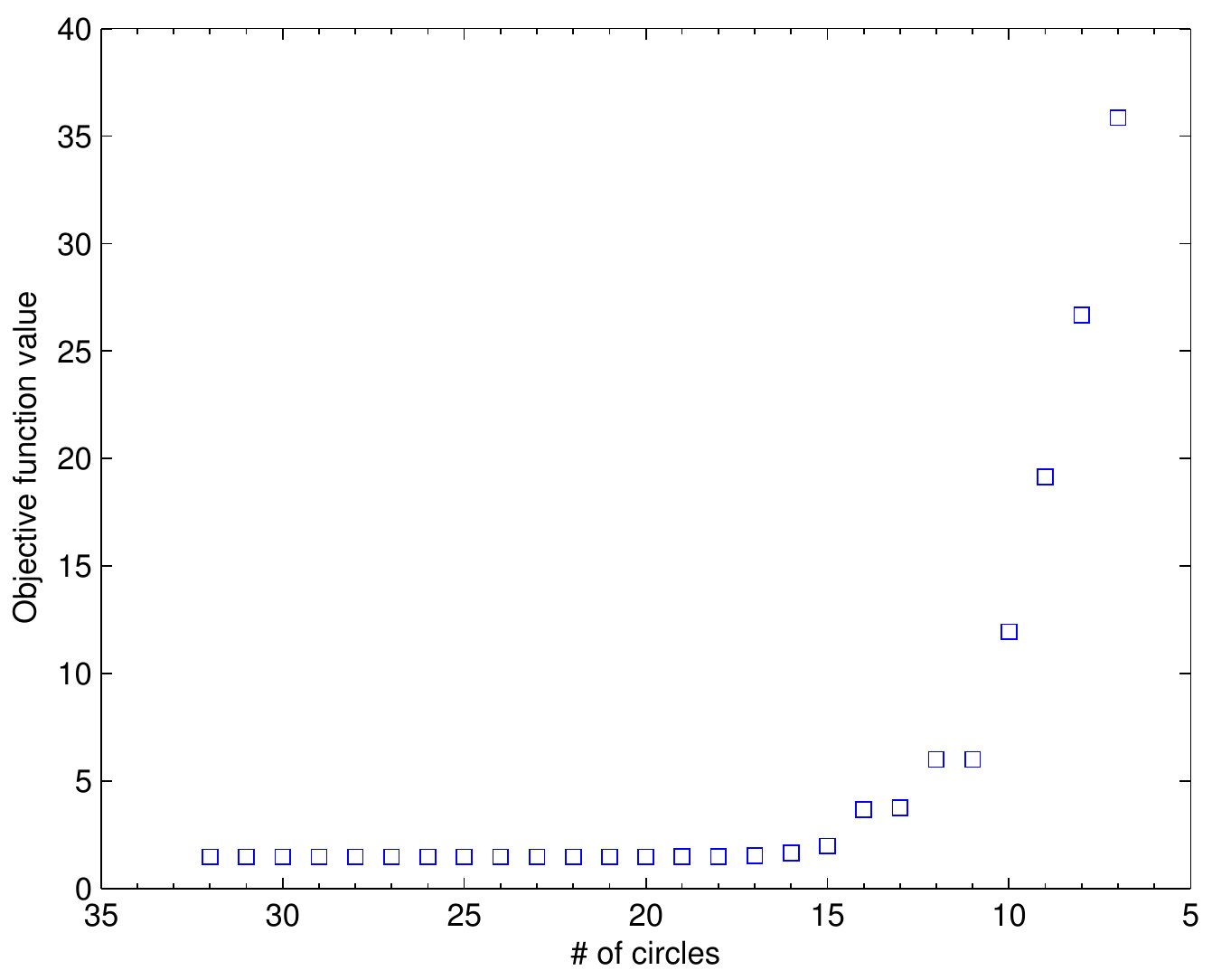}
  }
  \caption[Objective function behavior versus the grid pitch size]{Objective function behavior versus the grid pitch size
    (in pixels):
    (a) 21---the correct value, (b) 20, (d) 19, (f) 18, (c) 22, (e)
    23, (g) 24.}
  \label{fig:objectiv-function-gridptich}
\end{figure}

\section{Concluding remarks}
\label{sec:sparse-concluding-remarks}
In this chapter, we presented a technique facilitating reconstruction
of sub-wavelength features, along with phase-retrieval at the
sub-wavelength scale, at an unprecedented resolution for single-shot
experiments. That is, we have taken coherent lensless imaging into the
sub-wavelength scale, and demonstrated sub-wavelength CDI from
intensity measurements only. The method relies on prior knowledge
that the sample is sparse in a known basis (circles on a grid, in our
examples). We emphasize that sparsity is what makes our
phase retrieval work---the other assumptions used in the algorithm
(non-negativity, bounded support and the known basis) alone are not
sufficient. It is important to note that most natural and artificial
objects are sparse, in some basis. The information does not
necessarily have to be sparse in real space---it can be sparse in any
mathematical basis whose relation to the measurement basis is known,
for example, the wavelet basis or the gradient of the field intensity, given
that this basis is sufficiently uncorrelated with the measurements. In
all these cases our technique can provide a major improvement by
``looking beyond the resolution limit'' in a single-shot
experiment. Since our approach is purely algorithmic, it can be
applied to every optical microscope and imaging system as a simple
computerized image processing tool, delivering results in real time
with practically no additional hardware. The fact that our technique
works in a single-shot holds the promise for ultrafast sub-wavelength
imaging---one could capture a series of ultrafast blurred images, and
then off-line processing will reveal their sub-wavelength features,
which could vary from one frame to the next.  Finally, we note that
our technique is general, and can be extended also to other,
non-optical, microscopes, such as atomic force microscope,
scanning-tunnelling microscope, magnetic microscopes, and other
imaging systems. We believe that the microscopy technique presented
here holds the promise to revolutionize the world of microscopy with
just minor adjustments to current technology---sparse sub-wavelength
images could be recovered by making efficient use of their available
degrees of freedom. Last but not least, we emphasize that our approach
is more general than the particular subject of optical sub-wavelength
imaging. It is in fact a universal scheme for recovering information
beyond the cut-off of the response function of a general system,
relying only on a priori knowledge that the information is sparse in a
known basis.  Our preliminary theoretical and experimental results
indicate, unequivocally, that our method offers an improvement by
orders of magnitude beyond the most sophisticated deconvolution
methods. In a similar vein, we believe that our method can be applied
for spectral analysis, offering a means to recover the fine details of
atomic lines, as long as they are sparse (that is, do not form bands). In
principle, the ideas described here can be generalized to any
sensing/detection/data acquisition schemes, provided only that the 
information is sparse in a known basis, and that the measurements are
taken in a basis sufficiently uncorrelated to it.

%%% Local Variables: 
%%% mode: latex
%%% TeX-master: "../thesis"
%%% End: 

%% file: afterword/afterword.tex
\chapter{Afterword}
\label{cha:afterword}

An alternative title for this thesis could be ``Prior information in
the phase retrieval problem''. This also can describe the main line of
the work presented here. In fact we discovered and showed how to use
two powerful priors: approximately known Fourier phase, and sparsity
of the sought signal. The material in each chapter represents
an idea and the main results in a succinct form that is suitable for 
publication in a scientific journal. Therefore, some of the
results were not included in this thesis. Part of them is available as
technical reports listed below
\begin{itemize}
\item \shortcite{osherovich10algorithms}
\item \shortcite{osherovich10simultaneous}
\item \shortcite{osherovich09image}
\item \shortcite{osherovich08signal}
\end{itemize}

Other forms of prior knowledge, for example, defocused/blurred version of the
sought signal can be found in~\shortcite{osherovich10numerical}, which
summarizes our work done in  collaboration with KLA Tencor Inc. 

We should also mention a standalone work done
in~\shortcite{osherovich09designing}, where we used Fienup's HIO
algorithm to design an overcomplete dictionary in a way that its atoms
(columns) are maximally uncorrelated. The method produces excellent
results that are significantly better than the results produced by
currently used algorithms.

%%% Local Variables: 
%%% mode: latex
%%% TeX-master: "../thesis"
%%% End: 